\documentclass[authoryear,review]{elsarticle}
\usepackage{appendix}
\usepackage{setspace}
\usepackage[font=small,  margin=2cm]{caption}
\usepackage[tbtags]{amsmath}
\usepackage{amsthm,amssymb,amsfonts}
\usepackage{natbib}
\usepackage{lscape}
\usepackage{subfigure}
\usepackage{makecell}
\usepackage{exscale}
\usepackage{booktabs}
\usepackage{threeparttable}
\usepackage{array}
\usepackage{fullpage}
\usepackage{url}
\usepackage{algorithm}
\usepackage{verbatim}
\usepackage{bm}
\usepackage{smile}
\usepackage{mathtools}
\usepackage{wrapfig}
\usepackage{lipsum}
\usepackage{mathrsfs}
\usepackage{relsize}
\usepackage{dsfont}
\usepackage{multirow}
\usepackage{apalike}
\usepackage{chngcntr}
\usepackage{graphicx}
\usepackage{algorithmic}
\usepackage[english]{babel}
\usepackage[usenames,dvipsnames,svgnames,table]{xcolor}
\usepackage[colorlinks=true,
linkcolor=blue,
urlcolor=blue,
citecolor=blue]{hyperref}

\journal{}
\newtheorem{assumption}{Assumption}

\newcommand{\beq}{\begin{equation}}
\newcommand{\eeq}{\end{equation}}
\renewcommand{\l}{\left}
\renewcommand{\r}{\right}

\numberwithin{equation}{section}
\numberwithin{theorem}{section}
\numberwithin{corollary}{section}
\counterwithout{asmp}{section}
\numberwithin{definition}{section}

\input{tcilatex}
\begin{document}

\begin{frontmatter}
\title{Robust Tensor Factor Analysis }
		\author[myfirstaddress]{Matteo Barigozzi}
		\address[myfirstaddress]{Department of Economics, Universita' di Bologna,
Italy}

			\author[mysecondaddress]{Yong He\corref{cor1}}
		\address[mysecondaddress]{Institute for Financial Studies, Shandong University, Jinan, 250100, China}
\cortext[cor1]{Corresponding author. All authors contributed equally to this work.}
		\ead{heyongsdu.edu.cn}
	\author[mythirdaddress]{Lingxiao Li}
		\address[mythirdaddress]{School of Mathematics, Shandong University, Jinan, 250100, China}
		
			\author[myfourthaddress]{Lorenzo Trapani}
		\address[myfourthaddress]{School of Economics, University of Leicester, UK}

\begin{abstract}

We consider (robust) inference in the context of a factor model for
tensor-valued sequences. We study the consistency of the estimated common
factors and loadings space when using estimators based on minimising
quadratic loss functions. Building on the observation that such loss
functions are adequate only if sufficiently many moments exist, we extend
our results to the case of heavy-tailed distributions by considering
estimators based on minimising the Huber loss function, which uses an $L_{1}$%
-norm weight on outliers. We show that such class of estimators is robust to
the presence of heavy tails, even when only the second moment of the data
exists. We also propose a modified
version of the eigenvalue-ratio principle to estimate the dimensions of the core tensor and show the consistency of the resultant estimators without any condition on the relative  rates of
divergence of the sample size and dimensions. Extensive numerical studies are conducted to show the advantages of the proposed methods over the state-of-the-art ones especially under the heavy-tailed cases. An import/export dataset of a variety of commodities across multiple countries is analyzed to show the practical usefulness of the proposed robust estimation procedure. An R package ``RTFA" implementing the proposed methods is available on R CRAN.
\end{abstract}
\begin{keyword}
Tensor data; Factor model; Heavy-tailed data;
Quadratic loss; Huber loss.
		\end{keyword}
	\end{frontmatter}


\section{Introduction\label{intro}}

In this paper, we study (robust) inference in the context of a factor model
for tensor-valued sequences, viz.
\begin{equation}
\mathcal{X}_{t}=\mathcal{F}_{t}\times _{1}\mathbf{A}_{1}\times _{2}\cdots
\times _{K}\mathbf{A}_{K}+\mathcal{E}_{t}=\mathcal{F}_{t}\times _{k=1}^{K}%
\mathbf{A}_{k}+\mathcal{E}_{t},  \label{tfm}
\end{equation}%
where $\mathcal{X}_{t}$ is a $p_{1}\times p_{2}\times \dots \times p_{K}$%
-dimensional tensor observed at $1\leq t\leq T$. In (\ref{tfm}), $\mathbf{A}%
_{k}$ is a $p_{k}\times r_{k}$ matrix of loadings with $1\leq k\leq K$, $%
\mathcal{F}_{t}$ is a common core tensor of dimensions $r_{1}\times \cdots
\times r_{K}$, and $\mathcal{E}_{t}$ is an idiosyncratic component of
dimensions $p_{1}\times p_{2}\times \dots \times p_{K}$.

The literature on inference for tensor-valued sequences has been growing
very rapidly in the past few years, and has now become one of the most
active research areas in statistics and machine learning. Datasets collected
in the form of high-order, multidimensional arrays arise in virtually all
applied sciences, including \textit{inter alia}: computer vision data (%
\citealp{panagakis2021tensor}); neuroimaging data (see e.g. %
\citealp{zhou2013tensor}, \citealp{Ji2021Brain} and %
\citealp{Chen2021Simultaneous}); macroeconomic indicators (%
\citealp{Chen2021Factor}) and financial data (\citealp{Han2021Rank}, %
\citealp{He2021Matrix}); recommender systems (see e.g. %
\citealp{entezari2021tensor}, and the various references therein); and data
arising in psychometrics (\citealp{carroll1970analysis}, %
\citealp{douglas1980candelinc}) and chemometrics (see %
\citealp{tomasi2005parafac}, and also the references in %
\citealp{acar2011scalable}). We also refer to the papers by %
\citet{kolda2009tensor} and \citet{bi2021tensors} for a review of further
applications. As \citet{lu2011survey} put it: \textquotedblleft Increasingly
large amount of multidimensional data are being generated on a daily basis
in many applications\textquotedblright ~(p. 1540). However, the large
dimensionality involved (and the huge computational power required) in the
analysis of tensor-valued data proves an important challenge, and several
contributions have been developed in order to parsimoniously model a tensor
dataset $\left\{ \mathcal{X}_{t},1\leq t\leq T\right\} $, and in order to
extract the signal it contains. A natural way of modelling tensors is through a low-dimensional projection on a space of common
factors (core tensor), such as the Tucker decomposition type factor model in
(\ref{tfm}). Indeed, factor models have proven very effective in the context
of vector-valued  series, and we refer to \citet{stock2002forecasting}, %
\citet{bai2003inferential} and \citet{fan2013large}, and also to subsequent
representative work on estimating the number of factors including, e.g., %
\citet{bai2002determining}, \citet{onatski2010determining}, %
\citet{ahnhorenstein13}, \citet{Trapani2018A} and \cite{yu2019robust}. In
recent years, the literature has extended the tools developed for the
analysis of vector-valued data to the context of matrix-valued data: in
particular, we refer to the seminal contribution by \cite{wang2019factor},
who proposed Matrix Factor Model (MFM) exploiting the double low-rank
structure of matrix-valued observations (see also \citealp{fan2021}, %
\citealp{Yu2021Projected} and \citealp{hkty}). In contrast with this
plethora of contributions, the statistical analysis of factor models for
tensor-valued data is still in its infancy, with few exceptions: we refer to %
\citet{Chen2021Factor}, \citet{han2020tensor} \citet{chen2022rank}, %
\citet{he2022statistical}, and \citet{zhang2022tucker} for the Tucker
decomposition, and \citet{Han2021CP} and \citet{chang2021modelling} for the
CP decomposition.

One restriction common to virtually all the contributions cited above is that
statistical inference is developed under the assumption of the data
admitting at least four moments. Whilst this is convenient when developing
the limiting theory, it can be argued that such an assumption is unlikely to
be realistic in various contexts: as mentioned above, for example, financial
data lend themselves naturally to being modelled as tensor-valued series,
but the existence of high-order moments in financial data, particularly at
high frequency, is difficult to justify (see e.g. \citealp{cont2001empirical}%
, and \citealp{degiannakis2021superkurtosis}). Similarly, heavy tails are
encountered in income data (\citealp{sarpy}), macroeconomics (%
\citealp{ibragimov2018heavy}), urban studies (\citealp{gabaix1999zipf}), and
in insurance, telecommunication network traffic and meteorology (see e.g.
the book by \citealp{embrechts2013modelling}). Contributions that develop
inference in the context for factor models with heavy tails are rare: the
main attempts to relax moment restrictions are \citet{FLW18} and %
\citet{he2022large} (see also \citealp{yu2019robust} and \citealp{bct} for
estimating the number of factors). However, both \citet{FLW18} and %
\citet{he2022large} impose a shape constraint on the data - namely, that
they follow an elliptical distribution - which, albeit natural, may not be
satisfied by all datasets; an alternative to the use of a specific
distributional shape, one could consider the use of estimators based on $%
L_{1}$-norm loss functions such as the Huber loss function (see the original
paper by \citealp{huber1992robust}, and also \citealp{he2023huber} for an
application to vector-valued  series). In the high-dimensional
regression setting, \cite{zhou2018new} derive a nonasymptotic concentration
results for an adaptive Huber estimator with the tuning parameter adapted to
sample size, dimension, and the variance of the noise, see also \cite%
{sun2020adaptive} and \cite{wang2021new}.

\subsection{Contributions and paper organization}

In this paper, we propose an advance in both areas mentioned above through
two contributions. Firstly, we study estimation of the loadings $\left\{
\mathbf{A}_{k},1\leq k\leq K\right\} $ and of the common factors $\left\{
\mathcal{F}_{t},1\leq t\leq T\right\} $ based on minimising a square loss
function; we show the consistency of the estimates, and derive their rates
under standard moment assumptions. Secondly, in order to be able to analyse
heavy-tailed data, we study estimation of the loadings $\left\{ \mathbf{A}%
_{k},1\leq k\leq K\right\} $ and of the common factors $\left\{ \mathcal{F}%
_{t},1\leq t\leq T\right\} $ using a Huber loss function. To the best of our
knowledge, this is the first time that an estimator for a tensor factor
model is proposed that can be used in the presence of heavy tails, with no
restrictions on the shape of the distribution. As a by-product, we also
offer a novel methodology, based on a modified version of the
eigenvalue-ratio principle, in order to determine the number of common
factors. Although the asymptotic theory is reported in Section \ref{theory},
here we offer a heuristic preview of our findings. We show that, when the
fourth moment exists, the Least Squares (LS) estimator is consistent as can
be expected; estimators based on the Huber loss function are also
consistent, however, and at the same rate as the LS\ estimator, save for
(extreme) cases where one (or more) of the cross-sectional dimension is very
small. Conversely, the robust, Huber loss-based estimator is consistent even
when the fourth moment does not exist, and indeed even when only the second
moent exists; in this case, estimation is slower than it would be if more
moments existed, but consistency can still be guaranteed. From a practical
point of view, our estimators are easy to implement and require minimal
coding. An R package \textquotedblleft RTFA" implementing the proposed
methods is available on R CRAN \footnote{%
\url{https://cran.r-project.org/web/packages/RTFA/index.html}}.

\bigskip

The remainder of the paper is organised as follows. We provide the details
of estimators based on quadratic loss and Huber loss functions in Section %
\ref{method} (see Sections \ref{ls} and \ref{huber} respectively), and
discuss our estimators of the number of common factors in Section \ref%
{sec:2.3}. We report and discuss the main assumptions, and the rates of
convergence and the relevant asymptotic theory, in Section \ref{theory}. We
validate ou results by means of a comprehensive set of simulations in
Section \ref{sec:4}, and we ilustrate our approach through an application to
import/export data in Section \ref{empirics}. Section \ref{conclusion}
concludes the paper and all proofs are relegated to the supplent.

\subsection{Notation}

Throughout the paper, we extensively use the following notation: $%
p=p_{1}p_{2}\cdots p_{K}$, $p_{-k}=p/p_{k}$, $r=r_{1}r_{2}\cdots r_{K}$ and $%
r_{-k}=r/r_{k}$. Given a tensor $\mathcal{X}\in \mathbb{R}^{p_{1}\times
p_{2}\times \dots \times p_{K}}$, and a matrix $\mathbf{A}\in \mathbb{R}%
^{d\times p_{k}}$, we denote the mode-$k$ product as (the tensor of size $%
p_{1}\times \cdots \times p_{k-1}\times d\times p_{k+1}\times \cdots p_{K}$)
$\mathcal{X}\times _{k}\mathbf{A}$, defined element-wise as%
\begin{equation*}
\left( \mathcal{X}\times _{k}\mathbf{A}\right) _{i_{1},\cdots
,i_{k-1},j,i_{k+1},\cdots ,i_{K}}=\sum_{i_{k}=1}^{p_{k}}x_{i_{1},\cdots
,i_{k},\cdots ,i_{K}}a_{j,i_{k}}.
\end{equation*}%
The mode-$k$ unfolding matrix of $\mathcal{X}$ is denoted by mat$_{k}(%
\mathcal{X})$, and it arranges all $p_{-k}$ mode-$k$ fibers of $\mathcal{X}$
to be the columns of a $p_{k}\times p_{-k}$ matrix. As far as matrices and
matrix operations are concerned: given a matrix $\mathbf{A}$, $\mathbf{A}%
^{\top }$ is the transpose of $\mathbf{A}$; Tr$(\mathbf{A})$ is the trace of
$\mathbf{A}$; $\Vert \mathbf{A}\Vert _{F}$ is the Frobenious norm of $%
\mathbf{A}$; $\mathbf{A}\otimes \mathbf{B}$ denotes the Kronecker product
between matrices $\mathbf{A}$ and $\mathbf{B}$; and $\mathbf{I}_{k}$
represents a $k$-order identity matrix. Finite, positive constants whose
value can change from line to line are indicated as $c_{0}$, $c_{1}$,...
Other, relevant notation is introduced later on in the paper.

\section{Methodology\label{method}}

Recall the tensor factor model of (\ref{tfm}), viz.%
\begin{equation*}
\mathcal{X}_{t}=\mathcal{F}_{t}\times _{1}\mathbf{A}_{1}\times _{2}\cdots
\times _{K}\mathbf{A}_{K}+\mathcal{E}_{t},
\end{equation*}%
which can be re-written more compactly as%
\begin{equation*}
\mathcal{X}_{t}=\mathcal{F}_{t}\times _{k=1}^{K}\mathbf{A}_{k}+\mathcal{E}%
_{t}=\mathcal{S}_{t}+\mathcal{E}_{t},
\end{equation*}%
with $\mathcal{S}_{t}$ representing the signal component. In this section,
we discuss estimation of (\ref{tfm}) using two (families of) methodologies:
one based on minimising a quadratic loss function (Section \ref{ls}), and
one based on minimising a Huber-type loss function (Section \ref{huber}).
The emphasis is on applicability/implementation, and all the relevant theory
is relegated to the next section.

In order to ensure identifiability, and without loss of generality, we will
assume the following (standard) restriction
\begin{equation}
\frac{1}{p_{k}}\mathbf{A}_{k}^{\top }\mathbf{A}_{k}=\mathbf{I}_{r_{k}},
\label{identification-ak}
\end{equation}%
for all $1\leq k\leq K$. Henceforth, we will also use extensively the
short-hand notation%
\begin{eqnarray}
\text{mat}_{k}\left( \mathcal{X}_{t}\right) &=&\mathbf{X}_{k,t},  \label{xk}
\\
\text{mat}_{k}\left( \mathcal{F}_{t}\right) &=&\mathbf{F}_{k,t}.  \label{fk}
\end{eqnarray}

\subsection{Estimation based on quadratic loss functions\label{ls}}

Under the identifiability condition (\ref{identification-ak}), $\{\mathbf{A}%
_{k},1\leq k\leq K\}$ and $\left\{ \mathcal{F}_{t},1\leq t\leq T\right\} $\
is estimated by minimizing the following Least Squares loss:
\begin{gather}
\min_{\mathbf{A}_{1},\cdots ,\mathbf{A}_{K},\left\{ \mathcal{F}_{t}\right\}
_{t=1}^{T}}L_{1}\left( \mathbf{A}_{1},\cdots ,\mathbf{A}_{K},\left\{
\mathcal{F}_{t}\right\} _{t=1}^{T}\right)  \label{equ:LL} \\
\text{s.t. }\frac{1}{p_{k}}\mathbf{A}_{k}^{\top }\mathbf{A}_{k}=\mathbf{I}%
_{r_{k}},1\leq k\leq K,  \notag
\end{gather}%
where
\begin{equation*}
L_{1}\left( \mathbf{A}_{1},\cdots ,\mathbf{A}_{K},\left\{ \mathcal{F}%
_{t}\right\} _{t=1}^{T}\right)
=\frac{1}{T}\sum_{t=1}^{T}\left\Vert \frac{\mathcal{%
X}_{t}-\mathcal{F}_{t}\times _{k=1}^{K}\mathbf{A}_{k}}{\sqrt{p}}\right\Vert _{F}^{2}.
\end{equation*}%
Note that
\begin{equation*}
\left\Vert \mathcal{X}_{t}-\mathcal{F}_{t}\times _{k=1}^{K}\mathbf{A}%
_{k}\right\Vert _{F}^{2}=\left\Vert \text{mat}_{k}\left( \mathcal{X}%
_{t}\right) -\text{mat}_{k}\left( \mathcal{F}_{t}\times _{k=1}^{K}\mathbf{A}%
_{k}\right) \right\Vert _{F}^{2}=\left\Vert \text{mat}_{k}\left( \mathcal{X}%
_{t}\right) -\mathbf{A}_{k}\text{mat}_{k}\left( \mathcal{F}_{t}\right)
\left( \otimes _{j\in \lbrack K]\backslash \{k\}}\mathbf{A}_{j}\right)
^{\top }\right\Vert _{F}^{2},
\end{equation*}%
where $\otimes _{j\in \lbrack K]\backslash \{k\}}\mathbf{A}_{j}=\mathbf{A}%
_{K}\otimes \mathbf{A}_{K-1}\otimes \cdots \otimes \mathbf{A}_{k+1}\otimes
\mathbf{A}_{k-1}\otimes \cdots \otimes \mathbf{A}_{1}$. Let
\begin{equation*}
\mathbf{B}_{k}=\otimes _{j\in \lbrack K]\backslash \{k\}}\mathbf{A}_{j};
\end{equation*}%
then $\mathbf{B}_{k}$ satisfies $\mathbf{B}_{k}^{\top }\mathbf{B}_{k}/p_{-k}=%
\mathbf{I}_{r_{-k}}$ for all $1\leq k\leq K$. Recalling that mat$_{k}\left(
\mathcal{X}_{t}\right) =\mathbf{X}_{k,t}$, and mat$_{k}\left( \mathcal{F}%
_{t}\right) =\mathbf{F}_{k,t}$, we can finally write%
\begin{eqnarray*}
L_{1}\left( \mathbf{A}_{1},\cdots ,\mathbf{A}_{K},\left\{ \mathcal{F}%
_{t}\right\} _{t=1}^{T}\right) &=&\frac{1}{Tp}\sum_{t=1}^{T}\left\Vert
\mathbf{X}_{k,t}-\mathbf{A}_{k}\mathbf{F}_{k,t}\mathbf{B}_{k}^{\top
}\right\Vert _{F}^{2} \\
&=&\frac{1}{Tp}\sum_{t=1}^{T}\left( \text{Tr}\left( \mathbf{X}_{k,t}^{\top }%
\mathbf{X}_{k,t}\right) -2\text{Tr}\left( \mathbf{X}_{k,t}^{\top }\mathbf{A}%
_{k}\mathbf{F}_{k,t}\mathbf{B}_{k}^{\top }\right) +p\text{Tr}\left( \mathbf{F%
}_{k,t}^{\top }\mathbf{F}_{k,t}\right) \right) \\
&=&L_{1}\left( \mathbf{A}_{k},\mathbf{B}_{k},\left\{ \mathbf{F}%
_{k,t}\right\} _{t=1}^{T}\right) .
\end{eqnarray*}%
The function $L_{1}\left( \mathbf{A}_{k},\mathbf{B}_{k},\left\{ \mathbf{F}%
_{k,t}\right\} _{t=1}^{T}\right) $\ can be minimised by concentrating out $%
\left\{ \mathbf{F}_{k,t}\right\} _{t=1}^{T}$. Indeed, given $\left\{ \mathbf{%
A}_{k},\mathbf{B}_{k}\right\} $, and minimising with respect to $\left\{
\mathbf{F}_{k,t}\right\} _{t=1}^{T}$, we receive the following first order
conditions, for all $1\leq t\leq T$%
\begin{equation*}
\frac{\partial }{\partial \mathbf{F}_{k,t}}L_{1}\left( \mathbf{A}_{k},%
\mathbf{B}_{k},\left\{ \mathbf{F}_{k,t}\right\} _{t=1}^{T}\right) =0,
\end{equation*}%
whose solutions are%
\begin{equation}
\mathbf{F}_{k,t}=\frac{1}{p}\mathbf{A}_{k}^{\top }\mathbf{X}_{k,t}\mathbf{B}%
_{k},  \label{fk-sol-ls}
\end{equation}%
for all $1\leq t\leq T$. Plugging (\ref{fk-sol-ls}) into $L_{1}\left(
\mathbf{A}_{k},\mathbf{B}_{k},\left\{ \mathbf{F}_{k,t}\right\}
_{t=1}^{T}\right) $, the minimisation problem becomes%
\begin{gather}
\min_{\mathbf{A}_{1},\cdots ,\mathbf{A}_{K}}L_{1}\left( \mathbf{A}_{k},%
\mathbf{B}_{k}\right)  \label{min-l1} \\
\text{s.t. }\frac{1}{p_{k}}\mathbf{A}_{k}^{\top }\mathbf{A}_{k}=\mathbf{I}%
_{r_{k}},\text{ \ and }\frac{1}{p_{-k}}\mathbf{B}_{k}^{\top }\mathbf{B}_{k}=%
\mathbf{I}_{r_{-k}},\text{\ }1\leq k\leq K,  \notag
\end{gather}%
where
\begin{equation*}
L_{1}\left( \mathbf{A}_{k},\mathbf{B}_{k}\right) =\frac{1}{Tp}%
\sum_{t=1}^{T}\left( \text{Tr}\left( \mathbf{X}_{k,t}^{\top }\mathbf{X}%
_{k,t}\right) -2\text{Tr}\left( \mathbf{X}_{k,t}^{\top }\mathbf{A}_{k}%
\mathbf{F}_{k,t}\mathbf{B}_{k}^{\top }\right) +p\text{Tr}\left( \mathbf{F}%
_{k,t}^{\top }\mathbf{F}_{k,t}\right) \right) .
\end{equation*}%
Hence, the minimisation problem ultimately becomes
\begin{equation}
\min_{\{\mathbf{A}_{k},\mathbf{B}_{k}\}}\mathcal{L}_{1},
\label{min-l1-final}
\end{equation}%
where $\mathcal{L}_{1}$ is the Lagrangian associated with (\ref{min-l1}),
given by%
\begin{equation}
\mathcal{L}_{1}=L_{1}\left( \mathbf{A}_{k},\mathbf{B}_{k}\right) +\text{Tr}%
\left( \Lambda _{\mathbf{A}}\left( \frac{1}{p_{k}}\mathbf{A}_{k}^{\top }%
\mathbf{A}_{k}-\mathbf{I}_{r_{k}}\right) \right) +\text{Tr}\left( \Lambda _{%
\mathbf{B}}\left( \frac{1}{p_{-k}}\mathbf{B}_{k}^{\top }\mathbf{B}_{k}-%
\mathbf{I}_{r_{-k}}\right) \right) ,  \label{lagrange-ls}
\end{equation}%
with the Lagrange multipliers $\Lambda _{\mathbf{A}}$ and $\Lambda _{\mathbf{%
B}}$ being two symmetric matrices. Based on (\ref{min-l1-final}) and (\ref%
{lagrange-ls}), we can derive the following KKT conditions%
\begin{equation}
\left\{
\begin{array}{ccc}
\frac{\partial }{\partial \mathbf{A}_{k}}\mathcal{L}_{1}= & -\frac{2}{Tp^2}%
\sum_{t=1}^{T}\mathbf{X}_{k,t}\mathbf{B}_{k}\mathbf{B}_{k}^{\top }\mathbf{X}%
_{k,t}^{\top }\mathbf{A}_{k}+\frac{2}{p_{k}}\mathbf{A}_{k}\Lambda _{\mathbf{A%
}} & =0 \\
\frac{\partial }{\partial \mathbf{B}_{k}}\mathcal{L}_{1}= & -\frac{2}{Tp^2}%
\sum_{t=1}^{T}\mathbf{X}_{k,t}^{\top }\mathbf{A}_{k}\mathbf{A}_{k}^{\top }%
\mathbf{X}_{k,t}\mathbf{B}_{k}+\frac{2}{p_{-k}}\mathbf{B}_{k}\Lambda _{%
\mathbf{B}} & =0%
\end{array}%
\right. .  \label{kkt-ls}
\end{equation}%
Then, the following equations hold%
\begin{equation}
\left\{
\begin{array}{cc}
\left( \frac{1}{Tpp_{-k}}\sum_{t=1}^{T}\mathbf{X}_{k,t}\mathbf{B}_{k}\mathbf{B%
}_{k}^{\top }\mathbf{X}_{k,t}^{\top }\right) \mathbf{A}_{k} & =\mathbf{A}%
_{k}\Lambda _{\mathbf{A}} \\
\left( \frac{1}{Tpp_{k}}\sum_{t=1}^{T}\mathbf{X}_{k,t}^{\top }\mathbf{A}_{k}%
\mathbf{A}_{k}^{\top }\mathbf{X}_{k,t}\right) \mathbf{B}_{k} & =\mathbf{B}%
_{k}\Lambda _{\mathbf{B}}%
\end{array}%
\right. .  \label{kkt-1-1}
\end{equation}%
Defining, for short%
\begin{eqnarray}
\mathbf{M}_{k} &=&\frac{1}{Tpp_{-k}}\sum_{t=1}^{T}\mathbf{X}_{k,t}\mathbf{B}%
_{k}\mathbf{B}_{k}^{\top }\mathbf{X}_{k,t}^{\top },  \label{equ:Mk} \\
\mathbf{M}_{-k} &=&\frac{1}{Tpp_{k}}\sum_{t=1}^{T}\mathbf{X}_{k,t}^{\top }%
\mathbf{A}_{k}\mathbf{A}_{k}^{\top }\mathbf{X}_{k,t},  \notag
\end{eqnarray}%
then (\ref{kkt-1-1}) can alternatively be written as%
\begin{equation}
\left\{
\begin{array}{cc}
\mathbf{M}_{k}\mathbf{A}_{k} & =\mathbf{A}_{k}\Lambda _{\mathbf{A}} \\
\mathbf{M}_{-k}\mathbf{B}_{k} & =\mathbf{B}_{k}\Lambda _{\mathbf{B}}%
\end{array}%
\right. .  \label{kkt-1-2}
\end{equation}%
Define the first $r_{k}$ eigenvalues of $\mathbf{M}_{k}$ as $\left\{ \lambda
_{k,1},\cdots ,\lambda _{k,r_{k}}\right\} $, and the corresponding
eigenvectors as $\left\{ \mathbf{u}_{k,1},\cdots ,\mathbf{u}%
_{k,r_{k}}\right\} $. Then it is easy to see that (\ref{kkt-1-2}) is
satisfied by%
\begin{eqnarray*}
\widehat{\Lambda }_{\mathbf{A}} &=&\text{diag}\left\{ \lambda _{k,1},\cdots
,\lambda _{k,r_{k}}\right\} , \\
\widehat{\mathbf{A}}_{k} &=&p_{k}^{1/2}\left\{ \mathbf{u}_{k,1}|\cdots |%
\mathbf{u}_{k,r_{k}}\right\} ,
\end{eqnarray*}%
and similarly, defining $r_{k}$ eigenvalues of $\mathbf{M}_{-k}$ as $\left\{
\lambda _{-k,1},\cdots ,\lambda _{-k,r_{k}}\right\} $, and the corresponding
eigenvectors as $\left\{ \mathbf{u}_{-k,1},\cdots ,\mathbf{u}%
_{-k,r_{k}}\right\} $, (\ref{kkt-1-2}) is also satisfied by%
\begin{eqnarray*}
\widehat{\Lambda }_{\mathbf{B}} &=&\text{diag}\left\{ \lambda _{-k,1},\cdots
,\lambda _{-k,r_{k}}\right\} , \\
\widehat{\mathbf{B}}_{k} &=&p_{-k}^{1/2}\left\{ \mathbf{u}_{-k,1}|\cdots |%
\mathbf{u}_{-k,r_{k}}\right\} .
\end{eqnarray*}%
With $\widehat{\mathbf{B}}_{k}$, we are able to define the feasible version
of $\mathbf{M}_{k}$ defined in (\ref{equ:Mk}), viz.%
\begin{equation}
\widehat{\mathbf{M}}_{k}=\frac{1}{Tpp_{-k}}\sum_{t=1}^{T}\mathbf{X}_{k,t}%
\widehat{\mathbf{B}}_{k}\widehat{\mathbf{B}}_{k}^{\top }\mathbf{X}%
_{k,t}^{\top },  \label{m_hat_k}
\end{equation}%
and $\widehat{\mathbf{M}}_{-k}$ an be defined along similar lines. Hence, it
is clear that the estimation of $\Ab_{k}$ relies on the unknown $\left\{ \Ab%
_{j}\right\} _{j\neq k}$; in our Algorithm \ref{alg111} below, we propose to
initialise the estimation of $\left\{ \Ab_{j}\right\} _{j=1}^{K}$ by using
the Initial Estimator (IE) of \citet{he2022statistical}, defined as $%
\widehat{\Ab}_{k}^{(0)}=\sqrt{p_{k}}\widehat{\Ub}_{k}$, where $\widehat{\Ub}%
_{k}$ is the $p_{k}\times r_{k}$ matrix having as columns the $r_{k}$
leading normalized eigenvectors of $\sum_{t=1}^{T}\Xb_{k,t}\Xb_{k,t}^{\top
}/(Tp)$. When computing the estimators of $\left\{ \Ab_{j}\right\} _{j\neq
k} $ in order to estimate $\Ab_{k}$, we recommend using a projection-based
method like e.g. the Projection Estimation (PE) of \citet{he2022statistical}%
, or the Iterative Projected mode-wise PCA estimation (IPmoPCA) of \cite%
{zhang2022tucker}.

\begin{algorithm}[htbp]
	\caption{Iterative Projection Estimation Algorithm for Tensor Factor Model}
	\label{alg111}
	\begin{algorithmic}[1]	
		
		\REQUIRE tensor data $\{\cX_t\}_{t=1}^T$, factor numbers $\{r_k\}_{k=1}^K$, initial estimation of loading matrices $\{\widehat\Ab_k^{(0)}\}_{k=1}^K$, maximum iterative step $m$
		
		\ENSURE loading matrices $\{\widehat{\Ab}_k\}_{k=1}^K$, factor tensors $\{\widehat\cF_t\}_{t=1}^T$
		
		
		\STATE compute  $\widehat{\Bb}_k^{(s)}=\otimes_{j=k+1}^{K}\widehat{\Ab}_j^{(s-1)}\otimes_{j=1}^{k-1}\widehat{\Ab}_j^{(s)}$;

		\STATE compute $\widehat{\Mb}_k^{(s)}$ based on (\ref{m_hat_k}), viz. $\widehat{\Mb}_k^{(s)}=\sum_{t=1}^{T}%
		\mathbf{X}_{k,t}\widehat{\mathbf{B}}_{k}^{(s)}
		\widehat{\mathbf{B}}_{k}^{(s)\top}\mathbf{X}_{k,t}^{\top }/(Tp_{-k})$,
		renew $\widehat{\Ab}_k^{(s)}$ as $\sqrt{p_k}$ times the matrix with columns being the  first $r_k$ eigenvectors  of $\widehat{\Mb}_k^{(s)}$;
		
		\STATE repeat steps 1 to 2 until convergence, or up to the maximum number of iterations, output the last step estimators as $\{\widehat{\Ab}_k\}_{k=1}^K$, and the corresponding factor tensors  $\{\widehat\cF_t=\cX_t\times_{k=1}^K\l(\widehat\Ab_k\r)^\top/p\}_{t=1}^T$.
	\end{algorithmic} 	
\end{algorithm}

\subsection{Estimation based on the Huber loss functions\label{huber}}

As mentioned in the introduction, in the presence of heavy tails it may be
more appropriate to consider a loss function which dampens outliers. Here,
we propose the following, $L_{1}$-norm based loss function, known as \textit{%
Huber loss function}%
\begin{equation}
H_{\tau }\left( x\right) =%
\begin{cases}
\frac{1}{2}x^{2} & \text{if } \left\vert x\right\vert \leq \tau, \\
\tau \left\vert x\right\vert -\frac{1}{2}\tau ^{2} & \text{if } \left\vert x\right\vert
>\tau.
\end{cases}%
  \label{h-tau}
\end{equation}%

The penalty imposed on an error $x$ is quadratic up to a threshold $\tau $,
and \textquotedblleft only\textquotedblright\ linear thereafter. Based on (%
\ref{h-tau}), we define the following minimisation problem%
\begin{gather}
\min_{\mathbf{A}_{1},\cdots ,\mathbf{A}_{K},\left\{ \mathcal{F}_{t}\right\}
_{t=1}^{T}}L_{2}\left( \mathbf{A}_{1},\cdots ,\mathbf{A}_{K},\left\{
\mathcal{F}_{t}\right\} _{t=1}^{T}\right)  \label{min-l2} \\
\text{s.t. }\frac{1}{p_{k}}\mathbf{A}_{k}^{\top }\mathbf{A}_{k}=\mathbf{I}%
_{r_{k}},\text{ \ and }\frac{1}{p_{-k}}\mathbf{B}_{k}^{\top }\mathbf{B}_{k}=%
\mathbf{I}_{r_{-k}},\text{\ }1\leq k\leq K,  \notag
\end{gather}%
where%
\begin{eqnarray*}
L_{2}\left( \mathbf{A}_{1},\cdots ,\mathbf{A}_{K},\left\{ \mathcal{F}%
_{t}\right\} _{t=1}^{T}\right) &=&\frac{1}{T}\sum_{t=1}^{T}H_{\tau }\left(
\left\Vert \frac{\mathcal{X}_{t}-\mathcal{F}_{t}\times _{k=1}^{K}\mathbf{A}%
_{k}}{\sqrt{p}}\right\Vert _{F}\right) \\
&=&\frac{1}{T}\sum_{t=1}^{T}H_{\tau }\left( \left\Vert \frac{\mathbf{X}_{k,t}-%
\mathbf{A}_{k}\mathbf{F}_{k,t}\mathbf{B}_{k}^{\top }}{\sqrt{p}}\right\Vert _{F}\right)
\\
&=&L_{2}\left( \mathbf{A}_{k},\mathbf{B}_{k},\left\{ \mathbf{F}%
_{k,t}\right\} _{t=1}^{T}\right) .
\end{eqnarray*}%
Using (\ref{h-tau}), it is clear that
\begin{align}
H_{\tau }& \left( \left\Vert\frac{ \mathbf{X}_{k,t}-\mathbf{A}_{k}\mathbf{F}_{k,t}%
\mathbf{B}_{k}^{\top }}{\sqrt{p}}\right\Vert _{F}\right)  \notag  \label{hb} \\
& =\left\{
\begin{array}{cc}
\frac{1}{2p}\left( \text{Tr}\left( \mathbf{X}_{k,t}^{\top }\mathbf{X}%
_{k,t}\right) -2\text{Tr}\left( \mathbf{X}_{k,t}^{\top }\mathbf{A}_{k}%
\mathbf{F}_{k,t}\mathbf{B}_{k}^{\top }\right) +p\text{Tr}\left( \mathbf{F}%
_{k,t}^{\top }\mathbf{F}_{k,t}\right) \right) &\text{if } \left\Vert \frac{\mathbf{X}_{k,t}-\mathbf{A}_{k}\mathbf{F}_{k,t}%
	\mathbf{B}_{k}^{\top }}{\sqrt{p}}\right\Vert _{F}\leq
\tau, \nonumber \\
\tau \sqrt{\frac{1}{p}\left(\text{Tr}\left( \mathbf{X}_{k,t}^{\top }\mathbf{X}_{k,t}\right) -2%
\text{Tr}\left( \mathbf{X}_{k,t}^{\top }\mathbf{A}_{k}\mathbf{F}_{k,t}%
\mathbf{B}_{k}^{\top }\right) +p\text{Tr}\left( \mathbf{F}_{k,t}^{\top }%
\mathbf{F}_{k,t}\right) \right)}-\frac{1}{2}\tau ^{2} &  \text{if } \left\Vert \frac{\mathbf{X}_{k,t}-\mathbf{A}_{k}\mathbf{F}_{k,t}%
\mathbf{B}_{k}^{\top }}{\sqrt{p}}\right\Vert _{F}>\tau.%
\end{array}%
\right.  \\
&
\end{align}%
Similarly to the results in the previous section, the first order conditions
\begin{equation*}
\frac{\partial }{\partial \mathbf{F}_{k,t}}L_{2}\left( \mathbf{A}_{k},%
\mathbf{B}_{k},\left\{ \mathbf{F}_{k,t}\right\} _{t=1}^{T}\right) =0,
\end{equation*}%
yield%
\begin{equation}
\mathbf{F}_{k,t}=\frac{1}{p}\mathbf{A}_{k}^{\top }\mathbf{X}_{k,t}\mathbf{B}%
_{k},  \label{f-hub}
\end{equation}%
whence, substituting (\ref{f-hub}) into (\ref{hb}), it follows that the
concentrated Huber loss at each $t$ can be expressed as%
\begin{align*}
& H_{\tau }\left( \left\Vert \frac{\mathbf{X}_{k,t}-\mathbf{A}_{k}\mathbf{F}_{k,t}%
\mathbf{B}_{k}^{\top }}{\sqrt{p}}\right\Vert _{F}\right) \\
& =\left\{
\begin{array}{cc}
\frac{1}{2p}\left( \text{Tr}\left( \mathbf{X}_{k,t}^{\top }\mathbf{X}%
_{k,t}\right) -\frac{1}{p}\text{Tr}\left( \mathbf{X}_{k,t}^{\top }\mathbf{A}%
_{k}\mathbf{A}_{k}^{\top }\mathbf{X}_{k,t}\mathbf{B}_{k}\mathbf{B}_{k}^{\top
}\right) \right) &\text{if } \left\Vert \frac{\mathbf{X}_{k,t}-\mathbf{A}_{k}\mathbf{F}_{k,t}%
\mathbf{B}_{k}^{\top }}{\sqrt{p}}\right\Vert _{F}\leq \tau, \\
\tau \sqrt{\frac{1}{p}\left(\text{Tr}\left( \mathbf{X}_{k,t}^{\top }\mathbf{X}_{k,t}\right) -%
\frac{1}{p}\text{Tr}\left( \mathbf{X}_{k,t}^{\top }\mathbf{A}_{k}\mathbf{A}%
_{k}^{\top }\mathbf{X}_{k,t}\mathbf{B}_{k}\mathbf{B}_{k}^{\top }\right)\right) }-%
\frac{1}{2}\tau ^{2} &\text{if } \left\Vert \frac{\mathbf{X}_{k,t}-\mathbf{A}_{k}\mathbf{F}_{k,t}%
	\mathbf{B}_{k}^{\top }}{\sqrt{p}}\right\Vert _{F}>\tau.%
\end{array}%
\right.
\end{align*}%
When $\left\Vert \frac{\mathbf{X}_{k,t}-\mathbf{A}_{k}\mathbf{F}_{k,t}%
	\mathbf{B}_{k}^{\top }}{\sqrt{p}}\right\Vert _{F}\leq \tau $, it holds that%
\begin{equation}
\left\{
\begin{array}{cc}
\frac{\partial }{\partial \mathbf{A}_{k}}H_{\tau }\left( \left\Vert \frac{\mathbf{X%
}_{k,t}-\mathbf{A}_{k}\mathbf{F}_{k,t}\mathbf{B}_{k}^{\top }}{\sqrt{p}}\right\Vert
_{F}\right) = & -\frac{1}{p^2}\mathbf{X}_{k,t}\mathbf{B}_{k}\mathbf{B}%
_{k}^{\top }\mathbf{X}_{k,t}^{\top }\mathbf{A}_{k} \\
\frac{\partial }{\partial \mathbf{B}_{k}}H_{\tau }\left( \left\Vert \frac{\mathbf{X%
}_{k,t}-\mathbf{A}_{k}\mathbf{F}_{k,t}\mathbf{B}_{k}^{\top }}{\sqrt{p}}\right\Vert
_{F}\right) = & -\frac{1}{p^2}\mathbf{X}_{k,t}^{\top }\mathbf{A}_{k}\mathbf{A}%
_{k}^{\top }\mathbf{X}_{k,t}\mathbf{B}_{k}%
\end{array}%
\right. .  \label{huber-low}
\end{equation}%
Conversely, when $\left\Vert \frac{\mathbf{X}_{k,t}-\mathbf{A}_{k}\mathbf{F}_{k,t}%
	\mathbf{B}_{k}^{\top }}{\sqrt{p}}\right\Vert _{F}>\tau $, it holds that
\begin{equation}
\left\{
\begin{array}{cc}
\frac{\partial }{\partial \mathbf{A}_{k}}H_{\tau }\left( \left\Vert \frac{\mathbf{X%
	}_{k,t}-\mathbf{A}_{k}\mathbf{F}_{k,t}\mathbf{B}_{k}^{\top }}{\sqrt{p}}\right\Vert
_{F}\right) = & -\frac{\tau }{p^2}\frac{\mathbf{X}_{k,t}\mathbf{B}_{k}\mathbf{B%
}_{k}^{\top }\mathbf{X}_{k,t}^{\top }\mathbf{A}_{k}}{\sqrt{\left(\text{Tr}\left(
\mathbf{X}_{k,t}^{\top }\mathbf{X}_{k,t}\right) -\frac{1}{p}\text{Tr}\left(
\mathbf{X}_{k,t}^{\top }\mathbf{A}_{k}\mathbf{A}_{k}^{\top }\mathbf{X}_{k,t}%
\mathbf{B}_{k}\mathbf{B}_{k}^{\top }\right)\right)/p }} \\
\frac{\partial }{\partial \mathbf{B}_{k}}H_{\tau }\left( \left\Vert \frac{\mathbf{X%
	}_{k,t}-\mathbf{A}_{k}\mathbf{F}_{k,t}\mathbf{B}_{k}^{\top }}{\sqrt{p}}\right\Vert
_{F}\right) = & -\frac{\tau }{p^2}\frac{\mathbf{X}_{k,t}^{\top }\mathbf{A}_{k}%
\mathbf{A}_{k}^{\top }\mathbf{X}_{k,t}\mathbf{B}_{k}}{\sqrt{\left(\text{Tr}\left(
\mathbf{X}_{k,t}^{\top }\mathbf{X}_{k,t}\right) -\frac{1}{p}\text{Tr}\left(
\mathbf{X}_{k,t}^{\top }\mathbf{A}_{k}\mathbf{A}_{k}^{\top }\mathbf{X}_{k,t}%
\mathbf{B}_{k}\mathbf{B}_{k}^{\top }\right)\right)/p }}%
\end{array}%
\right. .  \label{huber-high}
\end{equation}%
By the above, we write the Lagrangian of the concentrated version of (\ref%
{min-l2}) as follows%
\begin{equation}
\min_{\{\mathbf{A}_{k},\mathbf{B}_{k}\}}\mathcal{L}_{2},  \label{l2}
\end{equation}%
where $\mathcal{L}_{2}$ is the Lagrangian associated with (\ref{min-l1}),
given by%
\begin{equation}
\mathcal{L}_{2}=L_{2}\left( \mathbf{A}_{k},\mathbf{B}_{k}\right) +\text{Tr}%
\left( \Lambda _{\mathbf{A}}^{H}\left( \frac{1}{p_{k}}\mathbf{A}_{k}^{\top }%
\mathbf{A}_{k}-\mathbf{I}_{r_{k}}\right) \right) +\text{Tr}\left( \Lambda _{%
\mathbf{B}}^{H}\left( \frac{1}{p_{-k}}\mathbf{B}_{k}^{\top }\mathbf{B}_{k}-%
\mathbf{I}_{r_{-k}}\right) \right) ,  \label{l2-def-1}
\end{equation}%
where the Lagrange multipliers $\Lambda _{\mathbf{A}}^{H}$ and $\Lambda _{%
\mathbf{B}}^{H}$ are two symmetric matrices, and%
\begin{equation}
L_{2}\left( \mathbf{A}_{k},\mathbf{B}_{k}\right) =\frac{1}{T}%
\sum_{t=1}^{T}H_{\tau }\left( \left\Vert \frac{\mathbf{X%
	}_{k,t}-\mathbf{A}_{k}\mathbf{F}_{k,t}\mathbf{B}_{k}^{\top }}{\sqrt{p}}\right\Vert
_{F}\right) .
\label{l2-def}
\end{equation}%
Hence, we have the following KKT conditions%
\begin{equation}
\left\{
\begin{array}{ccc}
\frac{\partial }{\partial \mathbf{A}_{k}}\mathcal{L}_{2}= & -\left( \frac{1}{%
T}\sum_{t=1}^{T}w_{k,t}^{h}\mathbf{X}_{k,t}\mathbf{B}_{k}\mathbf{B}%
_{k}^{\top }\mathbf{X}_{k,t}^{\top }\mathbf{A}_{k}\right) \mathbf{A}_{k}+%
\frac{2}{p_{k}}\mathbf{A}_{k}\Lambda _{\mathbf{A}}^{H} & =0 \\
\frac{\partial }{\partial \mathbf{B}_{k}}\mathcal{L}_{2}= & -\left( \frac{1}{%
T}\sum_{t=1}^{T}w_{k,t}^{h}\mathbf{X}_{k,t}^{\top }\mathbf{A}_{k}\mathbf{A}%
_{k}^{\top }\mathbf{X}_{k,t}\mathbf{B}_{k}\right) \mathbf{B}_{k}+\frac{2}{%
p_{-k}}\mathbf{B}_{k}\Lambda _{\mathbf{B}}^{H} & =0%
\end{array}%
\right. ,  \label{kkt-huber}
\end{equation}%
where we have defined the weights $w_{k,t}^{h}$ as
\begin{equation*}
w_{k,t}^{h}=\left\{
\begin{array}{cc}
\frac{1}{p^2} &\text{if } \left\Vert \frac{\mathbf{X%
	}_{k,t}-\mathbf{A}_{k}\mathbf{F}_{k,t}\mathbf{B}_{k}^{\top }}{\sqrt{p}}\right\Vert
_{F}\leq \tau, \\
\frac{\tau }{p^2}\frac{1}{\sqrt{\left(\text{Tr}\left( \mathbf{X}_{k,t}^{\top }%
\mathbf{X}_{k,t}\right) -\frac{1}{p}\text{Tr}\left( \mathbf{X}_{k,t}^{\top }%
\mathbf{A}_{k}\mathbf{A}_{k}^{\top }\mathbf{X}_{k,t}\mathbf{B}_{k}\mathbf{B}%
_{k}^{\top }\right)\right)/p }} &\text{if } \left\Vert \frac{\mathbf{X%
}_{k,t}-\mathbf{A}_{k}\mathbf{F}_{k,t}\mathbf{B}_{k}^{\top }}{\sqrt{p}}\right\Vert
_{F}>\tau.
\end{array}%
\right.
\end{equation*}%
Letting $\widetilde{w}_{k,t}^{H}=\left( p^2/2\right) w_{k,t}^{h}$ and
\begin{eqnarray}
\mathbf{M}_{k}^{H} &=&\frac{1}{Tpp_{-k}}\sum_{t=1}^{T}\widetilde{w}_{k,t}^{H}%
\mathbf{X}_{k,t}\mathbf{B}_{k}\mathbf{B}_{k}^{\top }\mathbf{X}_{k,t}^{\top },
\label{equ:MkH} \\
\mathbf{M}_{-k}^{H} &=&\frac{1}{Tpp_{k}}\sum_{t=1}^{T}\widetilde{w}_{k,t}^{H}%
\mathbf{X}_{k,t}^{\top }\mathbf{A}_{k}\mathbf{A}_{k}^{\top }\mathbf{X}_{k,t},
\notag
\end{eqnarray}%
the minimisation problem can be re-written equivalently as%
\begin{equation}
\left\{
\begin{array}{cc}
\mathbf{M}_{k}^{H}\mathbf{A}_{k} & =\mathbf{A}_{k}\Lambda _{\mathbf{A}}^{H}
\\
\mathbf{M}_{-k}^{H}\mathbf{B}_{k} & =\mathbf{B}_{k}\Lambda _{\mathbf{B}}^{H}%
\end{array}%
\right. .  \label{kkt-2-2}
\end{equation}%
Hence, defining the first $r_{k}$ eignvalues of $\mathbf{M}_{k}^{H}$ as $%
\{\lambda _{k,1}^{H},\cdots ,\lambda _{k,r_{k}}^{H}\}$ and the corresponding
eigenvectors as $\{\mathbf{u}_{k,1}^{H},\cdots ,\mathbf{u}_{k,r_{k}}^{H}\}$,
(\ref{kkt-2-2}) is satisfied by%
\begin{eqnarray*}
\widehat{\Lambda }_{\mathbf{A}}^{H} &=&\text{diag}\left\{ \lambda
_{k,1}^{H},\cdots ,\lambda _{k,r_{k}}^{H}\right\} , \\
\widehat{\mathbf{A}}_{k}^{H} &=&p_{k}^{1/2}\left\{ \mathbf{u}%
_{k,1}^{H}|\cdots |\mathbf{u}_{k,r_{k}}^{H}\right\} ,
\end{eqnarray*}%
and similarly, defining $r_{k}$ eigenvalues of $\mathbf{M}_{-k}^{H}$ as $%
\left\{ \lambda _{-k,1}^{H},\cdots ,\lambda _{-k,r_{k}}^{H}\right\} $, and
the corresponding eigenvectors as $\left\{ \mathbf{u}_{-k,1}^{H},\cdots ,%
\mathbf{u}_{-k,r_{k}}^{H}\right\} $
\begin{eqnarray*}
\widehat{\Lambda }_{\mathbf{B}}^{H} &=&\text{diag}\left\{ \lambda
_{-k,1}^{H},\cdots ,\lambda _{-k,r_{k}}^{H}\right\} , \\
\widehat{\mathbf{B}}_{k}^{H} &=&p_{-k}^{1/2}\left\{ \mathbf{u}%
_{-k,1}^{H}|\cdots |\mathbf{u}_{-k,r_{k}}^{H}\right\} .
\end{eqnarray*}%
With $\widehat{\mathbf{B}}_{k}^{H}$, we are able to define the feasible
version of $\mathbf{M}_{k}^{H}$ defined in (\ref{equ:MkH}), viz.%
\begin{equation}
\widehat{\mathbf{M}}_{k}^{H}=\frac{1}{Tpp_{-k}}\sum_{t=1}^{T}\widehat{w}%
_{k,t}^{H}\mathbf{X}_{k,t}\widehat{\mathbf{B}}_{k}^{H}\left( \widehat{%
\mathbf{B}}_{k}^{H}\right) ^{\top }\mathbf{X}_{k,t}^{\top },
\label{m_hat_kH}
\end{equation}%
based on the feasible weights%
\begin{equation}
\widehat{w}_{k,t}^{H}=\left\{
\begin{array}{cc}
\frac{1}{2} &\text{if } \left\Vert\frac{ \mathbf{X}_{k,t}-\widehat{\mathbf{A}}_{k}^{H}%
\widehat{\mathbf{F}}_{k,t}^{H}\left( \widehat{\mathbf{B}}_{k}^{H}\right) ^{\top
}}{\sqrt{p}}\right\Vert _{F}\leq \tau, \\
\frac{\tau }{2}\frac{1}{\sqrt{\left(\text{Tr}\left( \mathbf{X}_{k,t}^{\top }%
\mathbf{X}_{k,t}\right) -\frac{1}{p}\text{Tr}\left( \mathbf{X}_{k,t}^{\top }%
\widehat{\mathbf{A}}_{k}^H(\widehat{\mathbf{A}}_{k}^H)^{\top }\mathbf{X}_{k,t}\widehat{\mathbf{B}}_{k}^H(\widehat{\mathbf{B}}%
_{k}^H)^{\top }\right)\right)/p }} &\text{if } \left\Vert \frac{\mathbf{X}_{k,t}-\widehat{\mathbf{A}}%
_{k}^{H}\widehat{\mathbf{F}}_{k,t}^{H}\left( \widehat{\mathbf{B}}_{k}^{H}\right)
^{\top }}{\sqrt{p}}\right\Vert _{F}>\tau;%
\end{array}%
\right.   \label{weights_H}
\end{equation}%
$\widehat{\mathbf{M}}_{-k}^{H}$ also can be defined along similar lines.
Similarly to the case of Least Squares loss, the estimation of the loading
matrices $\Ab_{k}$ relies on the unobservable $\left\{ \Ab_{j}\right\}
_{j=1}^{K}$ (note that an initial estimate of $\Ab_{k}$ itself is required
in the computation of $\widehat{w}_{k,t}^{H}$). In our Algorithm \ref{alg2}
below, we propose an iterative weighted projection-based procedure to obtain
the robust estimation of loading matrices and factor tensor based on Huber
loss function, which we call Robust Tensor Factor Analysis (RTFA). Similarly
to Algorithm \ref{alg111}, we recommend to use the Initial Estimator (IE) of %
\citet{he2022statistical} as initialisation, and subsequently the Projection
Estimation (PE) of \citet{he2022statistical}, or the Iterative Projected
mode-wise PCA estimation (IPmoPCA) of \cite{zhang2022tucker}.

\begin{algorithm}[htbp]
	\caption{Robust Tensor Factor Analysis (RTFA)}
	\label{alg2}
	\begin{algorithmic}[1]	
		
		\REQUIRE tensor data $\{\cX_t\}_{t=1}^T$, factor numbers $\{r_k\}_{k=1}^K$, initial estimation of loading matrices $\{\widehat\Ab_k^{(0)}\}_{k=1}^K$, maximum iterative step $m$
		
		\ENSURE loading matrices $\{\widehat{\Ab}_k^H\}_{k=1}^K$, factor tensors $\{\widehat\cF_t^H\}_{t=1}^T$
		
		
		\STATE compute  $\widehat{\Bb}_k^{(s)}=\otimes_{j=k+1}^{K}\widehat{\Ab}_j^{(s-1)}\otimes_{j=1}^{k-1}\widehat{\Ab}_j^{(s)}$;
		
		\STATE calculate the feasible weights $\widehat{w}_{k,t}^{H(s)}$ defined in (\ref{weights_H}) using $\widehat{\Ab}_k^{(s-1)}$, $\widehat{\Bb}_k^{(s)}$;
		
		\STATE compute $\widehat{\Mb}_k^{H(s)}$ based on (\ref{m_hat_kH}), viz. $\widehat{\Mb}_k^{H(s)}=\sum_{t=1}^{T}\widehat{w}_{k,t}^{H(s)}%
		\mathbf{X}_{k,t}\widehat{\mathbf{B}}_{k}^{(s)}
		\widehat{\mathbf{B}}_{k}^{(s)\top}\mathbf{X}_{k,t}^{\top }/(Tp_{-k})$,
		renew $\widehat{\Ab}_k^{(s)}$ as $\sqrt{p_k}$ times the matrix with columns being the  first $r_k$ eigenvectors  of $\widehat{\Mb}_k^{H(s)}$;
		
		\STATE repeat steps 1 to 3 until convergence, or up to the maximum number of iterations, output the last step estimators as $\{\widehat{\Ab}_k^H\}_{k=1}^K$, and the corresponding factor tensors  $\{\widehat\cF_t^H=\cX_t\times_{k=1}^K\l(\widehat\Ab_k^{H}\r)^\top/p\}_{t=1}^T$.
	\end{algorithmic} 	
\end{algorithm}


\subsection{Estimation of the numbers of factors\label{sec:2.3}}

Arguably, the first step in the estimation of the core tensor is the
estimation of its dimensions $r_{k}$, $1\leq k\leq K$. Several methodologies
can be proposed to this end, which are, in essence, extensions of existing
results derived in the case of vector-valued time series. Here, we propose
an estimator of $r_{k}$, $1\leq k\leq K$, based on a modified version of the
eigenvalue-ratio principle (\citealp{lam2012factor}; %
\citealp{ahnhorenstein13}). Specifically, we propose the following family of
(infeasible) criteria
\begin{equation}
	ER\left( r_{k}\right) =\operatorname{argmax}_{j\leq r_{\max }}\frac{\lambda _{j}(\mathbf{M}%
		_{k}^{w})}{\lambda _{j+1}(\mathbf{M}_{k}^{w})+c\widetilde{\omega }_{k}},
	\label{equ:2}
\end{equation}%
where: $\mathbf{M}_{k}^{w}=\sum_{t=1}^{T}w_{k,t}\mathbf{X}_{k,t}{\mathbf{B}}_{k}%
{\mathbf{B}}_{k}^{\top }\mathbf{X}_{k,t}^{\top }/\left(
Tpp_{-k}\right) $ is a weighted projection covariance matrix  of $\{\mathcal{X}_t\}$ with some weights $w_{k,t}$'s, $r_{\max }$ is a predetermined positive constant greater
than $\{r_{k},k=1,\cdots ,K\}$, $0<c<\infty $ is a user-chosen constant, and
$\widetilde{\omega }_{k}$ is a sequence such that, as $\min \left\{
T,p_{1},...,p_{K}\right\} \rightarrow \infty $, $\widetilde{\omega }%
_{k}=o\left( 1\right) $ (we refer to Section \ref{nfactors} for specific
examples and discussion of these tuning parameters).

The intuition underpinning (\ref{equ:2}) is the same as in %
\citet{ahnhorenstein13}: if the common factors are strong enough, the first $%
r_{k}$ eigenvalues of $\mathbf{M}_{k}^w$ will be significantly larger than the
rest. However, the denominator of $ER\left( r_{k}\right) $ contains the
further term $c\widetilde{w}_{k}$. This is because there is no theoretical
guarantee that, for some $j>r_{k}$, $\lambda _{j+1}\left( \mathbf{M}%
_{k}^w\right) $, will not be very small, thus artificially inflating the ratio
$\lambda _{j}\left( \mathbf{M}_{k}^w\right) /\lambda _{j+1}\left( \mathbf{M}^w%
_{k}\right) $. The presence of $c\widetilde{w}_{k}$ (which can be chosen to
be of the same order of magnitude as the upper bound for $\lambda
_{j+1}\left( \mathbf{M}_{k}^w\right) $ when $j\geq r_{k}$) serves the purpose
of \textquotedblleft weighing down\textquotedblright\ the eigenvalue ratio
and avoid such degeneracy - see also \citet{chang2021modelling}. As we show
in Section \ref{nfactors}, no restrictions are needed on the relative rates
of divergence of $T$ and $p_{1},...,p_{K}$ as they pass to infinity,
contrary to \citet{ahnhorenstein13}: hence our estimator can be applied for
all values of $\left\{ T,p_{1},...,p_{K}\right\} $, as long as $\min \left\{
T,p_{1},...,p_{K}\right\} \rightarrow \infty $.

As far as implementation is concerned, based on the discussion in the
previous sections, we propose two methods to estimate the number of factors.
The first one is to replace $\mathbf{M}_{k}^w$ in (\ref{equ:2}) with $\widehat{%
\mathbf{M}}_{k}$ defined in (\ref{m_hat_k}), based on the Least Squares loss
estimate; the second one is to use $\widehat{\mathbf{M}}_{k}^{H}$ defined in
(\ref{m_hat_kH}), based on the Huber loss estimate. Since the computation of
$\widehat{\mathbf{M}}_{k}$ and $\widehat{\mathbf{M}}_{k}^{H}$ requires $\{%
\widehat{\mathcal{F}}_{t}\}_{t=1}^{T}$ and $\{\widehat{\mathcal{F}}%
_{t}^{H}\}_{t=1}^{T}$, and therefore it requires an initial estimate of the
dimensions of the tensor factor, we propose two iterative methods to obtain,
respectively, the projection-based estimators $\{\widehat{r}_{k},k=1,\cdots
,K\}$, and the robust, Huber loss based estimator $\{\widehat{r}%
_{k}^{H},k=1,\cdots ,K\}$. Details are in Algorithms \ref{alg3} and \ref%
{alg4} below, respectively. Similarly to the iterative estimation of loading
matrices and tensor factors, we use the initial estimator IE of %
\citet{he2022statistical} as the initial estimation in the numerical
simulations. One novelty of the algorithms is that, at each iteration $s+1$,
we re-estimate the loadings $\mathbf{A}_{k}$\ using a deliberately inflated
dimension given by $\widehat{r}_{k}^{\left( s\right) }+2$, in order to avoid
the risk of underestimation. We refer to Section \ref{sec:4.4} for
guidelines on the choice of $c$ in (\ref{equ:2}).%

\begin{algorithm}[htbp]
	\caption{Iterative Projected Estimation of factor numbers based on Eigenvalue-Ratio (iPE-ER)}
	\label{alg3}
	\begin{algorithmic}[1]	
		
		\REQUIRE tensor data $\{\cX_t\}_{t=1}^T$, the initial estimators  $\{\widehat{\Ab}_k^{(0)}\}_{k=1}^K$, maximum number $r_{\max}$, maximum iterative step $m$
		
		\ENSURE factor numbers $\{\hat{r}_k\}_{k=1}^K$
		
		\STATE initialize: $r_k^{(0)}=r_{\max},k=1,\cdots,K$;
		
		
		\STATE compute $\widehat{\Bb}_k^{(s)}=\otimes_{j=k+1}^{K}\widehat{\Ab}_j^{(s-1)}\otimes_{j=1}^{k-1}\widehat{\Ab}_j^{(s)}$;
		
		\STATE compute $\widehat{\Mb}_k^{(s)}$ as per (\ref{m_hat_k}), obtain $\hat{r}_k^{(s)}$;
		
		\STATE renew $\widehat{\Ab}_k^{(s)}$ as $\sqrt{p_k}$ times the matrix with columns being the  first $\hat{r}_k^{(s)}+2$ eigenvectors  of $\widehat{\Mb}_k^{(s)}$;
		
		\STATE repeat steps 2 to 4 until convergence, or up to the maximum number of iterations, output the last step estimator as $\{\hat{r}_k^{\text{iPE-ER}}\}_{k=1}^K$.
	\end{algorithmic} 	
\end{algorithm}

\begin{algorithm}[htbp]
	\caption{Robust Tensor Factor Analysis of factor numbers based on Eigenvalue-Ratio (RTFA-ER)}
	\label{alg4}
	\begin{algorithmic}[1]	
		
		\REQUIRE tensor data $\{\cX_t,t=1,\cdots,T\}$, the initial estimators  $\{\widehat{\Ab}_k^{(0)},k=1,\cdots,K\}$, maximum number $r_{\max}$, maximum iterative step $m$
		
		\ENSURE factor numbers $\{\hat{r}_k^{\text{RTFA-ER}},k=1,\cdots,K\}$
		
		\STATE initialize: $r_k^{(0)}=r_{\max},k=1,\cdots,K$;
		
		
		\STATE compute $\widehat{\Bb}_k^{(s)}=\otimes_{j=k+1}^{K}\widehat{\Ab}_j^{(s-1)}\otimes_{j=1}^{k-1}\widehat{\Ab}_j^{(s)}$;
		
		\STATE use $\widehat{\Ab}_k^{(s-1)}$ and $\widehat{\Bb}_k^{(s)}$ to calculate weights $\widehat{w}_{k,t}^{H(s)}$;
		
		\STATE compute $\widehat{\Mb}_k^{H(s)}$ as per (\ref{m_hat_kH}), obtain $\hat{r}_k^{(s)}$;
		
		\STATE renew $\widehat{\Ab}_k^{(s)}$ as $\sqrt{p_k}$ times the matrix with columns being the  first $\hat{r}_k^{(s)}+2$ eigenvectors  of $\widehat{\Mb}_k^{H(s)}$;
		
		\STATE repeat steps 2 to 5 until convergence, or up to the maximum number of iterations, output the last step estimator as $\{\hat{r}_k^{\text{RTFA-ER}}\}_{k=1}^K$.
	\end{algorithmic} 	
\end{algorithm}

\section{Assumptions and rates of convergence\label{theory}}

In this section, we study the asymptotic theory of the proposed estimators.
Throughout the section, we use the following notation. We define, for short,
a generic value on the parameter space as%
\begin{equation*}
\theta =\left\{ \mathbf{A}_{1},\ldots,\mathbf{A}_{K},\mathcal{F}_{1},\ldots,%
\mathcal{F}_{T}\right\} ,
\end{equation*}%
with $\mathbf{A}_{k}=\left( \mathbf{a}_{k,1}|\ldots|\mathbf{a}%
_{k,p_{k}}\right) ^{\top}$, and we denote the parameter space as%
\begin{align*}
\Theta =& \,\bigg\{ \theta :\mathbf{a}_{k,j}\in \mathcal{A}_{k}\subset
\mathbb{R}^{r_{k}}\text{, }1\leq j\leq p_k, 1\leq k\leq K,\text{ such that }%
\frac{1}{p_{k}}\mathbf{A}_{k}^{\top }\mathbf{A}_{k}=\mathbf{I}_{r_{k}}\text{%
, }1\leq k\leq K, \text{ and }\\
& \left. \mathcal{F}_{t}\in \mathcal{F}\subset \mathbb{R}%
^{r_{1}\times ...\times r_{K}}\text{, such that }\frac{1}{T}\sum_{t=1}^{T}\Fb%
_{k,t}\Fb_{k,t}^{\top }\text{ is a }r_{k}\times r_{k}\text{ positive diagonal matrix with }\right. \\
& \text{non-increasing diagonal elements, }%
1\leq t\leq T\bigg\} .
\end{align*}%
Finally, we denote the true value of a set of parameters using the subscript
\textquotedblleft $_{0}$\textquotedblright , e.g.
\begin{equation*}
\theta _{0}=\left\{ \mathbf{A}_{01},...,\mathbf{A}_{0K},\mathcal{F}_{01},...,%
\mathcal{F}_{0T}\right\} .
\end{equation*}

\subsection{Assumptions\label{assumptions}}

We now introduce our main assumptions; from the outset we note that we will
consider two sets of assumptions - one set for the case of sufficiently
light-tailed data, and another set for the case of datasets with heavier
tails.

\begin{assumption}
\label{as-1}\textit{It holds that: \textit{(i) }$\left\{ \mathcal{A}_{k}%
\text{, }1\leq k\leq K\right\} $ and $\mathcal{F}$ are closed sets and $%
\theta _{0}\in \Theta $; \textit{(ii)} there exists constants $%
0<c_{k}<\infty $ such that, for all $1\leq k\leq K$, $\left\Vert \mathbf{A}%
_{0k}\right\Vert _{F}=c_{k}p_{k}^{1/2}$; \textit{(iii)} $\left\{ \mathbf{F}%
_{0k,t},1\leq t\leq T\right\} $ is a nonrandom sequence such that (a) $%
\left\Vert \mathbf{F}_{0k,t}\right\Vert _{F}<\infty $, and (b) for $1\leq
k\leq K$
\begin{equation*}
\frac{1}{T}\sum_{t=1}^{T}\mathbf{F}_{0k,t}\mathbf{F}_{0k,t}^{\top}=\bSigma%
_{0T,k} \rightarrow \mathbf{\Sigma }_{0k} \text{ as }T\rightarrow\infty,
\end{equation*}%
where $\mathbf{\Sigma }_{0k}$ is an $r_{k}\times r_{k}$ positive definite,
diagonal matrix with non-increasing diagonal elements $\sigma_{k,01},\cdots,%
\sigma_{0,kr_{k}}$, where $\mathbf{F}_{0k,t}=$mat$_{k}\left( \mathcal{F}%
_{0t}\right) $.}
\end{assumption}

\begin{assumption}
\label{as-2}\textit{It holds that $e_{t,i_{1},...,i_{K}}|\left\{ \mathcal{F}%
_{0,t},1\leq t\leq T\right\} $ is independent across $t$ and $%
i_{1},...,i_{K} $.}
\end{assumption}

Assumption \ref{as-1} is standard in this literature. Part \textit{(ii)}, in
essence, requires that the common factors be \textquotedblleft
strong\textquotedblright\ or \textquotedblleft pervasive\textquotedblright\
across the dimension $k$, by requiring that the sum of the squares of the
loadings, $\mathbf{A}_{0k}$, grow as fast as the dimension $p_{k}$. This is
a typical requirement in the context of factor models for vector-valued (see
e.g. \citealp{stock2002forecasting} and \citealp{bai2003inferential}) and
matrix-valued (see e.g. \citealp{wang2019factor}) time series, and it could
potentially be relaxed to consider \textquotedblleft
weak(er)\textquotedblright\ common factor structures - see, e.g. %
\citet{uematsu2022inference}, \textit{inter alia}, for the vector-valued
case, and \citet{hkty} for the matrix-valued case. Part \textit{(iii)} of
the assumption stipulates that the common factors are fixed; whilst this is
natural, since in our context we need to estimate them, we would like to
point out that in the context of vector-valued and matrix-valued time
series, usually it is assumed that the common factors are random variables.
In our context, this could also be done, at the price of more complicated
algebra, and we spell out all our assumptions in terms of conditioning upon $%
\left\{ \mathcal{F}_{0,t},1\leq t\leq T\right\} $. The assumption of serial and cross-sectional independence of the idiosyncratic components is required for some technical results in our proofs (e.g., we need a Hoeffding-type bound for sum of sub-Gaussian variables), and in principle it could be relaxed at the price of more complicated high-level assumptions (for example, in the case of the aforementioned Hoeffding-type bound, see \citealp{van2002hoeffding}).

We now consider two alternative assumptions to characterise the tails of $%
e_{t,i_{1},...,i_{K}}$. Let $\Gamma _{e}$ denote the support of $%
e_{t,i_{1},...,i_{K}}|\left\{ \mathcal{F}_{0,t},1\leq t\leq T\right\} $,
i.e.
\begin{equation*}
\Gamma _{e}=\big\{ x:P\left( e_{t,i_{1},...,i_{K}}|\left\{ \mathcal{F}%
_{0,t},1\leq t\leq T\right\} \in \left( x\pm \epsilon \right) ,\text{ for
all }\epsilon >0\right) \big\} ,
\end{equation*}%
and denote the half lines as $\mathbb{R}^{+}=\left\{ x:x>0\right\} $ and $%
\mathbb{R}^{-}=\left\{ x:x<0\right\} $. The number of elements contained in $%
\Gamma _{e}$ is denoted as Card$\left( \Gamma _{e}\right) $, with the
convention that Card$\left( \Gamma _{e}\right) =\infty $ if $\Gamma _{e}$\
does not contain a finite number of elements.

\begin{assumption}
\label{as-3}\textit{It holds that: \textit{(i) }$\mathbb{E}\left(
e_{t,i_{1},...,i_{K}}|\mathcal{F}_{0,t}\right) =0$ for all $t$ and $%
i_{1},...,i_{K}$; \textit{(ii)} there exists a $\nu >0$ such that, for all $%
d=1,2,...$%
\begin{equation*}
\mathbb{E}\left( \left\vert e_{t,i_{1},...,i_{K}}\right\vert ^{2d}|\left\{
\mathcal{F}_{0,t},1\leq t\leq T\right\} \right) \leq \frac{\left( 2d\right) !%
}{2^{d}d!}\nu ^{2d};
\end{equation*}%
\textit{(iii)} $\left\Vert \mathcal{E}_{t}\right\Vert _{F}>0$ a.s., for all $%
1\leq t\leq T$.}
\end{assumption}

\begin{assumption}
\label{as-4}\textit{It holds that: \textit{(i) }$\Gamma _{e}\cap \mathbb{R}%
^{+}\neq \varnothing $, and $\Gamma _{e}\cap \mathbb{R}^{-}\neq \varnothing $%
, with Card$\left( \Gamma _{e}\right) \geq 3$; \textit{(ii)} $\mathbb{E}%
\left( e_{t,i_{1},...,i_{K}}|\mathcal{F}_{0,t}\right) =0$ for all $t$ and $%
i_{1},...,i_{K}$; \textit{(iii)}
\begin{equation*}
\mathbb{E}\left( \left\vert e_{t,i_{1},...,i_{K}}\right\vert ^{2+\epsilon
}|\left\{ \mathcal{F}_{0,t},1\leq t\leq T\right\} \right) <\infty ,
\end{equation*}%
for some $\epsilon >0$.}
\end{assumption}

Assumption \ref{as-3} contains a weak exogeneity requirement that the
idiosyncratic component $e_{t,i_{1},...,i_{K}}$ has zero mean conditional
upon $\mathcal{F}_{0,t}$, and states that the conditional distribution of $%
e_{t,i_{1},...,i_{K}}$ has sub-Gaussian tails (see e.g. \citealp{wainwright}%
). Assumption \ref{as-4} relaxes Assumption \ref{as-3}, requiring the
existence of fewer moments - indeed, the assumption requires, essentially,
only the finiteness of the second moments. Part \textit{(i)} of the
assumption is\ a technical condition which we require to be able to use the
results in \citet{jing2008towards}, and it essentially requires that the
distribution function of $e_{t,i_{1},...,i_{K}}|\left\{ \mathcal{F}%
_{0,t},1\leq t\leq T\right\} $ has a common support covering an open
neighbourhood of the origin.

\subsection{Rates of convergence\label{rates}}

We now report the main results of this paper, namely the consistency (up to
a linear transformation) of the estimated common factors and loadings
spaces. We will consider, in particular, estimators based on minimising
quadratic loss functions discussed in Section \ref{ls}, defined as%
\begin{equation*}
\widehat{\theta }=\arg \min_{\theta \in \Theta }\frac{1}{T}%
\sum_{t=1}^{T}\left\Vert \mathcal{X}_{t}-\mathcal{F}_{t}\times _{k=1}^{K}%
\mathbf{A}_{k}\right\Vert _{F}^{2},
\end{equation*}%
and estimators based on minimising the Huber loss function discussed in
Section \ref{huber}, viz.%
\begin{equation*}
\widehat{\theta }^{H}=\arg \min_{\theta \in \Theta }\frac{1}{T}%
\sum_{t=1}^{T}H_{\tau }\left( \left\Vert \mathcal{X}_{t}-\mathcal{F}%
_{t}\times _{k=1}^{K}\mathbf{A}_{k}\right\Vert _{F}\right) .
\end{equation*}%
In addition to the factors $\left\{ \mathcal{F}_{t}\right\} _{t=1}^{T}$ and
the loadings $\left\{ \mathbf{A}_{k},1\leq k\leq K\right\} $, we also
estimate the common components $\mathcal{S}_{0t}=\mathcal{F}_{0t}\times
_{k=1}^{K}\mathbf{A}_{0k}$.

We begin with the rates of convergence of $\widehat{\theta }$. Define the
following quantities%
\begin{equation*}
\widehat{\mathbf{S}}_{k}=\text{sgn}\left( \Ab_{0k}^\top\widehat\Ab%
_k/p_k\right) ,
\end{equation*}%
and
\begin{equation}
L=\min \left\{ p,Tp_{-1},...,Tp_{-K}\right\} .  \label{elle}
\end{equation}

\begin{theorem}
\label{ls-theorem}{\it We assume that Assumptions \ref{as-1}-\ref{as-3} are
satisfied. Then, for all $1 \leq k \leq K$, it holds that%
\begin{eqnarray*}
\frac{1}{p_{k}}\left\Vert \widehat{\mathbf{A}}_{k}-\mathbf{A}_{0k}\widehat{%
\mathbf{S}}_{k}\right\Vert _{F}^{2} &=&O_{P}\left( \frac{1}{L}\right) , \\
\frac{1}{T}\sum_{t=1}^{T}\left\Vert \widehat{\mathcal{F}}_{t}-%
\mathcal{F}_{0t}\times _{k=1}^{K}\widehat{\mathbf{S}}_{k} \right\Vert
_{F}^{2} &=&O_{P}\left( \frac{1}{L}\right) .
\end{eqnarray*}%
}
\end{theorem}

The following result states the consistency of the estimator of $\cS_{0t}$
defined as $\widehat\cS_t=\widehat\cF_t\times_{k=1}^K\widehat\Ab_k$.

\begin{corollary}
	\label{ls-cc}{\it We assume that Assumptions \ref{as-1}-\ref{as-3} are
	satisfied. Then, for all $1 \leq k \leq K$, it holds that $\sum_{t=1}^{T}\left\Vert\widehat{\mathcal{S}}_{t}-%
		\mathcal{S}_{0t}\right\Vert
		_{F}^{2}/(Tp)=O_P(1/L)$.%
	}
\end{corollary}

The results in Theorem \ref{ls-theorem} and Corollary \ref{ls-cc} state the $%
L_{2}$-norm consistency of the Least Squares estimator; under sufficiently
many moments, the estimators are consistent at rate $O_{P}\left(
L^{-1}\right) $, modulo the sign indeterminacy indicated by the presence of $%
\widehat{\mathbf{S}}_{k}$. We would like to point out that this result is
different - albeit obviously related - compared to the theory in %
\citet{han2020tensor} and \citet{he2022statistical}: the rates in Theorem %
\ref{ls-theorem} are a joint convergence result, whereas the estimators in %
\citet{han2020tensor} and \citet{he2022statistical} focus on individual
convergence rates.

We now turn to studying $\widehat{\theta }^{H}$. Let%
\begin{equation*}
\widehat{\mathbf{S}}_{k}^{H}=\text{sgn}\left( \Ab_{0k}^{\top }\widehat{\Ab}%
_{k}^{H}/p_{k}\right) ,
\end{equation*}%
and further define%
\begin{eqnarray}
L^{\ast } &=&\min \left\{
p,p_{-1}^{2},...,p_{-K}^{2},Tp_{-1},...,Tp_{-K}\right\} ,  \label{elle-star}
\\
L^{\ast \ast } &=&\min \left\{ p_{-1},...,p_{-K}\right\} ,
\label{elle-star2}
\end{eqnarray}%
and the common component estimator $\widehat{\cS}_{t}^{H}=\widehat{\cF}%
_{t}^{H}\times _{k=1}^{K}\widehat{\Ab}_{k}^{H}$.

\begin{theorem}
	\label{huber-theorem-1}{\it We assume that Assumptions \ref{as-1}-\ref{as-3} are
satisfied. Then, for all $1 \leq k \leq K$, it holds that%
\begin{eqnarray*}
\frac{1}{p_{k}}\left\Vert \widehat{\mathbf{A}}_{k}^{H}-\mathbf{A}_{0k}%
\widehat{\mathbf{S}}_{k}^{H}\right\Vert _{F}^{2} &=&O_{P}\left( \frac{1}{%
L^{\ast }}\right) , \\
\frac{1}{T}\sum_{t=1}^{T}\left\Vert \widehat{\mathcal{F}}_{t}^{H}-%
\mathcal{F}_{0t}\times _{k=1}^{K}\widehat{\mathbf{S}}_{k}^{H}
\right\Vert _{F}^{2} &=&O_{P}\left( \frac{1}{L^{\ast }}\right) .
\end{eqnarray*}%
}
\end{theorem}

\begin{corollary}
	\label{huber-cc-1}{\it We assume that Assumptions \ref{as-1}-\ref{as-3} are
	satisfied. Then, for all $1 \leq k \leq K$, it holds that $\sum_{t=1}^{T}\left\Vert\widehat{\mathcal{S}}_{t}^H-%
		\mathcal{S}_{0t}\right\Vert
		_{F}^{2}/(Tp)=O_{P}\left( 1/L^{\ast }\right)$.%
}
\end{corollary}

Theorem \ref{huber-theorem-1} and Corollary \ref{huber-theorem-1} are the
counterparts to Theorem \ref{ls-theorem} and Corollary \ref{ls-cc}
respectively, and are based exactly on the same assumptions. The rates are
also virtually the same, at least upon excluding the case of
\textquotedblleft very small\textquotedblright\ cross-sectional dimensions,
where the Least Squares estimator has a faster rate of convergence.

\begin{theorem}
	\label{huber-theorem-2}{\it We assume that Assumptions \ref{as-1}, \ref{as-2} and %
	\ref{as-4} are satisfied. Then, for all $1\leq k\leq K$, it holds that
	
	- if $\epsilon $ in Assumption \ref{as-4} satisfies $\epsilon \geq 2$%
	\begin{eqnarray*}
		\frac{1}{p_{k}}\left\Vert \widehat{\mathbf{A}}_{k}^{H}-\mathbf{A}_{0k}%
		\widehat{\mathbf{S}}_{k}^{H}\right\Vert _{F}^{2} &=&O_{P}\left( \frac{1}{%
			L^{\ast }}\right) , \\
		\frac{1}{T}\sum_{t=1}^{T}\left\Vert \widehat{\mathcal{F}}_{t}^{H}-%
		\mathcal{F}_{0t}\times _{k=1}^{K}\widehat{\mathbf{S}}_{k}^{H}
		\right\Vert _{F}^{2} &=&O_{P}\left( \frac{1}{L^{\ast }}\right) ;
	\end{eqnarray*}%
	
	- if $\epsilon $ in Assumption \ref{as-4} satisfies $0\leq \epsilon <2$%
	\begin{eqnarray*}
		\frac{1}{p_{k}}\left\Vert \widehat{\mathbf{A}}_{k}^{H}-\mathbf{A}_{0k}%
		\widehat{\mathbf{S}}_{k}^{H}\right\Vert _{F}^{2} &=&O_{P}\left( \frac{1}{%
			L^{\ast \ast }}\right) , \\
		\frac{1}{T}\sum_{t=1}^{T}\left\Vert \widehat{\mathcal{F}}_{t}^{H}-%
		\mathcal{F}_{0t}\times _{k=1}^{K}\widehat{\mathbf{S}}_{k}^{H}
		\right\Vert _{F}^{2} &=&O_{P}\left( \frac{1}{L^{\ast \ast }}\right) .
	\end{eqnarray*}%
}
\end{theorem}

\begin{corollary}
\label{huber-cc-2}{\it We assume that Assumptions \ref{as-1}, \ref{as-2} and %
\ref{as-4} are satisfied. Then it holds that

- if $\epsilon $ in Assumption \ref{as-4} satisfies $\epsilon \geq 2$, $\sum_{t=1}^{T}\left\Vert\widehat{\mathcal{S}}_{t}^H-%
	\mathcal{S}_{0t}\right\Vert
	_{F}^{2}/(Tp)=O_{P}(1/L^{\ast })$;%

- if $\epsilon $ in Assumption \ref{as-4} satisfies $0\leq \epsilon <2$, $\sum_{t=1}^{T}\left\Vert\widehat{\mathcal{S}}_{t}^H-%
	\mathcal{S}_{0t}\right\Vert
	_{F}^{2}/(Tp)=O_{P}(1/L^{\ast \ast })$.%
}
\end{corollary}

According to Theorem \ref{huber-theorem-2} and Corollary \ref{huber-cc-2},
the rates of convergence of $\widehat{\theta }^{H}$ and the corresponding
common components in the presence of heavy tails are the same as under
sub-Gaussian tails if at least the first four moments exist. If only fewer
moments exist, the Huber loss estimator has slower rates of convergence
because the signal contained in $\left\{ \mathcal{F}_{t}\right\} _{t=1}^{T}$
is drowned out by the heavy tails of the noise $\mathcal{E}_{t}$, but it is
still consistent, which indicates the robustness of the Huber estimator.

\subsection{Asymptotic properties of the estimators of the numbers of
factors \label{nfactors}}

In order to make the infeasible estimator defined in (\ref{equ:2}) feasible,
define $\omega _{k}=\min \left\{ p_{k}T,p_{-k}^{2},L\right\} $ and
\begin{equation*}
\widehat{r}_{k}^{\text{iPE-ER}}=\argmax_{j\leq r_{\max }}\frac{\lambda
_{j}\left( \widehat{\mathbf{M}}_{k}\right) }{\lambda _{j+1}\left( \widehat{%
\mathbf{M}}_{k}\right) +c\omega _{k}^{-1/2}},
\end{equation*}%
and%
\begin{equation*}
\widehat{r}_{k}^{\text{RTFA-ER}}=\argmax_{j\leq r_{\max }}\frac{\lambda
_{j}\left( \widehat{\mathbf{M}}_{k}^{H}\right) }{\lambda _{j+1}\left(
\widehat{\mathbf{M}}_{k}^{H}\right) +c\widetilde{L}^{-1/2}},
\end{equation*}%
where $\widetilde{L}=L^{\ast }I\left( \epsilon \geq 2\right) +L^{\ast \ast
}I\left( 0<\epsilon <2\right) $ - i.e., $\widetilde{L}=L^{\ast }$ or $%
L^{\ast \ast }$ according as $\epsilon \geq 2$\ or not, and $0<c<\infty $ is
a user-chosen quantity whose tuning we defer to Section \ref{sec:4.4}. The
extra terms in $\widehat{r}_{k}^{\text{iPE-ER}}$ and $\widehat{r}_{k}^{\text{%
RTFA-ER}}$\ are based on Lemmas \ref{er-ls} and \ref{er-h}; indeed, based on
those lemmas, they may not be the sharpest bounds, but they suffice for our
purposes.

The following theorems stipulate that $\widehat{r}_{k}^{\text{iPE-ER}}$ and $%
\widehat{r}_{k}^{\text{RTFA-ER}}$ are consistent estimators of $r_{k}$.

\begin{theorem}
\label{ls-fn}\it{We assume that Assumptions \ref{as-1}-\ref{as-3} hold. Then, for all $1\leq k\leq K$, it
holds that%
\begin{equation*}
\lim_{\min \left\{ T,p_{1},...,p_{K}\right\} \rightarrow \infty }\mathbb{P}%
\left( \widehat{r}_{k}^{\text{iPE-ER}}=r_{k}\right) =1,
\end{equation*}%
 }
\end{theorem}

\begin{theorem}
\label{huber-fn}\it{We assume that Assumptions \ref{as-1}, \ref{as-2} and \ref%
{as-4} hold, and that $\left\Vert \mathbf{E}_{k,0}\right\Vert
_{F}$ is identically distributed across $t$ with bounded density, with $P\left( \left\Vert \mathbf{E}_{k,0}\right\Vert
_{F}/\sqrt{p}\leq \tau \right) >0$. Then, for all $1\leq k\leq K$, it holds that%
\begin{equation*}
\lim_{\min \left\{ T,p_{1},...,p_{K}\right\} \rightarrow \infty }\mathbb{P}%
\left( \widehat{r}_{k}^{\text{RTFA-ER}}=r_{k}\right) =1.
\end{equation*}%
}
\end{theorem}

The two theorems state that both $\widehat{r}_{k}^{\text{iPE-ER}}$ and $%
\widehat{r}_{k}^{\text{RTFA-ER}}$ are consistent estimators of $r_{k}$;
consistency holds as long as $\min \left\{ T,p_{1},\ldots ,p_{K}\right\}
\rightarrow \infty $, with no restrictions needed on the relative rates of
divergence of the sample sizes. As far as Theorem \ref{huber-fn} is concerned,
the assumption that the idiosyncratic component is identically distributed
is done merely for simplicity, and extensions to the case(s) of
heterogeneity and serial dependence are possible - in essence, by making
appeal to a suitable Law of Large Numbers.

\section{Simulations\label{sec:4}}

In this section, we investigate the finite sample performance of the
estimators discussed in the previous sections. In particular, we focus on
the Robust Tensor Factor Analysis (RTFA), and evaluate its performance at
estimating the loading matrices, the common components and the number of
factors, against several alternative estimators: the Initial Estimator (IE)
and the Projection Estimator (PE) of \citet{he2022statistical}, the
Iterative Projected mode-wise PCA estimation (IPmoPCA) of \cite%
{zhang2022tucker}, and the Time series Outer-Product Unfolding Procedure
(TOPUP) and Time series Inner-Product Unfolding Procedure (TIPUP) with their
iteration versions (iTOPUP and iTIPUP) proposed by \cite{Chen2021Factor}. In
the following, all iterative procedures are initialised using the the
Initial Estimator (IE) of \citet{he2022statistical}.

\subsection{Data generation\label{sec:4.1}}

The tensor observations are generated following the 3-order tensor factor
model:
\begin{equation*}
\cX_{t}=\cF_{t}\times _{1}\Ab_{1}\times _{2}\Ab_{2}\times _{3}\Ab_{3}+\cE%
_{t}.
\end{equation*}%
We set $r_{1}=r_{2}=r_{3}=3$, draw the entries of $\Ab_{1}$, $\Ab_{2}$ and $%
\Ab_{3}$ independently from uniform distribution $\cU(-1,1)$, and let
\begin{equation}
\begin{gathered} \operatorname{Vec}\left(\cF_{t}\right)=\phi \times
\operatorname{Vec}\left(\cF_{t-1}\right)+\sqrt{1-\phi^{2}} \times
\epsilon_{t}, \epsilon_{t} \stackrel{i . i . d}{\sim} \cN\left(\mathbf{0},
\Ib_{r_1r_2r_3}\right), \\ \operatorname{Vec}\left(\cE_{t}\right)=\psi
\times \operatorname{Vec}\left(\cE_{t-1}\right)+\sqrt{1-\psi^{2}} \times
\operatorname{Vec}\left(\cU_{t}\right), \end{gathered}  \label{dgp}
\end{equation}%
where we consider two scenarios: in the first one, $\cU_{t}$ is drawn from a
tensor normal distribution, and in the second one it is drawn from a tensor-$%
t$ distribution. If $\cU_{t}$ is from a tensor normal distribution $\cT\cN(%
\cM,\bSigma_{1},\bSigma_{2},\bSigma_{3})$, then $\func{Vec}(\cU_{t})~{\sim }~%
\cN(\func{Vec}(\cM),\bSigma_{3}\otimes \bSigma_{2}\otimes \bSigma_{1})$; on
the other hand, if $\cU_{t}$ is from a tensor-$t$ distribution $t_{\nu }(\cM,%
\bSigma_{1},\bSigma_{2},\bSigma_{3})$, then $\func{Vec}(\cU_{t})$ is drawn
from a multivariate $t$ distribution $t_{\nu }(\cM,\bSigma_{3}\otimes \bSigma%
_{2}\otimes \bSigma_{1})$. We further set: $\cM=0$; $\bSigma_{k}$ to be a
matrix with ones on its main diagonal, and $1/p_{k}$ in the off-diagonal
entries, for $k=1,2,3$. The parameters $\phi $ and $\psi $ in (\ref{dgp})
control for temporal correlations of $\cF_{t}$ and $\cE_{t}$. By setting $%
\phi $ and $\psi $ unequal to zero, common factors have cross-correlations,
and idiosyncratic noises have both cross-correlations and weak
autocorrelations. For RTFA, the Huber loss threshold parameter $\tau $ is
set as the median of $\Vert \cX_{t}-\cF_{t}\times _{1}\widehat{\Ab}%
_{1}\times _{2}\widehat{\Ab}_{2}\times _{3}\widehat{\Ab}_{3}\Vert _{F}$,
where $\widehat{\Ab}_{1}$, $\widehat{\Ab}_{2}$ and $\widehat{\Ab}_{3}$ are
the initial estimators of $\Ab_{1}$, $\Ab_{2}$ and $\Ab_{3}$.

Finally, we note that in Sections \ref{sec:4.2} and \ref{sec:4.3} it is
assumed that factor numbers are known; the performance of our estimators for
the number of factors is investigated in Section \ref{sec:4.4}. All the
following simulation results are based on $1,000$ repetitions.

\begin{table}[!h]
\caption{Averaged estimation error and standard deviations of loading spaces
for Settings A, B and C under Tensor Normal distribution over 1000
replications. }
\label{tab:1}\renewcommand{\arraystretch}{0.5} \centering
\scalebox{1}{\begin{tabular}{ccccccccc}
			\toprule[2pt]
			Evaluation	&	$p_1$	&	$p_2$	&	$p_3$	&	$T$	&	RTFA	&	IPmoPCA	&	iTOPUP	&	iTIPUP	\\
			\midrule
			$\cD(\widehat{\Ab}_1,\Ab_1)$	&	10	&	10	&	10	&	20	&	0.0444(0.01714)	&	0.0444(0.01716)	&	0.0677(0.02622)	&	0.3564(0.12453)	\\
			&		&		&		&	50	&	0.0288(0.01154)	&	0.0288(0.01155)	&	0.0544(0.01850)	&	0.3039(0.12303)	\\
			&		&		&		&	100	&	0.0220(0.00902)	&	0.0219(0.00899)	&	0.0502(0.01870)	&	0.2357(0.10816)	\\
			&		&		&		&	200	&	0.0176(0.00840)	&	0.0175(0.00838)	&	0.0477(0.01866)	&	0.1519(0.06752)	\\
			&	100	&	10	&	10	&	20	&	0.0404(0.00606)	&	0.0404(0.00606)	&	0.0584(0.00876)	&	0.3495(0.09297)	\\
			&		&		&		&	50	&	0.0253(0.00369)	&	0.0253(0.00369)	&	0.0476(0.00694)	&	0.2941(0.09262)	\\
			&		&		&		&	100	&	0.0126(0.00179)	&	0.0126(0.00179)	&	0.0435(0.00626)	&	0.2350(0.08630)	\\
			&		&		&		&	200	&	0.0126(0.00188)	&	0.0126(0.00188)	&	0.0405(0.00611)	&	0.1548(0.05279)	\\
			&	20	&	20	&	20	&	20	&	0.0199(0.00328)	&	0.0199(0.00328)	&	0.0294(0.00546)	&	0.2551(0.11100)	\\
			&		&		&		&	50	&	0.0125(0.00206)	&	0.0125(0.00206)	&	0.0237(0.00425)	&	0.1927(0.09755)	\\
			&		&		&		&	100	&	0.0088(0.00145)	&	0.0088(0.00145)	&	0.0218(0.00391)	&	0.1363(0.08070)	\\
			&		&		&		&	200	&	0.0063(0.00108)	&	0.0063(0.00108)	&	0.0204(0.00364)	&	0.0758(0.03423)	\\
			$\cD(\widehat{\Ab}_2,\Ab_2)$	&	10	&	10	&	10	&	20	&	0.0430(0.01496)	&	0.0430(0.01497)	&	0.0647(0.02183)	&	0.3456(0.12170)	\\
			&		&		&		&	50	&	0.0289(0.01233)	&	0.0289(0.01234)	&	0.0554(0.02402)	&	0.3066(0.12251)	\\
			&		&		&		&	100	&	0.0221(0.01050)	&	0.0221(0.01049)	&	0.0505(0.01959)	&	0.2344(0.11219)	\\
			&		&		&		&	200	&	0.0177(0.00959)	&	0.0177(0.00958)	&	0.0483(0.01974)	&	0.1530(0.06836)	\\
			&	100	&	10	&	10	&	20	&	0.0124(0.00344)	&	0.0124(0.00344)	&	0.0192(0.00572)	&	0.2039(0.11585)	\\
			&		&		&		&	50	&	0.0079(0.00217)	&	0.0079(0.00217)	&	0.0159(0.00494)	&	0.1510(0.10617)	\\
			&		&		&		&	100	&	0.0040(0.00110)	&	0.0040(0.00110)	&	0.0148(0.00461)	&	0.0958(0.07905)	\\
			&		&		&		&	200	&	0.0041(0.00122)	&	0.0041(0.00122)	&	0.0134(0.00393)	&	0.0492(0.02894)	\\
			&	20	&	20	&	20	&	20	&	0.0198(0.00334)	&	0.0198(0.00334)	&	0.0293(0.00537)	&	0.2551(0.11190)	\\
			&		&		&		&	50	&	0.0125(0.00210)	&	0.0125(0.00210)	&	0.0237(0.00420)	&	0.2007(0.10736)	\\
			&		&		&		&	100	&	0.0088(0.00140)	&	0.0088(0.00140)	&	0.0220(0.00399)	&	0.1375(0.08115)	\\
			&		&		&		&	200	&	0.0063(0.00105)	&	0.0063(0.00106)	&	0.0205(0.00384)	&	0.0783(0.03367)	\\
			$\cD(\widehat{\Ab}_3,\Ab_3)$	&	10	&	10	&	10	&	20	&	0.0440(0.01639)	&	0.0440(0.01639)	&	0.0654(0.02285)	&	0.3539(0.12488)	\\
			&		&		&		&	50	&	0.0293(0.01573)	&	0.0293(0.01572)	&	0.0550(0.01986)	&	0.3049(0.12164)	\\
			&		&		&		&	100	&	0.0222(0.01003)	&	0.0222(0.01005)	&	0.0507(0.01943)	&	0.2350(0.10515)	\\
			&		&		&		&	200	&	0.0176(0.01166)	&	0.0175(0.01168)	&	0.0478(0.02054)	&	0.1515(0.07578)	\\
			&	100	&	10	&	10	&	20	&	0.0124(0.00341)	&	0.0124(0.00341)	&	0.0189(0.00568)	&	0.2050(0.12253)	\\
			&		&		&		&	50	&	0.0079(0.00230)	&	0.0079(0.00230)	&	0.0158(0.00468)	&	0.1500(0.10853)	\\
			&		&		&		&	100	&	0.0041(0.00119)	&	0.0041(0.00119)	&	0.0147(0.00457)	&	0.0944(0.07445)	\\
			&		&		&		&	200	&	0.0040(0.00131)	&	0.0040(0.00131)	&	0.0138(0.00451)	&	0.0492(0.02638)	\\
			&	20	&	20	&	20	&	20	&	0.0197(0.00319)	&	0.0197(0.00319)	&	0.0295(0.00564)	&	0.2617(0.11524)	\\
			&		&		&		&	50	&	0.0125(0.00200)	&	0.0125(0.00200)	&	0.0238(0.00452)	&	0.1945(0.10286)	\\
			&		&		&		&	100	&	0.0088(0.00144)	&	0.0088(0.00144)	&	0.0219(0.00401)	&	0.1328(0.07698)	\\
			&		&		&		&	200	&	0.0063(0.00104)	&	0.0063(0.00104)	&	0.0205(0.00383)	&	0.0772(0.03516)	\\
			
			\bottomrule[2pt]
	\end{tabular}}
\end{table}

\begin{table}[!h]
\caption{Averaged estimation errors and standard deviations of loading
spaces for Settings A, B and C under Tensor $t_3$ distribution over 1000
replications. }
\label{tab:2}\renewcommand{\arraystretch}{0.5} \centering
\scalebox{1}{\begin{tabular}{ccccccccc}
			\toprule[2pt]
			Evaluation	&	$p_1$	&	$p_2$	&	$p_3$	&	$T$	&	RTFA	&	IPmoPCA	&	iTOPUP	&	iTIPUP	\\
			\midrule
			$\cD(\widehat{\Ab}_1,\Ab_1)$	&	10	&	10	&	10	&	20	&	0.0834(0.11570)	&	0.1663(0.19281)	&	0.2019(0.19822)	&	0.4492(0.15666)	\\
			&		&		&		&	50	&	0.0480(0.05786)	&	0.1317(0.17159)	&	0.1583(0.17561)	&	0.4185(0.16319)	\\
			&		&		&		&	100	&	0.0373(0.03886)	&	0.1080(0.15089)	&	0.1458(0.17116)	&	0.3591(0.17313)	\\
			&		&		&		&	200	&	0.0323(0.02847)	&	0.0910(0.13705)	&	0.1376(0.16456)	&	0.2633(0.16956)	\\
			&	100	&	10	&	10	&	20	&	0.0509(0.09346)	&	0.1043(0.17912)	&	0.1263(0.18374)	&	0.4186(0.15705)	\\
			&		&		&		&	50	&	0.0254(0.03927)	&	0.0690(0.15127)	&	0.1023(0.16260)	&	0.3703(0.15725)	\\
			&		&		&		&	100	&	0.0164(0.01805)	&	0.0540(0.13712)	&	0.0836(0.14229)	&	0.2853(0.14941)	\\
			&		&		&		&	200	&	0.0112(0.00225)	&	0.0388(0.11309)	&	0.0720(0.13155)	&	0.1842(0.11782)	\\
			&	20	&	20	&	20	&	20	&	0.0199(0.04163)	&	0.0455(0.11617)	&	0.0503(0.10430)	&	0.2859(0.14501)	\\
			&		&		&		&	50	&	0.0124(0.03068)	&	0.0354(0.11309)	&	0.0423(0.10085)	&	0.2349(0.14240)	\\
			&		&		&		&	100	&	0.0079(0.00187)	&	0.0204(0.07970)	&	0.0358(0.08455)	&	0.1576(0.11309)	\\
			&		&		&		&	200	&	0.0058(0.00139)	&	0.0156(0.06583)	&	0.0367(0.09113)	&	0.0975(0.08937)	\\
			$\cD(\widehat{\Ab}_2,\Ab_2)$	&	10	&	10	&	10	&	20	&	0.0831(0.12082)	&	0.1663(0.20363)	&	0.2010(0.20489)	&	0.4514(0.16433)	\\
			&		&		&		&	50	&	0.0487(0.06591)	&	0.1294(0.17971)	&	0.1562(0.17988)	&	0.4135(0.16926)	\\
			&		&		&		&	100	&	0.0380(0.04225)	&	0.1080(0.16396)	&	0.1424(0.17216)	&	0.3600(0.17577)	\\
			&		&		&		&	200	&	0.0344(0.03886)	&	0.0891(0.14623)	&	0.1401(0.17474)	&	0.2618(0.17082)	\\
			&	100	&	10	&	10	&	20	&	0.0257(0.09058)	&	0.0760(0.18092)	&	0.0820(0.17510)	&	0.2982(0.18181)	\\
			&		&		&		&	50	&	0.0106(0.04496)	&	0.0509(0.15435)	&	0.0665(0.15780)	&	0.2489(0.19057)	\\
			&		&		&		&	100	&	0.0062(0.01828)	&	0.0406(0.13901)	&	0.0533(0.14272)	&	0.1600(0.16268)	\\
			&		&		&		&	200	&	0.0043(0.00182)	&	0.0297(0.11620)	&	0.0441(0.13454)	&	0.0877(0.12974)	\\
			&	20	&	20	&	20	&	20	&	0.0201(0.04920)	&	0.0481(0.13532)	&	0.0497(0.10947)	&	0.2821(0.14975)	\\
			&		&		&		&	50	&	0.0124(0.03446)	&	0.0358(0.12025)	&	0.0435(0.11028)	&	0.2293(0.14320)	\\
			&		&		&		&	100	&	0.0078(0.00171)	&	0.0216(0.09336)	&	0.0373(0.09692)	&	0.1574(0.11931)	\\
			&		&		&		&	200	&	0.0057(0.00125)	&	0.0162(0.07606)	&	0.0378(0.10205)	&	0.0987(0.09844)	\\
			$\cD(\widehat{\Ab}_3,\Ab_3)$	&	10	&	10	&	10	&	20	&	0.0817(0.11664)	&	0.1667(0.20234)	&	0.1975(0.20149)	&	0.4490(0.16158)	\\
			&		&		&		&	50	&	0.0498(0.06719)	&	0.1288(0.17687)	&	0.1539(0.17948)	&	0.4233(0.16859)	\\
			&		&		&		&	100	&	0.0384(0.04332)	&	0.1040(0.15592)	&	0.1431(0.17293)	&	0.3598(0.17328)	\\
			&		&		&		&	200	&	0.0339(0.03448)	&	0.0907(0.14822)	&	0.1353(0.16943)	&	0.2593(0.16919)	\\
			&	100	&	10	&	10	&	20	&	0.0260(0.09286)	&	0.0738(0.17765)	&	0.0811(0.17356)	&	0.3002(0.17988)	\\
			&		&		&		&	50	&	0.0106(0.04682)	&	0.0488(0.14885)	&	0.0666(0.15908)	&	0.2387(0.18475)	\\
			&		&		&		&	100	&	0.0060(0.01838)	&	0.0385(0.13252)	&	0.0536(0.14379)	&	0.1658(0.16711)	\\
			&		&		&		&	200	&	0.0042(0.00201)	&	0.0294(0.11584)	&	0.0441(0.13170)	&	0.0881(0.12675)	\\
			&	20	&	20	&	20	&	20	&	0.0200(0.04897)	&	0.0478(0.13359)	&	0.0500(0.10906)	&	0.2904(0.1496)	\\
			&		&		&		&	50	&	0.0124(0.03369)	&	0.0362(0.12095)	&	0.0431(0.10735)	&	0.2391(0.14844)	\\
			&		&		&		&	100	&	0.0078(0.00168)	&	0.0212(0.09063)	&	0.0369(0.09449)	&	0.1583(0.12042)	\\
			&		&		&		&	200	&	0.0056(0.00123)	&	0.0157(0.07138)	&	0.0380(0.10159)	&	0.0988(0.09740)	\\
			
			\bottomrule[2pt]
	\end{tabular}}
\end{table}

\subsection{Verifying the convergence rates for loading spaces\label{sec:4.2}%
}

In this section, we compare the performance of RTFA, PE, IE, IPmoPCA, TOPUP,
TIPUP, iTOPUP and iTIPUP (with $h_{0}=1$), as far as estimating loading
spaces is concerned. We would like to note that preliminary, unreported
simulations indicate that, in all setting considered, the iterative
estimation methods perform better than the non-iterative ones; hence, in
order to save space, we only report results concerning RTFA, IPmoPCA, iTOPUP
and iTIPUP.

We consider the following three settings:

\vspace{0.5em}

\textbf{Setting A:~} $p_1=p_2=p_3=10,~\phi=0.1,~\psi=0.1,~T \in
(20,50,100,200)$

\vspace{0.5em}

\textbf{Setting B:~} $p_1=100,~p_2=p_3=10,~\phi=0.1,~\psi=0.1,~T \in
(20,50,100,200)$

\vspace{0.5em}

\textbf{Setting C:~} $p_1=p_2=p_3=20,~\phi=0.1,~\psi=0.1,~T \in
(20,50,100,200)$



\vspace{0.5em}

Due to the identifiability issue of factor model, the performance of the
candidates methods is evaluated by comparing the distance between the
estimated loading space and the true loading space, which is
\begin{equation*}
\cD(\widehat{\Ab}_k, \Ab_k)=\left(1-\frac{1}{r_k} \mathrm{{Tr}\left(\widehat{%
\Qb}_k\widehat{\Qb}_k^{\top} \Qb_k \Qb_k^{\top}\right)}\right)^{1 /
2},~k=1,2,3,
\end{equation*}
where $\Qb_k$ and $\widehat{\Qb}_k$ are the left singular-vector matrices of
the true loading matrix $\Ab_k$ and its estimator $\widehat{\Ab}_k$. The
distance is always in the interval $[0,1]$. When $\Ab_k$ and $\widehat{\Ab}%
_k $ span the same space, the distance $\cD(\widehat{\Ab}_k, \Ab_k)$ is
equal to 0, while is equal to 1 when the two spaces are orthogonal.

\bigskip

Table \ref{tab:1} shows the averaged estimation errors and standard
deviations under Settings A, B and C with idiosyncratic components following
tensor normal distribution. As expected, all these methods benefit from the
increase in the dimensions $p_{1}$, ..., $p_{K}$ and $T$. Compared with
iTOPUP and iTIPUP, the RTFA and the IPmoPCA estimators deliver a better
performance especially when the idiosyncratic components have (weak)
cross-sectional and serial correlation. We note that the RTFA and IPmoPCA
estimators behave in a similar way in the case of Gaussian idiosyncratic
components. This result is not surprising: the RTFA estimator is based on
the eigenstructure of the weighted sample covariance matrix of the projected
data, while IPmoPCA is based on the unweighted sample covariance matrix.
However, in the case of Gaussian idiosyncratic components, the weights used
in the RTFA estimator tend to be equal; hence, the RTFA and IPmoPCA
estimators have virtually the same performance (incidentally, both share the
benefits brought about from iterative projection). Conversely, when
considering the case where the idiosyncratic components follow the tensor $%
t_{3}$ distribution, Table \ref{tab:2} indicates that, although all methods
improve as $p_{1}$, ..., $p_{K}$ and $T$ increase, the RTFA estimator
clearly outperforms all the other methods, across all settings. In summary,
the distilled essence of our simulations is that when data have light tail,
the RTFA estimator and the IPmoPCA estimator perform similarly; conversely,
the RTFA\ estimator offers the best performance for data with heavy tails,
which indicates that RTFA is more robust and is adaptive to the tail
properties of data compared with other non-weighted projection procedures.

\begin{table}[!h]
\caption{Averaged estimation errors and and standard deviations of common
components for Settings A, B and C over 1000 replications. }
\label{tab:3}\renewcommand{\arraystretch}{0.5} \centering
\scalebox{1}{\begin{tabular}{ccccccccc}
			\toprule[2pt]
			Distribution & $p_1$	&	$p_2$	&	$p_3$	&	$T$  &  RTFA &	IPmoPCA	&	iTOPUP	&	iTIPUP	\\
			\midrule
			normal	&	10	&	10	&	10	&	20	&	0.0314(0.00417)	&	0.0314(0.00417)	&	0.0362(0.00499)	&	0.2877(0.11770)	\\
			&		&		&		&	50	&	0.0290(0.00345)	&	0.0290(0.00345)	&	0.0335(0.00417)	&	0.2399(0.12044)	\\
			&		&		&		&	100	&	0.0286(0.00353)	&	0.0286(0.00353)	&	0.0328(0.00412)	&	0.1611(0.09365)	\\
			&		&		&		&	200	&	0.0281(0.00332)	&	0.0281(0.00332)	&	0.0322(0.00392)	&	0.0835(0.04398)	\\
			&	100	&	10	&	10	&	20	&	0.0044(0.00047)	&	0.0044(0.00047)	&	0.0064(0.00079)	&	0.1902(0.09065)	\\
			&		&		&		&	50	&	0.0034(0.00033)	&	0.0034(0.00033)	&	0.0052(0.00063)	&	0.1339(0.08519)	\\
			&		&		&		&	100	&	0.0030(0.00028)	&	0.0030(0.00028)	&	0.0048(0.00056)	&	0.0791(0.06332)	\\
			&		&		&		&	200	&	0.0029(0.00028)	&	0.0029(0.00028)	&	0.0045(0.00050)	&	0.0316(0.02330)	\\
			&	20	&	20	&	20	&	20	&	0.0044(0.00034)	&	0.0044(0.00034)	&	0.0055(0.00049)	&	0.1826(0.09781)	\\
			&		&		&		&	50	&	0.0038(0.00028)	&	0.0038(0.00028)	&	0.0048(0.00040)	&	0.1195(0.08365)	\\
			&		&		&		&	100	&	0.0036(0.00025)	&	0.0036(0.00025)	&	0.0046(0.00034)	&	0.0626(0.05947)	\\
			&		&		&		&	200	&	0.0035(0.00023)	&	0.0035(0.00023)	&	0.0044(0.00033)	&	0.0207(0.01746)	\\
			
			\midrule
			$t_3$	&	10	&	10	&	10	&	20	&	0.2617(4.86213)	&	0.4230(4.88769)	&	0.4146(4.84075)	&	0.6991(4.82388)	\\
			&		&		&		&	50	&	0.0973(0.73454)	&	0.2186(0.93345)	&	0.2202(0.78275)	&	0.5236(0.87847)	\\
			&		&		&		&	100	&	0.0426(0.04327)	&	0.1315(0.38312)	&	0.1507(0.35565)	&	0.3754(0.37653)	\\
			&		&		&		&	200	&	0.0412(0.02450)	&	0.1251(0.36703)	&	0.1565(0.39378)	&	0.2541(0.38040)	\\
			&	100	&	10	&	10	&	20	&	0.0770(0.71436)	&	0.1902(0.86077)	&	0.1526(0.62744)	&	0.3926(0.78267)	\\
			&		&		&		&	50	&	0.0213(0.29108)	&	0.1270(0.58298)	&	0.1263(0.54295)	&	0.3202(0.53778)	\\
			&		&		&		&	100	&	0.0099(0.20305)	&	0.0856(0.49088)	&	0.0939(0.49068)	&	0.2050(0.49981)	\\
			&		&		&		&	200	&	0.0034(0.00294)	&	0.0769(0.56020)	&	0.0902(0.56730)	&	0.1234(0.56034)	\\
			&	20	&	20	&	20	&	20	&	0.0213(0.34443)	&	0.0740(0.46448)	&	0.0535(0.35076)	&	0.2612(0.41123)	\\
			&		&		&		&	50	&	0.0044(0.00415)	&	0.0742(0.69496)	&	0.0719(0.68009)	&	0.2208(0.67815)	\\
			&		&		&		&	100	&	0.0040(0.00185)	&	0.0349(0.26854)	&	0.0353(0.23290)	&	0.1081(0.25400)	\\
			&		&		&		&	200	&	0.0041(0.00245)	&	0.0311(0.25765)	&	0.0387(0.26898)	&	0.0581(0.26084)	\\
			
			\bottomrule[2pt]
	\end{tabular}}
\end{table}

\subsection{Estimation error for common components\label{sec:4.3}}

In this section, we compare the performance of several estimators at
estimating the common component $\cS_{t}$, under Setting A, B, C defined in
the previous set of simulations. As in the previous section, we have tried
the following estimators: RTFA, PE, IE, TOPUP, TIPUP, iTOPUP and iTIPUP
(with $h_{0}=1$); however, preliminary simulations clearly showed that
non-iterative methods perform much worse than iterative ones, and therefore
we report only the results with the RTFA, IPmoPCA, iTOPUP and iTIPUP
estimators. We use the metric of Mean Squared Error (MSE) to evaluate the
performance of different procedures, i.e.,
\begin{equation*}
\mathrm{MSE}=\frac{1}{Tp}\sum_{t=1}^{T}\left\Vert \widehat{\cS}_{t}-\cS%
_{t}\right\Vert _{F}^{2}.
\end{equation*}

Table \ref{tab:3} shows the averaged estimation errors and standard
deviations of estimating the common components under Settings A, B and C.
Similarly to the conclusions in the previous section, all methods benefit
from increasing $p_{1}$, ..., $p_{K}$ and $T$\ as expected, and the RTFA and
IPmoPCA estimators outperform the iTOPUP and iTIPUP ones. In particular,
RTFA performs comparably with IPmoPCA under Gaussian idiosyncratic
components, whereas it delivers a significantly better performance than
other methods when the idiosyncratic components follow a $t_{3}$
distribution.

\begin{table}[!h]
\caption{ The frequencies of exact estimation of the numbers of factors
under Settings A, C and D over 1000 replications.}
\label{tab:4}\renewcommand{\arraystretch}{0.5} \centering
\begin{tabular}{cccccccc}
\toprule[2pt] Distribution & $p_1=p_2=p_3$ & $T$ & RTFA-ER & iPE-ER & TCorTh
& iTOP-ER & iTIP-ER \\
\midrule normal & 10 & 20 & 0.826 & 0.829 & 0.000 & 0.776 & 0.046 \\
&  & 50 & 0.854 & 0.853 & 0.034 & 0.789 & 0.097 \\
&  & 100 & 0.858 & 0.859 & 0.139 & 0.811 & 0.302 \\
&  & 200 & 0.848 & 0.851 & 0.318 & 0.796 & 0.577 \\
& 20 & 20 & 0.997 & 0.997 & 0.628 & 1.000 & 0.215 \\
&  & 50 & 0.999 & 0.999 & 0.941 & 1.000 & 0.434 \\
&  & 100 & 1.000 & 1.000 & 0.990 & 1.000 & 0.766 \\
&  & 200 & 1.000 & 1.000 & 0.994 & 1.000 & 0.975 \\
& 30 & 20 & 1.000 & 1.000 & 0.988 & $\backslash$ & 0.321 \\
&  & 50 & 1.000 & 1.000 & 1.000 & $\backslash$ & 0.573 \\
&  & 100 & 1.000 & 1.000 & 1.000 & $\backslash$ & 0.839 \\
&  & 200 & 1.000 & 1.000 & 1.000 & $\backslash$ & 0.991 \\
\midrule $t_3$ & 10 & 20 & 0.520 & 0.381 & 0.003 & 0.332 & 0.030 \\
&  & 50 & 0.565 & 0.367 & 0.138 & 0.341 & 0.060 \\
&  & 100 & 0.557 & 0.394 & 0.340 & 0.361 & 0.106 \\
&  & 200 & 0.568 & 0.398 & 0.462 & 0.347 & 0.216 \\
& 20 & 20 & 0.967 & 0.841 & 0.523 & 0.825 & 0.184 \\
&  & 50 & 0.968 & 0.856 & 0.135 & 0.867 & 0.385 \\
&  & 100 & 0.986 & 0.918 & 0.026 & 0.903 & 0.688 \\
&  & 200 & 0.988 & 0.924 & 0.002 & 0.915 & 0.897 \\
& 30 & 20 & 0.941 & 0.890 & 0.273 & $\backslash$ & 0.333 \\
&  & 50 & 0.952 & 0.934 & 0.068 & $\backslash$ & 0.502 \\
&  & 100 & 0.987 & 0.942 & 0.006 & $\backslash$ & 0.811 \\
&  & 200 & 0.999 & 0.967 & 0.001 & $\backslash$ & 0.954 \\
\bottomrule[2pt] &  &  &  &  &  &  &
\end{tabular}%
\end{table}

\subsection{Estimating the numbers of factors\label{sec:4.4}}

We investigate the empirical performances of several methodologies to
estimate the number of common factors. In this set of simulations, we also
consider the setting

\vspace{0.5em}

\textbf{Setting D:~} $p_{1}=p_{2}=p_{3}=30,~\phi =0.1,~\psi =0.1,~T\in
(20,50,100,200)$

\vspace{0.5em}

We compare the performance of our iterative Projected Estimator based on the
eigenvalue ratio (iPE-ER), alongside the robust version (RTFA-ER), against
seveal competitors including: the Total mode-k Correlation Thresholding
(TCorTh) of \cite{lam2021rank}, and the methods based on both information
criteria and the Eigenvalue Ratio principle based on the TOPUP, TIPUP,
iTOPUP and iTIPUP estimators of \cite{Han2021Rank}, which we abbreviate as
TOP-IC, TIP-IC, iTOP-IC, iTIP-IC, TOP-ER, TIP-ER, iTOP-ER, and iTIP-ER
respectively with fixed $h_{0}=1$. Preliminary results showed that the
accuracy of methods based on information criteria is almost zero under all
settings, which indicates that information criteria in general are not
suitable to estimate the number of factors in our setting. Similarly, TOP-ER
and TIP-ER are uniformly dominated by their iterative counterparts, and
therefore we only present the results obtained using iTOP-ER and iTIP-ER. We
set $r_{\max }=8$ for all estimation procedures. As far as both iPE-ER and
RTFA-ER are concerned, inspired by \citet{hallinliska07} and \citet{ABC10}
we explored the stability of the estimated number of factors as the tuning
coefficient $c$ in (\ref{equ:2}) changes from $c=0$ to $c$ set equal to the
largest eigenvalues of $\widehat{\mathbf{M}}_{k}$ and $\widehat{\mathbf{M}}%
_{k}^{H}$. Estimates are entirely insensitive to $c$, which indicates that
our methodology is completely robust to $c$; hence, hereafter we report
results using $c=0$.

\vspace{0.5em}

Table \ref{tab:4} presents the frequencies of exact estimation over $1,000$
replications under Settings A, C and D. When the idiosyncratic component
follows a Gaussian distribution, the accuracy of all five methods improve as
$p_{1}$, ..., $p_{K}$ and $T$\ increase, as expected. Among them, iTIP-ER
delivers the worst performance, essentially due to the fact that the TIPUP
method is not suitable in the presence of serial correlation in the
idiosyncratic component. Similarly, TCorTh does not perform well when $p_{1}$%
, ..., $p_{K}$ and $T$ are small, but its performance becomes comparable
with that of the other methodologies as $p_{1}$, ..., $p_{K}$ and $T$
increase. The performances of RTFA-ER and iPE-ER are comparable, and they
are both better than those of iTIP-ER and TCorTh. On the other hand, the
accuracy of iTOP-ER is comparable with that of iPE-ER and RTFA-ER; this can
be expected, since iTOP-ER is based on the tensor outer product, and it uses
as much information as possible. It is worth noting, though, that iTOP-ER
requires much longer run time and larger storage space, since the TOPUP
procedure needs to deal with a large high-order tensor, such as $\mathbb{R}%
^{p_{1}\times p_{2}\times p_{3}\times p_{1}\times p_{2}\times p_{3}\times 1}$
while $h_{0}=1$. Therefore, the use of iTOP-ER may present some difficulties
in practice. Finally, as can be expected, when the idiosyncratic component
follows a heavy-tailed, $t_{3}$ distribution, the accuracy of RTFA-ER is
significantly higher than that of all other methods, which indicates that
RTFA-ER is a robust method for estimating the numbers of factors, especially
in the cases of heavy-tailed data.

\section{Empirical analysis\label{empirics}}

In this section, we study import/export data of a variety of commodities
across multiple countries, which is also analyzed in \cite{Chen2021Factor}.
This set of data is made up of the monthly total value (in US dollars) of
imports and exports of 15 kinds of commodities in 22 European and American
countries from January 2014 to December 2022, thus extending the sample used
by \cite[January 2010 to December 2016]{Chen2021Factor} to cover also the
period of COVID-19 pandemic. The 15 commodities are Animal \& Animal
Products (HS code 01-05), Vegetable Products (06-15), Foodstuffs (16-24),
Mineral Products (25-27), Chemicals \& Allied Industries (28-38), Plastics
\& Rubbers (39-40), Raw Hides, Skins, Leather \& Furs (41-43), Wood \& Wood
Products (44-49), Textiles (50-63), Footwear \& Headgear (64-67), Stone \&
Glass (68-71), Metals (72-83), Machinery \& Electrical (84-85),
Transportation (86-89), and Miscellaneous (90-97), and the 22 countries are
Belgium (BE), Bulgaria (BG), Canada (CA), Denmark (DK), Finland (FI), France
(FR), Germany (DE), Greece (GR), Hungary (HU), Iceland (IS), Ireland (IR),
Italy (IT), Mexico (MX), Norway (NO), Poland (PO), Portugal (PT), Spain
(ES), Sweden (SE), Switzerland (CH), Turkey (TR), United States of America
(US) and United Kingdom (GB). Notice that, by construction, each country's
import and export with itself is zero. We compute a three-month moving average to
reduce the impact of occasional transactions for large traded or unusual
shipping delays.

Hence, we ultimately get a $22\times 22\times 15$ three-way tensor time
series $\{\cX_{t}\}_{t=1}^{T}$ with length $T=106$, where the $(i,j,k)$-th
element of $\cX_{t}$ represents the three-month average of exports of the $i$%
-th country to the $j$-th country for the $k$-th good at the $t$-th month.
By absorbing the time dimension, we may stack $\cX_{t}$ into an four-way
tensor $\cY\in \RR^{p_{1}\times p_{2}\times p_{3}\times T}$, with time $t$
as the fourth mode, referred to as the time-mode. We assume that, the
tensors $\cX_{t}\in \RR^{p_{1}\times p_{2}\times p_{3}}$ have the following
factor structure
\begin{equation*}
\cX_{t}=\cF_{t}\times _{1}\Ab_{1}\times _{2}\Ab_{2}\times _{3}\Ab_{3}+\cE%
_{t},\quad t=1,\cdots ,106,\quad p_{1}=p_{2}=22,\quad p_{3}=15,
\end{equation*}%
where $\Ab_{k}\in \RR^{p_{k}\times r_{k}}$ are the loading matrices, $\cF%
_{t}\in \RR^{r_{1}\times r_{2}\times r_{3}}$ is the factor tensor with
factor numbers $r_{k}<p_{k}$, $k=1,2,3$, to be determined. Equivalently we
write $\cY=\cG\times _{1}\Ab_{1}\times _{2}\Ab_{2}\times _{3}\Ab_{3}+\cR$,
where $\cG$ and $\cR$ are obtained by stacking $\cF_{t}$ and $\cE_{t}$ with
time $t$ as the fourth mode.

\begin{figure}[t!]
\subfigure[]{
		\includegraphics[width=8cm,height=7cm]{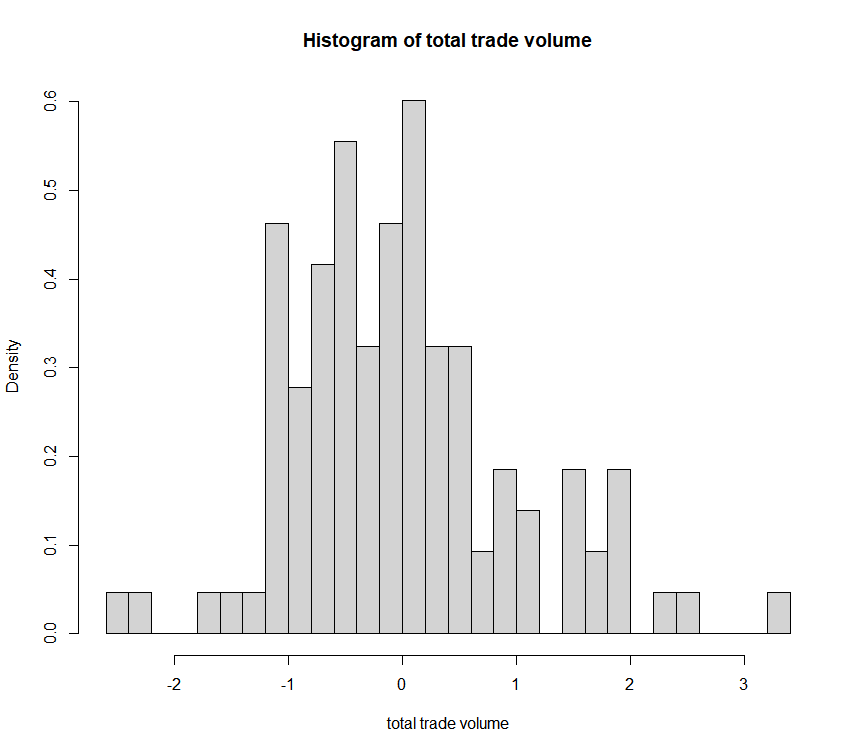} \label{fig:1(a)}
	} \hspace{2mm}
\subfigure[]{
		\includegraphics[width=8cm,height=7cm]{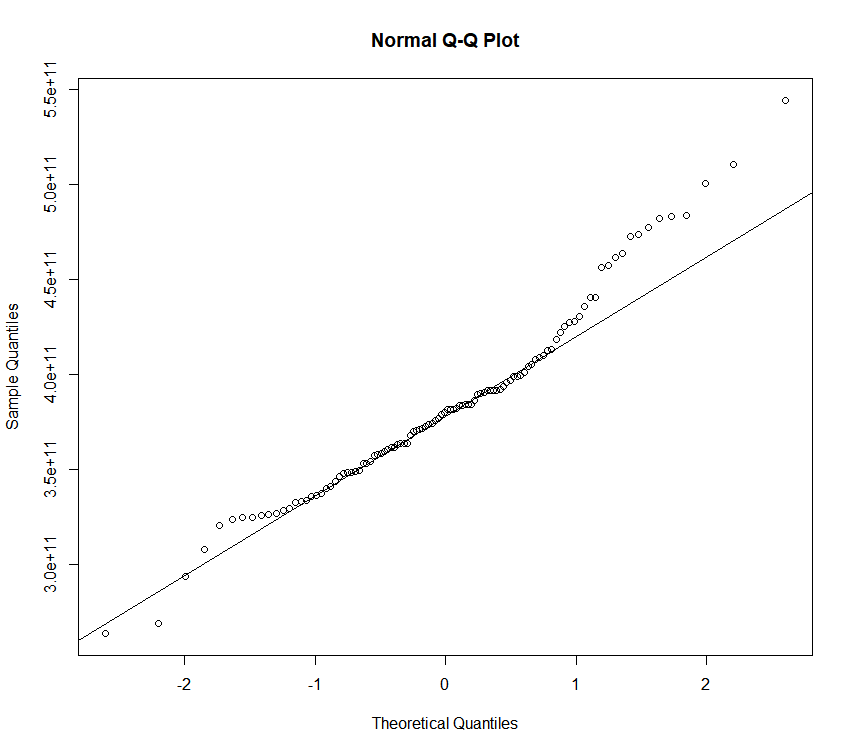}
	    \label{fig:1(b)}
    }
\caption{Histogram (a) and normal Q-Q plot (b) of total trade volume of 15
commodities between 22 countries at each time point.}
\end{figure}


\begin{figure}[t!]
\centering
\includegraphics[height=7cm,width=12.5cm]{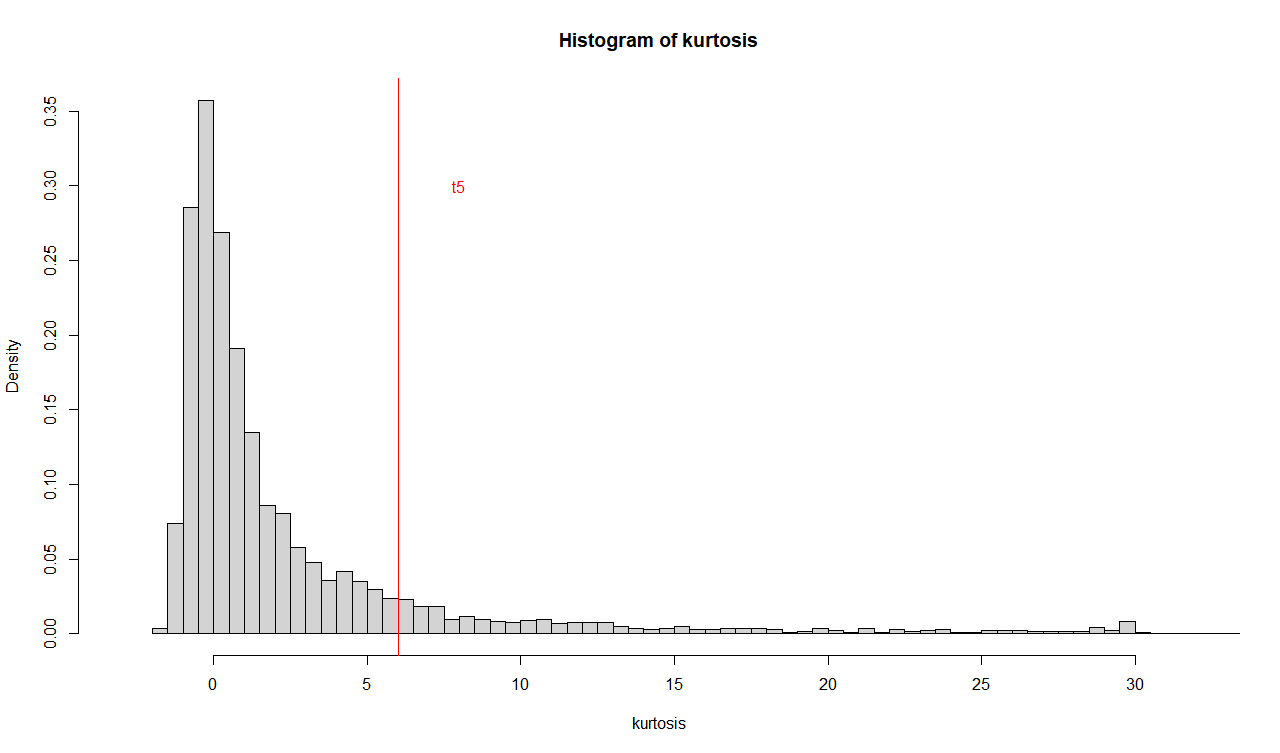}
\caption{Histogram of the sample kurtosis for $22\times22\times15$ trade
volumes and the red dashed line is the theoretical kurtosis of $t_5$
distribution.}
\label{fig:2}
\end{figure}

Figure \ref{fig:1(a)} and \ref{fig:1(b)} show the histogram and normal Q-Q
plot of the total trade volume at each time point $t$ by adding the volume
of trade between 22 countries for 15 commodities, respectively. It can be
seen that the total trade volume deviates from the normal distribution and
is skewed to the right. Figure \ref{fig:2} shows the histogram of sample
kurtosis of $22\times22\times15$ variables with respect to the trade volume
of 15 commodities among 22 countries, where the red vertical line represents
the theoretical kurtosis of the $t_5$ distribution. It suggests that this
dataset has heavy tails.

In order to assess the goodness of fit, we use the relative Mean Squared
Error (MSE) defined as
\begin{equation*}
\mathrm{MSE}=\frac{\sum_{t=1}^{T}\Vert \cX_{t}-\widehat{\cS}_{t}\Vert
_{F}^{2}}{\sum_{t=1}^{T}\Vert \cX_{t}\Vert _{F}^{2}},
\end{equation*}
where $\widehat{\cS}_{t}=\widehat\cF_t\times_1\widehat\Ab_1\times_2\widehat\Ab_2\times_3\widehat\Ab_3$
 is the estimated three-month average import-export factor driven tensor component of the data
tensor at $t$-th moment. When the MSE is small, this means that there is
strong evidence of comovements among the variables, thus indicating that the
factor model is an adequate representation of the data.
Table \ref{tab:5} shows the MSE of the three-month moving average estimated
by different methods given different combinations of factor numbers.
Specifically, we compare our RTFA method with the following estimators:
Initial Estimator (IE) and the
	Projection Estimator (PE) by \citet{he2022statistical}, Iterative Projected
	mode-wise PCA estimation (IPmoPCA) by \cite{zhang2022tucker}, and the Time series
	Inner-Product Unfolding Procedure (TIPUP) with their iteration versions
	(iTIPUP) by \cite{Chen2021Factor}.  It can be seen that the estimation
accuracy of the non-iterative method is inferior to that of the
corresponding iterative method, and the best one is the proposed RTFA method.

\begin{table}[t!]
\caption{ The MSE for the three-month moving average data. $r_1=r_2=r_3=r$
is the number of factors. }
\label{tab:5}\renewcommand{\arraystretch}{0.5} \centering
\begin{tabular}{ccccccc}
\toprule[2pt] r & IE & PE & IPmoPCA & RTFA & TIPUP & iTIPUP \\
\midrule 4 & 0.005052 & 0.003243 & 0.002973 & \textbf{0.002958} & 0.005055 &
0.002991 \\
5 & 0.004675 & 0.003061 & 0.002446 & \textbf{0.002403} & 0.004676 & 0.002445
\\
6 & 0.003885 & 0.002028 & 0.002025 & \textbf{0.001999} & 0.003734 & 0.002022
\\
\bottomrule[2pt] &  &  &  &  &  &
\end{tabular}%
\end{table}

\begin{table}[t!]
\caption{ The averaged MSE of rolling validation for the three-month moving
average data. $12n$ is the sample size of the training set. $r_1=r_2=r_3=r$
is the number of factors. }
\label{tab:6}\renewcommand{\arraystretch}{0.5} \centering
\begin{tabular}{cccccccc}
\toprule[2pt] $n$ & $r$ & IE & PE & IPmoPCA & RTFA & TIPUP & iTIPUP \\
\midrule 1 & 4 & 0.305648 & 0.234050 & 0.207419 & \textbf{0.207050} &
0.306730 & 0.218064 \\
2 & 4 & 0.306979 & 0.232142 & 0.190624 & \textbf{0.190185} & 0.307584 &
0.192854 \\
3 & 4 & 0.307025 & 0.217422 & 0.186855 & \textbf{0.186556} & 0.307332 &
0.187705 \\
1 & 5 & 0.261384 & 0.176364 & 0.172087 & \textbf{0.172043} & 0.258143 &
0.172437 \\
2 & 5 & 0.260488 & 0.177568 & 0.167676 & 0.167616 & 0.258174 & \textbf{%
0.166238} \\
3 & 5 & 0.259934 & 0.166453 & 0.149322 & \textbf{0.149246} & 0.258230 &
0.149851 \\
1 & 6 & 0.164834 & 0.114599 & 0.114138 & \textbf{0.113976} & 0.159747 &
0.114077 \\
2 & 6 & 0.148967 & 0.111139 & \textbf{0.110729 } & 0.110743 & 0.145592 &
0.110896 \\
3 & 6 & 0.163283 & \textbf{0.108324} & 0.108366 & 0.108390 & 0.159581 &
0.108561 \\
\bottomrule[2pt] &  &  &  &  &  &  &
\end{tabular}%
\end{table}

Then we use a rolling-validation
procedure as in \cite{wang2019factor} to futher compare these methods.
For
each year $t$ from 2017 to 2022, we repeatedly use the $n$ (bandwidth) year
observations prior to $t$ to fit the tensor factor model and estimate the
loading matrices. The estimated loading matrices are then used to estimate
the factor tensors and the corresponding residuals of the 12 months in the
current year. Specifically, let $\cX_{t}^{i}$ be the observed import-export tensor of month $i$ in year $t$ and $\widehat{\cS}_{t}^{i}=\widehat\cF_t^{i}\times_1\widehat\Ab_1^{i,n}\times_2\widehat\Ab_2^{i,n}\times_3\widehat\Ab_3^{i,n}$,
where $\{\Ab_k^{i,n}\}_{k=1}^K$ are the loading matrices estimated based on $n$ years of observations prior to month $i$ in year $t$, and further
define
\begin{equation*}
\mathrm{MSE}_{t}=\frac{\sum_{i=1}^{12}\left\Vert \cX_{t}^{i}-\widehat{\cS}_{t}^{i}\right\Vert _{F}^{2}}{\sum_{i=1}^{12}\left\Vert \cX_{t}^{i}\right\Vert _{F}^{2}}
\end{equation*}as the mean squared error.  Table \ref{tab:6} compares the
mean values of MSE of various estimation methods for different combinations
of bandwidth $n$ and factor number $r_{1}=r_{2}=r_{3}=r$. The estimation
errors of IPmoPCA, RTFA and iTIPUP methods are very close together, and the
results in Table \ref{tab:6} show that these three methods perform better
than the other methods. It is worth noting that our proposed RTFA performs
better than the other methods in almost all considered settings. This
indicates that RTFA is particularly suitable for capturing comovements in
this heavy-tailed data.

\begin{table}[!h]
\caption{ Estimated loading matrix $\widehat\Ab_3^\top$ for the commodity
mode.}
\label{tab:7}\renewcommand{\arraystretch}{0.5}
\resizebox{\textwidth}{!}{
		\centering
		\begin{tabular}{c|ccccccccccccccc}
			\toprule[2pt]
		Factor	&	Animal	&	Vegetable	&	Food	&	Mineral	&	Chemicals	&	Plastics	&	Leather	&	Wood	&	Textiles	&	Footwear	&	Stone	&	Metals	&	Machinery	&	Transport	&	Misc	\\
			\midrule
			1	&	0	&	3	&	-1	&	0	&	-1	&	1	&	0	&	-3	&	1	&	0	&	1	&	0	&	\textbf{29}	&	0	&	6	\\
			2	&	0	&	0	&	0	&	\textbf{30}	&	1	&	-1	&	0	&	3	&	-1	&	0	&	0	&	1	&	0	&	0	&	2	\\
			3	&	0	&	-3	&	1	&	-2	&	\textbf{29}	&	3	&	1	&	0	&	0	&	0	&	-2	&	3	&	0	&	0	&	6	\\
			4	&	0	&	-1	&	4	&	0	&	0	&	-2	&	0	&	2	&	0	&	0	&	-1	&	1	&	0	&	\textbf{30}	&	2	\\
			5	&	6	&	6	&	6	&	0	&	-2	&	\textbf{19}	&	0	&	6	&	4	&	1	&	0	&	17	&	1	&	0	&	-9	\\
			6	&	1	&	4	&	4	&	-1	&	2	&	-4	&	0	&	5	&	1	&	1	&	\textbf{29}	&	1	&	-1	&	0	&	2	\\
			\bottomrule[2pt]
	\end{tabular}}
\end{table}


\begin{figure}[htbp]
\centering
\includegraphics[width=8cm,height=7cm]{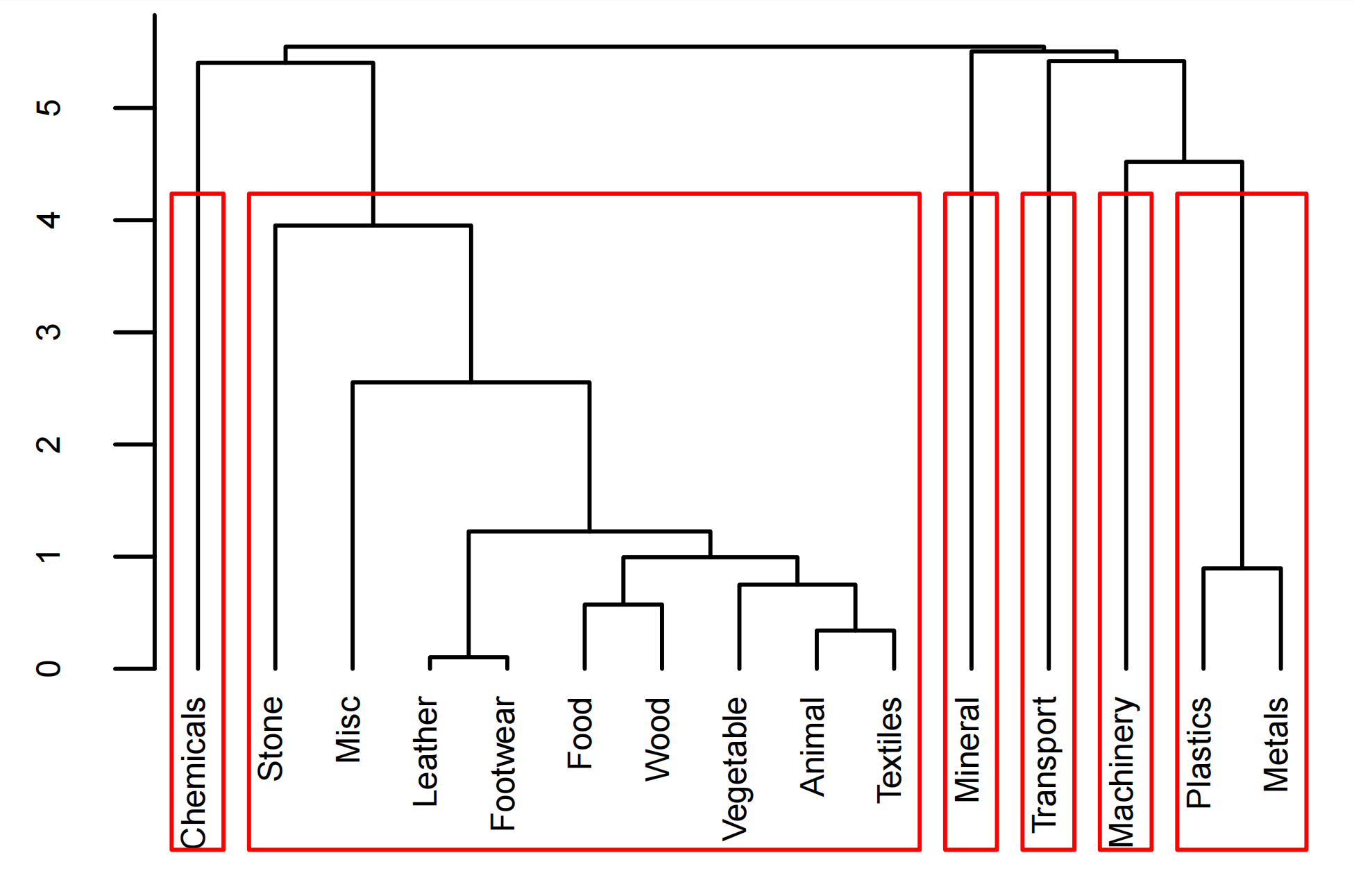}
\caption{Clustering of product categories by their loading coefficients.}
\label{fig:3}
\end{figure}

We then  use RTFA to further analyze the
export-import tensor time series data and, hereafter, we set $r_1=r_2=4,~r_3=6$ as  in \cite%
{Chen2021Factor}.

Table \ref{tab:7} shows the estimated loadings $\widehat{\Ab}_{3}^\top$ for
the commodity mode obtained using the RTFA procedure. For better
interpretation, the loading matrix is rotated using the varimax procedure
(see e.g. \citealp[Chapter 9.6]{mardia1979multivariate}) and all numbers are
multiplied by 30 and then truncated to integers for clearer viewing %
The sign of each row of $\widehat{\Ab}_{3}^\top$ (column of $\widehat{\Ab}_{3}$) is set in such a way that the largest entry (in boldface) is positive.
It can be seen from Table \ref{tab:7} that there is a group structure of
these six category factors. Hence, each factor is represents a
\textquotedblleft condensed product group" in \cite{Chen2021Factor}. Factor
1, 2, 3, 4 and 6 can be interpreted as Machinery \& Electrical factor,
Mineral Products factor, Chemicals \& Allied Industries factor,
Transportation factor and Stone \& Glass factor, because these factors are
mainly loaded on these commodity categories. Factor 5 can be viewed as a
mixing factor, with Plastics \& Rubbers and Metals as main load.
The clustering of the product categories according to
their loading vectors, the 15 columns of $\widehat\Ab_3^\top$ each of dimension 6, is shown in Figure \ref{fig:3}. The distance between two product categories $i$ and $j$ is defined by Euclidean distance, which is $d_{ij}^2=\sum_{k=1}^6\l((\widehat\Ab_3)_{ki}-(\widehat\Ab_3)_{kj}\r)^2/\widehat{\operatorname{Var}}(\widehat\Ab_3)_{k\cdot}=d_{ji}^2$, where $\widehat{\operatorname{Var}}(\widehat\Ab_3)_{k\cdot}$ is the sample variance of the $k$th row of $\widehat\Ab_3^\top$. In this study, we adopt  the ``complete linkage method"  {(see e.g. \citealp{Lawrence1974})} for clustering in this study.

As pointed out by \cite{Chen2021Factor}, each frontal slice of the factor $%
\cF_{t}$, which is equal to $\cG_{\cdot,\cdot,i_3,t}$, can be regarded as
the trade of $i_3$ commodity types between several virtual export hubs and
virtual import hubs. Each commodity transaction can be seen as first the
product is transported from the exporting country to an export hub, then
exported from the export hub to an import hub, and finally reaches the
importing country from the import hub. Each row of $\cG_{\cdot,\cdot,i_3,t}$
represents an exit hub and each column represents an import hub. The
corresponding estimated loading matrices $\widehat{\Ab}_1^\top$ and $%
\widehat{\Ab}_2^\top$ computed by the RTFA method, can reflect the trading
activities of each country through each export hub and import hub,
respectively. They are reported in Tables \ref{tab:8} and \ref{tab:9},
respectively. The scale and sign of the entries of this tables is set in the same way as for Table \ref{tab:7} above.

\begin{table}[!h]
\caption{ Estimated loading matrix $\widehat\Ab_1^\top$ for the export mode.}
\label{tab:8}\renewcommand{\arraystretch}{0.5}
\resizebox{\textwidth}{!}{
		\centering
		\begin{tabular}{c|cccccccccccccccccccccc}
			\toprule[2pt]
		Factor	&	BE	&	BU	&	CA	&	DK	&	FI	&	FR	&	DE	&	GR	&	HU	&	IS	&	IR	&	IT	&	MX	&	NO	&	PO	&	PT	&	ES	&	SE	&	CH	&	ER	&	US	&	GB	\\
			\midrule
			1	&	-1	&	0	&	0	&	0	&	0	&	1	&	3	&	0	&	0	&	0	&	-3	&	1	&	\textbf{30}	&	0	&	0	&	0	&	-1	&	0	&	-1	&	0	&	0	&	1	\\
			2	&	0	&	0	&	0	&	0	&	0	&	0	&	-2	&	0	&	1	&	0	&	-2	&	1	&	0	&	0	&	1	&	0	&	0	&	0	&	0	&	0	&	\textbf{30}	&	0	\\
			3	&	1	&	0	&	\textbf{30}	&	0	&	0	&	0	&	-1	&	0	&	0	&	0	&	2	&	0	&	0	&	1	&	0	&	0	&	0	&	0	&	1	&	0	&	0	&	1	\\
			4	&	7	&	0	&	0	&	2	&	1	&	9	&	\textbf{22}	&	0	&	2	&	0	&	9	&	8	&	-2	&	1	&	4	&	1	&	6	&	2	&	7	&	2	&	1	&	6	\\
			\bottomrule[2pt]
	\end{tabular}}
\end{table}

\begin{table}[!h]
\caption{ Estimated loading matrix $\widehat\Ab_2^\top$ for the import
fiber. }
\label{tab:9}\renewcommand{\arraystretch}{0.5}
\resizebox{\textwidth}{!}{
		\centering
		\begin{tabular}{c|cccccccccccccccccccccc}
			\toprule[2pt]
		Factor	&	BE	&	BU	&	CA	&	DK	&	FI	&	FR	&	DE	&	GR	&	HU	&	IS	&	IR	&	IT	&	MX	&	NO	&	PO	&	PT	&	ES	&	SE	&	CH	&	ER	&	US	&	GB	\\
			\midrule
			1	&	0	&	0	&	0	&	0	&	0	&	0	&	0	&	0	&	0	&	0	&	0	&	0	&	0	&	0	&	0	&	0	&	0	&	0	&	0	&	0	&	\textbf{30}	&	0	\\
			2	&	2	&	0	&	\textbf{28}	&	0	&	0	&	3	&	7	&	0	&	0	&	0	&	4	&	1	&	-2	&	1	&	1	&	0	&	1	&	0	&	5	&	1	&	0	&	1	\\
			3	&	7	&	1	&	-7	&	2	&	1	&	14	&	10	&	1	&	4	&	0	&	1	&	11	&	1	&	1	&	6	&	2	&	7	&	4	&	6	&	3	&	0	&	\textbf{14}	\\
			4	&	-2	&	0	&	3	&	0	&	0	&	-1	&	1	&	0	&	0	&	0	&	-2	&	-2	&	\textbf{29}	&	0	&	-1	&	0	&	0	&	0	&	-5	&	0	&	0	&	4	\\
			\bottomrule[2pt]
	\end{tabular}}
\end{table}

\begin{figure}[htbp]
\subfigure[]{
		\includegraphics[width=8cm,height=7cm]{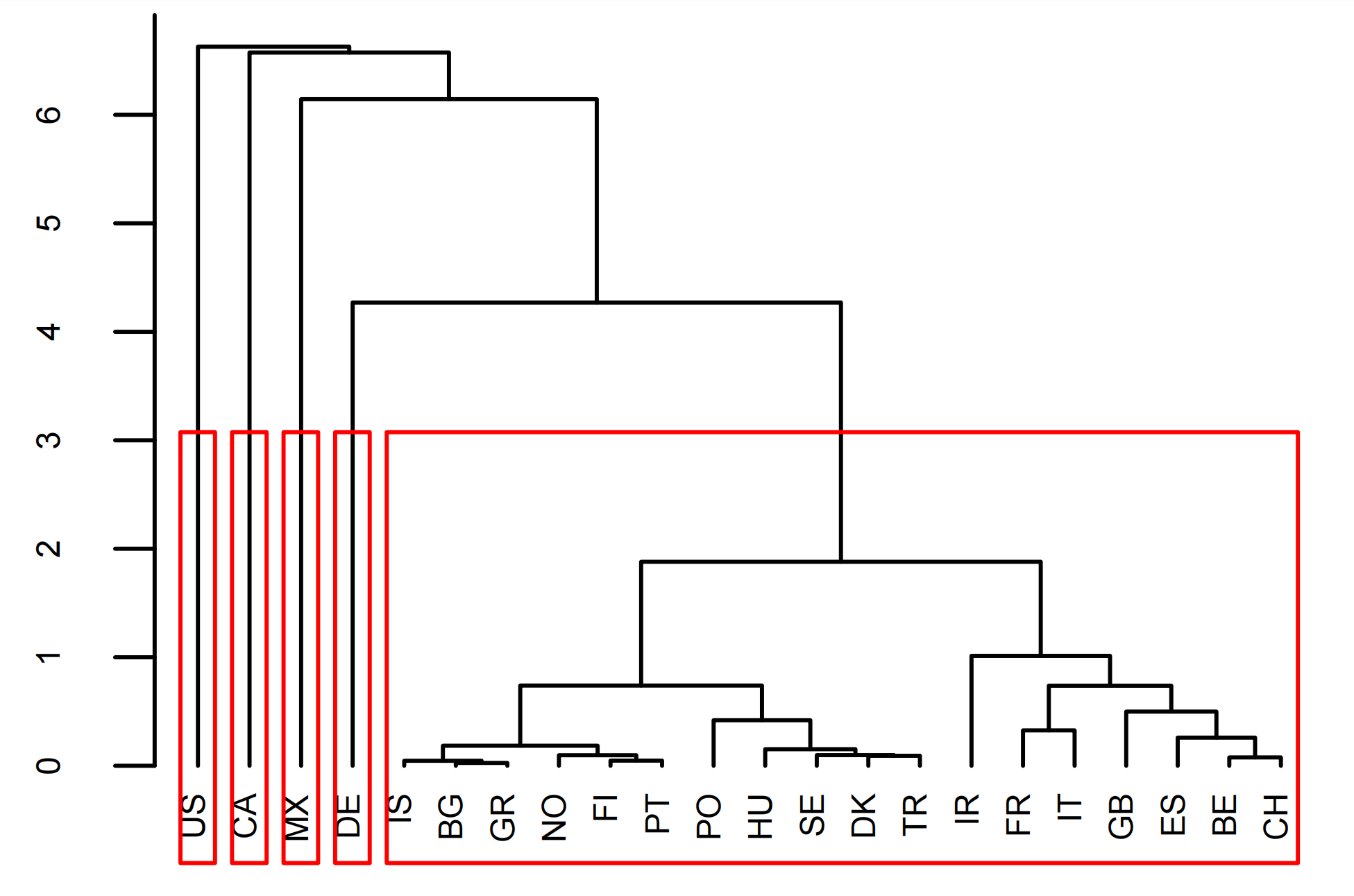} \label{fig:4(a)}
	} \hspace{2mm}
\subfigure[]{
		\includegraphics[width=8cm,height=7cm]{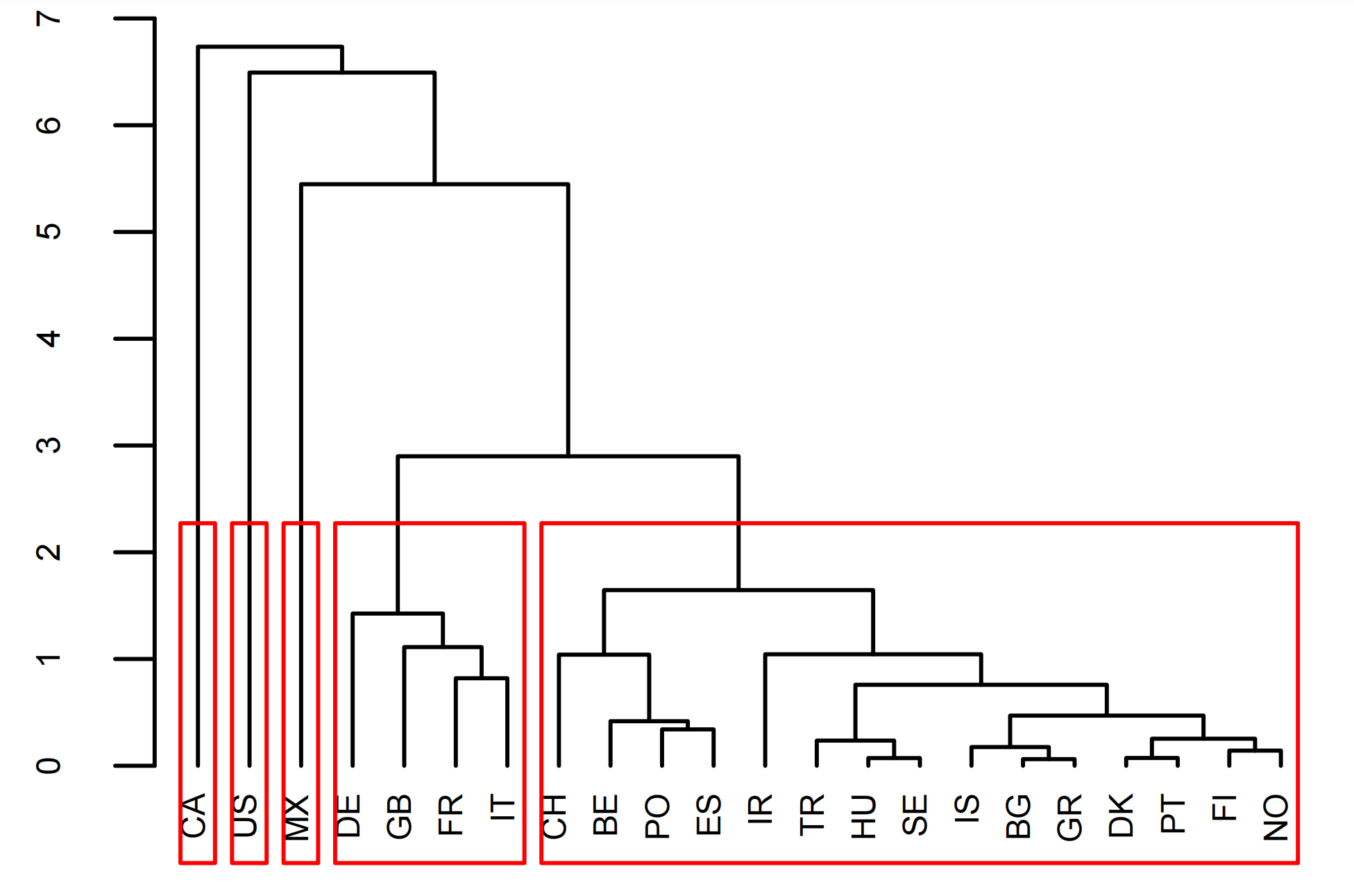}
		\label{fig:4(b)}
	}
\caption{Clustering of countries by their export (a) and import (b) loading
coefficients.}
\end{figure}

The loading matrices are also rotated via varimax procedure.%
Mexico, the United States of America and Canada heavily load on virtual
export hubs E1, E2 and E3, respectively. European countries mainly load on
export hub E4, and Germany occupies an important position. The United States
of America, Mexico and Canda heavily load on virtual import hubs I1, I2 and
I4. European countries mainly load on import hub I3, led by France and
Britain. Figures \ref{fig:4(a)} and \ref{fig:4(b)} show the clustering of
loadings for $\widehat\Ab_1$ and $\widehat\Ab_2$ respectively.

For exporting activities, the three countries in the America are very
different from the European countries, and Germany is divided into a
separate group and other European countries into another. For importing
activities, the three North American countries still differ markedly from
Europe. Germany, the United Kingdom, France and Italy are grouped into one
cluster, and the other countries are grouped into another.

\begin{figure}[htbp]
\centering
\includegraphics[height=5.5cm,width=16.5cm]{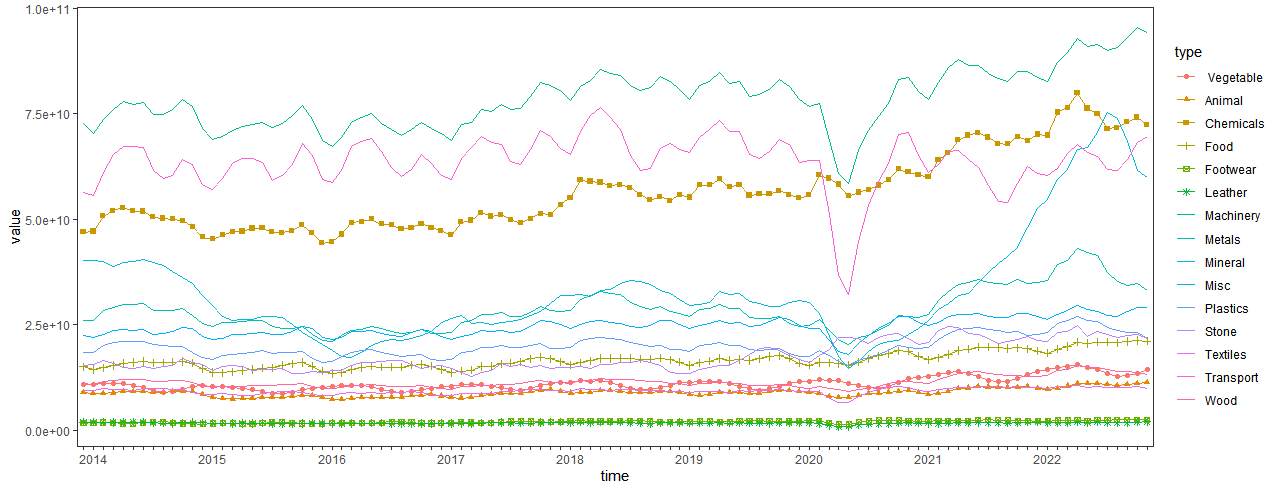}
\caption{Total trade volume of each commodity between 22 countries at each
moment.}
\label{fig:5}
\end{figure}
\begin{figure}[htbp]
\centering
\includegraphics[height=5.5cm,width=16.5cm]{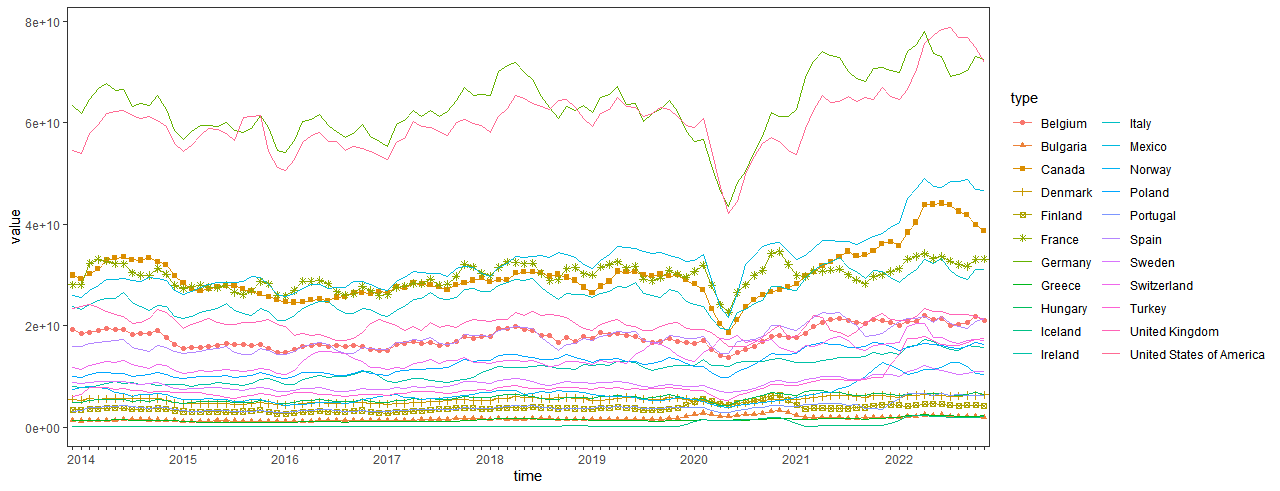}
\caption{Each country's total export trade with 21 other countries for 15
commodities at each moment.}
\label{fig:6}
\end{figure}
\begin{figure}[htbp]
\centering
\includegraphics[height=5.5cm,width=16.5cm]{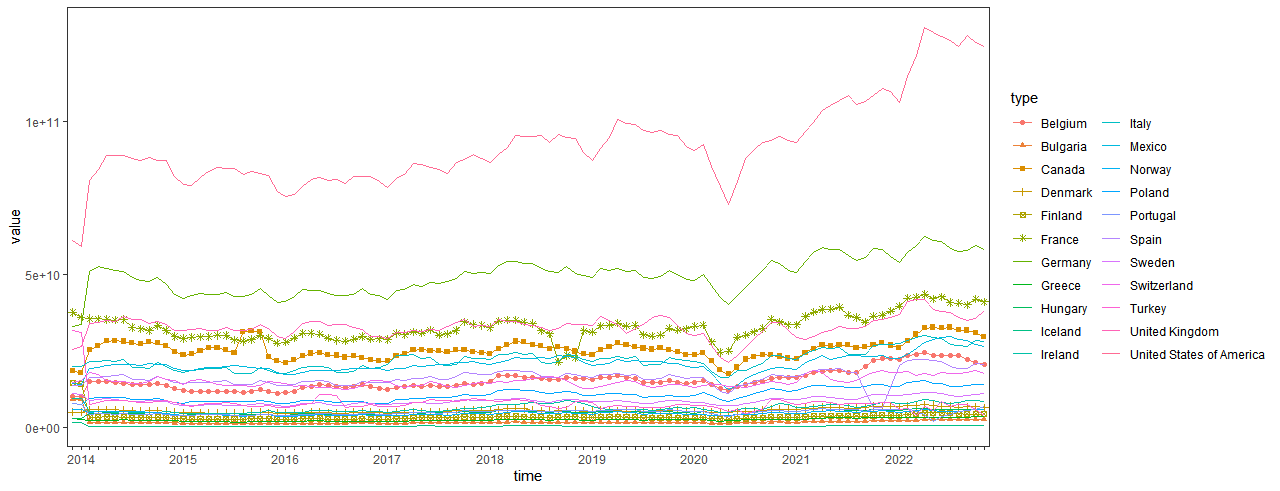}
\caption{Each country's total import trade with 21 other countries for 15
commodities at each moment.}
\label{fig:7}
\end{figure}

Figure \ref{fig:5} shows the total volume of trade between the 22 countries
for each commodity at each time. We can see that the product categories with
the largest volume of import-export are Machinery \& Electrical,
Transportation and Chemicals \& Allied Industries. Each country's total
export trade volumes and import trade volumes with 21 other countries for 15
commodities at each moment are shown in Figures \ref{fig:6} and \ref{fig:7}.
The largest importers are Germany, the United States of America, followed by
Mexico, Canada, France and Italy. The largest exporters are the United States
of America, Germany, followed by France and the United Kingdom. All three
charts show larger fluctuations in 2020, which is understandable because of
COVID-19.
And around 2021, the potential factor model of the data may
change. We will analyze this change by dividing the data into two groups for
2014-2019 and 2020-2022, written as period $\cI$ and $\cI\cI$. Although period $\cI\cI$ is relatively short we recall that when estimating the
loadings the effective sample size is $Tp_{-k}$ which is instead quite large.

\begin{table}[!h]
\caption{ Estimated loadings matrix $\widehat\Ab_3^\top$ for the commodity
mode in periods $\cI$ and $\cI\cI$.}
\label{tab:10}\renewcommand{\arraystretch}{0.5}
\resizebox{\textwidth}{!}{
		\centering
		\begin{tabular}{c|c|ccccccccccccccc}
			\toprule[2pt]
			Period  &      Factor    &	Animal  	&	Vegetable	&	Food	&	Mineral	&	Chemicals	&	Plastics	&	Leather	&	Wood	&	Textiles	&	Footwear	&	Stone	&	Metals	&	Machinery	&	Transport	&	Misc	\\
			\midrule
			$\cI$	&	1	&	0	&	2	&	-1	&	0	&	-1	&	1	&	0	&	-2	&	1	&	0	&	0	&	0	&	\textbf{29}	&	1	&	6	\\
			&	2	&	0	&	1	&	-1	&	\textbf{30}	&	1	&	-1	&	0	&	3	&	-1	&	0	&	1	&	1	&	0	&	0	&	1	\\
			&	3	&	0	&	-3	&	1	&	-1	&	\textbf{29}	&	2	&	0	&	1	&	0	&	0	&	0	&	2	&	0	&	0	&	7	\\
			&	4	&	0	&	0	&	4	&	0	&	0	&	-2	&	0	&	2	&	0	&	0	&	-1	&	1	&	0	&	\textbf{29}	&	2	\\
			&	5	&	6	&	6	&	6	&	0	&	0	&	\textbf{20}	&	0	&	6	&	5	&	1	&	1	&	16	&	1	&	0	&	-9	\\
			&	6	&	2	&	4	&	1	&	-1	&	0	&	-4	&	0	&	3	&	1	&	1	&	\textbf{29}	&	0	&	0	&	0	&	1	\\
			\midrule
			$\cI\cI$	&	1	&	-1	&	3	&	0	&	0	&	1	&	2	&	0	&	6	&	0	&	0	&	0	&	-1	&	\textbf{29}	&	2	&	5	\\
			&	2	&	0	&	1	&	0	&	\textbf{30}	&	2	&	-1	&	0	&	2	&	-1	&	0	&	0	&	1	&	0	&	-1	&	1	\\
			&	3	&	0	&	-2	&	1	&	-3	&	\textbf{29}	&	1	&	1	&	1	&	1	&	0	&	0	&	1	&	-1	&	3	&	6	\\
			&	4	&	6	&	6	&	6	&	0	&	0	&	\textbf{19}	&	0	&	7	&	4	&	0	&	0	&	17	&	2	&	-1	&	-8	\\
			&	5	&	1	&	0	&	4	&	0	&	-4	&	-4	&	0	&	7	&	0	&	0	&	0	&	2	&	0	&	\textbf{28}	&	3	\\
			&	6	&	1	&	-1	&	1	&	0	&	0	&	-1	&	0	&	5	&	1	&	0	&	\textbf{29}	&	0	&	1	&	-2	&	2	\\
			\bottomrule[2pt]
	\end{tabular}}
\end{table}

\begin{table}[!h]
\caption{ Estimated loadings matrix $\widehat\Ab_1^\top$ for the export mode
in periods $\cI$ and $\cI\cI$.}
\label{tab:11}\renewcommand{\arraystretch}{0.5}
\resizebox{\textwidth}{!}{
		\centering
		\begin{tabular}{c|c|cccccccccccccccccccccc}
			\toprule[2pt]
			Period  &   Factor   &	BE	&	BU	&	CA	&	DK	&	FI	&	FR	&	DE	&	GR	&	HU	&	IS	&	IR	&	IT	&	MX	&	NO	&	PO	&	PT	&	ES	&	SE	&	CH	&	ER	&	US	&	GB	\\
			\midrule
			$\cI$	&	1	&	-2	&	0	&	1	&	0	&	0	&	0	&	3	&	0	&	0	&	0	&	-2	&	1	&	\textbf{30}	&	0	&	0	&	0	&	-1	&	0	&	-1	&	0	&	0	&	1	\\
			&	2	&	0	&	0	&	-1	&	0	&	0	&	0	&	-2	&	0	&	1	&	0	&	-1	&	1	&	0	&	0	&	1	&	0	&	0	&	0	&	0	&	0	&	\textbf{30}	&	0	\\
			&	3	&	1	&	0	&	\textbf{30}	&	0	&	0	&	0	&	0	&	0	&	0	&	0	&	2	&	0	&	-1	&	1	&	0	&	0	&	0	&	0	&	1	&	0	&	1	&	1	\\
			&	4	&	6	&	0	&	-1	&	2	&	1	&	9	&	\textbf{23}	&	0	&	2	&	0	&	7	&	7	&	-1	&	1	&	4	&	1	&	6	&	2	&	6	&	2	&	1	&	6	\\
			\midrule
			$\cI\cI$	&	-1	&	-1	&	0	&	0	&	0	&	0	&	0	&	3	&	0	&	0	&	0	&	-3	&	1	&	\textbf{30}	&	0	&	0	&	0	&	0	&	0	&	-1	&	0	&	0	&	1	\\
			&	2	&	0	&	0	&	0	&	0	&	0	&	1	&	-1	&	0	&	1	&	0	&	-2	&	1	&	-1	&	1	&	1	&	0	&	1	&	0	&	0	&	0	&	\textbf{30}	&	1	\\
			&	-3	&	0	&	0	&	\textbf{30}	&	0	&	0	&	0	&	-1	&	0	&	0	&	0	&	1	&	0	&	0	&	0	&	-1	&	0	&	0	&	0	&	0	&	0	&	0	&	1	\\
			&	-4	&	7	&	0	&	0	&	2	&	1	&	8	&	\textbf{19}	&	0	&	2	&	0	&	13	&	8	&	-1	&	1	&	4	&	1	&	5	&	2	&	9	&	2	&	1	&	5	\\
			\bottomrule[2pt]
	\end{tabular}}
\end{table}

\begin{table}[!h]
\caption{ Estimated loadings matrix $\widehat\Ab_2^\top$ for the import mode
in periods $\cI$ and $\cI\cI$.}
\label{tab:12}\renewcommand{\arraystretch}{0.5}
\resizebox{\textwidth}{!}{
		\centering
		\begin{tabular}{c|c|cccccccccccccccccccccc}
			\toprule[2pt]
			Period  &   Factor   &	BE	&	BU	&	CA	&	DK	&	FI	&	FR	&	DE	&	GR	&	HU	&	IS	&	IR	&	IT	&	MX	&	NO	&	PO	&	PT	&	ES	&	SE	&	CH	&	ER	&	US	&	GB	\\
			\midrule
			$\cI$	&	1	&	0	&	0	&	0	&	0	&	0	&	0	&	0	&	0	&	0	&	0	&	0	&	0	&	0	&	0	&	0	&	0	&	0	&	0	&	0	&	0	&	\textbf{30}	&	0	\\
			&	2	&	-1	&	0	&	3	&	0	&	0	&	1	&	3	&	0	&	0	&	0	&	-3	&	-1	&	\textbf{29}	&	0	&	-1	&	0	&	-1	&	0	&	-4	&	0	&	0	&	4	\\
			&	3	&	7	&	1	&	-5	&	2	&	1	&	\textbf{15}	&	10	&	1	&	4	&	0	&	1	&	10	&	0	&	1	&	6	&	2	&	8	&	4	&	7	&	4	&	0	&	\textbf{15}	\\
			&	4	&	1	&	0	&	\textbf{29}	&	0	&	0	&	2	&	5	&	0	&	-1	&	0	&	5	&	1	&	-2	&	1	&	0	&	0	&	1	&	0	&	3	&	1	&	0	&	1	\\
			\midrule
			$\cI\cI$	&	1	&	0	&	0	&	1	&	0	&	0	&	0	&	0	&	0	&	0	&	0	&	0	&	0	&	0	&	0	&	0	&	0	&	0	&	0	&	0	&	0	&	\textbf{30}	&	0	\\
			&	2	&	-1	&	0	&	\textbf{21}	&	0	&	0	&	0	&	6	&	0	&	-1	&	0	&	2	&	-1	&	20	&	0	&	-1	&	0	&	1	&	0	&	-1	&	0	&	0	&	0	\\
			&	3	&	11	&	1	&	-1	&	3	&	1	&	\textbf{16}	&	14	&	1	&	4	&	0	&	2	&	12	&	-2	&	1	&	7	&	2	&	7	&	4	&	8	&	3	&	0	&	1	\\
			&	4	&	-2	&	0	&	-1	&	0	&	0	&	-1	&	-1	&	0	&	0	&	0	&	-1	&	-1	&	2	&	0	&	0	&	0	&	0	&	0	&	2	&	0	&	0	&	\textbf{30}	\\
			\bottomrule[2pt]
	\end{tabular}}
\end{table}

\begin{figure}[htbp]
\label{fig:8}
\subfigure[]{
		\includegraphics[width=8cm,height=7cm]{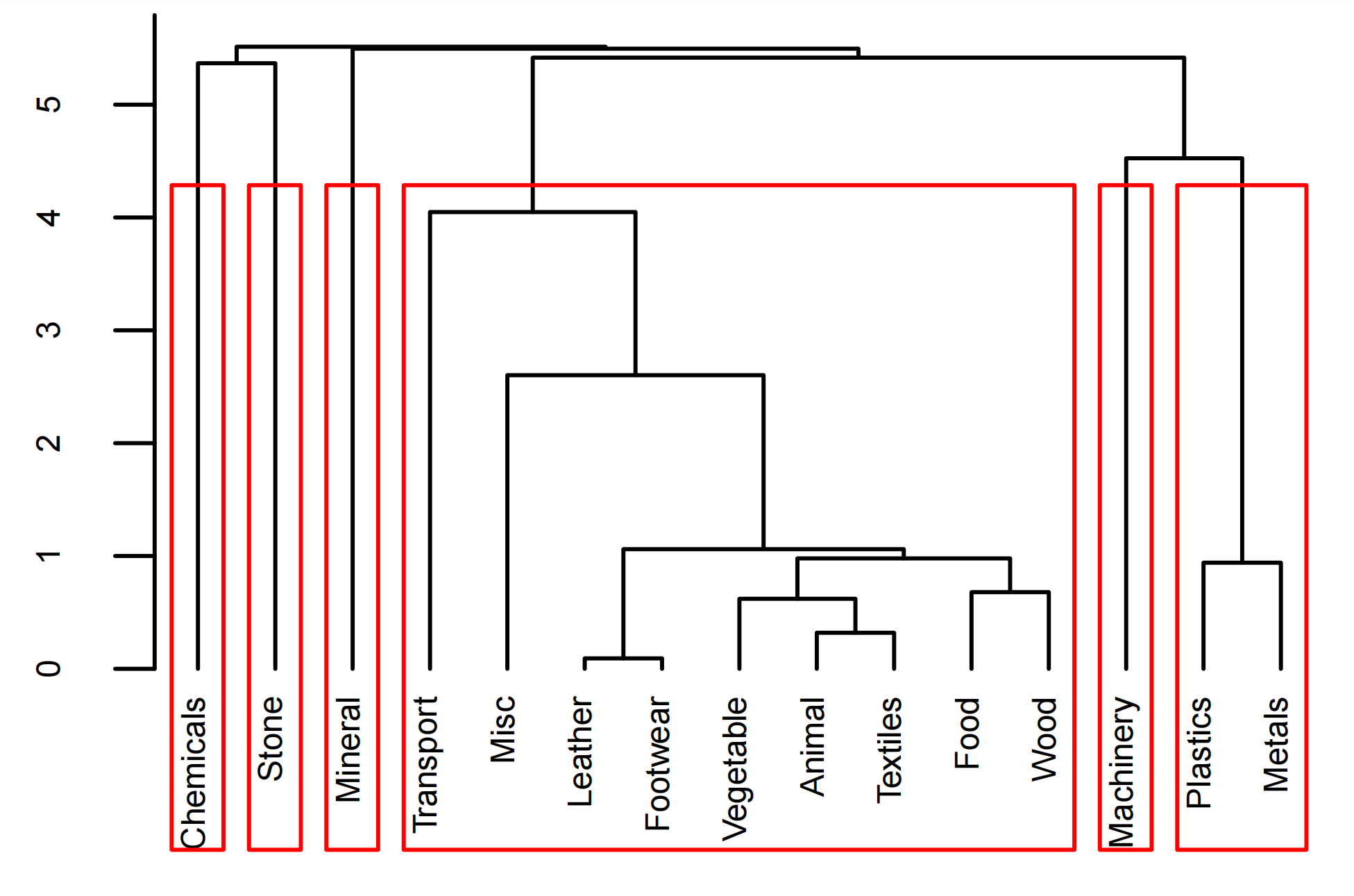} \label{fig:8(a)}
	} \hspace{2mm}
\subfigure[]{
		\includegraphics[width=8cm,height=7cm]{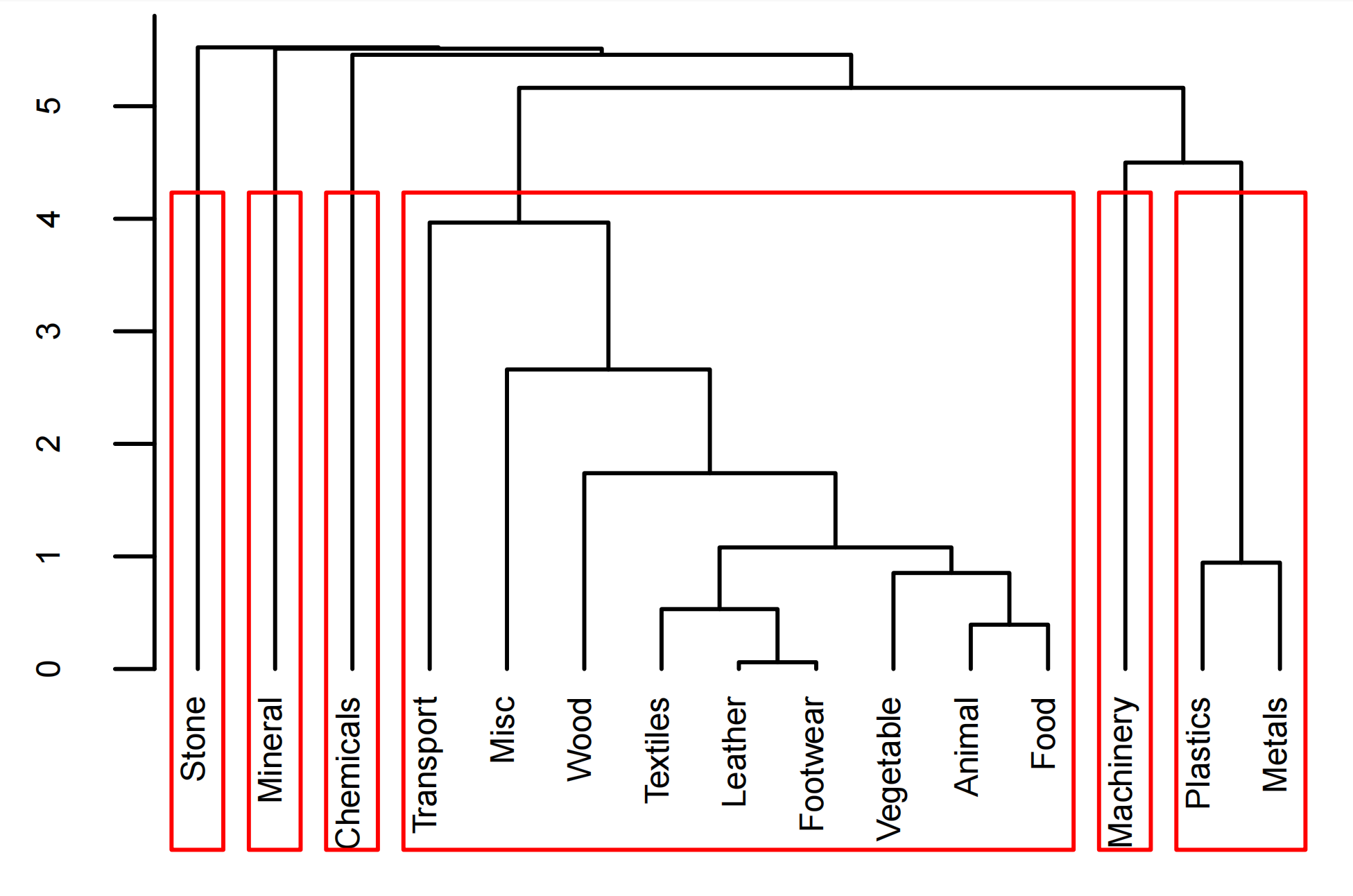}
		\label{fig:8(b)}
	}
\caption{Clustering of product categories by their loading coefficients in
periods $\cI$ (a) and $\cI\cI$ (b).}
\end{figure}

\begin{figure}[htbp]
\label{fig:9}
\subfigure[]{
		\includegraphics[width=8cm,height=7cm]{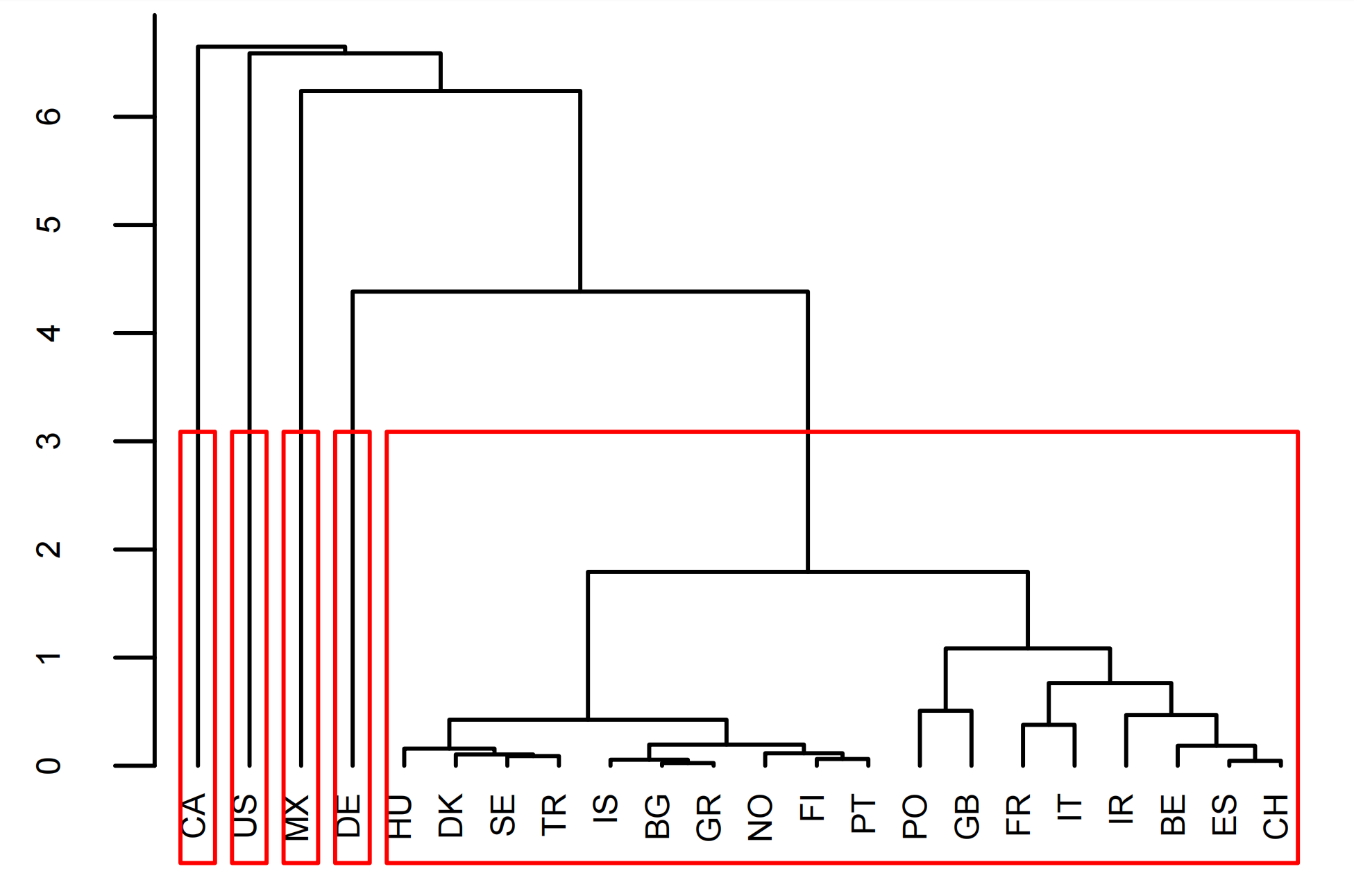} \label{fig:9(a)}
	} \hspace{2mm}
\subfigure[]{
		\includegraphics[width=8cm,height=7cm]{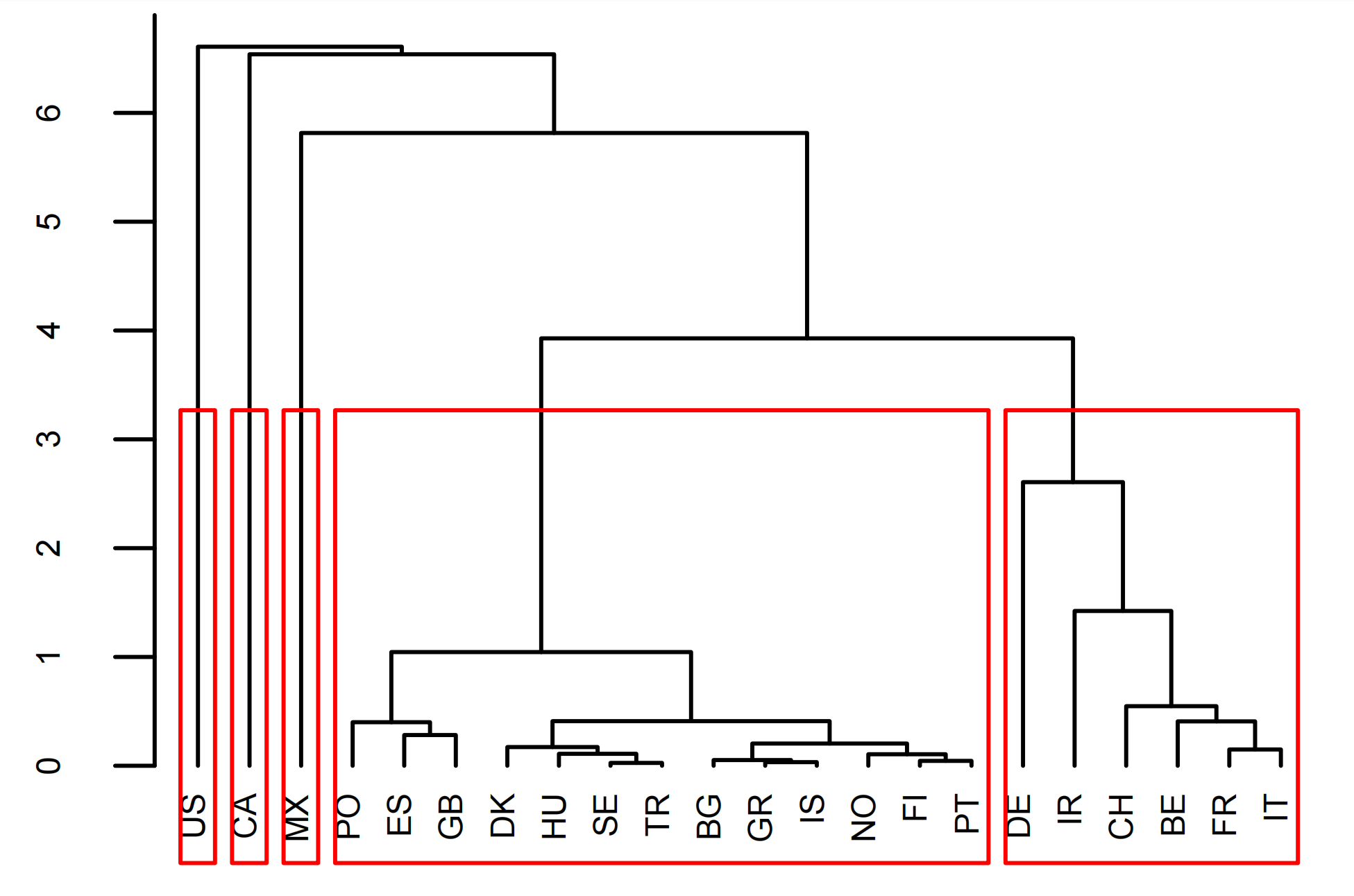}
		\label{fig:9(b)}
	}
\caption{Clustering of countries by their export loading coefficients in
periods $\cI$ (a) and $\cI\cI$ (b).}
\end{figure}
\begin{figure}[htbp]
\label{fig:10}
\subfigure[]{
		\includegraphics[width=8cm,height=7cm]{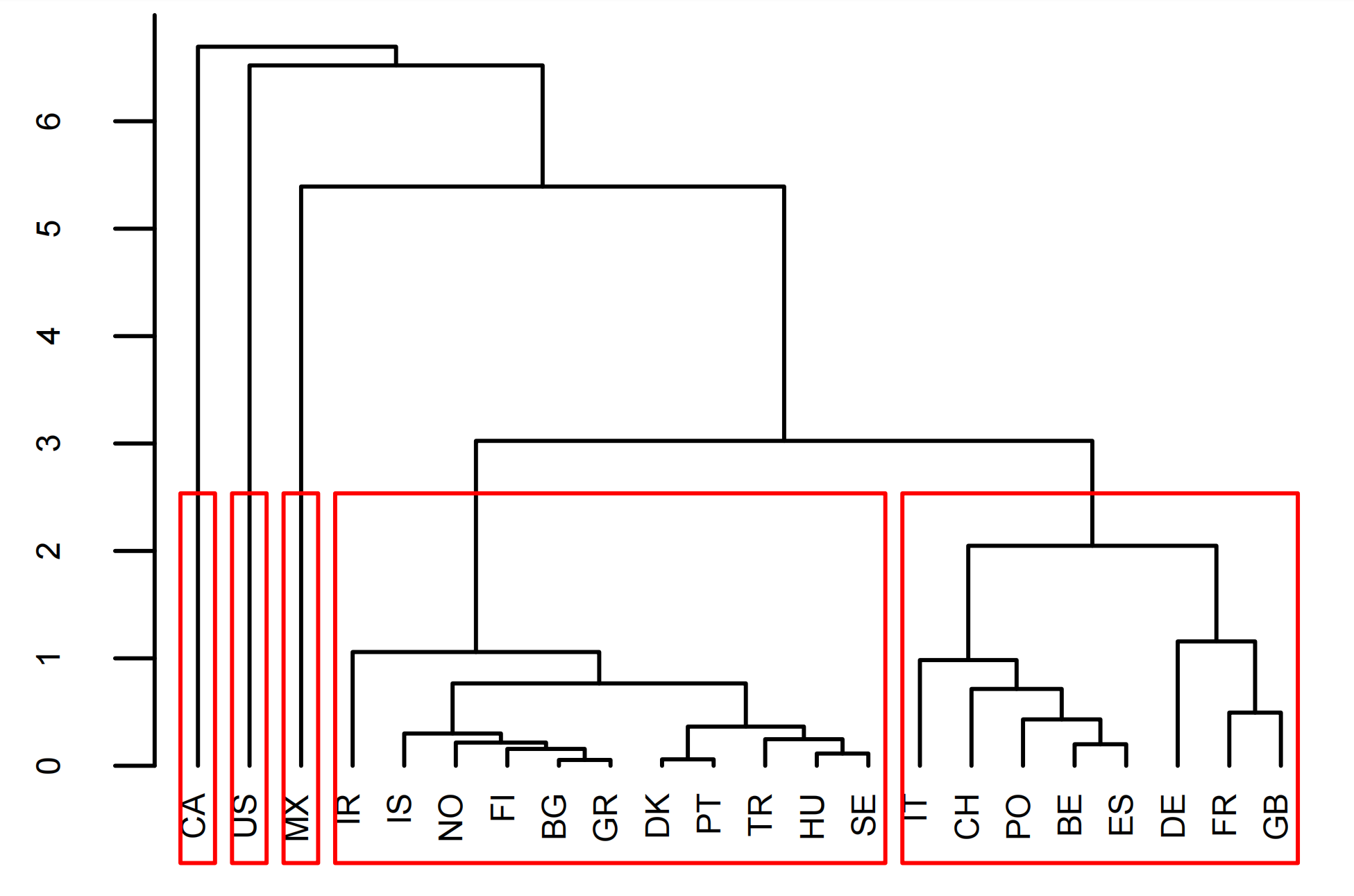} \label{fig:10(a)}
	} \hspace{2mm}
\subfigure[]{
		\includegraphics[width=8cm,height=7cm]{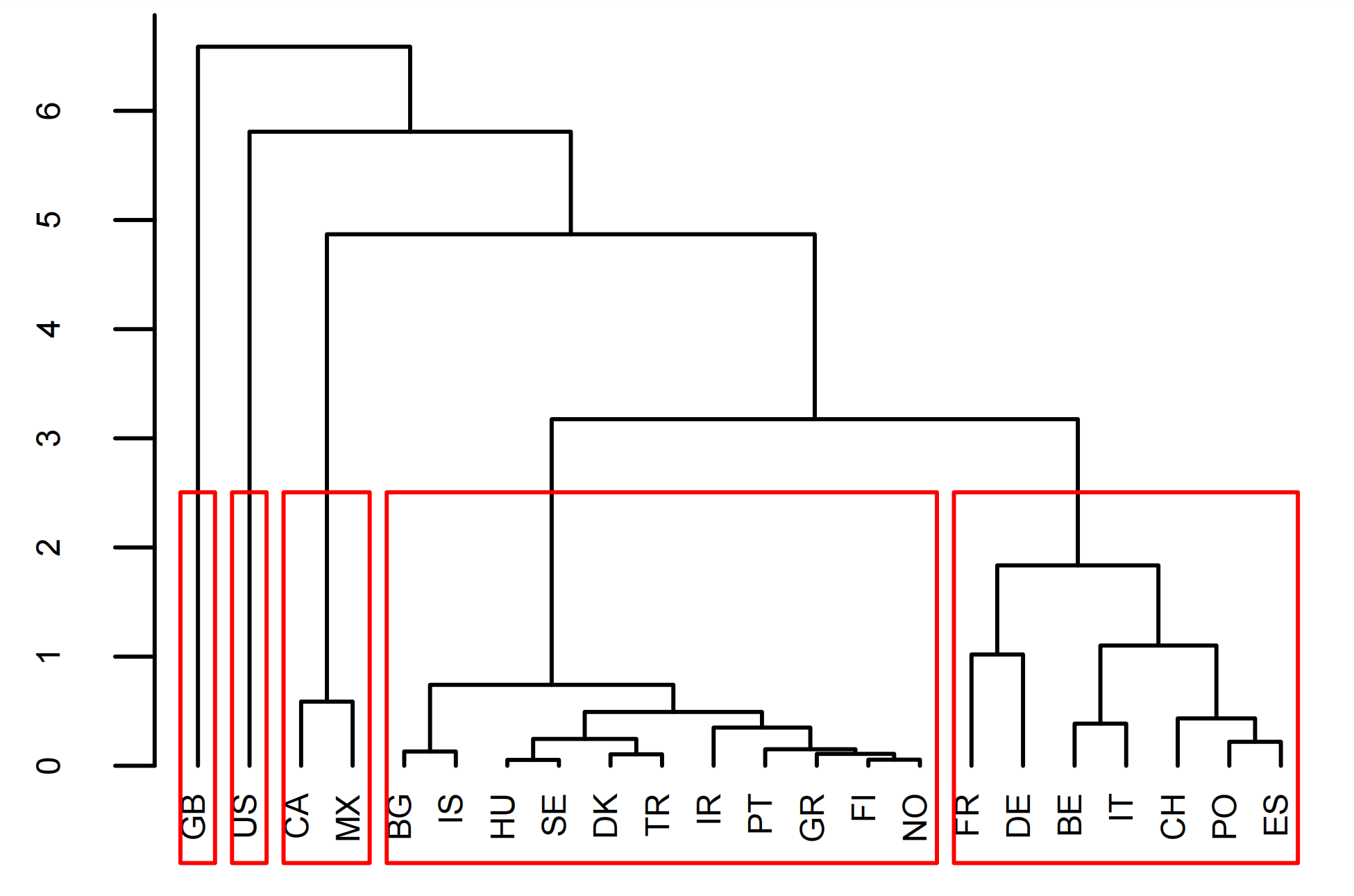}
		\label{fig:10(b)}
	}
\caption{Clustering of countries by their import loading coefficients in
periods $\cI$ (a) and $\cI\cI$ (b).}
\end{figure}

Tables \ref{tab:10}, \ref{tab:11} and \ref{tab:12} show the estimated
loading matrices $\widehat\Ab_3$, $\widehat\Ab_1$, and $\widehat\Ab_2$,
respectively, of import-export data for two periods, and their corresponding
clustering results is shown in Figures \ref{fig:8(a)} and \ref{fig:8(b)}, %
\ref{fig:9(a)} and \ref{fig:9(b)}, and \ref{fig:10(a)} and \ref{fig:10(b)},
respectively.

For product categories and exporting countries, the estimation of loading
matrices in the two periods are very similar. The clustering results also
illustrate this point. It is suggested that there are some difference in
estimation of $\Ab_2$. In period $\cI$, the United States of America, Mexico
and Canada heavily load on virtual import hubs I1, I2 and I3, European
countries mainly load on import hub I3, dominated by France and the United
Kingdom. In period $\cI\cI$, the United States of America and the United
Kingdom heavily load on virtual import hubs I1 and I4, American countries
Canada and Mexico mainly load on import hub I2 while hub I3 is mainly loaded by
European countries France and Germany.

\section{Conclusion\label{conclusion}}

In this paper, we study inference in the context of a factor model for
tensor-valued time series, from the perspective of both Least Squares and
Huber Loss. As far as the former is concerned, we investigate the
consistency of the estimated common factors and loadings space when using
estimators based on minimising \textit{quadratic} loss functions. Building
on the observation that such loss functions are adequate only if
sufficiently many moments exist, we extend our results to the case of
heavy-tailed distributions by considering estimators based on minimising the
\textit{Huber} loss function, which uses an $L_{1}$-norm weight on outliers.
We show that such class of estimators is robust to the presence of heavy
tails, even when only the second moment of the data exists. We also propose
a modified version of the eigenvalue-ratio principle to estimate the
dimensions of the factors tensor, and show the consistency of such
estimators without any condition on the relative rates of divergence of the
dimensions $p_{1}$, ..., $p_{K}$ and $T$. Extensive numerical results show
the proposed methods performs better than the state-of-the-art ones, and
that our proposed methodology is particularly effective when data exhibit
heavy tails. In the paper, we also show that the iterative version of our
estimators performs very well; deriving theoretical guarantees for the
estimators in the solution path of the iterative algorithm is a very
interesting, and challenging, topic, which is currently under investigation
by the authors.

\bibliographystyle{plain}
\bibliography{Ref}

\newpage

\setcounter{footnote}{0} \clearpage
\appendix

\section{Preliminary lemmas\label{lemmas}}

Henceforth, we use the following notation. Recall that, by (\ref{elle})
\begin{equation*}
L=\min \left\{ p,Tp_{-1},...,Tp_{-K}\right\} .
\end{equation*}%
Further, let $~r=r_{1}\cdots r_{k}$; we define
\begin{equation}
M=\sum_{k=1}^{K}p_{k}r_{k}+Tr,  \label{def-m}
\end{equation}%
as the total number of parameters to be estimated. We define the function
\begin{equation}
\mathbb{M}(\theta )=\frac{1}{Tp}\sum_{t=1}^{T}\sum_{i_{1}=1}^{p_{1}}\cdots
\sum_{i_{K}=1}^{p_{K}}\left( x_{t,i_{1},\cdots ,i_{K}}-\mathcal{F}_{t}\times
_{k=1}^{K}\boldsymbol{a}_{k,i_{k}}^{\top }\right) ^{2}=\frac{1}{Tp}%
\sum_{t=1}^{T}\left\Vert \mathcal{X}_{t}-\mathcal{F}_{t}\times _{k=1}^{K}%
\mathbf{A}_{k}\right\Vert _{F}^{2};  \label{m-theta}
\end{equation}%
it can be readily seen that minimizing $\mathbb{M}(\theta )$ is equivalent
to minimizing the Least Squares loss $L_{1}\left( \mathbf{A}_{1},\cdots ,%
\mathbf{A}_{K},\mathcal{F}_{t}\right) $, whence $\widehat{\theta }$ is
(alternatively) defined as
\begin{equation*}
\widehat{\theta }=\left( \widehat{\mathbf{A}}_{1},\cdots ,\widehat{\mathbf{A}%
}_{K},\widehat{\mathcal{F}}_{1},\cdots ,\widehat{\mathcal{F}}_{T},\right) =%
\underset{\theta \in \Theta }{\arg \min }\mathbb{M}(\theta ).
\end{equation*}%
In addition to $\mathbb{M}(\theta )$, we introduce the quantities%
\begin{eqnarray}
\mathbb{M}^{\ast }\left( \theta \right) &=&\frac{1}{Tp}\sum_{t=1}^{T}%
\sum_{i_{1}=1}^{p_{1}}...\sum_{i_{K}=1}^{p_{K}}w\left( \boldsymbol{a}%
_{1,i_{1}},,...,\boldsymbol{a}_{K,i_{K}},\mathcal{F}_{t}\right) ,
\label{m-star} \\
\overline{\mathbb{M}}^{\ast }\left( \theta \right) &=&\mathbb{E}\left(
\mathbb{M}^{\ast }\left( \theta \right) \right) =\frac{1}{Tp}%
\sum_{t=1}^{T}\sum_{i_{1}=1}^{p_{1}}...\sum_{i_{K}=1}^{p_{K}}\mathbb{E}\left[
w\left( \boldsymbol{a}_{1,i_{1}},,...,\boldsymbol{a}_{K,i_{K}},\mathcal{F}%
_{t}\right) \right] ,  \label{m-bar} \\
\mathbb{W}\left( \theta \right) &=&\mathbb{M}^{\ast }\left( \theta \right) -%
\overline{\mathbb{M}}^{\ast }\left( \theta \right) ,  \label{w}
\end{eqnarray}%
where we have defined%
\begin{equation}
w\left( \boldsymbol{a}_{1,i_{1}},\cdots ,\mathbf{b}_{K,i_{K}},\mathcal{F}%
_{t}\right) =\left( x_{t,i_{1},\cdots ,i_{K}}-\mathcal{F}_{t}\times
_{k=1}^{K}\boldsymbol{a}_{k,i_{k}}^{\top }\right) ^{2}-\left(
x_{t,i_{1},\cdots ,i_{K}}-\mathcal{F}_{0t}\times _{k=1}^{K}\boldsymbol{a}%
_{0k,i_{k}}^{\top }\right) ^{2}.  \label{w-function}
\end{equation}

We will make extensive use of the semimetric%
\begin{eqnarray}
d\left( \theta _{a},\theta _{b}\right) &=&\sqrt{\frac{1}{Tp}%
\sum_{t=1}^{T}\sum_{i_{1}=1}^{p_{1}}...\sum_{i_{K}=1}^{p_{K}}\left( \mathcal{%
F}_{at}\times _{k=1}^{K}\boldsymbol{a}_{ak,i_{k}}^{\top }-\mathcal{F}%
_{bt}\times _{k=1}^{K}\boldsymbol{a}_{bk,i_{k}}^{\top }\right) ^{2}}
\label{d-metric} \\
&=&\sqrt{\frac{1}{Tp}\sum_{t=1}^{T}\left\Vert \mathcal{F}_{at}\times
_{k=1}^{K}\boldsymbol{A}_{ak}-\mathcal{F}_{bt}\times _{k=1}^{K}\boldsymbol{A}%
_{bk}\right\Vert _{F}^{2}.}  \notag
\end{eqnarray}

Further, we let $C(\cdot ,g,\mathcal{G})$ and $D(\cdot ,g,\mathcal{G})$ be
the covering number and the packing number, respectively, of space $\mathcal{%
G}$ equipped with the semimetric $g$. Throughout the proof, we denote
positive, finite constants as $c_{0}$, $c_{1}$, ... and their value may
change from line to line.

For a random variable $X$ and a non-decreasing, convex function $\psi $ such
that $\psi \left( 0\right) =0$, we define the Orlicz norm as $\Vert X\Vert
_{\psi }=\inf \{C>0:\mathbb{E}\psi (|X|/C)\leq 1\}$. In particular, based on
the function $y=\psi ^{2}\left( x\right) $, we use the notation%
\begin{equation*}
\Vert X\Vert _{\psi ^{2}}=\inf \{C>0:\mathbb{E}\exp \left( \frac{\left\vert
X\right\vert }{C}\right) ^{2}\leq 1\}.
\end{equation*}

\begin{lemma}\label{semi-1}{\it
We assume that Assumptions \ref{as-1}-\ref{as-3} hold. Then,
as $\min \left\{ p_{1},...,p_{K},T\right\} \rightarrow \infty $, it holds
that%
\begin{equation*}
d\left( \widehat{\theta },\theta _{0}\right) =o_{P}\left( 1\right) .
\end{equation*}}
\end{lemma}

\begin{proof}
We begin by noting that, for any $\mathbf{A}_{k}$, $1\leq k\leq K$, and for
any $\mathcal{F}_{t}\in \mathcal{F}$, it holds that, using Assumptions \ref%
{as-2} and \ref{as-3}%
\begin{eqnarray}
&&\mathbb{E}\left[ w\left( \boldsymbol{a}_{1,i_{1}},,...,\boldsymbol{a}%
_{K,i_{K}},\mathcal{F}_{t}\right) \right]   \label{exp-w} \\
&=&\mathbb{E}\left[ \left( x_{t,i_{1},...,i_{K}}-\mathcal{F}_{t}\times
_{k=1}^{K}\boldsymbol{a}_{k,i_{k}}^{\top }\right) ^{2}-\left(
x_{t,i_{1},...,i_{K}}-\mathcal{F}_{0t}\times _{k=1}^{K}\boldsymbol{a}%
_{0k,i_{k}}^{\top }\right) ^{2}\right]   \notag \\
&=&\mathbb{E}\left( \mathcal{F}_{t}\times _{k=1}^{K}\boldsymbol{a}%
_{k,i_{k}}^{\top }-\mathcal{F}_{0t}\times _{k=1}^{K}\boldsymbol{a}%
_{0k,i_{k}}^{\top }\right) \left( \mathcal{F}_{t}\times _{k=1}^{K}%
\boldsymbol{a}_{k,i_{k}}^{\top }-\mathcal{F}_{0t}\times _{k=1}^{K}%
\boldsymbol{a}_{0k,i_{k}}^{\top }+2e_{t,i_{1},...,i_{K}}\right)   \notag \\
&=&\left( \mathcal{F}_{t}\times _{k=1}^{K}\boldsymbol{a}_{k,i_{k}}^{\top }-%
\mathcal{F}_{0t}\times _{k=1}^{K}\boldsymbol{a}_{0k,i_{k}}^{\top }\right)
^{2}.  \notag
\end{eqnarray}%
Then, recalling (\ref{m-bar}), for any $\theta \in \Theta $ it follows that%
\begin{equation*}
\overline{\mathbb{M}}^{\ast }(\theta )=d^{2}\left( \theta ,\theta
_{0}\right) .
\end{equation*}%
By definition, it holds that
\begin{equation*}
\mathbb{M}^{\ast }(\widehat{\theta })=\mathbb{M}(\widehat{\theta })-\mathbb{M%
}\left( \theta _{0}\right) \leq 0,
\end{equation*}%
and therefore, using the definition of (\ref{w})
\begin{equation*}
\mathbb{W}(\widehat{\theta })+\overline{\mathbb{M}}^{\ast }(\widehat{\theta }%
)=\mathbb{M}^{\ast }(\widehat{\theta })\leq 0.
\end{equation*}%
Hence, it ultimately follows that
\begin{equation*}
0\leq d^{2}\left( \widehat{\theta },\theta _{0}\right) =\overline{\mathbb{M}}%
^{\ast }(\widehat{\theta })\leq \sup_{\theta \in \Theta }|\mathbb{W}(\theta
)|,
\end{equation*}%
which entails that it is sufficient to prove that
\begin{equation}
\sup_{\theta \in \Theta }|\mathbb{W}(\theta )|=o_{P}(1).  \label{eqW}
\end{equation}%
For any $\theta \in \Theta $, define $\theta ^{\ast }=\left( \mathbf{A}%
_{1}^{\ast },\cdots ,\mathbf{A}_{K}^{\ast },\mathcal{F}_{1}^{\ast },\cdots ,%
\mathcal{F}_{T}^{\ast }\right) $, where, for all $1\leq k\leq K$
\begin{equation}
\boldsymbol{a}_{k,i_{k}}^{\ast }=\left\{ \boldsymbol{a}_{k,(i)}:i\leq
J_{k},\left\Vert \boldsymbol{a}_{k,(i)}-\boldsymbol{a}_{k,i_{k}}\right\Vert
_{2}\leq \epsilon /C_{1}\right\} ,  \label{ball-1}
\end{equation}%
and%
\begin{equation}
\mathcal{F}_{t}^{\ast }=\left\{ \mathcal{F}_{(h)}:h\leq J_{H},\left\Vert
\mathcal{F}_{(h)}-\mathcal{F}_{t}\right\Vert _{F}\leq \epsilon
/C_{1}\right\} .  \label{ball-2}
\end{equation}%
Hence, using the fact that $\mathbb{W}(\theta )=\mathbb{W}\left( \theta
^{\ast }\right) +\mathbb{W}(\theta )-\mathbb{W}\left( \theta ^{\ast }\right)
$, (\ref{eqW}) follows if we show that%
\begin{eqnarray}
\sup_{\theta \in \Theta }|\mathbb{W}(\theta )-\mathbb{W}\left( \theta ^{\ast
}\right) | &=&o_{P}(1),  \label{eqWa} \\
\sup_{\theta \in \Theta }|\mathbb{W}(\theta ^{\ast })| &=&o_{P}(1).
\label{eqWb}
\end{eqnarray}%
We begin with (\ref{eqWa}). Choose $C_{1}$ large enough such that $%
\left\Vert \boldsymbol{a}_{0k,i_{k}}\right\Vert _{2}$, $\left\Vert \mathcal{F%
}_{0t}\right\Vert _{F}$, $\left\Vert \boldsymbol{a}_{k,i_{k}}\right\Vert _{2}
$, and $\left\Vert \mathcal{F}_{t}\right\Vert _{F}$ are all $\leq C_{1}$ for
all $i_{1},\cdots ,i_{K},t$ and $1\leq k\leq K$. Let $B_{i}(C_{1})$ denote a
Euclidean ball in $\mathbb{R}^{i}$ with radius $K$, for all $i=r_{1},\cdots
,r_{K}$; and let $B_{r_{1}\times r_{2}\times \cdots \times r_{K}}(C_{1})$
denote a Euclidean ball in $\mathbb{R}^{r_{1}\times \cdots \times r_{K}}$
with radius $C_{1}$. For any $\epsilon >0$, let $\mathbf{a}_{k,(1)},\cdots ,%
\mathbf{a}_{k,\left( J_{k}\right) }$ be the maximal set of points in $%
B_{r_{k}}(C_{1})$ such that $\left\Vert \mathbf{a}_{k,(i)}-\mathbf{a}%
_{k,(h)}\right\Vert _{2}>\epsilon /C_{1}$, $\forall i\neq h$; and let $%
\mathcal{F}_{(1)},\cdots ,\mathcal{F}_{\left( H\right) }$ be the maximal set
of points in $B_{r_{1}\times \cdots \times r_{K}}(C_{1})$ such that $%
\left\Vert \mathcal{F}_{(t)}-\mathcal{F}_{(h)}\right\Vert _{F}>\epsilon /C$,
for $\forall t\neq h$. Then the packing numbers of $B_{r_{k}}(C_{1})$, with $%
1\leq k\leq K$, and of $B_{r_{1}\times \cdots \times r_{K}}(C_{1})$ are $%
D_{k}(C_{1}/\epsilon )^{r_{k}}$, and $D_{F}(C_{1}/\epsilon )^{r_{1}\times
\cdots \times r_{K}}$, respectively - note that we denote the constants as $%
D_{k}$, $1\leq k\leq K$, and $D_{F}$ respectively. Note that%
\begin{eqnarray}
&&\mathbb{W}(\theta )-\mathbb{W}\left( \theta ^{\ast }\right)
\label{eqWa-dec} \\
&=&\frac{1}{Tp}\sum_{t=1}^{T}\sum_{i_{1}=1}^{p_{1}}...\sum_{i_{K}=1}^{p_{K}}%
\left[ w\left( \boldsymbol{a}_{1,i_{1}},,...,\boldsymbol{a}_{K,i_{K}},%
\mathcal{F}_{t}\right) -w\left( \boldsymbol{a}_{1,i_{1}}^{\ast },,...,%
\boldsymbol{a}_{K,i_{K}}^{\ast },\mathcal{F}_{t}^{\ast }\right) \right]
\notag \\
&&-\frac{1}{Tp}\sum_{t=1}^{T}\sum_{i_{1}=1}^{p_{1}}...\sum_{i_{K}=1}^{p_{K}}%
\left( \mathbb{E}\left[ w\left( \boldsymbol{a}_{1,i_{1}},,...,\boldsymbol{a}%
_{K,i_{K}},\mathcal{F}_{t}\right) \right] -\mathbb{E}\left[ w\left(
\boldsymbol{a}_{1,i_{1}}^{\ast },,...,\boldsymbol{a}_{K,i_{K}}^{\ast },%
\mathcal{F}_{t}^{\ast }\right) \right] \right) .  \notag
\end{eqnarray}%
By standard algebra, it holds that
\begin{eqnarray}
&&\mathbb{E}\left\vert w\left( \boldsymbol{a}_{1,i_{1}},,...,\boldsymbol{a}%
_{K,i_{K}},\mathcal{F}_{t}\right) -w\left( \boldsymbol{a}_{1,i_{1}}^{\ast
},,...,\boldsymbol{a}_{K,i_{K}}^{\ast },\mathcal{F}_{t}^{\ast }\right)
\right\vert   \label{bound-exp} \\
&=&2\mathbb{E}\left( \left\vert \mathcal{F}_{t}^{\ast }\times _{k=1}^{K}%
\boldsymbol{a}_{k,i_{k}}^{\ast \top }-\mathcal{F}_{t}\times _{k=1}^{K}%
\boldsymbol{a}_{k,i_{k}}^{\top }\right\vert \left\vert x_{t,i_{1},...,i_{K}}-%
\mathcal{F}_{0t}\times _{k=1}^{K}\boldsymbol{a}_{0k,i_{k}}^{\top
}\right\vert \right)   \notag \\
&=&2\mathbb{E}\left( \left\vert \mathcal{F}_{t}^{\ast }\times _{k=1}^{K}%
\boldsymbol{a}_{k,i_{k}}^{\ast \top }-\mathcal{F}_{t}\times _{k=1}^{K}%
\boldsymbol{a}_{k,i_{k}}^{\top }\right\vert \left\vert
e_{t,i_{1},...,i_{K}}\right\vert \right)   \notag \\
&\leq &c_{0}\left\vert \mathcal{F}_{t}^{\ast }\times _{k=1}^{K}\boldsymbol{a}%
_{k,i_{k}}^{\ast \top }-\mathcal{F}_{t}\times _{k=1}^{K}\boldsymbol{a}%
_{k,i_{k}}^{\top }\right\vert ,  \notag
\end{eqnarray}%
having used Assumptions \ref{as-2} and \ref{as-3}. Note now that%
\begin{eqnarray}
&&\left\vert \mathcal{F}_{t}^{\ast }\times _{k=1}^{K}\boldsymbol{a}%
_{k,i_{k}}^{\ast \top }-\mathcal{F}_{t}\times _{k=1}^{K}\boldsymbol{a}%
_{k,i_{k}}^{\top }\right\vert   \label{bound-factors} \\
&=&\left\vert \mathcal{F}_{t}\times _{k=1}^{K}\left( \boldsymbol{a}%
_{k,i_{k}}^{\top }-\boldsymbol{a}_{k,i_{k}}^{\ast \top }\right) +\left(
\mathcal{F}_{t}-\mathcal{F}_{t}^{\ast }\right) \times _{k=1}^{K}\boldsymbol{a%
}_{k,i_{k}}^{\ast \top }\right\vert   \notag \\
&=&\left\vert \sum_{l=1}^{K}\mathcal{F}_{t}\times _{k=1}^{l-1}\left(
\boldsymbol{a}_{k,i_{k}}^{\top }-\boldsymbol{a}_{k,i_{k}}^{\ast \top
}\right) \times _{l}\left( \boldsymbol{a}_{l,i_{l}}^{\top }-\boldsymbol{a}%
_{l,i_{l}}^{\ast \top }\right) \times _{k=l+1}^{K}\left( \boldsymbol{a}%
_{k,i_{k}}^{\top }-\boldsymbol{a}_{k,i_{k}}^{\ast \top }\right) +\left(
\mathcal{F}_{t}-\mathcal{F}_{t}^{\ast }\right) \times _{k=1}^{K}\boldsymbol{a%
}_{k,i_{k}}^{\ast \top }\right\vert   \notag \\
&\leq &\sum_{l=1}^{K}\left\Vert \mathcal{F}_{t}\right\Vert
_{F}\prod\limits_{k=1,k\neq l}^{K}\left\Vert \boldsymbol{a}%
_{k,i_{k}}\right\Vert _{2}\left\Vert \boldsymbol{a}_{l,i_{l}}-\boldsymbol{a}%
_{l,i_{l}}^{\ast }\right\Vert _{2}+\prod\limits_{k=1}^{K}\left\Vert
\boldsymbol{a}_{k,i_{k}}^{\ast }\right\Vert _{2}\left\Vert \mathcal{F}_{t}-%
\mathcal{F}_{t}^{\ast }\right\Vert _{F}  \notag \\
&\leq &\left( K+1\right) C_{1}^{K-1}\epsilon ,  \notag
\end{eqnarray}%
where we have used, repeatedly, (\ref{ball-1}) and (\ref{ball-2}), and the
fact that, by definition of $C_{1}$, $\left\Vert \mathcal{F}_{t}\right\Vert
_{F}\leq C_{1}$, $\left\Vert \boldsymbol{a}_{k,i_{k}}\right\Vert _{2}\leq
C_{1}$ and $\left\Vert \boldsymbol{a}_{k,i_{k}}^{\ast }\right\Vert _{2}\leq
C_{1}$. Considering now (\ref{eqWa-dec}), the first term on the right-hand
side is bounded by%
\begin{eqnarray*}
&&\frac{1}{Tp}\sum_{t=1}^{T}\sum_{i_{1}=1}^{p_{1}}...\sum_{i_{K}=1}^{p_{K}}%
\mathbb{E}\left\vert w\left( \boldsymbol{a}_{1,i_{1}},,...,\boldsymbol{a}%
_{K,i_{K}},\mathcal{F}_{t}\right) -w\left( \boldsymbol{a}_{1,i_{1}}^{\ast
},,...,\boldsymbol{a}_{K,i_{K}}^{\ast },\mathcal{F}_{t}^{\ast }\right)
\right\vert  \\
&\leq &c_{0}\left\vert \mathcal{F}_{t}^{\ast }\times _{k=1}^{K}\boldsymbol{a}%
_{k,i_{k}}^{\ast \top }-\mathcal{F}_{t}\times _{k=1}^{K}\boldsymbol{a}%
_{k,i_{k}}^{\top }\right\vert \leq c_{0}\left( K+1\right)
C_{1}^{K-1}\epsilon ,
\end{eqnarray*}%
where the last passage follows from (\ref{bound-factors}); the same applies
to the second term on the right-hand side of (\ref{eqWa-dec}); putting
everything together, it finally follows that
\begin{equation}
\sup_{\theta \in \Theta }\left\vert \mathbb{W}\left( \theta \right) -\mathbb{%
W}\left( \theta ^{\ast }\right) \right\vert =\epsilon O_{P}\left( 1\right) .
\label{eqW1}
\end{equation}%
We now turn to (\ref{eqWb}). Recall that the sub-Gaussian factor is defined
in Assumption \ref{as-3} as $\nu ^{2}$; then, for any $\lambda \in \mathbb{R}
$, it holds that%
\begin{eqnarray*}
&&\mathbb{E}\left[ \exp \left( \lambda \left( w\left( \boldsymbol{a}%
_{1,i_{1}}^{\ast },,...,\boldsymbol{a}_{K,i_{K}}^{\ast },\mathcal{F}%
_{t}^{\ast }\right) -\mathbb{E}w\left( \boldsymbol{a}_{1,i_{1}}^{\ast },,...,%
\boldsymbol{a}_{K,i_{K}}^{\ast },\mathcal{F}_{t}^{\ast }\right) \right)
\right) \right]  \\
&=&\mathbb{E}\left[ \exp \left( 2\lambda e_{t,i_{1},...,i_{K}}\left(
\mathcal{F}_{0t}\times _{k=1}^{K}\boldsymbol{a}_{0k,i_{k}}^{\top }-\mathcal{F%
}_{t}^{\ast }\times _{k=1}^{K}\boldsymbol{a}_{k,i_{k}}^{\ast \top }\right)
\right) \right]  \\
&\leq &c_{1}\exp \left( \frac{4\lambda ^{2}\left( \mathcal{F}_{0t}\times
_{k=1}^{K}\boldsymbol{a}_{0k,i_{k}}^{\top }-\mathcal{F}_{t}^{\ast }\times
_{k=1}^{K}\boldsymbol{a}_{k,i_{k}}^{\ast \top }\right) ^{2}}{2}\nu
^{2}\right) ,
\end{eqnarray*}%
having used Assumption \ref{as-3} in the last passage. Thus, $w\left(
\boldsymbol{a}_{1,i_{1}}^{\ast },,...,\boldsymbol{a}_{K,i_{K}}^{\ast },%
\mathcal{F}_{t}^{\ast }\right) -\mathbb{E}w\left( \boldsymbol{a}%
_{1,i_{1}}^{\ast },,...,\boldsymbol{a}_{K,i_{K}}^{\ast },\mathcal{F}%
_{t}^{\ast }\right) $ is a sub-Gaussian random variable with variance factor
equal to $4\nu ^{2}\left( \mathcal{F}_{0t}\times _{k=1}^{K}\boldsymbol{a}%
_{0k,i_{k}}^{\top }-\mathcal{F}_{t}^{\ast }\times _{k=1}^{K}\boldsymbol{a}%
_{k,i_{k}}^{\ast \top }\right) ^{2}$. Using the Hoeffding's bound for sums
of sub-Gaussian variables, it follows that
\begin{eqnarray*}
&&\mathbb{P}\left( \left\vert \mathbb{W}\left( \theta ^{\ast }\right)
\right\vert >c\right)  \\
&\leq &2\exp \left( -\frac{c^{2}\left( Tp\right) ^{2}}{8\nu
^{2}\sum_{t=1}^{T}\sum_{i_{1}=1}^{p_{1}}...\sum_{i_{K}=1}^{p_{K}}\left(
\mathcal{F}_{0t}\times _{k=1}^{K}\boldsymbol{a}_{0k,i_{k}}^{\top }-\mathcal{F%
}_{t}^{\ast }\times _{k=1}^{K}\boldsymbol{a}_{k,i_{k}}^{\ast \top }\right)
^{2}}\right)  \\
&=&2\exp \left( -\frac{c^{2}Tp}{8\nu ^{2}d^{2}\left( \theta ^{\ast },\theta
_{0}\right) }\right)
\end{eqnarray*}%
Hence, by Lemma 2.2.1 in \citet{vanderVaart1996}%
\begin{equation*}
\left\Vert \mathbb{W}\left( \theta ^{\ast }\right) \right\Vert _{\psi
_{2}}\leq \sqrt{24\nu ^{2}}\left( \frac{d^{2}\left( \theta ^{\ast },\theta
_{0}\right) }{Tp}\right) ^{1/2},
\end{equation*}%
where $\theta ^{\ast }$ can take at most $%
\prod_{k=1}^{K}J_{k}^{p_{k}}J_{H}^{T}\leq
c_{0}\prod_{k=1}^{K}(C_{1}/\epsilon )^{p_{k}r_{k}}(C_{1}/\epsilon
)^{Tr}=(C_{1}/\epsilon )^{\sum_{k=1}^{K}p_{k}r_{k}+Tr}$ different values,
and $d\left( \theta ^{\ast },\theta _{0}\right) \leq c_{0}K$ by (\ref%
{bound-factors}).Hence, using the maximal inequality for Orlicz norms (see
Lemma 2.2.2 in \citealp{vanderVaart1996}) we have
\begin{equation}
\mathbb{E}\sup_{\theta \in \Theta }\left\vert \mathbb{W}\left( \theta ^{\ast
}\right) \right\vert \leq \left\Vert \sup_{\theta \in \Theta }\left\vert
\mathbb{W}\left( \theta ^{\ast }\right) \right\vert \right\Vert _{\psi
_{2}}\lesssim \sqrt{M}\sqrt{\log (C_{1}/\epsilon )}/\sqrt{Tp}\lesssim \sqrt{%
\log (C_{1}/\epsilon )}/\sqrt{L}.  \label{eqW2}
\end{equation}%
Finally, for any $\delta >0$%
\begin{eqnarray*}
&&\mathbb{P}\left( \sup_{\theta \in \Theta }\left\vert \mathbb{W}\left(
\theta \right) \right\vert >\delta \right)  \\
&\leq &\mathbb{P}\left( \sup_{\theta \in \Theta }\left\vert \mathbb{W}\left(
\theta ^{\ast }\right) \right\vert >\delta /2\right) +\mathbb{P}\left(
\sup_{\theta \in \Theta }\left\vert \mathbb{W}\left( \theta \right) -\mathbb{%
W}\left( \theta ^{\ast }\right) \right\vert >\delta /2\right)  \\
&\leq &\frac{2}{\delta }\mathbb{E}\sup_{\theta \in \Theta }\left\vert
\mathbb{W}\left( \theta ^{\ast }\right) \right\vert +\mathbb{P}\left(
\sup_{\theta \in \Theta }\left\vert \mathbb{W}\left( \theta \right) -\mathbb{%
W}\left( \theta ^{\ast }\right) \right\vert >\delta /2\right) ,
\end{eqnarray*}%
having used Markov's inequality in the last line. Combining (\ref{eqW1}) and
(\ref{eqW2}), and recalling that $\epsilon $ in (\ref{eqW1}) is arbitrary,
the desired result follows.
\end{proof}

\begin{lemma}
\label{lemma2}{\it We assume that Assumptions \ref{as-1}-\ref{as-3} hold. Define
the set $\Theta \left( \delta \right) =\left\{ \theta \in \Theta :d\left(
\theta ,\theta _{0}\right) \leq \delta \right\} $. Then, for all $\theta \in
\Theta \left( \delta \right) $, it holds that%
\begin{equation*}
\sum_{k=1}^{K}\frac{1}{p_{k}}\left\Vert \mathbf{A}_{k}-\mathbf{A}_{0k}%
\mathbf{S}_{k}\right\Vert _{F}^{2}+\frac{1}{T}\sum_{t=1}^{T}\left\Vert
\mathcal{F}_{t}-\mathcal{F}_{0t}\times _{k=1}^{K}%
\mathbf{S}_{k} \right\Vert _{F}^{2}\leq C_{3}\delta ^{2},
\end{equation*}%
where
\begin{equation*}
\mathbf{S}_{k}=\text{\upshape sgn}\left( \Ab_{0k}^\top\Ab_k/p_k\right) .
\end{equation*}
}
\end{lemma}

\begin{proof}
Let $\mathbf{U}_{k}\in \mathbb{R}^{r_{k}\times r_{k}}$ - where $1\leq k\leq
K $ - is a diagonal matrix whose diagonal elements are either $1$ or $-1$,
and recall that $\mathbf{A}_{k}^{\top }\mathbf{A}_{k}=\mathbf{A}_{0k}^{\top }%
\mathbf{A}_{0k}=\mathbf{I}_{k}$; further, Assumption \ref{as-1} entails that
$\left\Vert \mathcal{F}_{0t}\right\Vert _{F}^{2}\leq C_{4}$ for every $t$.
We will use the following facts
\begin{equation}
\left\Vert \mathcal{F}_{t}-\mathcal{F}_{0t}\times _{k=1}^{K}\mathbf{U}%
_{k}\right\Vert _{F}^{2}=\frac{1}{p}\left\Vert \left( \mathcal{F}_{t}-%
\mathcal{F}_{0t}\times _{k=1}^{K}\mathbf{U}_{k}\right) \times _{k=1}^{K}%
\mathbf{A}_{k}\right\Vert _{F}^{2}.  \label{fact-2}
\end{equation}%
It holds that%
\begin{eqnarray*}
&&\sum_{t=1}^{T}\left\Vert \left( \mathcal{F}_{t}-\mathcal{F}_{0t}\times
_{k=1}^{K}\mathbf{U}_{k}\right) \right\Vert _{F}^{2} \\
&\leq &\frac{1}{p}\sum_{t=1}^{T}\left\Vert \left( \mathcal{F}_{t}-\mathcal{F}%
_{0t}\times _{k=1}^{K}\mathbf{U}_{k}\right) \times _{k=1}^{K}\mathbf{A}%
_{k}\right\Vert _{F}^{2} \\
&=&\frac{1}{p}\sum_{t=1}^{T}\left\Vert \mathcal{F}_{t}\times _{k=1}^{K}%
\mathbf{A}_{k}-\mathcal{F}_{0t}\times _{k=1}^{K}\mathbf{A}_{0k}+\mathcal{F}%
_{0t}\times _{k=1}^{K}\mathbf{A}_{0k}-\mathcal{F}_{0t}\times
_{k=1}^{K}\left( \mathbf{A}_{k}\mathbf{U}_{k}\right) \right\Vert _{F}^{2} \\
&\leq &\frac{2}{p}\sum_{t=1}^{T}\left\Vert \mathcal{F}_{t}\times _{k=1}^{K}%
\mathbf{A}_{k}-\mathcal{F}_{0t}\times _{k=1}^{K}\mathbf{A}_{0k}\right\Vert
_{F}^{2}+\frac{2}{p}\sum_{t=1}^{T}\left\Vert \mathcal{F}_{0t}\times
_{k=1}^{K}\mathbf{A}_{0k}-\mathcal{F}_{0t}\times _{k=1}^{K}\left( \mathbf{A}%
_{k}\mathbf{U}_{k}\right) \right\Vert _{F}^{2} \\
&=&2Td^{2}\left( \theta ,\theta _{0}\right) +\frac{2}{p}%
\sum_{t=1}^{T}\left\Vert \mathcal{F}_{0t}\times _{k=1}^{K}\mathbf{A}_{0k}-%
\mathcal{F}_{0t}\times _{k=1}^{K}\left( \mathbf{A}_{k}\mathbf{U}_{k}\right)
\right\Vert _{F}^{2}.
\end{eqnarray*}%
Further note that%
\begin{eqnarray*}
&&\sum_{t=1}^{T}\left\Vert \mathcal{F}_{0t}\times _{k=1}^{K}\mathbf{A}_{0k}-%
\mathcal{F}_{0t}\times _{k=1}^{K}\left( \mathbf{A}_{k}\mathbf{U}_{k}\right)
\right\Vert _{F}^{2} \\
&=&\sum_{t=1}^{T}\left\Vert \mathcal{F}_{0t}\times _{k=1}^{K}\mathbf{A}_{0k}-%
\mathcal{F}_{0t}\times _{1}\left( \mathbf{A}_{1}\mathbf{U}_{1}\right) \times
_{k=2}^{K}\mathbf{A}_{0k}+\mathcal{F}_{0t}\times _{1}\left( \mathbf{A}_{1}%
\mathbf{U}_{1}\right) \times _{k=2}^{K}\mathbf{A}_{0k}-\mathcal{F}%
_{0t}\times _{k=1}^{K}\left( \mathbf{A}_{k}\mathbf{U}_{k}\right) \right\Vert
_{F}^{2} \\
&=&\sum_{t=1}^{T}\left\Vert \mathcal{F}_{0t}\times _{1}\left( \mathbf{A}%
_{01}-\mathbf{A}_{1}\mathbf{U}_{1}\right) \times _{k=2}^{K}\mathbf{A}_{0k}+%
\mathcal{F}_{0t}\times _{1}\left( \mathbf{A}_{1}\mathbf{U}_{1}\right) \times
_{2}\left( \mathbf{A}_{02}-\mathbf{A}_{2}\mathbf{U}_{2}\right) \times
_{k=3}^{K}\mathbf{A}_{0k}+...\right. \\
&&\left. ...+\mathcal{F}_{0t}\times _{k=1}^{K}\left( \mathbf{A}_{k}\mathbf{U}%
_{k}\right) \times _{k}\left( \mathbf{A}_{0K}-\mathbf{A}_{K}\mathbf{U}%
_{K}\right) \right\Vert _{F}^{2} \\
&\leq &K\sum_{t=1}^{T}\sum_{l=1}^{K}\left\Vert \mathcal{F}_{0t}\times
_{k=1}^{l-1}\left( \mathbf{A}_{k}\mathbf{U}_{k}\right) \times _{l}\left(
\mathbf{A}_{0l}-\mathbf{A}_{l}\mathbf{U}_{l}\right) \times _{k=l+1}^{K}%
\mathbf{A}_{0k}\right\Vert _{F}^{2} \\
&\leq &c_{0}KT\sum_{k=1}^{K}p_{-k}\left\Vert \mathbf{A}_{k}-\mathbf{A}_{0k}%
\mathbf{U}_{k}\right\Vert _{F}^{2}.
\end{eqnarray*}%
Hence it follows that%
\begin{equation*}
\frac{1}{T}\sum_{t=1}^{T}\left\Vert \left( \mathcal{F}_{t}-\mathcal{F}%
_{0t}\times _{k=1}^{K}\mathbf{U}_{k}\right) \right\Vert _{F}^{2}\leq
2d^{2}\left( \theta ,\theta _{0}\right) +2c_{0}K\sum_{k=1}^{K}\frac{1}{p_{k}}%
\left\Vert \mathbf{A}_{k}-\mathbf{A}_{0k}\mathbf{U}_{k}\right\Vert _{F}^{2}.
\end{equation*}%
Thus, for all $\theta \in \Theta (\delta )$%
\begin{eqnarray*}
&&\sum_{k=1}^{K}\frac{1}{p_{k}}\left\Vert \mathbf{A}_{k}-\mathbf{A}_{0k}%
\mathbf{U}_{k}\right\Vert _{F}^{2}+\frac{1}{T}\sum_{t=1}^{T}\left\Vert
\left( \mathcal{F}_{t}-\mathcal{F}_{0t}\times _{k=1}^{K}%
\mathbf{U}_{k}\right) \right\Vert _{F}^{2} \\
&\leq &2d^{2}\left( \theta ,\theta _{0}\right) +\left( 2c_{0}K+1\right)
\sum_{k=1}^{K}\frac{1}{p_{k}}\left\Vert \mathbf{A}_{k}-\mathbf{A}_{0k}%
\mathbf{U}_{k}\right\Vert _{F}^{2} \\
&\leq &2\delta ^{2}+\left( 2c_{0}K+1\right) \sum_{k=1}^{K}\frac{1}{p_{k}}%
\left\Vert \mathbf{A}_{k}-\mathbf{A}_{0k}\mathbf{U}_{k}\right\Vert _{F}^{2}.
\end{eqnarray*}%
Hence, we only need to study $p_{k}^{-1}\left\Vert \mathbf{A}_{k}-\mathbf{A}%
_{0k}\mathbf{U}_{k}\right\Vert _{F}^{2}$. It holds that%
\begin{eqnarray}
&&\frac{1}{p_{k}}\left\Vert \mathbf{A}_{k}-\mathbf{A}_{0k}\mathbf{U}%
_{k}\right\Vert _{F}^{2}  \label{dec-lemma2} \\
&\leq &\frac{2}{p_{k}}\left\Vert \mathbf{A}_{0k}\mathbf{U}_{k}-\frac{1}{p_{k}%
}\mathbf{A}_{k}\mathbf{A}_{k}^{\top }\mathbf{A}_{0k}\mathbf{U}%
_{k}\right\Vert _{F}^{2}+\frac{2}{p_{k}}\left\Vert \frac{1}{p_{k}}\mathbf{A}%
_{k}\mathbf{A}_{k}^{\top }\mathbf{A}_{0k}\mathbf{U}_{k}-\mathbf{A}%
_{k}\right\Vert _{F}^{2}  \notag \\
&=&\frac{2}{p_{k}}\left\Vert \mathbf{A}_{0k}-\frac{1}{p_{k}}\mathbf{A}_{k}%
\mathbf{A}_{k}^{\top }\mathbf{A}_{0k}\right\Vert _{F}^{2}+\frac{2}{p_{k}}%
\left\Vert \mathbf{A}_{k}\left( \frac{1}{p_{k}}\mathbf{A}_{k}^{\top }\mathbf{%
A}_{0k}\mathbf{U}_{k}-\mathbf{I}_{r_{k}}\right) \right\Vert _{F}^{2}  \notag
\\
&=&\frac{2}{p_{k}}\left\Vert \left( \mathbf{I}_{p_{k}}-\mathbf{P}_{\mathbf{A}%
_{k}}\right) \mathbf{A}_{0k}\right\Vert _{F}^{2}+2\left\Vert \frac{1}{p_{k}}%
\mathbf{A}_{k}^{\top }\mathbf{A}_{0k}-\mathbf{U}_{k}\right\Vert _{F}^{2}
\notag \\
&=&\frac{2}{p_{k}}\left\Vert \mathbf{M}_{\mathbf{A}_{k}}\mathbf{A}%
_{0k}\right\Vert _{F}^{2}+2\left\Vert \frac{1}{p_{k}}\mathbf{A}_{k}^{\top }%
\mathbf{A}_{0k}-\mathbf{U}_{k}\right\Vert _{F}^{2}=I+II,  \notag
\end{eqnarray}%
where we have defined $\mathbf{P}_{\mathbf{A}_{k}}=\mathbf{A}_{k}\left(
\mathbf{A}_{k}^{\top }\mathbf{A}_{k}\right) ^{-1}\mathbf{A}_{k}^{\top }$ and
$\mathbf{M}_{\mathbf{A}_{k}}=\mathbf{I}_{p_{k}}-\mathbf{P}_{\mathbf{A}_{k}}$%
. Consider $I$; letting $\mathbf{B}_{k}=\mathbf{A}_{K}\otimes ...\otimes
\mathbf{A}_{k+1}\otimes \mathbf{A}_{k-1}\otimes ...\otimes \mathbf{A}_{1}$
and $\mathbf{B}_{0k}=\mathbf{A}_{0K}\otimes ...\otimes \mathbf{A}%
_{0k+1}\otimes \mathbf{A}_{0k-1}\otimes ...\otimes \mathbf{A}_{01}$, it
holds that%
\begin{eqnarray*}
&&\frac{1}{Tp}\sum_{t=1}^{T}\left\Vert \mathbf{M}_{\mathbf{A}_{k}}\left(
\mathbf{A}_{k}\mathbf{F}_{k,t}\mathbf{B}_{k}^{\top }-\mathbf{A}_{0k}\mathbf{F%
}_{0k,t}\mathbf{B}_{0k}^{\top }\right) \right\Vert _{F}^{2} \\
&\leq &\frac{1}{Tp}\sum_{t=1}^{T}\text{rank}\left( \mathbf{M}_{\mathbf{A}%
_{k}}\left( \mathbf{A}_{k}\mathbf{F}_{k,t}\mathbf{B}_{k}^{\top }-\mathbf{A}%
_{0k}\mathbf{F}_{0k,t}\mathbf{B}_{0k}^{\top }\right) \right) \left\Vert
\mathbf{M}_{\mathbf{A}_{k}}\right\Vert _{2}^{2}\left\Vert \mathcal{F}%
_{t}\times _{k=1}^{K}\mathbf{A}_{k}-\mathcal{F}_{0t}\times _{k=1}^{K}\mathbf{%
A}_{0k}\right\Vert _{F}^{2} \\
&\leq &\frac{1}{Tp}\sum_{t=1}^{T}\left\Vert \mathcal{F}_{t}\times _{k=1}^{K}%
\mathbf{A}_{k}-\mathcal{F}_{0t}\times _{k=1}^{K}\mathbf{A}_{0k}\right\Vert
_{F}^{2}=d^{2}\left( \theta ,\theta _{0}\right) ,
\end{eqnarray*}%
and%
\begin{eqnarray*}
&&\frac{1}{Tp}\sum_{t=1}^{T}\left\Vert \mathbf{M}_{\mathbf{A}_{k}}\left(
\mathbf{A}_{k}\mathbf{F}_{k,t}\mathbf{B}_{k}^{\top }-\mathbf{A}_{0k}\mathbf{F%
}_{0k,t}\mathbf{B}_{0k}^{\top }\right) \right\Vert _{F}^{2} \\
&=&\frac{1}{Tp}\sum_{t=1}^{T}\left\Vert \mathbf{M}_{\mathbf{A}_{k}}\mathbf{A}%
_{0k}\mathbf{F}_{0k,t}\mathbf{B}_{0k}^{\top }\right\Vert _{F}^{2} \\
&=&\frac{1}{Tp_{k}}\sum_{t=1}^{T}\tr\left( \mathbf{F}_{0k,t}\mathbf{F}%
_{0k,t}^{\top }\mathbf{A}_{0k}^{\top }\mathbf{M}_{\mathbf{A}_{k}}\mathbf{A}%
_{0k}\right) \\
&\geq &\frac{1}{Tp_{k}}\sum_{t=1}^{T}\lambda _{\min }\left( \mathbf{F}_{0k,t}%
\mathbf{F}_{0k,t}^{\top }\right) \tr\left( \mathbf{A}_{0k}^{\top }\mathbf{M}_{%
\mathbf{A}_{k}}\mathbf{A}_{0k}\right) \\
&\geq &\lambda _{\min }\left( \frac{1}{T}\sum_{t=1}^{T}\mathbf{F}_{0k,t}%
\mathbf{F}_{0k,t}^{\top }\right) \frac{1}{p_{k}}\left\Vert \mathbf{M}_{%
\mathbf{A}_{k}}\mathbf{A}_{0k}\right\Vert _{F}^{2},
\end{eqnarray*}%
where we have used the fact that $\mathbf{B}_{0k}^{\top }\mathbf{B}_{0k}=%
\mathbf{I}_{K}\otimes ...\otimes \mathbf{I}_{k+1}\otimes \mathbf{I}%
_{k-1}\otimes ...\otimes \mathbf{I}_{1}$ in the third line, and Assumption %
\ref{as-1} in the last one. Combining these two bounds%
\begin{equation*}
\frac{1}{p_{k}}\left\Vert \mathbf{M}_{\mathbf{A}_{k}}\mathbf{A}%
_{0k}\right\Vert _{F}^{2}\leq \frac{d^{2}\left( \theta ,\theta _{0}\right) }{%
\lambda _{\min }\left( \frac{1}{T}\sum_{t=1}^{T}\mathbf{F}_{0k,t}\mathbf{F}%
_{0k,t}^{\top }\right) }\leq d^{2}\left( \theta ,\theta _{0}\right) \leq
\delta ^{2}.
\end{equation*}%
As far as $II$ is concerned, we define $\mathbf{Z}_{k}^{\top }=\mathbf{A}%
_{k}^{\top }\mathbf{A}_{0k}/p_{k}$ and
\begin{equation*}
\mathbf{V}_{k}=\mbox{diag}\left( \left( \mathbf{z}_{k,1}\mathbf{z}%
_{k,1}^{\top }\right) ^{-1/2},...,\left( \mathbf{z}_{k,r_{k}}\mathbf{z}%
_{k,r_{k}}^{\top }\right) ^{-1/2}\right) ,
\end{equation*}%
where $\mathbf{z}_{k,j}$ is the $j$-th row of $\mathbf{Z}_{k}$. Then we have%
\begin{eqnarray}
\left\Vert \frac{1}{p_{k}}\mathbf{A}_{k}^{\top }\mathbf{A}_{0k}-\mathbf{S}%
_{k}\right\Vert _{F}^{2} &=&\left\Vert \mathbf{Z}_{k}^{\top }-\mathbf{S}%
_{k}\right\Vert _{F}^{2}  \label{pert} \\
&\leq &\left\Vert \mathbf{Z}_{k}^{\top }\mathbf{V}_{k}-\mathbf{S}%
_{k}\right\Vert _{F}^{2}+\left\Vert \mathbf{Z}_{k}^{\top }\mathbf{V}_{k}-%
\mathbf{Z}_{k}^{\top }\right\Vert _{F}^{2}  \notag \\
&\leq &\left\Vert \mathbf{Z}_{k}^{\top }\mathbf{V}_{k}-\mathbf{S}%
_{k}\right\Vert _{F}^{2}+\left\Vert \mathbf{Z}_{k}\right\Vert
_{F}^{2}\left\Vert \mathbf{V}_{k}-\mathbf{I}_{r_{k}}\right\Vert _{F}^{2}.
\notag
\end{eqnarray}%
Using the definition of $\mathbf{S}_{k}$, it follows that%
\begin{equation}
\left\Vert \mathbf{Z}_{k}^{\top }\mathbf{V}_{k}-\mathbf{S}_{k}\right\Vert
_{F}^{2}=\left\Vert \mathbf{Z}_{k}^{\top }\mathbf{V}_{k}\mathbf{S}_{k}-%
\mathbf{I}_{r_{k}}\right\Vert _{F}^{2}\leq d^{2}\left( \theta ,\theta
_{0}\right) =\delta ^{2}.  \label{pert-1}
\end{equation}%
Also%
\begin{eqnarray*}
\left\Vert \mathbf{V}_{k}-\mathbf{I}_{r_{k}}\right\Vert _{F}^{2} &\leq
&\left\Vert \mathbf{Z}_{k}^{\top }\mathbf{Z}_{k}-\mathbf{I}%
_{r_{k}}\right\Vert _{F}^{2}=\left\Vert \frac{1}{p_{k}}\mathbf{A}_{k}^{\top }%
\mathbf{P}_{\mathbf{A}_{0k}}\mathbf{A}_{k}-\frac{1}{p_{k}}\mathbf{A}%
_{k}^{\top }\mathbf{A}_{k}\right\Vert _{F}^{2} \\
&=&\left\Vert \frac{1}{p_{k}}\mathbf{A}_{k}^{\top }\mathbf{M}_{\mathbf{A}%
_{0k}}\mathbf{A}_{k}\right\Vert _{F}^{2}\leq \frac{1}{p_{k}^{2}}\left\Vert
\mathbf{A}_{k}\right\Vert _{F}^{2}\left\Vert \mathbf{M}_{\mathbf{A}_{0k}}%
\mathbf{A}_{k}\right\Vert _{F}^{2}.
\end{eqnarray*}%
Similar passages as above yield%
\begin{equation*}
\frac{1}{p_{k}^{2}}\left\Vert \mathbf{M}_{\mathbf{A}_{0k}}\mathbf{A}%
_{k}\right\Vert _{F}^{2}\leq \frac{d^{2}\left( \theta ,\theta _{0}\right) }{%
\frac{1}{T}\sum_{t=1}^{T}\lambda _{\min }\left( \mathbf{F}_{k,t}\mathbf{F}%
_{k,t}^{\top }\right) }.
\end{equation*}%
By Lemma \ref{l-min}, it follows that $\frac{1}{T}\sum_{t=1}^{T}\lambda
_{\min }\left( \mathbf{F}_{k,t}\mathbf{F}_{k,t}^{\top }\right) $ is bounded
away from zero, and therefore%
\begin{equation}
\frac{1}{p_{k}^{2}}\left\Vert \mathbf{M}_{\mathbf{A}_{0k}}\mathbf{A}%
_{k}\right\Vert _{F}^{2}\leq c_{0}d^{2}\left( \theta ,\theta _{0}\right)
=c_{0}\delta ^{2}.  \label{pert-2}
\end{equation}%
Putting (\ref{pert-1}) and (\ref{pert-2}) in (\ref{pert}), it follows that%
\begin{equation}
\left\Vert \frac{1}{p_{k}}\mathbf{A}_{k}^{\top }\mathbf{A}_{0k}-\mathbf{S}%
_{k}\right\Vert _{F}^{2}\leq c_{0}\delta ^{2}.  \label{pert-0}
\end{equation}%
Putting all together, the desired result follows.
\end{proof}

\begin{lemma}
\label{l-min}{\it We assume that the assumptions of Lemma \ref{lemma2} hold. Then
it holds that}
\begin{equation*}
\frac{1}{T}\sum_{t=1}^{T}\lambda _{\min }\left( \mathbf{F}_{k,t}\mathbf{F}%
_{k,t}^{\top }\right) >0.
\end{equation*}

\begin{proof}
Recall that we have defined $\mathbf{Z}_{k}=p_{k}^{-1}\mathbf{A}_{k}^{\top }%
\mathbf{A}_{0k}$. We begin with some preliminary facts. Firstly, note that
\begin{eqnarray*}
&&\frac{1}{Tp}\sum_{t=1}^{T}\left\Vert \mathbf{P}_{\mathbf{A}_{k}}\left(
\mathbf{A}_{k}\mathbf{F}_{k,t}\mathbf{B}_{k}^{\top }-\mathbf{A}_{0k}\mathbf{F%
}_{0k,t}\mathbf{B}_{0k}^{\top }\right) \right\Vert _{F}^{2} \\
&\leq &\frac{1}{Tp}\sum_{t=1}^{T}\left\Vert \mathbf{P}_{\mathbf{A}%
_{k}}\right\Vert _{F}^{2}\left\Vert \mathbf{A}_{k}\mathbf{F}_{k,t}\mathbf{B}%
_{k}^{\top }-\mathbf{A}_{0k}\mathbf{F}_{0k,t}\mathbf{B}_{0k}^{\top
}\right\Vert _{F}^{2} \\
&=&\frac{r_{k}}{Tp}\sum_{t=1}^{T}\left\Vert \mathbf{A}_{k}\mathbf{F}_{k,t}%
\mathbf{B}_{k}^{\top }-\mathbf{A}_{0k}\mathbf{F}_{0k,t}\mathbf{B}_{0k}^{\top
}\right\Vert _{F}^{2}\leq c_{0}d^{2}\left( \theta ,\theta _{0}\right) ,
\end{eqnarray*}%
which entails that%
\begin{eqnarray}
&&\frac{1}{Tp_{-k}}\sum_{t=1}^{T}\left\Vert \mathbf{F}_{k,t}\mathbf{B}%
_{k}^{\top }-\frac{1}{p_{k}}\mathbf{A}_{k}^{\top }\mathbf{A}_{0k}\mathbf{F}%
_{0k,t}\mathbf{B}_{0k}^{\top }\right\Vert _{F}^{2}  \label{equ:FB1} \\
&=&\frac{1}{Tp}\sum_{t=1}^{T}\left\Vert \mathbf{A}_{k}\mathbf{F}_{k,t}%
\mathbf{B}_{k}^{\top }-\frac{1}{p_{k}}\mathbf{A}_{k}\mathbf{A}_{k}^{\top }%
\mathbf{A}_{0k}\mathbf{F}_{0k,t}\mathbf{B}_{0k}^{\top }\right\Vert _{F}^{2}
\notag \\
&=&\frac{1}{Tp}\sum_{t=1}^{T}\left\Vert \mathbf{A}_{k}\mathbf{F}_{k,t}%
\mathbf{B}_{k}^{\top }-\mathbf{P}_{\mathbf{A}_{k}}\mathbf{A}_{0k}\mathbf{F}%
_{0k,t}\mathbf{B}_{0k}^{\top }\right\Vert _{F}^{2}  \notag \\
&=&\frac{1}{Tp}\sum_{t=1}^{T}\left\Vert \mathbf{P}_{\mathbf{A}_{k}}\left(
\mathbf{A}_{k}\mathbf{F}_{k,t}\mathbf{B}_{k}^{\top }-\mathbf{A}_{0k}\mathbf{F%
}_{0k,t}\mathbf{B}_{0k}^{\top }\right) \right\Vert _{F}^{2}\leq
c_{0}d^{2}\left( \theta ,\theta _{0}\right) .  \notag
\end{eqnarray}%
Similarly,
\begin{equation}
\frac{1}{Tp_{-k}}\sum_{t=1}^{T}\left\Vert \mathbf{F}_{0k,t}\mathbf{B}%
_{0k}^{\top }-\frac{1}{p_{k}}\mathbf{A}_{0k}^{\top }\mathbf{A}_{k}\mathbf{F}%
_{k,t}\mathbf{B}_{k}^{\top }\right\Vert _{F}^{2}\lesssim d^{2}\left( \theta
,\theta _{0}\right) .  \label{equ:FB2}
\end{equation}%
Note that $\mathbf{A}_{k}\mathbf{Z}_{k}=p_{k}^{-1}\mathbf{A}_{k}\mathbf{A}%
_{k}^{\top }\mathbf{A}_{0k}=\mathbf{P}_{\mathbf{A}_{k}}\mathbf{A}_{0k}$,
thus
\begin{eqnarray}
\mathbf{I}_{r_{k}} &=&\frac{1}{p_{k}}\mathbf{A}_{k}^{\top }\mathbf{A}_{k}=%
\frac{1}{p_{k}}\mathbf{A}_{0k}^{\top }\mathbf{A}_{0k}\pm \mathbf{Z}%
_{k}^{\top }\left( \frac{1}{p_{k}}\mathbf{A}_{k}^{\top }\mathbf{A}%
_{k}\right) \mathbf{Z}_{k}  \label{equ:I1} \\
&=&\mathbf{Z}_{k}^{\top }\mathbf{Z}_{k}+\frac{1}{p_{k}}\mathbf{A}_{0k}^{\top
}\mathbf{A}_{0k}-\frac{1}{p_{k}}\mathbf{A}_{0k}^{\top }\mathbf{A}_{k}\mathbf{%
Z}_{k}  \notag \\
&=&\mathbf{Z}_{k}^{\top }\mathbf{Z}_{k}+\mathbf{A}_{0k}^{\top }\frac{1}{p_{k}%
}\left( \mathbf{A}_{0k}-\mathbf{A}_{k}\mathbf{Z}_{k}\right)  \notag \\
&=&\mathbf{Z}_{k}^{\top }\mathbf{Z}_{k}+\frac{1}{p_{k}}\mathbf{A}_{0k}^{\top
}\mathbf{M}_{\mathbf{A}_{k}}\mathbf{A}_{0k}.  \notag
\end{eqnarray}%
Similarly,
\begin{equation}
\mathbf{I}_{r_{k}}=\mathbf{Z}_{k}\mathbf{Z}_{k}^{\top }+\mathbf{A}_{k}^{\top
}\frac{1}{p_{k}}\left( \mathbf{A}_{k}-\mathbf{A}_{0k}\mathbf{Z}_{k}^{\top
}\right) =\mathbf{Z}_{k}\mathbf{Z}_{k}^{\top }+\frac{1}{p_{k}}\mathbf{A}%
_{k}^{\top }\mathbf{M}_{\mathbf{A}_{0k}}\mathbf{A}_{k}.  \label{equpd}
\end{equation}%
In addition, it holds that%
\begin{eqnarray}
&&\frac{1}{Tp_{-k}}\sum_{t=1}^{T}\left( \mathbf{F}_{0k,t}\mathbf{B}%
_{0k}^{\top }\right) \left( \mathbf{F}_{0k,t}\mathbf{B}_{0k}^{\top }\right)
^{\top }  \label{equ:I1b} \\
&=&\frac{1}{Tp_{-k}}\sum_{t=1}^{T}\left( \mathbf{F}_{0k,t}\mathbf{B}%
_{0k}^{\top }\right) \left( \mathbf{F}_{0k,t}\mathbf{B}_{0k}^{\top }\right)
^{\top }-\mathbf{Z}_{k}^{\top }\left( \frac{1}{Tp_{-k}}\sum_{t=1}^{T}\left(
\mathbf{F}_{k,t}\mathbf{B}_{k}^{\top }\right) \left( \mathbf{F}_{k,t}\mathbf{%
B}_{k}^{\top }\right) ^{\top }\right) \mathbf{Z}_{k}  \notag \\
&&+\mathbf{Z}_{k}^{\top }\left( \frac{1}{Tp_{-k}}\sum_{t=1}^{T}\left(
\mathbf{F}_{k,t}\mathbf{B}_{k}^{\top }\right) \left( \mathbf{F}_{k,t}\mathbf{%
B}_{k}^{\top }\right) ^{\top }\right) \mathbf{Z}_{k}  \notag \\
&=&\mathbf{Z}_{k}^{\top }\left( \frac{1}{Tp_{-k}}\sum_{t=1}^{T}\left(
\mathbf{F}_{k,t}\mathbf{B}_{k}^{\top }\right) \left( \mathbf{F}_{k,t}\mathbf{%
B}_{k}^{\top }\right) ^{\top }\right) \left( \mathbf{Z}_{k}^{\top }\right)
^{-1}\mathbf{Z}_{k}^{\top }\mathbf{Z}_{k}+\frac{1}{Tp_{-k}}%
\sum_{t=1}^{T}\left( \mathbf{F}_{0k,t}\mathbf{B}_{0k}^{\top }-\mathbf{Z}_{k}%
\mathbf{F}_{k,t}\mathbf{B}_{k}^{\top }\right) \left( \mathbf{F}_{0k,t}%
\mathbf{B}_{0k}^{\top }\right) ^{\top }  \notag \\
&&+\frac{1}{Tp_{-k}}\sum_{t=1}^{T}\mathbf{Z}_{k}\mathbf{F}_{k,t}\mathbf{B}%
_{k}^{\top }\left( \mathbf{F}_{0k,t}\mathbf{B}_{0k}^{\top }-\mathbf{Z}_{k}%
\mathbf{F}_{k,t}\mathbf{B}_{k}^{\top }\right) ^{\top }  \notag \\
&=&\mathbf{Z}_{k}^{\top }\left( \frac{1}{Tp_{-k}}\sum_{t=1}^{T}\left(
\mathbf{F}_{k,t}\mathbf{B}_{k}^{\top }\right) \left( \mathbf{F}_{k,t}\mathbf{%
B}_{k}^{\top }\right) ^{\top }\right) \left( \mathbf{Z}_{k}^{\top }\right)
^{-1}  \notag \\
&&\mathbf{+Z}_{k}^{\top }\left( \frac{1}{Tp_{-k}}\sum_{t=1}^{T}\left(
\mathbf{F}_{k,t}\mathbf{B}_{k}^{\top }\right) \left( \mathbf{F}_{k,t}\mathbf{%
B}_{k}^{\top }\right) ^{\top }\right) \left( \mathbf{Z}_{k}^{\top }\right)
^{-1}\left( \mathbf{Z}_{k}^{\top }\mathbf{Z}_{k}-\mathbf{I}_{r_{k}}\right)
\notag \\
&&+\frac{1}{Tp_{-k}}\sum_{t=1}^{T}\left( \mathbf{F}_{0k,t}\mathbf{B}%
_{0k}^{\top }-\mathbf{Z}_{k}\mathbf{F}_{k,t}\mathbf{B}_{k}^{\top }\right)
\left( \mathbf{F}_{0k,t}\mathbf{B}_{0k}^{\top }\right) ^{\top }+\frac{1}{%
Tp_{-k}}\sum_{t=1}^{T}\mathbf{Z}_{k}\mathbf{F}_{k,t}\mathbf{B}_{k}^{\top
}\left( \mathbf{F}_{0k,t}\mathbf{B}_{0k}^{\top }-\mathbf{Z}_{k}\mathbf{F}%
_{k,t}\mathbf{B}_{k}^{\top }\right) ^{\top }  \notag
\end{eqnarray}%
Then, combining (\ref{equ:I1}) and (\ref{equ:I1b}), it follows that
\begin{equation*}
\left( \frac{1}{Tp_{-k}}\sum_{t=1}^{T}\mathbf{F}_{0k,t}\mathbf{B}_{0k}^{\top
}\left( \mathbf{F}_{0k,t}\mathbf{B}_{0k}^{\top }\right) ^{\top }+\mathbf{D}%
_{k}\right) \mathbf{Z}_{k}^{\top }=\mathbf{Z}_{k}^{\top }\left( \frac{1}{%
Tp_{-k}}\sum_{t=1}^{T}\mathbf{F}_{k,t}\mathbf{B}_{k}^{\top }\left( \mathbf{F}%
_{k,t}\mathbf{B}_{k}^{\top }\right) ^{\top }\right) ,
\end{equation*}%
where%
\begin{eqnarray*}
\mathbf{D}_{k} &=&\mathbf{Z}_{k}^{\top }\left( \frac{1}{Tp_{-k}}%
\sum_{t=1}^{T}\left( \mathbf{F}_{k,t}\mathbf{B}_{k}^{\top }\right) \left(
\mathbf{F}_{k,t}\mathbf{B}_{k}^{\top }\right) ^{\top }\right) \left( \mathbf{%
Z}_{k}^{\top }\right) ^{-1}\left( \mathbf{I}_{r_{k}}-\mathbf{Z}_{k}^{\top }%
\mathbf{Z}_{k}\right) \\
&&-\frac{1}{Tp_{-k}}\sum_{t=1}^{T}\left( \mathbf{F}_{0k,t}\mathbf{B}%
_{0k}^{\top }-\mathbf{Z}_{k}\mathbf{F}_{k,t}\mathbf{B}_{k}^{\top }\right)
\left( \mathbf{F}_{0k,t}\mathbf{B}_{0k}^{\top }\right) ^{\top } \\
&&-\frac{1}{Tp_{-k}}\sum_{t=1}^{T}\mathbf{Z}_{k}\mathbf{F}_{k,t}\mathbf{B}%
_{k}^{\top }\left( \mathbf{F}_{0k,t}\mathbf{B}_{0k}^{\top }-\mathbf{Z}_{k}%
\mathbf{F}_{k,t}\mathbf{B}_{k}^{\top }\right) ^{\top } \\
&=&\mathbf{Z}_{k}^{\top }\left( \frac{1}{Tp_{-k}}\sum_{t=1}^{T}\left(
\mathbf{F}_{k,t}\mathbf{B}_{k}^{\top }\right) \left( \mathbf{F}_{k,t}\mathbf{%
B}_{k}^{\top }\right) ^{\top }\right) \left( \mathbf{Z}_{k}^{\top }\right)
^{-1}\left( \frac{1}{p_{k}}\mathbf{A}_{0k}^{\top }\mathbf{M}_{\mathbf{A}_{k}}%
\mathbf{A}_{0k}\right) \\
&&-\frac{1}{Tp_{-k}}\sum_{t=1}^{T}\left( \mathbf{F}_{0k,t}\mathbf{B}%
_{0k}^{\top }-\mathbf{Z}_{k}\mathbf{F}_{k,t}\mathbf{B}_{k}^{\top }\right)
\left( \mathbf{F}_{0k,t}\mathbf{B}_{0k}^{\top }\right) ^{\top } \\
&&-\frac{1}{Tp_{-k}}\sum_{t=1}^{T}\mathbf{Z}_{k}\mathbf{F}_{k,t}\mathbf{B}%
_{k}^{\top }\left( \mathbf{F}_{0k,t}\mathbf{B}_{0k}^{\top }-\mathbf{Z}_{k}%
\mathbf{F}_{k,t}\mathbf{B}_{k}^{\top }\right) ^{\top }.
\end{eqnarray*}%
Combining the above with (\ref{equ:FB1}) and (\ref{equ:FB2}), it follows that%
\begin{eqnarray*}
\left\Vert \mathbf{D}_{k}\right\Vert _{F}^{2} &\leq &3\left\Vert \mathbf{Z}%
_{k}^{\top }\left( \frac{1}{Tp_{-k}}\sum_{t=1}^{T}\left( \mathbf{F}_{k,t}%
\mathbf{B}_{k}^{\top }\right) \left( \mathbf{F}_{k,t}\mathbf{B}_{k}^{\top
}\right) ^{\top }\right) \left( \mathbf{Z}_{k}^{\top }\right) ^{-1}\left(
\frac{1}{p_{k}}\mathbf{A}_{0k}^{\top }\mathbf{M}_{\mathbf{A}_{k}}\mathbf{A}%
_{0k}\right) \right\Vert _{F}^{2} \\
&&+3\left\Vert \frac{1}{Tp_{-k}}\sum_{t=1}^{T}\left( \mathbf{F}_{0k,t}%
\mathbf{B}_{0k}^{\top }-\mathbf{Z}_{k}\mathbf{F}_{k,t}\mathbf{B}_{k}^{\top
}\right) \left( \mathbf{F}_{0k,t}\mathbf{B}_{0k}^{\top }\right) ^{\top
}\right\Vert _{F}^{2} \\
&&+3\left\Vert \frac{1}{Tp_{-k}}\sum_{t=1}^{T}\mathbf{Z}_{k}\mathbf{F}_{k,t}%
\mathbf{B}_{k}^{\top }\left( \mathbf{F}_{0k,t}\mathbf{B}_{0k}^{\top }-%
\mathbf{Z}_{k}\mathbf{F}_{k,t}\mathbf{B}_{k}^{\top }\right) ^{\top
}\right\Vert _{F}^{2} \\
&\leq &c_{0}\left\Vert \mathbf{Z}_{k}\right\Vert _{F}^{2}\frac{p_{-k}}{%
T^{2}p_{-k}^{2}}\left\Vert \sum_{t=1}^{T}\mathbf{F}_{k,t}\mathbf{F}%
_{k,t}^{\top }\right\Vert _{F}^{2}\left( \left\Vert \mathbf{Z}%
_{k}\right\Vert _{F}^{2}\right) ^{-1}\left\Vert \frac{1}{p_{k}}\mathbf{A}%
_{0k}^{\top }\mathbf{M}_{\mathbf{A}_{k}}\mathbf{A}_{0k}\right\Vert _{F}^{2}
\\
&&+c_{1}\frac{1}{T^{2}p_{-k}^{2}}\left( \sum_{t=1}^{T}\left\Vert \mathbf{F}%
_{0k,t}\mathbf{B}_{0k}^{\top }-\mathbf{Z}_{k}\mathbf{F}_{k,t}\mathbf{B}%
_{k}^{\top }\right\Vert _{F}^{2}\right) \left( \sum_{t=1}^{T}\left\Vert
\mathbf{F}_{0k,t}\mathbf{B}_{0k}^{\top }\right\Vert _{F}^{2}\right) \\
&&+c_{2}\frac{1}{T^{2}p_{-k}^{2}}\left( \sum_{t=1}^{T}\left\Vert \mathbf{F}%
_{k,t}\mathbf{B}_{k}^{\top }\right\Vert _{F}^{2}\right) \left(
\sum_{t=1}^{T}\left\Vert \mathbf{F}_{0k,t}\mathbf{B}_{0k}^{\top }-\mathbf{Z}%
_{k}\mathbf{F}_{k,t}\mathbf{B}_{k}^{\top }\right\Vert _{F}^{2}\right) \\
&\leq &c_{3}d^{2}\left( \theta ,\theta _{0}\right) .
\end{eqnarray*}%
Note that the matrix $\sum_{t=1}^{T}\mathbf{F}_{0k,t}\mathbf{F}_{0k,t}^{\top
}/T$ is symmetric, and therefore diagonalisable, for each $T$; hence, the
matrix of its eigenvectors is invertible, with condition number $\kappa _{%
\mathbf{F}}<\infty $. Hence, by the Bauer-Fike theorem (%
\citealp{golub2013matrix}), there exists an eigenvalue of $\sum_{t=1}^{T}%
\mathbf{F}_{0k,t}\mathbf{F}_{0k,t}^{\top }/T$, say $\mu \left( \sum_{t=1}^{T}%
\mathbf{F}_{0k,t}\mathbf{F}_{0k,t}^{\top }/T\right) $, such that%
\begin{eqnarray*}
&&\left\vert \lambda _{\min }\left( \frac{1}{T}\sum_{t=1}^{T}\mathbf{F}_{k,t}%
\mathbf{F}_{k,t}^{\top }\right) -\mu \left( \frac{1}{T}\sum_{t=1}^{T}\mathbf{%
F}_{0k,t}\mathbf{F}_{0k,t}^{\top }\right) \right\vert \\
&=&\left\vert \lambda _{\min }\left( \frac{1}{T}\sum_{t=1}^{T}\left( \mathbf{%
F}_{k,t}\mathbf{B}_{k}^{\top }\right) \left( \mathbf{F}_{k,t}\mathbf{B}%
_{k}^{\top }\right) ^{\top }\right) -\mu \left( \frac{1}{T}%
\sum_{t=1}^{T}\left( \mathbf{F}_{0k,t}\mathbf{B}_{0k}^{\top }\right) \left(
\mathbf{F}_{0k,t}\mathbf{B}_{0k}^{\top }\right) ^{\top }\right) \right\vert
\\
&\leq &\kappa _{\mathbf{F}}\left\Vert \mathbf{D}_{k}\right\Vert _{2}\leq
\kappa _{\mathbf{F}}\left\Vert \mathbf{D}_{k}\right\Vert _{F}\leq \kappa _{%
\mathbf{F}}d\left( \theta ,\theta _{0}\right) .
\end{eqnarray*}%
Using Lemma \ref{semi-1}, it follows that
\begin{equation*}
\left\vert \lambda _{\min }\left( \frac{1}{T}\sum_{t=1}^{T}\mathbf{F}_{k,t}%
\mathbf{F}_{k,t}^{\top }\right) -\mu \left( \frac{1}{T}\sum_{t=1}^{T}\mathbf{%
F}_{0k,t}\mathbf{F}_{0k,t}^{\top }\right) \right\vert =o\left( 1\right) ,
\end{equation*}%
which entails that $\lambda _{\min }\left( \sum_{t=1}^{T}\mathbf{F}_{k,t}%
\mathbf{F}_{k,t}^{\top }/T\right) $ is bounded from below by a positive
constant on account of Assumption \ref{as-1}. Furthermore, using Weyl's
inequality, it follows that
\begin{equation*}
\frac{1}{T}\sum_{t=1}^{T}\lambda _{\min }\left( \mathbf{F}_{k,t}\mathbf{F}%
_{k,t}^{\top }\right) \geq \lambda _{\min }\left( \frac{1}{T}\sum_{t=1}^{T}%
\mathbf{F}_{k,t}\mathbf{F}_{k,t}^{\top }\right) >0,
\end{equation*}%
whence the desired result.
\end{proof}
\end{lemma}

\begin{lemma}
\label{exp-sup}{\it We assume that Assumptions \ref{as-1}-\ref{as-3} hold. Then
it holds that%
\begin{equation*}
\mathbb{E}\left( \sup_{\theta \in \Theta \left( \delta \right) }\left\vert
\mathbb{W}\left( \theta \right) \right\vert \right) \leq c_{0}\frac{\delta }{%
L^{1/2}},
\end{equation*}%
where recall that $\Theta \left( \delta \right) =\left\{ \theta \in \Theta
:d\left( \theta ,\theta _{0}\right) \leq \delta \right\} $.}

\begin{proof}
Repeating the proof of Lemma \ref{semi-1}, it can be shown that, for any $%
\theta _{a},~\theta _{b}\in \Theta $, it holds that
\begin{equation*}
\left\Vert \sqrt{Tp}\left\vert \mathbb{W}\left( \theta _{a}\right) -\mathbb{W%
}\left( \theta _{b}\right) \right\vert \right\Vert _{\psi _{2}}\leq
c_{0}d\left( \theta _{a},\theta _{b}\right) .
\end{equation*}%
By Assumption \ref{as-3}, the process $\mathbb{W}(\theta )$ is sub-Gaussian.
Further, Assumption \ref{as-1}\textit{(i)} entails that $\left( \cup
_{k=1}^{K}\mathbf{A}_{k}\right) \cup \mathcal{F}$ is closed, since it is the finite union of closed sets; since $\Theta $
is defined as a closed subset of $\left( \cup _{k=1}^{K}\mathbf{A}%
_{k}\right) \cup \mathcal{F}$, it follows that $\Theta $ also is closed. Then, by the continuity of $\mathbb{W}(\theta )$,
applying condition ($S_{3}^{\prime }$) on p. 171 in \citet{loeve2}, it
follows that $\mathbb{W}(\theta )$ is a separable process. Hence, applying
Corollary 2.2.8 in \citet{vanderVaart1996}
\begin{equation*}
\mathbb{E}\left( \sup_{\theta \in \Theta (\delta )}\sqrt{Tp}|\mathbb{W}%
(\theta )|\right) \lesssim \left\Vert \sup_{\theta \in \Theta (\delta )}%
\sqrt{Tp}|\mathbb{W}(\theta )|\right\Vert _{\psi _{2}}\lesssim
\int_{0}^{\delta }\sqrt{\log D(\epsilon ,d,\Theta (\delta ))}d\epsilon ,
\end{equation*}%
where recall that $D(\epsilon ,d,\Theta (\delta ))$ is the packing number of
the set $\Theta (\delta )$. We now show that
\begin{equation*}
\int_{0}^{\delta }\sqrt{\log D(\epsilon ,d,\Theta (\delta ))}d\epsilon /%
\sqrt{Tp}=O_{P}\left( \delta /\sqrt{L}\right) .
\end{equation*}%
Based on Lemma \ref{lemma2}, and the definition of $\Theta (\delta )$, it is
easy to verify that
\begin{equation*}
\Theta (\delta )\subset \bigcup_{\mathbf{U}_{1}\in \mathcal{S}_{1},\cdots ,%
\mathbf{U}_{K}\in \mathcal{S}_{K}}\Theta \left( \delta ;\mathbf{U}%
_{1},\cdots ,\mathbf{U}_{K}\right) ,
\end{equation*}%
where (similarly to the notation used in Lemma \ref{lemma2}): $\mathcal{S}%
_{k}=\Big\{ \mathbf{U}_{k}\in \mathbb{R}^{r_{k}\times r_{k}}:\mathbf{U}_{k}=%
\mbox{diag}\left( u_{1},...,u_{r_{k}}\right) ,$ $u_{j}\in \{-1,1\}\text{ for }%
j=1,...,r_{k}\Big\} $ for $1\leq k\leq K$ and
\begin{equation*}
\Theta \left( \delta ,\mathbf{U}_{1},...,\mathbf{U}_{K}\right) =\left\{
\theta \in \Theta :\sum_{k=1}^{K}\frac{1}{p_{k}}\left\Vert \mathbf{A}_{k}-%
\mathbf{A}_{0k}\mathbf{U}_{k}\right\Vert _{F}^{2}+\frac{1}{T^{2}}\left\Vert
\sum_{t=1}^{T}\left( {\mathcal{F}}_{t}-{\mathcal{F}}_{0t}\times _{k=1}^{K}%
\mathbf{U}_{k}\right) \right\Vert _{F}^{2}\leq C_{5}\delta ^{2}\right\} .
\end{equation*}%
Because, for each $1\leq k\leq K$, there are $2^{r_{k}}$ elements in $%
\mathcal{S}_{k}$, we only to need study
\begin{equation*}
\int_{0}^{\delta }\sqrt{\log D\left( \epsilon ,d,\Theta \left( \delta ,%
\mathbf{U}_{1},\cdots ,\mathbf{U}_{K}\right) \right) }d\epsilon ,
\end{equation*}%
at a single $\mathbf{U}_{k}\in \mathcal{S}_{k}$. Without loss of generality,
set $\mathbf{U}_{k}=\mathbf{I}_{r_{k}}$. For any $\theta _{a},\theta _{b}\in
\Theta $, we have%
\begin{eqnarray*}
&&d\left( \theta _{a},\theta _{b}\right)  \\
&=&\sqrt{\frac{1}{Tp}\sum_{t=1}^{T}\left\Vert \mathcal{F}_{at}\times
_{k=1}^{K}\mathbf{A}_{ak}-\mathcal{F}_{bt}\times _{k=1}^{K}\mathbf{A}%
_{bk}\right\Vert _{F}^{2}} \\
&=&\sqrt{\frac{1}{Tp}\sum_{t=1}^{T}\left\Vert \mathcal{F}_{at}\times
_{1}\left( \mathbf{A}_{a1}-\mathbf{A}_{b1}\right) \times _{k=2}^{K}\mathbf{A}%
_{ak}+\left( \mathcal{F}_{at}\times _{k=2}^{K}\mathbf{A}_{ak}-\mathcal{F}%
_{bt}\times _{k=2}^{K}\mathbf{A}_{bk}\right) \times _{1}\mathbf{A}%
_{b1}\right\Vert _{F}^{2}} \\
&\leq &c_{0}\frac{1}{\sqrt{Tp}}\left\Vert \mathbf{A}_{a1}-\mathbf{A}%
_{b1}\right\Vert _{F}\prod\limits_{k=2}^{K}\left\Vert \mathbf{A}%
_{ak}\right\Vert _{F}\sqrt{\sum_{t=1}^{T}\left\Vert \mathcal{F}%
_{at}\right\Vert _{F}^{2}}+ \\
&&c_{1}\frac{1}{\sqrt{Tp}}\left\Vert \mathbf{A}_{b1}\right\Vert _{F}\sqrt{%
\sum_{t=1}^{T}\left\Vert \mathcal{F}_{at}\times _{k=1}^{K}\mathbf{A}_{ak}-%
\mathcal{F}_{bt}\times _{k=1}^{K}\mathbf{A}_{bk}\right\Vert _{F}^{2}} \\
&\leq &c_{2}\left( \frac{1}{\sqrt{p_{1}}}\left\Vert \mathbf{A}_{a1}-\mathbf{A%
}_{b1}\right\Vert _{F}+\frac{1}{\sqrt{Tp_{-1}}}\sqrt{\sum_{t=1}^{T}\left\Vert \mathcal{F}_{at}\times _{k=2}^{K}\mathbf{A}_{ak}-\mathcal{F}%
_{bt}\times _{k=2}^{K}\mathbf{A}_{bk}\right\Vert _{F}^{2}}\right)  \\
&\leq &c_{3}\left( \sum_{k=1}^{K}\frac{1}{\sqrt{p_{k}}}\Vert \mathbf{A}_{ak}-%
\mathbf{A}_{bk}{\Vert}_{F}+\frac{1}{\sqrt{T}}\sqrt{%
\sum_{t=1}^{T}\left\Vert \mathcal{F}_{at}-\mathcal{F}_{bt}\right\Vert
_{F}^{2}}\right)
\end{eqnarray*}%
Define
\begin{equation*}
d^{\ast }\left( \theta _{a},\theta _{b}\right) =\sqrt{\frac{1}{p_{k}}%
\sum_{k=1}^{K}\Vert \mathbf{A}_{ak}-\mathbf{A}_{bk}{\Vert}_{F}^{2}+%
\frac{1}{T}\sum_{t=1}^{T}\left\Vert \mathcal{F}_{at}-\mathcal{F}%
_{bt}\right\Vert _{F}^{2}}.
\end{equation*}%
Repeated applications of the Jensen's inequality yield that there exists a $%
0<C_{6}<\infty $ such that
\begin{equation*}
d\left( \theta _{a},\theta _{b}\right) \leq C_{6}\cdot d^{\ast }\left(
\theta _{a},\theta _{b}\right) .
\end{equation*}%
Furthermore, it also holds that $\Theta \left( \delta ;\mathbf{I}%
_{r_{1}},...,\mathbf{I}_{r_{k}}\right) \subset \Theta ^{\ast }(\delta
)=\left\{ \theta \in \Theta :d^{\ast }\left( \theta ,\theta _{0}\right) \leq
C_{7}\delta \right\} $, where $C_{7}=C_{5}C_{4}$. Then it follows that
\begin{equation}
D\left( \epsilon ,d,\Theta \left( \delta ;\mathbf{I}_{r_{1}},\cdots ,\mathbf{%
I}_{r_{_{K}}}\right) \right) \leq D\left( \epsilon ,d^{\ast },\Theta \left(
\delta ;\mathbf{I}_{r_{1}},\cdots ,\mathbf{I}_{r_{_{K}}}\right) \right) \leq
D\left( \epsilon /2,d^{\ast },\Theta ^{\ast }(\delta )\right) \leq C\left(
\epsilon /4,d^{\ast },\Theta ^{\ast }(\delta )\right) \text{,}  \label{equ:D}
\end{equation}%
where $C\left( \epsilon /4,d^{\ast },\Theta ^{\ast }(\delta )\right) $ is
the covering number of $\Theta ^{\ast }(\delta )$. We now find an upper
bound for $C\left( \epsilon /4,d^{\ast },\Theta ^{\ast }(\delta )\right) $.
Let $\eta =\epsilon /4$, and $\theta _{1}^{\ast },\cdots ,\theta _{J}^{\ast }
$ be the maximal set in $\Theta ^{\ast }(\delta )$ such that $d^{\ast
}\left( \theta _{j}^{\ast },\theta _{l}^{\ast }\right) >\eta $, for all $%
j\neq l$. Set $B(\theta ,c)=\left\{ \gamma \in \Theta :d^{\ast }(\gamma
,\theta )\leq c\right\} $. Then $B\left( \theta _{1}^{\ast },\eta \right)
,\cdots ,B\left( \theta _{J}^{\ast },\eta \right) $ cover $\Theta ^{\ast
}(\delta )$ and $C\left( \epsilon /4,d^{\ast },\Theta ^{\ast }(\delta
)\right) \leq J$. Moreover, $B\left( \theta _{1}^{\ast },\eta /4\right)
,\cdots ,B\left( \theta _{J}^{\ast },\eta /4\right) $ are disjoint and it
can be verified that
\begin{equation}
\bigcup_{j=1}^{J}B\left( \theta _{j}^{\ast },\eta /4\right) \subset \Theta
^{\ast }(\delta +\eta /4).  \label{setes}
\end{equation}%
The volume of a ball belonging in $\mathbb{R}^{M}$ (where recall that $M$ is
defined in (\ref{def-m})) defined by the semimetric $d^{\ast }$ and with
radius $c$ is equal to $h_{M}c^{M}$, where $h_{M}$ is a constant, so (\ref%
{setes}) implies%
\begin{equation*}
J\cdot h_{M}\left( \frac{\eta }{4}\right) ^{M}\leq h_{M}\left( C_{7}\left(
\delta +\frac{\eta }{4}\right) \right) ^{M}.
\end{equation*}%
Hence we have
\begin{equation}
J\leq \left( \frac{C_{7}(4\delta +\eta )}{\eta }\right) ^{M}=\left( \frac{%
C_{7}(16\delta +\epsilon )}{\epsilon }\right) ^{M}\leq \left( \frac{%
C_{8}\delta }{\epsilon }\right) ^{M},  \label{equ:J}
\end{equation}%
for $\epsilon \leq \delta $. Combining (\ref{equ:D}) and (\ref{equ:J}), we
have%
\begin{eqnarray*}
&&\int_{0}^{\delta }\sqrt{\log D(\epsilon ,d,\Theta \left( \delta ;\mathbf{I}%
_{r_{1}},\cdots ,\mathbf{I}_{r_{_{K}}}\right) )}d\epsilon  \\
&\leq &\int_{0}^{\delta }\sqrt{\log C\left( \epsilon /4,d^{\ast },\Theta
^{\ast }(\delta )\right) }d\epsilon  \\
&\leq &M^{1/2}\int_{0}^{\delta }\sqrt{\log \left( \frac{C_{8}\delta }{%
\epsilon }\right) }d\epsilon \leq c_{0}\delta M^{1/2}.
\end{eqnarray*}%
Hence, finally
\begin{equation*}
\mathbb{E}\left( \sup_{\theta \in \Theta (\delta )}|\mathbb{W}(\theta
)|\right) \leq c_{0}\int_{0}^{\delta }\sqrt{\log D(\epsilon ,d,\Theta
(\delta ))}d\epsilon /\sqrt{Tp}\lesssim \frac{\delta }{\sqrt{L}},
\end{equation*}%
which entails the desired result.
\end{proof}
\end{lemma}

\begin{lemma}
\label{distance-huber}\textit{We assume that Assumptions \ref{as-1}-\ref%
{as-3} hold. Then, as $\min \left\{ p_{1},...,p_{K},T\right\} \rightarrow
\infty $, it holds that%
\begin{equation*}
d\left( \widehat{\theta }^{H},\theta _{0}\right) =o_{P}\left( 1\right) .
\end{equation*}
}

\begin{proof}
The proof similar to Lemma \ref{semi-1}.
The $\left( \mathbf{A}_{1},\cdots ,\mathbf{A}%
_{K},\left\{ \mathcal{F}_{t}\right\} _{t=1}^{T}\right) $ minimizing $%
L_{2}\left( \mathbf{A}_{1},\cdots ,\mathbf{A}_{K},\left\{ \mathcal{F}%
_{t}\right\} _{t=1}^{T}\right) $ is the same as that minimizing $L_{3}\left(
\mathbf{A}_{1},\cdots ,\mathbf{A}_{K},\left\{ \mathcal{F}_{t}\right\}
_{t=1}^{T}\right) $ where
\begin{equation*}
L_{3}\left( \mathbf{A}_{1},\cdots ,\mathbf{A}_{K},\left\{ \mathcal{F}%
_{t}\right\} _{t=1}^{T}\right) =\frac{1}{T}\sum_{t=1}^{T}\left[ H_{\tau}\left( \left\Vert \frac{\mathcal{E}_{t}+\mathcal{F}_{0t}\times
_{k=1}^{K}\mathbf{A}_{0k}-\mathcal{F}_{t}\times _{k=1}^{K}\mathbf{A}_{k}}{%
\sqrt{p}}\right\Vert _{F}\right) -H_{\tau }\left( \left\Vert \frac{%
\mathcal{E}_{t}}{\sqrt{p}}\right\Vert _{F}\right) \right] .
\end{equation*}%
Let $\xi _{t}$ be a variable lying between $\left\Vert p^{-1/2}\mathcal{E}%
_{t}\right\Vert _{F}$ and $\left\Vert p^{-1/2}\left( \mathcal{E}_{t}+%
\mathcal{F}_{0t}\times _{k=1}^{K}\mathbf{A}_{0k}-\mathcal{F}_{t}\times
_{k=1}^{K}\mathbf{A}_{k}\right) \right\Vert _{F}$. Using the Mean Value
Theorem, it follows that%
\begin{eqnarray*}
&&L_{3}\left( \mathbf{A}_{1},\cdots ,\mathbf{A}_{K},\left\{ \mathcal{F}%
_{t}\right\} _{t=1}^{T}\right) \\
&=&\frac{1}{T}\sum_{t=1}^{T}H_{\tau}\left( \xi _{t}\right) \left(
\left\Vert \frac{\mathcal{E}_{t}+\mathcal{F}_{0t}\times _{k=1}^{K}\mathbf{A}%
_{0k}-\mathcal{F}_{t}\times _{k=1}^{K}\mathbf{A}_{k}}{\sqrt{p}}\right\Vert
_{F}-\left\Vert \frac{\mathcal{E}_{t}}{\sqrt{p}}\right\Vert _{F}\right) \\
&=&\frac{1}{T}\sum_{t=1}^{T}H_{\tau }\left( \xi _{t}\right) \frac{1%
}{\sqrt{p}}\frac{\sum_{i_{1}=1}^{p_{1}}...\sum_{i_{K}=1}^{p_{K}}2\left(
\mathcal{F}_{0t}\times _{k=1}^{K}\boldsymbol{a}_{0k,i_{k}}^{\top }-\mathcal{F%
}_{t}\times _{k=1}^{K}\boldsymbol{a}_{k,i_{k}}^{\top }\right)
e_{t,i_{1},...,i_{K}}}{\left\Vert \mathcal{E}_{t}+\mathcal{F}_{0t}\times
_{k=1}^{K}\mathbf{A}_{0k}-\mathcal{F}_{t}\times _{k=1}^{K}\mathbf{A}%
_{k}\right\Vert _{F}+\left\Vert \mathcal{E}_{t}\right\Vert _{F}} \\
&&+\frac{1}{T}\sum_{t=1}^{T}H_{\tau }\left( \left\Vert \frac{%
\mathcal{E}_{t}}{\sqrt{p}}\right\Vert _{F}\right) \frac{1}{\sqrt{p}}\frac{%
\sum_{i_{1}=1}^{p_{1}}...\sum_{i_{K}=1}^{p_{K}}\left( \mathcal{F}_{0t}\times
_{k=1}^{K}\boldsymbol{a}_{0k,i_{k}}^{\top }-\mathcal{F}_{t}\times _{k=1}^{K}%
\boldsymbol{a}_{k,i_{k}}^{\top }\right) ^{2}}{\left\Vert \mathcal{E}_{t}+%
\mathcal{F}_{0t}\times _{k=1}^{K}\mathbf{A}_{0k}-\mathcal{F}_{t}\times
_{k=1}^{K}\mathbf{A}_{k}\right\Vert _{F}+\left\Vert \mathcal{E}%
_{t}\right\Vert _{F}} \\
&=&I\left( \theta \right) +II\left( \theta \right) .
\end{eqnarray*}%
Recall the definition of $\theta ^{\ast }=\left( \mathbf{A}_{1}^{\ast
},\cdots ,\mathbf{A}_{K}^{\ast },\left\{ \mathcal{F}_{t}^{\ast }\right\}
_{t=1}^{T}\right) $ in the proof of Lemma \ref{semi-1}, and define%
\begin{equation*}
I\left( \theta ^{\ast }\right) =\frac{1}{T}\sum_{t=1}^{T}H_{\tau }\left( \xi _{t}\right) \frac{1}{\sqrt{p}}\frac{\sum_{i_{1}=1}^{p_{1}}...%
\sum_{i_{K}=1}^{p_{K}}2\left( \mathcal{F}_{0t}\times _{k=1}^{K}\boldsymbol{a}%
_{0k,i_{k}}^{\top }-\mathcal{F}_{t}^{\ast }\times _{k=1}^{K}\boldsymbol{a}%
_{k,i_{k}}^{\ast \top }\right) e_{t,i_{1},...,i_{K}}}{\left\Vert \mathcal{E}%
_{t}+\mathcal{F}_{0t}\times _{k=1}^{K}\mathbf{A}_{0k}-\mathcal{F}_{t}\times
_{k=1}^{K}\mathbf{A}_{k}\right\Vert _{F}+\left\Vert \mathcal{E}%
_{t}\right\Vert _{F}}.
\end{equation*}%
It holds that%
\begin{eqnarray*}
&&I\left( \theta \right) -I\left( \theta ^{\ast }\right) \\
&=&\frac{1}{T}\sum_{t=1}^{T}H_{\tau }\left( \xi _{t}\right) \frac{1%
}{\sqrt{p}}\frac{\sum_{i_{1}=1}^{p_{1}}...\sum_{i_{K}=1}^{p_{K}}2\left(
\mathcal{F}_{t}^{\ast }\times _{k=1}^{K}\boldsymbol{a}_{k,i_{k}}^{\ast \top
}-\mathcal{F}_{t}\times _{k=1}^{K}\boldsymbol{a}_{k,i_{k}}^{\top }\right)
e_{t,i_{1},...,i_{K}}}{\left\Vert \mathcal{E}_{t}+\mathcal{F}_{0t}\times
_{k=1}^{K}\mathbf{A}_{0k}-\mathcal{F}_{t}\times _{k=1}^{K}\mathbf{A}%
_{k}\right\Vert _{F}+\left\Vert \mathcal{E}_{t}\right\Vert _{F}};
\end{eqnarray*}%
using the Cauchy-Schwartz inequality%
\begin{eqnarray*}
&&\left\vert I\left( \theta \right) -I\left( \theta ^{\ast }\right)
\right\vert \\
&\leq &\left( \frac{1}{T}\sum_{t=1}^{T}\left( H_{\tau }\left( \xi
_{t}\right) \right) ^{2}\right) ^{1/2}\left( \frac{1}{Tp}\sum_{t=1}^{T}%
\left( \frac{\sum_{i_{1}=1}^{p_{1}}...\sum_{i_{K}=1}^{p_{K}}2\left( \mathcal{%
F}_{t}^{\ast }\times _{k=1}^{K}\boldsymbol{a}_{k,i_{k}}^{\ast \top }-%
\mathcal{F}_{t}\times _{k=1}^{K}\boldsymbol{a}_{k,i_{k}}^{\top }\right)
e_{t,i_{1},...,i_{K}}}{\left\Vert \mathcal{E}_{t}+\mathcal{F}_{0t}\times
_{k=1}^{K}\mathbf{A}_{0k}-\mathcal{F}_{t}\times _{k=1}^{K}\mathbf{A}%
_{k}\right\Vert _{F}+\left\Vert \mathcal{E}_{t}\right\Vert _{F}}\right)
^{2}\right) ^{1/2} \\
&\leq &c_{0}\left( \frac{1}{Tp}\sum_{t=1}^{T}\left(
\sum_{i_{1}=1}^{p_{1}}...\sum_{i_{K}=1}^{p_{K}}2\left( \mathcal{F}_{t}^{\ast
}\times _{k=1}^{K}\boldsymbol{a}_{k,i_{k}}^{\ast \top }-\mathcal{F}%
_{t}\times _{k=1}^{K}\boldsymbol{a}_{k,i_{k}}^{\top }\right)
e_{t,i_{1},...,i_{K}}\right) ^{2}\right) ^{1/2},
\end{eqnarray*}%
having exploited the fact that $H_{\tau }\left( \cdot \right) $ is
bounded, and Assumption \ref{as-3}\textit{(iii)}. Note that%
\begin{eqnarray*}
&&\mathbb{E}\left[ \frac{1}{Tp}\sum_{t=1}^{T}\left(
\sum_{i_{1}=1}^{p_{1}}...\sum_{i_{K}=1}^{p_{K}}\left( \mathcal{F}_{t}^{\ast
}\times _{k=1}^{K}\boldsymbol{a}_{k,i_{k}}^{\ast \top }-\mathcal{F}%
_{t}\times _{k=1}^{K}\boldsymbol{a}_{k,i_{k}}^{\top }\right)
e_{t,i_{1},...,i_{K}}\right) ^{2}\right] \\
&=&\mathbb{E}\left[ \frac{1}{Tp}\sum_{t=1}^{T}\sum_{i_{1},i_{1}^{\prime
}=1}^{p_{1}}...\sum_{i_{K},i_{K}^{\prime }=1}^{p_{K}}\left( \mathcal{F}%
_{t}^{\ast }\times _{k=1}^{K}\boldsymbol{a}_{k,i_{k}}^{\ast \top }-\mathcal{F%
}_{t}\times _{k=1}^{K}\boldsymbol{a}_{k,i_{k}}^{\top }\right) \times \right.
\\
&&\left. \left( \mathcal{F}_{t}^{\ast }\times _{k=1}^{K}\boldsymbol{a}%
_{k,i_{k}^{\prime }}^{\ast \top }-\mathcal{F}_{t}\times _{k=1}^{K}%
\boldsymbol{a}_{k,i_{^{\prime }k}}^{\top }\right)
e_{t,i_{1},...,i_{K}}e_{t,i_{^{\prime }1},...,i_{K}^{\prime }}\right] \\
&=&\frac{1}{Tp}\sum_{t=1}^{T}\sum_{i_{1},i_{1}^{\prime
}=1}^{p_{1}}...\sum_{i_{K},i_{K}^{\prime }=1}^{p_{K}}\left( \mathcal{F}%
_{t}^{\ast }\times _{k=1}^{K}\boldsymbol{a}_{k,i_{k}}^{\ast \top }-\mathcal{F%
}_{t}\times _{k=1}^{K}\boldsymbol{a}_{k,i_{k}}^{\top }\right) ^{2}\mathbb{E}%
\left( e_{t,i_{1},...,i_{K}}^{2}\right) \\
&=&c_{0}\frac{1}{Tp}\sum_{t=1}^{T}\sum_{i_{1},i_{1}^{\prime
}=1}^{p_{1}}...\sum_{i_{K},i_{K}^{\prime }=1}^{p_{K}}\left( \mathcal{F}%
_{t}^{\ast }\times _{k=1}^{K}\boldsymbol{a}_{k,i_{k}}^{\ast \top }-\mathcal{F%
}_{t}\times _{k=1}^{K}\boldsymbol{a}_{k,i_{k}}^{\top }\right) ^{2},
\end{eqnarray*}%
by Assumption \ref{as-3}. Using (\ref{bound-factors}), it follows that%
\begin{equation*}
\frac{1}{Tp}\sum_{t=1}^{T}\sum_{i_{1},i_{1}^{\prime
}=1}^{p_{1}}...\sum_{i_{K},i_{K}^{\prime }=1}^{p_{K}}\left( \mathcal{F}%
_{t}^{\ast }\times _{k=1}^{K}\boldsymbol{a}_{k,i_{k}}^{\ast \top }-\mathcal{F%
}_{t}\times _{k=1}^{K}\boldsymbol{a}_{k,i_{k}}^{\top }\right) ^{2}\leq
c_{0}\epsilon ,
\end{equation*}%
whence finally, by convexity%
\begin{equation}
\mathbb{E}\sup_{\theta \in \Theta }\left\vert I(\theta )-I\left( \theta
^{\ast }\right) \right\vert =\epsilon ^{1/2}O(1).  \label{th2-1}
\end{equation}%
We now define the following variables (and the corresponding short-hand
notation)
\begin{equation}
\zeta _{t}=\frac{2H_{\tau }^{\prime }\left( \xi _{t}\right) }{%
\left\Vert \mathcal{E}_{t}+\mathcal{F}_{0t}\times _{k=1}^{K}\mathbf{A}_{0k}-%
\mathcal{F}_{t}\times _{k=1}^{K}\mathbf{A}_{k}\right\Vert _{F}/\sqrt{p}%
+\left\Vert \mathcal{E}_{t}\right\Vert \Vert _{F}/\sqrt{p}}=\frac{2H_{\tau
}^{\prime }\left( \xi _{t}\right) }{\zeta _{td,1}},  \label{zeta}
\end{equation}%
and
\begin{equation}
\zeta _{t}^{0}=\frac{2H_{\tau }^{\prime }\left( \xi
_{t}^{0}\right) }{\sqrt{d_{t}^{2}(\theta ,\theta _{0})+\sigma ^{2}}+\sigma }=%
\frac{2H_{\tau }^{\prime }\left( \xi _{t}^{0}\right) }{\zeta
_{td,0}},  \label{zeta-0}
\end{equation}%
where: we have defined for short%
\begin{equation*}
\sigma ^{2}=\mathbb{E}\left( e_{t,i_{1},\cdots ,i_{K}}^{2}|\mathbf{F}\right)
,
\end{equation*}%
with $\mathbf{F}=\left\{ \mathcal{F}_{1},\ldots ,\mathcal{F}_{T}\right\} $; $%
\xi _{t}^{0}$ is the conditional mean of $\xi _{t}$ on $\mathbf{F}$; and
\begin{equation*}
d_{t}^{2}\left( \theta ,\theta _{0}\right) =\frac{1}{p}%
\sum_{i_{1}=1}^{p_{1}}\cdots \sum_{i_{K}=1}^{p_{K}}\left( \mathcal{F}%
_{0t}\times _{k=1}^{K}\mathbf{b}_{0k,i_{k}}^{\top }-\mathcal{F}_{t}\times
_{k=1}^{K}\mathbf{b}_{k,i_{k}}^{\top }\right) ^{2}.
\end{equation*}%
We also note that $\zeta _{td,0}>0$ a.s. by definition, and that Assumption %
\ref{as-3}\textit{(iii)} ensures also that $\zeta _{td,1}>0$ a.s.\ Let now $%
D_{t}=\zeta _{t}-\zeta _{t}^{0}$ and note that, by construction, $\left\vert
\zeta _{t}^{0}\right\vert \leq c_{0}$ a.s. Then $I(\theta )$ can be
rewritten as%
\begin{eqnarray}
I\left( \theta \right) &=&\frac{1}{Tp}\sum_{t=1}^{T}\sum_{i_{1}=1}^{p_{1}}%
\cdots \sum_{i_{K}=1}^{p_{K}}\zeta _{t}^{0}\left( \mathcal{F}_{0t}\times
_{k=1}^{K}\mathbf{a}_{0k,i_{k}}^{\top }-\mathcal{F}_{t}\times _{k=1}^{K}%
\mathbf{a}_{k,i_{k}}^{\top }\right) e_{t,i_{1},\cdots ,i_{K}}  \label{i-dec}
\\
&&+\frac{1}{Tp}\sum_{t=1}^{T}\sum_{i_{1}=1}^{p_{1}}\cdots
\sum_{i_{K}=1}^{p_{K}}D_{t}\left( \mathcal{F}_{0t}\times _{k=1}^{K}\mathbf{a}%
_{0k,i_{k}}^{\top }-\mathcal{F}_{t}\times _{k=1}^{K}\mathbf{a}%
_{k,i_{k}}^{\top }\right) e_{t,i_{1},\cdots ,i_{K}}  \notag \\
&=&I_{1}\left( \theta \right) +I_{2}\left( \theta \right) .  \notag
\end{eqnarray}%
Then for any fixed $\theta $ and for all $\lambda \in \mathbb{R}$ and all $%
\delta >0$, it holds that%
\begin{eqnarray*}
&&\mathbb{P}\left( I_{1}\left( \theta \right) >\delta \right) \\
&\leq &\exp \left( -\lambda \delta \right) \mathbb{E}\exp \left( \lambda
I_{1}\left( \theta \right) \right) \\
&=&\exp \left( -\lambda \delta \right) \prod\limits_{t=1}^{T}\mathbb{E}\exp
\left( \frac{\lambda \zeta _{t}^{0}}{Tp}\sum_{i_{1}=1}^{p_{1}}\cdots
\sum_{i_{K}=1}^{p_{K}}\left( \mathcal{F}_{0t}\times _{k=1}^{K}\mathbf{a}%
_{0k,i_{k}}^{\top }-\mathcal{F}_{t}\times _{k=1}^{K}\mathbf{a}%
_{k,i_{k}}^{\top }\right) e_{t,i_{1},\cdots ,i_{K}}\right) \\
&=&\exp \left( -\lambda \delta +\left( \lambda \zeta _{t}^{0}\right) ^{2}%
\frac{d^{2}\left( \theta ,\theta _{0}\right) }{Tp}\right) ,
\end{eqnarray*}%
having used the Hoeffding bound in the last passage. Setting $\lambda
=c\delta Tp/d^{2}\left( \theta ,\theta _{0}\right) $ for small enough $c$,
it follows that%
\begin{equation*}
\mathbb{P}\left( I_{1}\left( \theta \right) >\delta \right) \leq \exp \left(
-\left( c\delta ^{2}-c^{2}\delta ^{2}\left( \zeta _{t}^{0}\right)
^{2}\right) \frac{Tp}{d^{2}\left( \theta ,\theta _{0}\right) }\right) \leq
\exp \left( -C_{10}\frac{\delta ^{2}Tp}{d^{2}\left( \theta ,\theta
_{0}\right) }\right) ,
\end{equation*}%
for some $C_{10}>0$. Using the same arguments as in the proof of Lemma \ref%
{semi-1}, it can be shown that
\begin{equation*}
\Vert I_{1}(\theta )\Vert _{\psi _{2}}\lesssim \frac{1}{\sqrt{Tp}}d\left(
\theta ,\theta _{0}\right) ;
\end{equation*}%
combining this with (\ref{th2-1}), it also follows that
\begin{equation}
\mathbb{E}\sup_{\theta \in \Theta }|I_{1}(\theta ^{\ast })|=O(1/\sqrt{L}).
\label{th2-2}
\end{equation}%
Also, letting $\Theta (\delta )=\{\theta \in \Theta :d(\theta ,\theta
_{0})\leq \delta \}$, and following the proof of Lemma \ref{exp-sup}, it
holds that%
\begin{equation}
\mathbb{E}\sup_{\theta \in \Theta (\delta )}|I_{1}(\theta )|=O(\delta /\sqrt{%
L}).  \label{equ:subI1}
\end{equation}%
We now consider $I_{2}(\theta )$ in (\ref{i-dec}). We write%
\begin{eqnarray}
D_{t} &=&2\frac{H_{\tau }^{\prime }(\xi _{t})\zeta _{td,0}-H_{\tau
}^{\prime }(\xi _{t})\zeta _{td,1}+H_{\tau }^{\prime
}(\xi _{t})\zeta _{td,1}-H_{\tau }^{\prime }(\xi _{t}^{0})\zeta
_{td,0}}{\zeta _{td,0}\zeta _{td,1}}  \label{dt} \\
&=&-\frac{1}{p}\dfrac{4H_{\tau }^{\prime }(\xi
_{t})\sum_{i_{1}=1}^{p_{1}}\cdots \sum_{i_{K}=1}^{p_{K}}(\mathcal{F}%
_{0t}\times _{k=1}^{K}\mathbf{b}_{0k,i_{k}}^{\top }-\mathcal{F}_{t}\times
_{k=1}^{K}\mathbf{b}_{k,i_{k}}^{\top }\mathring{)}e_{t,i_{1},\cdots ,i_{K}}}{%
\zeta _{td,1}\zeta _{td,0}\left( \zeta _{td,1}-\Vert \mathcal{E}_{t}/\sqrt{p}%
\Vert _{F}+\zeta _{td,0}-\sigma \right) }  \notag \\
&&-2\frac{H_{\tau }^{\prime }(\xi _{t})\left( \Vert \mathcal{E}%
_{t}/\sqrt{p}\Vert _{F}^{2}-\sigma ^{2}\right) }{\zeta _{td,1}\zeta
_{td,0}\left( \zeta _{td,1}-\Vert \mathcal{E}_{t}/\sqrt{p}\Vert _{F}+\zeta
_{td,0}-\sigma \right) }-2\frac{H_{\tau }^{\prime }(\xi
_{t})\left( \Vert \mathcal{E}_{t}/\sqrt{p}\Vert _{F}-\sigma \right) }{\zeta
_{td,1}\zeta _{td,0}}+2\frac{H_{\tau }^{\prime }(\xi _{t})-H_{\tau
}^{\prime }(\xi _{t}^{0})}{\zeta _{td,0}}  \notag \\
&=&D_{t,1}+D_{t,2},  \notag
\end{eqnarray}%
with%
\begin{eqnarray*}
D_{t,1} &=&-4H_{\tau }^{\prime }(\xi _{t})\dfrac{%
p^{-1}\sum_{i_{1}=1}^{p_{1}}\cdots \sum_{i_{K}=1}^{p_{K}}\left( \mathcal{F}%
_{0t}\times _{k=1}^{K}\mathbf{b}_{0k,i_{k}}^{\top }-\mathcal{F}_{t}\times
_{k=1}^{K}\mathbf{b}_{k,i_{k}}^{\top }\right) e_{t,i_{1},\cdots ,i_{K}}}{%
\zeta _{td,1}\zeta _{td,0}(\zeta _{td,1}-\Vert \mathcal{E}_{t}/\sqrt{p}\Vert
_{F}+\zeta _{td,0}-\sigma \mathring{)}} \\
&=&-4H_{\tau }(\xi _{t})\dfrac{D_{t,1}^{(1)}}{%
D_{t,1}^{(2)}},
\end{eqnarray*}%
and $D_{t,2}$ the remainder in (\ref{dt}). Let
\begin{equation*}
I_{2}(\theta )=I_{2,1}(\theta )+I_{2,2}(\theta ),
\end{equation*}%
where $I_{2,1}(\theta )$ ($I_{2,2}(\theta )$) is defined by replacing $D_{t}$
with $D_{t,1}$ ($D_{t,2}$) in $I_{2}(\theta )$. Then
\begin{equation*}
|I_{2,1}(\theta )|=\frac{1}{T}\sum_{t=1}^{T}\dfrac{4H_{\tau }^{\prime }(\xi _{t})(D_{t,1}^{(1)})^{2}}{D_{t,1}^{\left(
2\right) }}\leq C_{11}\frac{1}{T}\sum_{t=1}^{T}(D_{t,1}^{(1)})%
^{2}.
\end{equation*}%
Similar passages as in the proof of Lemma \ref{semi-1} - with $T=1$ - yield $%
\mathbb{E}\sup_{\theta \in \Theta }\left\vert D_{t,1}^{(1)}\right\vert
=O\left( 1/\sqrt{\min \left\{ p_{-1},...,p_{-K}\right\} }\right) $ and $%
\mathbb{E}\sup_{\theta \in \Theta \left( \delta \right) }\left\vert
D_{t,1}^{(1)}\right\vert =O\left( \delta /\sqrt{\min \left\{
p_{-1},...,p_{-K}\right\} }\right) $, and therefore%
\begin{eqnarray}
&&\mathbb{E}\sqrt{\sup_{\theta \in \Theta }\left\vert I_{2,1}(\theta ^{\ast
})\right\vert }  \label{th2-3} \\
&\leq &\sqrt{\mathbb{E}\sup_{\theta \in \Theta }\left\vert I_{2,1}(\theta
^{\ast })\right\vert }  \notag \\
&\leq &c_{0}p^{-1/2}\sqrt{\frac{1}{T}\sum_{t=1}^{T}\sup_{\theta \in \Theta }%
\frac{1}{p}\sum_{i_{1}=1}^{p_{1}}\cdots \sum_{i_{K}=1}^{p_{K}}\left(
\mathcal{F}_{0t}\times _{k=1}^{K}\mathbf{b}_{0k,i_{k}}^{\top }-\mathcal{F}%
_{t}\times _{k=1}^{K}\mathbf{b}_{k,i_{k}}^{\top }\right) ^{2}}  \notag \\
&=&O\left( 1/\sqrt{\min \{p_{-1},p_{-2},\cdots ,p_{-K}\}}\right) ,  \notag
\end{eqnarray}%
and%
\begin{eqnarray*}
\mathbb{E}\sqrt{\sup_{\theta \in \Theta \left( \delta \right) }\left\vert
I_{2,1}(\theta )\right\vert } &\leq &\sqrt{\mathbb{E}\sup_{\theta \in \Theta
\left( \delta \right) }\left\vert I_{2,1}(\theta )\right\vert } \\
&\leq &c_{0}p^{-1/2}\sqrt{\frac{1}{T}\sum_{t=1}^{T}\sup_{\theta \in \Theta
\left( \delta \right) }\frac{1}{p}\sum_{i_{1}=1}^{p_{1}}\cdots
\sum_{i_{K}=1}^{p_{K}}\left( \mathcal{F}_{0t}\times _{k=1}^{K}\mathbf{b}%
_{0k,i_{k}}^{\top }-\mathcal{F}_{t}\times _{k=1}^{K}\mathbf{b}%
_{k,i_{k}}^{\top }\right) ^{2}} \\
&=&O\left( \delta /\sqrt{\min \{p_{-1},p_{-2},\cdots ,p_{-K}\}}\right) .
\end{eqnarray*}%
As fas as $I_{2,2}(\theta )$ is concerned%
\begin{eqnarray*}
\mathbb{E}\sup_{\theta \in \Theta }\left\vert I_{2,2}(\theta )\right\vert
&\leq &\mathbb{E}\sup_{\theta \in \Theta }\left\vert D_{t,2}\right\vert
\sup_{\theta \in \Theta }\left\vert \frac{1}{Tp}\sum_{t=1}^{T}%
\sum_{i_{1}=1}^{p_{1}}\cdots \sum_{i_{K}=1}^{p_{K}}\left( \mathcal{F}%
_{0t}\times _{k=1}^{K}\mathbf{b}_{0k,i_{k}}^{\top }-\mathcal{F}_{t}\times
_{k=1}^{K}\mathbf{b}_{k,i_{k}}^{\top }\right)
e_{t,i_{1},...,i_{K}}\right\vert \\
&\leq &\sqrt{\mathbb{E}\sup_{\theta \in \Theta }\left\vert
D_{t,2}\right\vert ^{2}}\sqrt{\mathbb{E}\sup_{\theta \in \Theta }\left(
\frac{1}{Tp}\sum_{t=1}^{T}\sum_{i_{1}=1}^{p_{1}}\cdots
\sum_{i_{K}=1}^{p_{K}}\left( \mathcal{F}_{0t}\times _{k=1}^{K}\mathbf{b}%
_{0k,i_{k}}^{\top }-\mathcal{F}_{t}\times _{k=1}^{K}\mathbf{b}%
_{k,i_{k}}^{\top }\right) e_{t,i_{1},...,i_{K}}\right) ^{2}} \\
&\leq &\frac{1}{T}\sum_{t=1}^{T}\sqrt{\mathbb{E}\sup_{\theta \in \Theta
}\left\vert D_{t,2}\right\vert ^{2}}\sqrt{\mathbb{E}\sup_{\theta \in \Theta
}d_{t}^{2}\left( \theta ,\theta _{0}\right) \frac{1}{p}%
\sum_{i_{1}=1}^{p_{1}}\cdots \sum_{i_{K}=1}^{p_{K}}e_{t,i_{1},...,i_{K}}^{2}}%
.
\end{eqnarray*}%
Applying the Delta method, it is easy to see that $\mathbb{E}\sup_{\theta
\in \Theta }\left\vert D_{t,2}\right\vert ^{2}=O\left( 1/p\right) $; using
Assumption \ref{as-3}, it now follows that%
\begin{equation}
\mathbb{E}\sup_{\theta \in \Theta }\left\vert I_{2,2}(\theta ^{\ast
})\right\vert =O\left( p^{-1/2}\right) ,  \label{th2-4}
\end{equation}%
and
\begin{equation*}
\mathbb{E}\sup_{\theta \in \Theta \left( \delta \right) }\left\vert
I_{2,2}(\theta )\right\vert =O\left( \delta p^{-1/2}\right) .
\end{equation*}%
Hence, using now (\ref{th2-1}), (\ref{th2-2}), (\ref{th2-3}) and (\ref{th2-4}%
), it holds that%
\begin{eqnarray*}
&&\mathbb{P}\left( \sup_{\theta \in \Theta }\left\vert I\left( \theta
\right) \right\vert \geq \epsilon ^{\prime }\right) \\
&\leq &\mathbb{P}\left( \sup_{\theta \in \Theta }\left\vert I\left( \theta
\right) -I\left( \theta ^{\ast }\right) \right\vert \geq \frac{\epsilon
^{\prime }}{4}\right) +\mathbb{P}\left( \sup_{\theta \in \Theta }\left\vert
I\left( \theta ^{\ast }\right) \right\vert \geq \frac{3\epsilon ^{\prime }}{4%
}\right) \\
&\leq &\mathbb{P}\left( \sup_{\theta \in \Theta }\left\vert I\left( \theta
\right) -I\left( \theta ^{\ast }\right) \right\vert \geq \frac{\epsilon
^{\prime }}{4}\right) +\mathbb{P}\left( \sup_{\theta \in \Theta }\left\vert
I_{1}\left( \theta ^{\ast }\right) \right\vert \geq \frac{\epsilon ^{\prime }%
}{4}\right) +\mathbb{P}\left( \sup_{\theta \in \Theta }\left\vert
I_{2,1}\left( \theta ^{\ast }\right) \right\vert \geq \frac{\epsilon
^{\prime }}{4}\right) +\mathbb{P}\left( \sup_{\theta \in \Theta }\left\vert
I_{2,2}\left( \theta ^{\ast }\right) \right\vert \geq \frac{\epsilon
^{\prime }}{4}\right) \\
&\leq &\frac{4}{\epsilon ^{\prime }}\mathbb{E}\sup_{\theta \in \Theta
}\left\vert I\left( \theta \right) -I\left( \theta ^{\ast }\right)
\right\vert +\frac{4}{\epsilon ^{\prime }}\mathbb{E}\sup_{\theta \in \Theta
}\left\vert I_{1}\left( \theta ^{\ast }\right) \right\vert +\frac{4}{%
\epsilon ^{\prime }}\mathbb{E}\sup_{\theta \in \Theta }\left\vert
I_{2,1}\left( \theta ^{\ast }\right) \right\vert +\frac{4}{\epsilon ^{\prime
}}\mathbb{E}\sup_{\theta \in \Theta }\left\vert I_{2,2}\left( \theta ^{\ast
}\right) \right\vert \\
&\leq &c_{0}\left( \frac{4}{\epsilon ^{\prime }}\epsilon ^{1/2}+\frac{4}{%
\epsilon ^{\prime }}L^{-1/2}+\frac{4}{\epsilon ^{\prime }}\left( \sqrt{\min
\{p_{-1},p_{-2},\cdots ,p_{-K}\}}\right) ^{-1/2}+\frac{4}{\epsilon ^{\prime }%
}p^{-1/2}\right) ;
\end{eqnarray*}%
upon noting that $\epsilon $ in (\ref{th2-1}) and $\epsilon ^{\prime }$ are
arbitrary, the result that $\sup_{\theta \in \Theta }|I(\theta )|=o_{P}(1)$
readily obtains. Because of the boundness of $\zeta _{t}$, $II(\theta
)=C_{12}d^{2}(\theta ,\theta _{0})$ for some constant $C_{12}>0$. Noting
that $L_{3}\left( \widehat{\theta }^{H}\right) =I\left( \widehat{\theta }%
^{H}\right) +II\left( \widehat{\theta }^{H}\right) \leq L_{3}\left( \theta
_{0}\right) =0$, then we finally obtain that
\begin{equation}
d^{2}\left( \widehat{\theta }^{H},\theta _{0}\right) \leq c_{0}II\left(
\widehat{\theta }^{H}\right) \leq \sup_{\theta \in \Theta }\left\vert
I(\theta )\right\vert =o_{P}(1).  \label{equ:d=op}
\end{equation}
\end{proof}
\end{lemma}

\begin{lemma}
\label{er-ls}We assume that Assumptions \ref{as-1}-\ref{as-3} hold. Then, as
$\min \left\{ T,p_{1},...,p_{K}\right\} \rightarrow \infty $, it holds that,
for all $j\leq r_{k}$%
\begin{equation}
\lambda _{j}\left( \widehat{\mathbf{M}}_{k}\right) =c_{k}+o_{P}\left(
1\right) ,  \label{er-ls-1}
\end{equation}%
for some $c_{k}>0$; for all $j>r_{k}$, it holds that%
\begin{equation}
\lambda _{j}\left( \widehat{\mathbf{M}}_{k}\right) =O_{P}\left( \omega
_{k}^{-1/2}\right) .  \label{er-ls-2}
\end{equation}%
\begin{proof}

The proof is similar to that of Lemma \ref{er-h} below and, similarly to
that lemma, not all rates are the sharpest although they suffice for our
purposes. Hence, we report only some passages in which the two proofs
differ. We begin by studying the consistency of the (overdimensioned) Least
Squares estimates of $\mathbf{A}_{k}$ and $\mathbf{B}_{k}$. For $\theta
_{a}\in \Theta ^{M}$, $\theta _{b}\in \Theta ^{M}$, let
\begin{equation*}
d_{M}\left( \theta _{a},\theta _{b}\right) =\left( Tp\right)
^{-1/2}\left\Vert \mathcal{F}_{at}\times _{k=1}^{K}\mathbf{A}_{ak}-\mathcal{F%
}_{bt}\times _{k=1}^{K}\mathbf{A}_{bt}\right\Vert _{F}.
\end{equation*}%
We denote%
\begin{equation}
\left( \mathcal{F}_{0t}^{M}\right) _{i_{1},...,i_{K}}=\left\{
\begin{array}{ll}
\left( \mathcal{F}_{0t}\right) _{i_{1},...,i_{K}} & 1\leq i_{1}\leq
r_{1},...,1\leq i_{K}\leq r_{K} \\
0 & \text{others}%
\end{array}%
\right. ,  \label{fM_0t}
\end{equation}%
and, for $m_{k}\geq r_{k}$
\begin{equation}
\mathbf{A}_{0k}^{M}=\left( \mathbf{A}_{0k}|\mathbf{0}_{p_{k}\times \left(
m_{k}-r_{k}\right) }\right) ,  \label{a0MK}
\end{equation}%
whence also%
\begin{equation}
\theta _{0}^{M}=\left( \mathbf{A}_{01}^{M},\cdots ,\mathbf{A}_{0K}^{M},%
\mathcal{F}_{01}^{M},\cdots ,\mathcal{F}_{0T}^{M}\right) ,  \label{theta0_M}
\end{equation}%
and $\mathbf{B}_{0k}^{M}$ can be defined accordingly. Let $\widehat{\theta }%
^{M}$ be the Least Squares estimator of $\theta _{0}^{M}$ satisfying (\ref%
{kkt-1-2}). By repeating \textit{verbatim} the proof of Theorem \ref%
{ls-theorem}, it can be shown that
\begin{equation*}
d_{M}\left( \widehat{\theta }^{M},\theta _{0}^{M}\right) =O_{P}\left(
L^{-1/2}\right) .
\end{equation*}%
Also, following the proof of Lemma \ref{lemma2}, for any $\theta ^{M}\in
\Theta ^{M}(\delta )$ with $\Theta ^{M}(\delta )=\left\{ \theta ^{M}\in
\Theta ^{M}:d_{M}\left( \theta ^{M},\theta _{0}^{M}\right) \leq \delta
\right\} $, it holds that, for all $1\leq k\leq K$
\begin{eqnarray}
\frac{1}{p_{k}}\left\Vert \widehat{\mathbf{A}}_{k}^{M}-\mathbf{A}_{0k}^{M}%
\widehat{\mathbf{S}}_{k}^{M}\right\Vert _{F}^{2} &\lesssim &\delta ^{2},
\label{nf1_app} \\
\frac{1}{T}\sum_{t=1}^{T}\left\Vert \widehat{\mathcal{F}}_{t}^{M}-\mathcal{F}%
_{0t}^{M}\times _{k=1}^{K}\widehat{\mathbf{S}}_{k}^{M}\right\Vert _{F}^{2}
&\lesssim &\delta ^{2},  \label{nf2_app}
\end{eqnarray}%
where $\widehat{\mathbf{S}}_{k}^{M}=\func{sgn}\left\{ \left( \widehat{%
\mathbf{A}}_{k}^{M}\right) ^{\top }\mathbf{A}_{0k}^{M}/p_{k}\right\} $, so
that%
\begin{equation}
\frac{1}{p_{k}}\left\Vert \widehat{\mathbf{A}}_{k}^{M}-\mathbf{A}_{0k}^{M}%
\widehat{\mathbf{S}}_{k}^{M}\right\Vert _{F}^{2}=O_{P}\left( L^{-1}\right) .
\label{ls_A}
\end{equation}%
Also, letting $\widehat{\mathbf{S}}_{-k}^{M}=\left( \otimes _{j=1,j\neq
k}^{K}\widehat{\mathbf{S}}_{j}^{M}\right) $, it follows that%
\begin{eqnarray}
&&\frac{1}{p_{-k}}\left\Vert \widehat{\mathbf{B}}_{k}^{M}-\mathbf{B}_{0k}^{M}%
\widehat{\mathbf{S}}_{-k}^{M}\right\Vert _{F}^{2}  \label{ls_B} \\
&\lesssim &\frac{1}{p_{-k}}\sum_{j=1,j\neq k}^{K}\left\Vert \left( \otimes
_{l=k+1,l\neq k}^{K}\mathbf{A}_{0l}^{M}\widehat{\mathbf{S}}_{l}^{M}\right)
\otimes \left( \widehat{\mathbf{A}}_{j}^{M}-\mathbf{A}_{0j}^{M}\widehat{%
\mathbf{S}}_{j}^{M}\right) \otimes \left( \otimes _{l=1,l\neq k}^{j-1}%
\widehat{\mathbf{A}}_{l}^{M}\right) \right\Vert _{F}^{2}  \notag \\
&\lesssim &\frac{1}{p_{-k}}\sum_{j=1,j\neq k}^{K}\left\Vert \left( \otimes
_{l=k+1,l\neq k}^{K}\mathbf{A}_{0l}^{M}\right) \otimes \left( \widehat{%
\mathbf{A}}_{j}^{M}-\mathbf{A}_{0j}^{M}\widehat{\mathbf{S}}_{j}^{M}\right)
\otimes \left( \otimes _{l=1,l\neq k}^{j-1}\mathbf{A}_{0l}^{M}\right)
\right\Vert _{F}^{2}  \notag \\
&\lesssim &\frac{1}{p_{-k}}\sum_{j=1,j\neq k}^{K}\left( \prod_{l=1,l\neq
j,k}^{K}\left\Vert \mathbf{A}_{0l}^{M}\right\Vert _{F}^{2}\right) \left\Vert
\widehat{\mathbf{A}}_{j}^{M}-\mathbf{A}_{0j}^{M}\widehat{\mathbf{S}}%
_{j}^{M}\right\Vert _{F}^{2}  \notag \\
&\lesssim &\sum_{j=1,j\neq k}^{K}\frac{1}{p_{j}}\left\Vert \widehat{\mathbf{A%
}}_{j}^{M}-\mathbf{A}_{0j}^{M}\widehat{\mathbf{S}}_{j}^{M}\right\Vert
_{F}^{2}=O_{P}\left( L^{-1}\right) .  \notag
\end{eqnarray}

We are now ready to prove the lemma. We will routinely assume that $%
r_{1}=...=r_{K}=1$ for simplicity and without loss of generality. In order
to make the proof as transparent as possible, we will use the notation%
\begin{equation}
\mathbf{X}_{k,t}=\mathbf{A}_{0k}^{M}\mathbf{F}_{0k,t}^{M}\left( \mathbf{B}%
_{0k}^{M}\right) ^{\top }+\mathbf{E}_{k,t},  \label{notation0M}
\end{equation}%
and recall that
\begin{equation}
\widehat{\mathbf{B}}_{k}^{M}=\widehat{\mathbf{B}}_{k}^{M}+\mathbf{B}_{0k}^{M}%
\widehat{\mathbf{S}}_{-k}^{M}-\mathbf{B}_{0k}^{M}\widehat{\mathbf{S}}%
_{-k}^{M}.  \label{bhat_KM}
\end{equation}%
We begin by noting that
\begin{eqnarray*}
\widehat{\mathbf{M}}_{k} &=&\frac{1}{pTp_{-k}}\sum_{t=1}^{T}\mathbf{X}_{k,t}%
\widehat{\mathbf{B}}_{k}^{M}\left( \widehat{\mathbf{B}}_{k}^{M}\right)
^{^{\top }}\mathbf{X}_{k,t}^{^{\top }} \\
&=&\frac{1}{pTp_{-k}}\sum_{t=1}^{T}\mathbf{A}_{0k}^{M}\mathbf{F}%
_{0k,t}^{M}\left( \mathbf{B}_{0k}^{M}\right) ^{\top }\widehat{\mathbf{B}}%
_{k}^{M}\left( \widehat{\mathbf{B}}_{k}^{M}\right) ^{^{\top }}\mathbf{B}%
_{0k}^{M}\left( \mathbf{F}_{0k,t}^{M}\right) ^{^{\top }}\left( \mathbf{A}%
_{0k}^{M}\right) ^{^{\top }} \\
&&+\frac{1}{pTp_{-k}}\sum_{t=1}^{T}\mathbf{A}_{0k}^{M}\mathbf{F}%
_{0k,t}^{M}\left( \mathbf{B}_{0k}^{M}\right) ^{\top }\widehat{\mathbf{B}}%
_{k}^{M}\left( \widehat{\mathbf{B}}_{k}^{M}\right) ^{^{\top }}\mathbf{E}%
_{k,t}^{^{\top }} \\
&&+\frac{1}{pTp_{-k}}\sum_{t=1}^{T}\mathbf{E}_{k,t}\widehat{\mathbf{B}}%
_{k}^{M}\left( \widehat{\mathbf{B}}_{k}^{M}\right) ^{^{\top }}\mathbf{B}%
_{0k}^{M}\left( \mathbf{F}_{0k,t}^{M}\right) ^{^{\top }}\left( \mathbf{A}%
_{0k}^{M}\right) ^{^{\top }} \\
&&+\frac{1}{pTp_{-k}}\sum_{t=1}^{T}\mathbf{E}_{k,t}\widehat{\mathbf{B}}%
_{k}^{M}\left( \widehat{\mathbf{B}}_{k}^{M}\right) ^{^{\top }}\mathbf{E}%
_{k,t}^{^{\top }} \\
&=&I+II+III+IV.
\end{eqnarray*}%
We begin by studying the spectrum of $I$. We write
\begin{eqnarray*}
I &=&\frac{1}{pTp_{-k}}\sum_{t=1}^{T}\mathbf{A}_{0k}^{M}\mathbf{F}%
_{0k,t}^{M}\left( \mathbf{B}_{0k}^{M}\right) ^{\top }\mathbf{B}_{0k}^{M}%
\widehat{\mathbf{S}}_{-k}^{M}\left( \mathbf{B}_{0k}^{M}\widehat{\mathbf{S}}%
_{-k}^{M}\right) ^{^{\top }}\mathbf{B}_{0k}^{M}\left( \mathbf{F}%
_{0k,t}^{M}\right) ^{^{\top }}\left( \mathbf{A}_{0k}^{M}\right) ^{^{\top }}
\\
&&+\frac{1}{pTp_{-k}}\sum_{t=1}^{T}\mathbf{A}_{0k}^{M}\mathbf{F}%
_{0k,t}^{M}\left( \mathbf{B}_{0k}^{M}\right) ^{\top }\mathbf{B}_{0k}^{M}%
\widehat{\mathbf{S}}_{-k}^{M}\left( \widehat{\mathbf{B}}_{k}^{M}-\mathbf{B}%
_{0k}^{M}\widehat{\mathbf{S}}_{-k}^{M}\right) ^{^{\top }}\mathbf{B}%
_{0k}^{M}\left( \mathbf{F}_{0k,t}^{M}\right) ^{^{\top }}\left( \mathbf{A}%
_{0k}^{M}\right) ^{^{\top }} \\
&&+\frac{1}{pTp_{-k}}\sum_{t=1}^{T}\mathbf{A}_{0k}^{M}\mathbf{F}%
_{0k,t}^{M}\left( \mathbf{B}_{0k}^{M}\right) ^{\top }\left( \widehat{\mathbf{%
B}}_{k}^{M}-\mathbf{B}_{0k}^{M}\widehat{\mathbf{S}}_{-k}^{M}\right) \left(
\mathbf{B}_{0k}^{M}\widehat{\mathbf{S}}_{-k}^{M}\right) ^{^{\top }}\mathbf{B}%
_{0k}^{M}\left( \mathbf{F}_{0k,t}^{M}\right) ^{^{\top }}\left( \mathbf{A}%
_{0k}^{M}\right) ^{^{\top }} \\
&&+\frac{1}{pTp_{-k}}\sum_{t=1}^{T}\mathbf{A}_{0k}^{M}\mathbf{F}%
_{0k,t}^{M}\left( \mathbf{B}_{0k}^{M}\right) ^{\top }\left( \widehat{\mathbf{%
B}}_{k}^{M}-\mathbf{B}_{0k}^{M}\widehat{\mathbf{S}}_{-k}^{M}\right) \left(
\widehat{\mathbf{B}}_{k}^{M}-\mathbf{B}_{0k}^{M}\widehat{\mathbf{S}}%
_{-k}^{M}\right) ^{^{\top }}\mathbf{B}_{0k}^{M}\left( \mathbf{F}%
_{0k,t}^{M}\right) ^{^{\top }}\left( \mathbf{A}_{0k}^{M}\right) ^{^{\top }}
\\
&=&I_{a}+I_{b}+I_{c}+I_{d}.
\end{eqnarray*}%
Consider the case $j\leq r_{k}$; it holds that%
\begin{equation*}
\lambda _{j}\left( I_{a}\right) =\lambda _{j}\left( \frac{1}{p_{k}}\mathbf{A}%
_{0k}\left( \frac{1}{T}\sum_{t=1}^{T}\mathbf{F}_{0k,t}\mathbf{F}%
_{0k,t}^{^{\top }}\right) \mathbf{A}_{0k}^{^{\top }}\right) ,
\end{equation*}%
and using the multiplicative Weyl's inequality (see e.g. Theorem 7 in %
\citealp{merikoski2004inequalities})%
\begin{equation*}
\lambda _{j}\left( I_{a}\right) \geq \lambda _{j}\left( \frac{\mathbf{A}%
_{0k}^{^{\top }}\mathbf{A}_{0k}}{p_{k}}\right) \lambda _{\min }\left( \frac{1%
}{T}\sum_{t=1}^{T}\mathbf{F}_{0k,t}\mathbf{F}_{0k,t}^{^{\top }}\right)
=c_{0}>0,
\end{equation*}%
for some $c_{0}>0$, on account of the normalisation $\mathbf{A}_{0k}^{\top }%
\mathbf{A}_{0k}/p_{k}=\mathbf{I}_{r_{k}}$ and of Assumption \ref{as-1}.
Further, by construction, $\lambda _{j}\left( I_{a}\right) =0$ for all $%
j>r_{k}$. Also (recall that we set $r_{k}=1$ for simplicity)%
\begin{eqnarray*}
\left\Vert I_{b}\right\Vert _{F} &=&\frac{1}{p}\left\Vert \frac{1}{T}%
\sum_{t=1}^{T}\left( \mathbf{F}_{0k,t}^{M}\right) ^{2}\mathbf{A}%
_{0k}^{M}\left( \widehat{\mathbf{B}}_{k}^{M}-\mathbf{B}_{0k}^{M}\widehat{%
\mathbf{S}}_{-k}^{M}\right) ^{^{\top }}\mathbf{B}_{0k}^{M}\left( \mathbf{A}%
_{0k}^{M}\right) ^{^{\top }}\right\Vert _{F} \\
&=&\left( \frac{1}{T}\sum_{t=1}^{T}\left( \mathbf{F}_{0k,t}^{M}\right)
^{2}\right) \frac{1}{p}\left\Vert \mathbf{A}_{0k}^{M}\left( \widehat{\mathbf{%
B}}_{k}^{M}-\mathbf{B}_{0k}^{M}\widehat{\mathbf{S}}_{-k}^{M}\right) ^{^{\top
}}\mathbf{B}_{0k}^{M}\left( \mathbf{A}_{0k}^{M}\right) ^{^{\top
}}\right\Vert _{F} \\
&\leq &c_{0}\frac{1}{p}\left\Vert \mathbf{A}_{0k}^{M}\right\Vert
_{F}^{2}\left\Vert \mathbf{B}_{0k}^{M}\right\Vert _{F}\left\Vert \widehat{%
\mathbf{B}}_{k}^{M}-\mathbf{B}_{0k}^{M}\widehat{\mathbf{S}}%
_{-k}^{M}\right\Vert _{F}=O_{P}\left( L^{-1/2}\right) ,
\end{eqnarray*}%
and the same can be shown for $\left\Vert I_{c}\right\Vert _{F}$ by
symmetry. Similarly%
\begin{eqnarray*}
\left\Vert I_{d}\right\Vert _{F} &=&\frac{1}{p}\left\Vert \frac{1}{Tp_{-k}}%
\sum_{t=1}^{T}\left( \mathbf{F}_{0k,t}^{M}\right) ^{2}\mathbf{A}%
_{0k}^{M}\left( \mathbf{B}_{0k}^{M}\right) ^{\top }\left( \widehat{\mathbf{B}%
}_{k}^{M}-\mathbf{B}_{0k}^{M}\widehat{\mathbf{S}}_{-k}^{M}\right) \left(
\widehat{\mathbf{B}}_{k}^{M}-\mathbf{B}_{0k}^{M}\widehat{\mathbf{S}}%
_{-k}^{M}\right) ^{^{\top }}\mathbf{B}_{0k}^{M}\left( \mathbf{A}%
_{0k}^{M}\right) ^{^{\top }}\right\Vert _{F} \\
&=&\left( \frac{1}{T}\sum_{t=1}^{T}\left( \mathbf{F}_{0k,t}^{M}\right)
^{2}\right) \frac{1}{p_{-k}p}\left\Vert \mathbf{A}_{0k}^{M}\left( \mathbf{B}%
_{0k}^{M}\right) ^{\top }\left( \widehat{\mathbf{B}}_{k}^{M}-\mathbf{B}%
_{0k}^{M}\widehat{\mathbf{S}}_{-k}^{M}\right) \left( \widehat{\mathbf{B}}%
_{k}^{M}-\mathbf{B}_{0k}^{M}\widehat{\mathbf{S}}_{-k}^{M}\right) ^{^{\top }}%
\mathbf{B}_{0k}^{M}\left( \mathbf{A}_{0k}^{M}\right) ^{^{\top }}\right\Vert
_{F} \\
&\leq &\left( \frac{1}{T}\sum_{t=1}^{T}\left( \mathbf{F}_{0k,t}^{M}\right)
^{2}\right) \frac{1}{p_{-k}p}\left\Vert \mathbf{A}_{0k}^{M}\right\Vert
_{F}^{2}\left\Vert \mathbf{B}_{0k}^{M}\right\Vert _{F}^{2}\left\Vert
\widehat{\mathbf{B}}_{k}^{M}-\mathbf{B}_{0k}^{M}\widehat{\mathbf{S}}%
_{-k}^{M}\right\Vert _{F}^{2}=O_{P}\left( L^{-1}\right) .
\end{eqnarray*}%
Using Weyls' inequality, this entails that%
\begin{equation*}
\lambda _{j}\left( I\right) =c_{0}+O_{P}\left( L^{-1/2}\right) ,
\end{equation*}%
for all $j\leq r_{k}$. By construction, $I$ has rank $r_{k}$, and therefore $%
\lambda _{j}\left( I\right) =0$ for all $j>r_{k}$. We now turn to bounding
the largest eigenvalues of $II$, $III$ and $IV$, setting $r_{k}=1$ for
simplicity. Recall $\left\Vert \widehat{\mathbf{A}}_{k}^{M}\right\Vert
_{F}=c_{0}p_{k}^{1/2}$. Hence
\begin{eqnarray*}
\left\Vert II\right\Vert _{F} &\leq &\frac{1}{p}\left\Vert \mathbf{A}%
_{0k}^{M}\left( \mathbf{B}_{0k}^{M}\right) ^{\top }\widehat{\mathbf{B}}%
_{k}^{M}\frac{1}{Tp_{-k}}\sum_{t=1}^{T}\left( \widehat{\mathbf{B}}%
_{k}^{M}\right) ^{^{\top }}\mathbf{E}_{k,t}^{^{\top }}\mathbf{F}%
_{0k,t}^{M}\right\Vert _{F} \\
&\leq &\frac{1}{p}\frac{\left\Vert \mathbf{B}_{0k}^{M}\right\Vert
_{F}\left\Vert \widehat{\mathbf{B}}_{k}^{M}\right\Vert _{F}}{p_{-k}}%
\left\Vert \frac{1}{T}\sum_{t=1}^{T}\mathbf{A}_{0k}^{M}\left( \widehat{%
\mathbf{B}}_{k}^{M}\right) ^{\top }\mathbf{E}_{k,t}^{^{\top }}\mathbf{F}%
_{0k,t}^{M}\right\Vert _{F} \\
&\leq &c_{0}\frac{1}{p}\left\Vert \frac{1}{T}\sum_{t=1}^{T}\mathbf{A}%
_{0k}^{M}\left( \widehat{\mathbf{B}}_{k}^{M}\right) ^{^{\top }}\mathbf{E}%
_{k,t}^{^{\top }}\mathbf{F}_{0k,t}^{M}\right\Vert _{F}.
\end{eqnarray*}%
Using again (\ref{bhat_KM}), we may write%
\begin{eqnarray*}
&&\left\Vert \frac{1}{T}\sum_{t=1}^{T}\mathbf{A}_{0k}^{M}\left( \widehat{%
\mathbf{B}}_{k}^{M}\right) ^{^{\top }}\mathbf{E}_{k,t}^{^{\top }}\mathbf{F}%
_{0k,t}^{M}\right\Vert _{F} \\
&\leq &\left\Vert \frac{1}{T}\sum_{t=1}^{T}\mathbf{A}_{0k}^{M}\left( \mathbf{%
B}_{0k}^{M}\widehat{\mathbf{S}}_{-k}^{M}\right) ^{^{\top }}\mathbf{E}%
_{k,t}^{^{\top }}\mathbf{F}_{0k,t}^{M}\right\Vert _{F}+\left\Vert \frac{1}{T}%
\sum_{t=1}^{T}\mathbf{A}_{0k}^{M}\left( \widehat{\mathbf{B}}_{k}^{M}-\mathbf{%
B}_{0k}^{M}\widehat{\mathbf{S}}_{-k}^{M}\right) ^{^{\top }}\mathbf{E}%
_{k,t}^{^{\top }}\mathbf{F}_{0k,t}^{M}\right\Vert _{F}.
\end{eqnarray*}%
It holds that%
\begin{equation*}
\left\Vert \sum_{t=1}^{T}\mathbf{A}_{0k}^{M}\left( \mathbf{B}_{0k}^{M}%
\widehat{\mathbf{S}}_{-k}^{M}\right) ^{^{\top }}\mathbf{E}_{k,t}^{^{\top }}%
\mathbf{F}_{0k,t}^{M}\right\Vert _{F}^{2}=\sum_{i,j=1}^{p_{k}}\left(
\sum_{t=1}^{T}a_{k,i}\sum_{h=1}^{p_{-k}}b_{k,h}e_{jh,k,t}\mathbf{F}%
_{k,t}\right) ^{2}=O_{P}\left( p_{k}^{2}p_{-k}T\right) ,
\end{equation*}%
and therefore%
\begin{equation*}
\left\Vert \frac{1}{T}\sum_{t=1}^{T}\mathbf{A}_{0k}^{M}\left( \mathbf{B}%
_{0k}^{M}\widehat{\mathbf{S}}_{-k}^{M}\right) ^{^{\top }}\mathbf{E}%
_{k,t}^{^{\top }}\mathbf{F}_{0k,t}^{M}\right\Vert _{F}=O_{P}\left(
p_{k}p_{-k}^{1/2}T^{-1/2}\right) .
\end{equation*}%
On the other hand%
\begin{eqnarray*}
\left\Vert \frac{1}{T}\sum_{t=1}^{T}\mathbf{A}_{0k}^{M}\left( \widehat{%
\mathbf{B}}_{k}^{M}-\mathbf{B}_{0k}^{M}\widehat{\mathbf{S}}_{-k}^{M}\right)
^{^{\top }}\mathbf{E}_{k,t}^{^{\top }}\mathbf{F}_{0k,t}^{M}\right\Vert
_{F}^{2} &\leq &\left\Vert \mathbf{A}_{0k}^{M}\right\Vert _{F}^{2}\left\Vert
\widehat{\mathbf{B}}_{k}^{M}-\mathbf{B}_{0k}^{M}\widehat{\mathbf{S}}%
_{-k}^{M}\right\Vert _{F}^{2}\left\Vert \frac{1}{T}\sum_{t=1}^{T}\mathbf{E}%
_{k,t}^{^{\top }}\mathbf{F}_{0k,t}^{M}\right\Vert _{F}^{2} \\
&=&O_{P}\left( 1\right) p_{k}p_{-k}L^{-1}T^{-1}p=O_{P}\left( p^{2}\left(
LT\right) ^{-1}\right) ,
\end{eqnarray*}%
and therefore%
\begin{equation*}
\left\Vert \frac{1}{T}\sum_{t=1}^{T}\mathbf{A}_{0k}^{M}\left( \widehat{%
\mathbf{B}}_{k}^{M}-\mathbf{B}_{0k}^{M}\widehat{\mathbf{S}}_{-k}^{M}\right)
^{^{\top }}\mathbf{E}_{k,t}^{^{\top }}\mathbf{F}_{0k,t}^{M}\right\Vert
_{F}=O_{P}\left( p\left( LT\right) ^{-1/2}\right) .
\end{equation*}%
Hence we finally have%
\begin{equation*}
\left\Vert II\right\Vert _{F}=\frac{1}{p}\left( O_{P}\left(
p_{k}p_{-k}^{1/2}T^{-1/2}\right) +O_{P}\left( p\left( LT\right)
^{-1/2}\right) \right) =O_{P}\left( \left( p_{k}T\right) ^{-1/2}\right)
+O_{P}\left( \left( LT\right) ^{-1/2}\right) ;
\end{equation*}%
the same can be shown for $\left\Vert III\right\Vert _{F}$, so that%
\begin{equation}
\left\vert \lambda _{\max }\left( II+III\right) \right\vert \leq \left\Vert
II+III\right\Vert _{F}\leq \left\Vert II\right\Vert _{F}+\left\Vert
III\right\Vert _{F}=O_{P}\left( \left( p_{k}T\right) ^{-1/2}\right)
+O_{P}\left( \left( LT\right) ^{-1/2}\right) .  \label{weyl-1}
\end{equation}%
Applying the same passages, \textit{mutatis mutandis}, as in the proof of
Lemma \ref{er-h}, it can be shown that
\begin{equation}
\left\Vert IV\right\Vert _{F}=O_{P}\left( p_{-k}^{-1}\right) +O_{P}\left(
L^{-1/2}\right)  \label{ivf}
\end{equation}%
Putting (\ref{weyl-1}) and (\ref{ivf}) together%
\begin{eqnarray*}
&&\left\vert \lambda _{\max }\left( II+III+IV\right) \right\vert \\
&\leq &\left\Vert II+III\right\Vert _{F}+\left\Vert IV\right\Vert
_{F}=O_{P}\left( \left( p_{k}T\right) ^{-1/2}\right) +O_{P}\left(
p_{-k}^{-1}\right) +O_{P}\left( L^{-1/2}\right) .
\end{eqnarray*}%
The desired result now follows from Weyl's inequality.
\end{proof}
\end{lemma}

\begin{lemma}
\label{er-h}We assume that Assumptions \ref{as-1}, \ref{as-2} and \ref{as-4}
hold. Then, as $\min \left\{ T,p_{1},...,p_{K}\right\} \rightarrow \infty $,
it holds that, for all $j\leq r_{k}$%
\begin{equation}
\lambda _{j}\left( \widehat{\mathbf{M}}_{k}^{H}\right) =c_{k}+o_{P}\left(
1\right) ,  \label{er-h-1}
\end{equation}%
for some $c_{k}>0$. For all $j>r_{k}$: (a) if, in Assumption \ref{as-4}, $%
\epsilon \geq 2$, then it holds that%
\begin{equation}
\lambda _{j}\left( \widehat{\mathbf{M}}_{k}^{H}\right) =O_{P}\left( \left(
L^{\ast }\right) ^{-1/2}\right) ;  \label{er-h-2}
\end{equation}%
(b) if, in Assumption \ref{as-4}, $0<\epsilon <2$, then it holds that%
\begin{equation}
\lambda _{j}\left( \widehat{\mathbf{M}}_{k}^{H}\right) =O_{P}\left( \left(
L^{\ast \ast }\right) ^{-1/2}\right) .  \label{er-h-3}
\end{equation}

\begin{proof}
Some arguments in the proof are rather repetitive and similar to the proof
of Lemma \ref{er-ls}, and we omit them when possible. We begin by defining
some preliminary notation and state some facts concerning the - deliberately
overdimensioned - estimation of $\mathbf{A}_{k}$. Let $\widetilde{L}=L^{\ast
}I\left( \epsilon >2\right) +L^{\ast \ast }I\left( 0<\epsilon \leq 2\right) $%
, where $\epsilon $ is defined in Assumption \ref{as-4}.

Recall (\ref{fM_0t})-(\ref{theta0_M}), and let $\widehat{\theta }^{M,H}$ be
the of $\theta _{0}^{M}$ satisfying (\ref{kkt-2-2}), by repeating verbatim
the proof of Theorem \ref{huber-theorem-1}, it can be shown that
\begin{equation}
d_{M}\left( \widehat{\theta }^{M,H},\theta _{0}^{M}\right) =O_{P}\left(
\widetilde{L}^{-1}\right) .  \label{nf3}
\end{equation}%
Also, following the proof of Lemma \ref{lemma2}, it can be shown that
\begin{eqnarray}
\frac{1}{p_{k}}\left\Vert \widehat{\mathbf{A}}_{k}^{M,H}-\mathbf{A}_{0k}^{M}%
\widehat{\mathbf{S}}_{k}^{M,H}\right\Vert _{F}^{2} &=&O_{P}\left( \widetilde{%
L}^{-1}\right) ,  \label{nf4_app} \\
\frac{1}{T}\sum_{t=1}^{T}\left\Vert \widehat{\mathcal{F}}_{t}^{M,H}-\mathcal{%
F}_{0t}^{M}\times _{k=1}^{K}\widehat{\mathbf{S}}_{k}^{M,H}\right\Vert
_{F}^{2} &=&O_{P}\left( \widetilde{L}^{-1}\right) ,  \label{nf44_app}
\end{eqnarray}%
where $\widehat{\mathbf{S}}_{k}^{M,H}=\func{sgn}\left\{ \left( \widehat{%
\mathbf{A}}_{k}^{M,H}\right) ^{\top }\mathbf{A}_{0k}^{M}/p_{k}\right\} $.
Further, letting $\widehat{\mathbf{S}}_{-k}^{M,H}=\left( \otimes _{j=1,j\neq
k}^{K}\widehat{\mathbf{S}}_{j}^{M,H}\right) $, it follows that%
\begin{eqnarray}
&&\frac{1}{p_{-k}}\left\Vert \widehat{\mathbf{B}}_{k}^{M,H}-\mathbf{B}%
_{0k}^{M}\widehat{\mathbf{S}}_{-k}^{M,H}\right\Vert _{F}^{2}  \label{nf5_app}
\\
&\lesssim &\frac{1}{p_{-k}}\sum_{j=1,j\neq k}^{K}\left\Vert \left( \otimes
_{l=k+1,l\neq k}^{K}\mathbf{A}_{0l}^{M}\widehat{\mathbf{S}}_{l}^{M,H}\right)
\otimes \left( \widehat{\mathbf{A}}_{j}^{M,H}-\mathbf{A}_{0j}^{M}\widehat{%
\mathbf{S}}_{j}^{M,H}\right) \otimes \left( \otimes _{l=1,l\neq k}^{j-1}%
\widehat{\mathbf{A}}_{l}^{M,H}\right) \right\Vert _{F}^{2}  \notag \\
&\lesssim &\frac{1}{p_{-k}}\sum_{j=1,j\neq k}^{K}\left\Vert \left( \otimes
_{l=k+1,l\neq k}^{K}\mathbf{A}_{0l}^{M}\right) \otimes \left( \widehat{%
\mathbf{A}}_{j}^{M,H}-\mathbf{A}_{0j}^{M}\widehat{\mathbf{S}}%
_{j}^{M,H}\right) \otimes \left( \otimes _{l=1,l\neq k}^{j-1}\mathbf{A}%
_{0l}^{M}\right) \right\Vert _{F}^{2}  \notag \\
&\lesssim &\frac{1}{p_{-k}}\sum_{j=1,j\neq k}^{K}\left( \prod_{l=1,l\neq
j,k}^{K}\left\Vert \mathbf{A}_{0l}^{M}\right\Vert _{F}^{2}\right) \left\Vert
\widehat{\mathbf{A}}_{j}^{M,H}-\mathbf{A}_{0j}^{M}\widehat{\mathbf{S}}%
_{j}^{M,H}\right\Vert _{F}^{2}  \notag \\
&\lesssim &\sum_{j=1,j\neq k}^{K}\frac{1}{p_{j}}\left\Vert \widehat{\mathbf{A%
}}_{j}^{M,H}-\mathbf{A}_{0j}^{M}\widehat{\mathbf{S}}_{j}^{M,H}\right\Vert
_{F}^{2}=O_{P}\left( \widetilde{L}^{-1}\right) .  \notag
\end{eqnarray}

We now study the spectrum of%
\begin{eqnarray*}
\widehat{\mathbf{M}}_{k}^{H} &=&\frac{1}{pTp_{-k}}\sum_{t=1}^{T}\widehat{w}%
_{k,t}^{H}\mathbf{X}_{k,t}\widehat{\mathbf{B}}_{k}^{M,H}\left( \widehat{%
\mathbf{B}}_{k}^{M,H}\right) ^{^{\top }}\mathbf{X}_{k,t}^{^{\top }} \\
&=&\frac{1}{pTp_{-k}}\sum_{t=1}^{T}\widehat{w}_{k,t}^{H}\mathbf{A}_{0k}^{M}%
\mathbf{F}_{0k,t}^{M}\left( \mathbf{B}_{0k}^{M}\right) ^{\top }\widehat{%
\mathbf{B}}_{k}^{M}\left( \widehat{\mathbf{B}}_{k}^{M,H}\right) ^{^{\top }}%
\mathbf{B}_{0k}^{M}\left( \mathbf{F}_{0k,t}^{M}\right) ^{^{\top }}\left(
\mathbf{A}_{0k}^{M}\right) ^{^{\top }} \\
&&+\frac{1}{pTp_{-k}}\sum_{t=1}^{T}\widehat{w}_{k,t}^{H}\mathbf{A}_{0k}^{M}%
\mathbf{F}_{0k,t}^{M}\left( \mathbf{B}_{0k}^{M}\right) ^{\top }\widehat{%
\mathbf{B}}_{k}^{M,H}\left( \widehat{\mathbf{B}}_{k}^{M,H}\right) ^{^{\top }}%
\mathbf{E}_{k,t}^{^{\top }} \\
&&+\frac{1}{pTp_{-k}}\sum_{t=1}^{T}\widehat{w}_{k,t}^{H}\mathbf{E}_{k,t}%
\widehat{\mathbf{B}}_{k}^{M}\left( \widehat{\mathbf{B}}_{k}^{M,H}\right)
^{^{\top }}\mathbf{B}_{0k}^{M}\left( \mathbf{F}_{0k,t}^{M}\right) ^{^{\top
}}\left( \mathbf{A}_{0k}^{M}\right) ^{^{\top }} \\
&&+\frac{1}{pTp_{-k}}\sum_{t=1}^{T}\widehat{w}_{k,t}^{H}\mathbf{E}_{k,t}%
\widehat{\mathbf{B}}_{k}^{M,H}\left( \widehat{\mathbf{B}}_{k}^{M,H}\right)
^{^{\top }}\mathbf{E}_{k,t}^{^{\top }} \\
&=&I+II+III+IV.
\end{eqnarray*}%
We will use again (\ref{notation0M}), and the notation%
\begin{equation}
\widehat{\mathbf{B}}_{k}^{M,H}=\widehat{\mathbf{B}}_{k}^{M,H}+\mathbf{B}%
_{0k}^{M}\widehat{\mathbf{S}}_{-k}^{M,H}-\mathbf{B}_{0k}^{M}\widehat{\mathbf{%
S}}_{-k}^{M,H}.  \label{bkMH}
\end{equation}%
begin by studying the spectrum of $I$. Note that
\begin{eqnarray*}
&&\frac{1}{pTp_{-k}}\sum_{t=1}^{T}\widehat{w}_{k,t}^{H}\mathbf{A}_{0k}^{M}%
\mathbf{F}_{0k,t}^{M}\left( \mathbf{B}_{0k}^{M}\right) ^{\top }\widehat{%
\mathbf{B}}_{k}^{M,H}\left( \widehat{\mathbf{B}}_{k}^{M,H}\right) ^{^{\top }}%
\mathbf{B}_{0k}^{M}\left( \mathbf{F}_{0k,t}^{M}\right) ^{^{\top }}\left(
\mathbf{A}_{0k}^{M}\right) ^{^{\top }} \\
&=&\mathbf{A}_{0k}^{M}\left( \frac{1}{pTp_{-k}}\sum_{t=1}^{T}\widehat{w}%
_{k,t}^{H}\mathbf{F}_{0k,t}^{M}\left( \mathbf{B}_{0k}^{M}\right) ^{\top }%
\widehat{\mathbf{B}}_{k}^{M,H}\left( \widehat{\mathbf{B}}_{k}^{M,H}\right)
^{^{\top }}\mathbf{B}_{0k}^{M}\left( \mathbf{F}_{0k,t}^{M}\right) ^{^{\top
}}\right) \left( \mathbf{A}_{0k}^{M}\right) ^{^{\top }},
\end{eqnarray*}%
has rank $r_{k}$\ by construction, and therefore $\lambda _{j}\left(
I\right) =0$ whenever $j>r_{k}$. We write%
\begin{eqnarray*}
I &=&\frac{1}{pTp_{-k}}\sum_{t=1}^{T}\widehat{w}_{k,t}^{H}\mathbf{A}_{0k}^{M}%
\mathbf{F}_{0k,t}^{M}\left( \mathbf{B}_{0k}^{M}\right) ^{\top }\mathbf{B}%
_{0k}^{M}\widehat{\mathbf{S}}_{-k}^{M,H}\left( \mathbf{B}_{0k}^{M}\widehat{%
\mathbf{S}}_{-k}^{M,H}\right) ^{^{\top }}\mathbf{B}_{0k}^{M}\left( \mathbf{F}%
_{0k,t}^{M}\right) ^{^{\top }}\left( \mathbf{A}_{0k}^{M}\right) ^{^{\top }}
\\
&&+\frac{1}{pTp_{-k}}\sum_{t=1}^{T}\widehat{w}_{k,t}^{H}\mathbf{A}_{0k}^{M}%
\mathbf{F}_{0k,t}^{M}\left( \mathbf{B}_{0k}^{M}\right) ^{\top }\mathbf{B}%
_{0k}^{M}\widehat{\mathbf{S}}_{-k}^{M,H}\left( \widehat{\mathbf{B}}%
_{k}^{M,H}-\mathbf{B}_{0k}^{M}\widehat{\mathbf{S}}_{-k}^{M,H}\right)
^{^{\top }}\mathbf{B}_{0k}^{M}\left( \mathbf{F}_{0k,t}^{M}\right) ^{^{\top
}}\left( \mathbf{A}_{0k}^{M}\right) ^{^{\top }} \\
&&+\frac{1}{pTp_{-k}}\sum_{t=1}^{T}\widehat{w}_{k,t}^{H}\mathbf{A}_{0k}^{M}%
\mathbf{F}_{0k,t}^{M}\left( \mathbf{B}_{0k}^{M}\right) ^{\top }\left(
\widehat{\mathbf{B}}_{k}^{M,H}-\mathbf{B}_{0k}^{M}\widehat{\mathbf{S}}%
_{-k}^{M,H}\right) \left( \mathbf{B}_{0k}^{M}\widehat{\mathbf{S}}%
_{-k}^{M,H}\right) ^{^{\top }}\mathbf{B}_{0k}^{M}\left( \mathbf{F}%
_{0k,t}^{M}\right) ^{^{\top }}\left( \mathbf{A}_{0k}^{M}\right) ^{^{\top }}
\\
&&+\frac{1}{pTp_{-k}}\sum_{t=1}^{T}\widehat{w}_{k,t}^{H}\mathbf{A}_{0k}^{M}%
\mathbf{F}_{0k,t}^{M}\left( \mathbf{B}_{0k}^{M}\right) ^{\top }\left(
\widehat{\mathbf{B}}_{k}^{M,H}-\mathbf{B}_{0k}^{M}\widehat{\mathbf{S}}%
_{-k}^{M,H}\right) \left( \widehat{\mathbf{B}}_{k}^{M,H}-\mathbf{B}_{0k}^{M}%
\widehat{\mathbf{S}}_{-k}^{M,H}\right) ^{^{\top }}\mathbf{B}_{0k}^{M}\left(
\mathbf{F}_{0k,t}^{M}\right) ^{^{\top }}\left( \mathbf{A}_{0k}^{M}\right)
^{^{\top }} \\
&=&I_{a}+I_{b}+I_{c}+I_{d}.
\end{eqnarray*}%
Consider the case $j\leq r_{k}$; it holds that%
\begin{equation*}
\lambda _{j}\left( I_{a}\right) =\lambda _{j}\left( \frac{1}{p_{k}}\mathbf{A}%
_{0k}\left( \frac{1}{T}\sum_{t=1}^{T}\widehat{w}_{k,t}^{H}\mathbf{F}_{0k,t}%
\mathbf{F}_{0k,t}^{^{\top }}\right) \mathbf{A}_{0k}^{\top }\right) ,
\end{equation*}%
and using the multiplicative Weyl's inequality (see e.g. Theorem 7 in %
\citealp{merikoski2004inequalities})%
\begin{equation*}
\lambda _{j}\left( I_{a}\right) \geq \lambda _{j}\left( \frac{\mathbf{A}%
_{0k}^{\top }\mathbf{A}_{0k}}{p_{k}}\right) \lambda _{\min }\left( \frac{1}{T%
}\sum_{t=1}^{T}\widehat{w}_{k,t}^{H}\mathbf{F}_{0k,t}\mathbf{F}%
_{0k,t}^{^{\top }}\right) =\lambda _{\min }\left( \frac{1}{T}\sum_{t=1}^{T}%
\widehat{w}_{k,t}^{H}\mathbf{F}_{0k,t}\mathbf{F}_{0k,t}^{^{\top }}\right) ,
\end{equation*}%
on account of the normalisation $\mathbf{A}_{0k}^{\top }\mathbf{A}%
_{0k}/p_{k}=\mathbf{I}_{r_{k}}$. We now show that%
\begin{equation}
\lambda _{\min }\left( \frac{1}{T}\sum_{t=1}^{T}\widehat{w}_{k,t}^{H}\mathbf{%
F}_{0k,t}\mathbf{F}_{0k,t}^{^{\top }}\right) >0.  \label{merik1}
\end{equation}%
Indeed, note that%
\begin{eqnarray*}
&&\frac{2}{T}\sum_{t=1}^{T}\widehat{w}_{k,t}^{H}\mathbf{F}_{0k,t}\mathbf{F}%
_{0k,t}^{^{\top }} \\
&=&\frac{1}{T}\sum_{t=1}^{T}\mathbf{F}_{0k,t}\mathbf{F}_{0k,t}^{^{\top
}}I\left( \left\Vert \mathbf{A}_{0k}^{M}\mathbf{F}_{0k,t}^{M}\left( \mathbf{B%
}_{0k}^{M}\right) ^{\top }-\widehat{\mathbf{A}}_{k}^{M,H}\widehat{\mathbf{F}}%
_{k,t}^{M,H}\left( \widehat{\mathbf{B}}_{k}^{M,H}\right) ^{\top }+\mathbf{E}%
_{k,t}\right\Vert _{F}/\sqrt{p}\leq \tau \right)  \\
&&+\frac{1}{T}\sum_{t=1}^{T}\widehat{w}_{k,t}^{H}\mathbf{F}_{0k,t}\mathbf{F}%
_{0k,t}^{^{\top }}I\left( \left\Vert \mathbf{A}_{0k}^{M}\mathbf{F}%
_{0k,t}^{M}\left( \mathbf{B}_{0k}^{M}\right) ^{\top }-\widehat{\mathbf{A}}%
_{k}^{M,H}\widehat{\mathbf{F}}_{k,t}^{M,H}\left( \widehat{\mathbf{B}}%
_{k}^{M,H}\right) ^{\top }+\mathbf{E}_{k,t}\right\Vert _{F}/\sqrt{p}>\tau \right)  \\
&\geq &\frac{1}{T}\sum_{t=1}^{T}\mathbf{F}_{0k,t}\mathbf{F}_{0k,t}^{^{\top
}}I\left( \left\Vert \mathbf{A}_{0k}^{M}\mathbf{F}_{0k,t}^{M}\left( \mathbf{B%
}_{0k}^{M}\right) ^{\top }-\widehat{\mathbf{A}}_{k}^{M,H}\widehat{\mathbf{F}}%
_{k,t}^{M,H}\left( \widehat{\mathbf{B}}_{k}^{M,H}\right) ^{\top }+\mathbf{E}%
_{k,t}\right\Vert _{F}/\sqrt{p}\leq \tau \right) ,
\end{eqnarray*}%
where $I\left( \cdot \right) $ is the indicator function. We begin by
finding a bound for%
\begin{eqnarray}
&&\frac{1}{T}\sum_{t=1}^{T}\mathbf{F}_{0k,t}\mathbf{F}_{0k,t}^{^{\top }}%
\left[ I\left( \left\Vert \mathbf{A}_{0k}^{M}\mathbf{F}_{0k,t}^{M}\left(
\mathbf{B}_{0k}^{M}\right) ^{\top }-\widehat{\mathbf{A}}_{k}^{M,H}\widehat{%
\mathbf{F}}_{k,t}^{M,H}\left( \widehat{\mathbf{B}}_{k}^{M,H}\right) ^{\top }+%
\mathbf{E}_{k,t}\right\Vert _{F}/\sqrt{p}\leq \tau \right) \right.   \label{boun-1} \\
&&\left. -I\left( \left\Vert \mathbf{E}_{k,t}-\widehat{\mathbf{A}}%
_{k}^{M,H}\left( \widehat{\mathbf{F}}_{k,t}^{M,H}-\widehat{\mathbf{S}}%
_{k}^{M,H}\mathbf{F}_{0k,t}^{M}\widehat{\mathbf{S}}_{-k}^{M,H}\right) \left(
\widehat{\mathbf{B}}_{k}^{M,H}\right) ^{\top }\right\Vert _{F}/\sqrt{p}\leq \tau
\right) \right] .  \notag
\end{eqnarray}%
Given that%
\begin{eqnarray*}
&&\left\Vert \mathbf{A}_{0k}^{M}\mathbf{F}_{0k,t}^{M}\left( \mathbf{B}%
_{0k}^{M}\right) ^{\top }-\widehat{\mathbf{A}}_{k}^{M,H}\widehat{\mathbf{F}}%
_{k,t}^{M,H}\left( \widehat{\mathbf{B}}_{k}^{M,H}\right) ^{\top }+\mathbf{E}%
_{k,t}\right\Vert _{F} \\
&\geq &\left\Vert \mathbf{E}_{k,t}-\widehat{\mathbf{A}}_{k}^{M,H}\left(
\widehat{\mathbf{F}}_{k,t}^{M,H}-\widehat{\mathbf{S}}_{k}^{M,H}\mathbf{F}%
_{0k,t}^{M}\widehat{\mathbf{S}}_{-k}^{M,H}\right) \left( \widehat{\mathbf{B}}%
_{k}^{M,H}\right) ^{\top }\right\Vert _{F}-\left\Vert \mathbf{A}_{0k}^{M}%
\mathbf{F}_{0k,t}^{M}\left( \mathbf{B}_{0k}^{M}\right) ^{\top }-\widehat{%
\mathbf{A}}_{k}^{M,H}\mathbf{F}_{0k,t}^{M}\widehat{\mathbf{S}}%
_{k}^{M,H}\left( \widehat{\mathbf{B}}_{k}^{M,H}\right) ^{\top }\right\Vert
_{F},
\end{eqnarray*}%
an upper bound for the expression in (\ref{boun-1}) is%
\begin{eqnarray*}
&&\frac{1}{T}\sum_{t=1}^{T}\mathbf{F}_{0k,t}\mathbf{F}_{0k,t}^{^{\top }}%
\left[ I\left( \left\Vert \mathbf{E}_{k,t}-\widehat{\mathbf{A}}%
_{k}^{M,H}\left( \widehat{\mathbf{F}}_{k,t}^{M,H}-\widehat{\mathbf{S}}%
_{k}^{M,H}\mathbf{F}_{0k,t}^{M}\widehat{\mathbf{S}}_{-k}^{M,H}\right) \left(
\widehat{\mathbf{B}}_{k}^{M,H}\right) ^{\top }\right\Vert _{F}/\sqrt{p}\right.
\right.  \\
&&\left. -\left\Vert \mathbf{A}_{0k}^{M}\mathbf{F}_{0k,t}^{M}\left( \mathbf{B%
}_{0k}^{M}\right) ^{\top }-\widehat{\mathbf{A}}_{k}^{M,H}\widehat{\mathbf{S}}%
_{k}^{M,H}\mathbf{F}_{0k,t}^{M}\widehat{\mathbf{S}}_{-k}^{M,H}\left( \widehat{\mathbf{B}}%
_{k}^{M,H}\right) ^{\top }\right\Vert _{F}/\sqrt{p}\leq \tau \right)  \\
&&\left. -I\left( \left\Vert \mathbf{E}_{k,t}-\widehat{\mathbf{A}}%
_{k}^{M,H}\left( \widehat{\mathbf{F}}_{k,t}^{M,H}-\widehat{\mathbf{S}}%
_{k}^{M,H}\mathbf{F}_{0k,t}^{M}\widehat{\mathbf{S}}_{-k}^{M,H}\right) \left(
\widehat{\mathbf{B}}_{k}^{M,H}\right) ^{\top }\right\Vert _{F}/\sqrt{p}\leq \tau
\right) \right] ,
\end{eqnarray*}%
and it is immediate to see that the term above is always non-negative. Also,
since
\begin{eqnarray*}
&&p^{-1/2}\left\Vert \mathbf{A}_{0k}^{M}\mathbf{F}_{0k,t}^{M}\left( \mathbf{B%
}_{0k}^{M}\right) ^{\top }-\widehat{\mathbf{A}}_{k}^{M,H}\widehat{\mathbf{F}}%
_{k,t}^{M,H}\left( \widehat{\mathbf{B}}_{k}^{M,H}\right) ^{\top }+\mathbf{E}%
_{k,t}\right\Vert _{F} \\
&\leq &p^{-1/2}\left\Vert \mathbf{E}_{k,t}-\widehat{\mathbf{A}}%
_{k}^{M,H}\left( \widehat{\mathbf{F}}_{k,t}^{M,H}-\widehat{\mathbf{S}}%
_{k}^{M,H}\mathbf{F}_{0k,t}^{M}\widehat{\mathbf{S}}_{-k}^{M,H}\right) \left(
\widehat{\mathbf{B}}_{k}^{M,H}\right) ^{\top }\right\Vert _{F} \\
&&+p^{-1/2}\left\Vert \mathbf{A}_{0k}^{M}\mathbf{F}_{0k,t}^{M}\left( \mathbf{%
B}_{0k}^{M}\right) ^{\top }-\widehat{\mathbf{A}}_{k}^{M,H}\widehat{\mathbf{S}}%
_{k}^{M,H}\mathbf{F}_{0k,t}^{M}\widehat{\mathbf{S}}_{-k}^{M,H}\left( \widehat{\mathbf{B}}%
_{k}^{M,H}\right) ^{\top }\right\Vert _{F},
\end{eqnarray*}%
a lower bound for the expression in (\ref{boun-1}) is%
\begin{eqnarray*}
&&\frac{1}{T}\sum_{t=1}^{T}\mathbf{F}_{0k,t}\mathbf{F}_{0k,t}^{^{\top }}%
\left[ I\left( \left\Vert \mathbf{E}_{k,t}-\widehat{\mathbf{A}}%
_{k}^{M,H}\left( \widehat{\mathbf{F}}_{k,t}^{M,H}-\widehat{\mathbf{S}}%
_{k}^{M,H}\mathbf{F}_{0k,t}^{M}\widehat{\mathbf{S}}_{-k}^{M,H}\right) \left(
\widehat{\mathbf{B}}_{k}^{M,H}\right) ^{\top }\right\Vert _{F}/\sqrt{p}\right.
\right.  \\
&&\left. +\left\Vert \mathbf{A}_{0k}^{M}\mathbf{F}_{0k,t}^{M}\left( \mathbf{B%
}_{0k}^{M}\right) ^{\top }-\widehat{\mathbf{A}}_{k}^{M,H}\widehat{\mathbf{S}}%
_{k}^{M,H}\mathbf{F}_{0k,t}^{M}\widehat{\mathbf{S}}_{-k}^{M,H}\left( \widehat{\mathbf{B}}%
_{k}^{M,H}\right) ^{\top }\right\Vert _{F}/\sqrt{p}\leq \tau \right)  \\
&&\left. -I\left( \left\Vert \mathbf{E}_{k,t}-\widehat{\mathbf{A}}%
_{k}^{M,H}\left( \widehat{\mathbf{F}}_{k,t}^{M,H}-\widehat{\mathbf{S}}%
_{k}^{M,H}\mathbf{F}_{0k,t}^{M}\widehat{\mathbf{S}}_{-k}^{M,H}\right) \left( \widehat{\mathbf{B}}%
_{k}^{M,H}\right) ^{\top }\right\Vert _{F}/\sqrt{p}\leq \tau \right) \right] ,
\end{eqnarray*}%
which is non-positive. Given that, by (\ref{nf4_app}) and (\ref{nf5_app}),
it follows that
\begin{equation*}
\left\Vert \mathbf{A}_{0k}^{M}\mathbf{F}_{0k,t}^{M}\left( \mathbf{B}%
_{0k}^{M}\right) ^{\top }-\widehat{\mathbf{A}}_{k}^{M,H}\widehat{\mathbf{S}}%
_{k}^{M,H}\mathbf{F}_{0k,t}^{M}\widehat{\mathbf{S}}_{-k}^{M,H}\left( \widehat{\mathbf{B}}_{k}^{M,H}\right)
^{\top }\right\Vert _{F}=O_{P}\left( p^{1/2}\widetilde{L}^{-1/2}\right) ,
\end{equation*}%
we have

\begin{eqnarray*}
&&E\left[ I\left( \left\Vert \mathbf{E}_{k,t}-\widehat{\mathbf{A}}%
_{k}^{M,H}\left( \widehat{\mathbf{F}}_{k,t}^{M,H}-\widehat{\mathbf{S}}%
_{k}^{M,H}\mathbf{F}_{0k,t}^{M}\widehat{\mathbf{S}}_{-k}^{M,H}\right) \left(
\widehat{\mathbf{B}}_{k}^{M,H}\right) ^{\top }\right\Vert _{F}/\sqrt{p}\right.
\right.  \\
&&\left. +\left\Vert \mathbf{A}_{0k}^{M}\mathbf{F}_{0k,t}^{M}\left( \mathbf{B%
}_{0k}^{M}\right) ^{\top }-\widehat{\mathbf{A}}_{k}^{M,H}\widehat{\mathbf{S}}%
_{k}^{M,H}\mathbf{F}_{0k,t}^{M}\widehat{\mathbf{S}}_{-k}^{M,H}\left( \widehat{\mathbf{B}}%
_{k}^{M,H}\right) ^{\top }\right\Vert _{F}/\sqrt{p}\leq \tau \right)  \\
&&\left. -I\left( \left\Vert \mathbf{E}_{k,t}-\widehat{\mathbf{A}}%
_{k}^{M,H}\left( \widehat{\mathbf{F}}_{k,t}^{M,H}-\widehat{\mathbf{S}}%
_{k}^{M,H}\mathbf{F}_{0k,t}^{M}\widehat{\mathbf{S}}_{-k}^{M,H}\right) \left(
\widehat{\mathbf{B}}_{k}^{M,H}\right) ^{\top }\right\Vert _{F}/\sqrt{p}\leq \tau
\right) \right]  \\
&=&P\left( \left\Vert \mathbf{E}_{k,t}-\widehat{\mathbf{A}}_{k}^{M,H}\left(
\widehat{\mathbf{F}}_{k,t}^{M,H}-\widehat{\mathbf{S}}_{k}^{M,H}\mathbf{F}%
_{0k,t}^{M}\widehat{\mathbf{S}}_{-k}^{M,H}\right) \left( \widehat{\mathbf{B}}%
_{k}^{M,H}\right) ^{\top }\right\Vert _{F}/\sqrt{p}\right.  \\
&&\left. +\left\Vert \mathbf{A}_{0k}^{M}\mathbf{F}_{0k,t}^{M}\left( \mathbf{B%
}_{0k}^{M}\right) ^{\top }-\widehat{\mathbf{A}}_{k}^{M,H}\widehat{\mathbf{S}}%
_{k}^{M,H}\mathbf{F}_{0k,t}^{M}\widehat{\mathbf{S}}_{-k}^{M,H}\left( \widehat{\mathbf{B}}%
_{k}^{M,H}\right) ^{\top }\right\Vert _{F}/\sqrt{p}\leq \tau \right)  \\
&&-P\left( \left\Vert \mathbf{E}_{k,t}-\widehat{\mathbf{A}}_{k}^{M,H}\left(
\widehat{\mathbf{F}}_{k,t}^{M,H}-\widehat{\mathbf{S}}_{k}^{M,H}\mathbf{F}%
_{0k,t}^{M}\widehat{\mathbf{S}}_{-k}^{M,H}\right) \left( \widehat{\mathbf{B}}%
_{k}^{M,H}\right) ^{\top }\right\Vert _{F}/\sqrt{p}\leq \tau \right)  \\
&=&P\left( \left\Vert \mathbf{E}_{k,t}-\widehat{\mathbf{A}}_{k}^{M,H}\left(
\widehat{\mathbf{F}}_{k,t}^{M,H}-\widehat{\mathbf{S}}_{k}^{M,H}\mathbf{F}%
_{0k,t}^{M}\widehat{\mathbf{S}}_{-k}^{M,H}\right) \left( \widehat{\mathbf{B}}%
_{k}^{M,H}\right) ^{\top }\right\Vert _{F}/\sqrt{p}\leq \tau \right) +o\left(
1\right)  \\
&&-P\left( \left\Vert \mathbf{E}_{k,t}-\widehat{\mathbf{A}}_{k}^{M,H}\left(
\widehat{\mathbf{F}}_{k,t}^{M,H}-\widehat{\mathbf{S}}_{k}^{M,H}\mathbf{F}%
_{0k,t}^{M}\widehat{\mathbf{S}}_{-k}^{M,H}\right) \left( \widehat{\mathbf{B}}%
_{k}^{M,H}\right) ^{\top }\right\Vert _{F}/\sqrt{p}\leq \tau \right)  \\
&=&o\left( 1\right) ,
\end{eqnarray*}%
and similarly%
\begin{eqnarray*}
&&E\left[ I\left( \left\Vert \mathbf{E}_{k,t}-\widehat{\mathbf{A}}%
_{k}^{M,H}\left( \widehat{\mathbf{F}}_{k,t}^{M,H}-\widehat{\mathbf{S}}%
_{k}^{M,H}\mathbf{F}_{0k,t}^{M}\widehat{\mathbf{S}}_{-k}^{M,H}\right) \left(
\widehat{\mathbf{B}}_{k}^{M,H}\right) ^{\top }\right\Vert /\sqrt{p}\right.
\right.  \\
&&\left. -\left\Vert \mathbf{A}_{0k}^{M}\mathbf{F}_{0k,t}^{M}\left( \mathbf{B%
}_{0k}^{M}\right) ^{\top }-\widehat{\mathbf{A}}_{k}^{M,H}\mathbf{F}%
_{0k,t}^{M}\widehat{\mathbf{S}}_{k}^{M,H}\left( \widehat{\mathbf{B}}%
_{k}^{M,H}\right) ^{\top }\right\Vert _{F}/\sqrt{p}\leq \tau \right)  \\
&&\left. -I\left( \left\Vert \mathbf{E}_{k,t}-\widehat{\mathbf{A}}%
_{k}^{M,H}\left( \widehat{\mathbf{F}}_{k,t}^{M,H}-\widehat{\mathbf{S}}%
_{k}^{M,H}\mathbf{F}_{0k,t}^{M}\widehat{\mathbf{S}}_{-k}^{M,H}\right) \left(
\widehat{\mathbf{B}}_{k}^{M,H}\right) ^{\top }\right\Vert _{F}/\sqrt{p}\leq \tau
\right) \right]  \\
&=&P\left( \left\Vert \mathbf{E}_{k,t}-\widehat{\mathbf{A}}_{k}^{M,H}\left(
\widehat{\mathbf{F}}_{k,t}^{M,H}-\widehat{\mathbf{S}}_{k}^{M,H}\mathbf{F}%
_{0k,t}^{M}\widehat{\mathbf{S}}_{-k}^{M,H}\right) \left( \widehat{\mathbf{B}}%
_{k}^{M,H}\right) ^{\top }\right\Vert _{F}/\sqrt{p}\right.  \\
&&\left. -\left\Vert \mathbf{A}_{0k}^{M}\mathbf{F}_{0k,t}^{M}\left( \mathbf{B%
}_{0k}^{M}\right) ^{\top }-\widehat{\mathbf{A}}_{k}^{M,H}\widehat{\mathbf{S}}%
_{k}^{M,H}\mathbf{F}_{0k,t}^{M}\widehat{\mathbf{S}}_{-k}^{M,H}\left( \widehat{\mathbf{B}}%
_{k}^{M,H}\right) ^{\top }\right\Vert _{F}/\sqrt{p}\leq \tau \right)  \\
&&-P\left( \left\Vert \mathbf{E}_{k,t}-\widehat{\mathbf{A}}_{k}^{M,H}\left(
\widehat{\mathbf{F}}_{k,t}^{M,H}-\widehat{\mathbf{S}}_{k}^{M,H}\mathbf{F}%
_{0k,t}^{M}\widehat{\mathbf{S}}_{-k}^{M,H}\right) \left( \widehat{\mathbf{B}}%
_{k}^{M,H}\right) ^{\top }\right\Vert _{F}/\sqrt{p}\leq \tau \right)  \\
&=&o\left( 1\right) .
\end{eqnarray*}%
Hence, by dominated convergence, we have shown that
\begin{eqnarray*}
&&\frac{1}{T}\sum_{t=1}^{T}\mathbf{F}_{0k,t}\mathbf{F}_{0k,t}^{^{\top
}}I\left( \left\Vert \mathbf{A}_{0k}^{M}\mathbf{F}_{0k,t}^{M}\left( \mathbf{B%
}_{0k}^{M}\right) ^{\top }-\widehat{\mathbf{A}}_{k}^{M,H}\widehat{\mathbf{F}}%
_{k,t}^{M,H}\left( \widehat{\mathbf{B}}_{k}^{M,H}\right) ^{\top }+\mathbf{E}%
_{k,t}\right\Vert _{F}/\sqrt{p}\leq \tau \right)  \\
&=&\frac{1}{T}\sum_{t=1}^{T}\mathbf{F}_{0k,t}\mathbf{F}_{0k,t}^{^{\top
}}I\left( \left\Vert \mathbf{E}_{k,t}-\mathbf{A}_{0k}^{M}\left( \widehat{%
\mathbf{F}}_{k,t}^{M,H}-\widehat{\mathbf{S}}_{k}^{M,H}\mathbf{F}_{0k,t}^{M}%
\widehat{\mathbf{S}}_{-k}^{M,H}\right) \left( \mathbf{B}_{0k}^{M}\right)
^{\top }\right\Vert _{F}/\sqrt{p}\leq \tau \right) +o_{P}\left( 1\right) .
\end{eqnarray*}%
Consider now%
\begin{equation*}
I\left( \left\Vert \mathbf{A}_{0k}^{M}\mathbf{F}_{0k,t}^{M}\left( \mathbf{B}%
_{0k}^{M}\right) ^{\top }-\widehat{\mathbf{A}}_{k}^{M,H}\widehat{\mathbf{F}}%
_{k,t}^{M,H}\left( \widehat{\mathbf{B}}_{k}^{M,H}\right) ^{\top }+\mathbf{E}%
_{k,t}\right\Vert _{F}/\sqrt{p}\leq \tau \right) -I\left( \left\Vert \mathbf{E}%
_{k,t}\right\Vert _{F}/\sqrt{p}\leq \tau \right) .
\end{equation*}%
This expression has (non-negative) upper bound given by%
\begin{equation*}
I\left( \left\Vert \mathbf{E}_{k,t}\right\Vert _{F}/\sqrt{p}-\left\Vert \mathbf{A}%
_{0k}^{M}\mathbf{F}_{0k,t}^{M}\left( \mathbf{B}_{0k}^{M}\right) ^{\top }-%
\widehat{\mathbf{A}}_{k}^{M,H}\widehat{\mathbf{F}}_{k,t}^{M,H}\left(
\widehat{\mathbf{B}}_{k}^{M,H}\right) ^{\top }\right\Vert _{F}/\sqrt{p}\leq \tau
\right) -I\left( \left\Vert \mathbf{E}_{k,t}\right\Vert _{F}/\sqrt{p}\leq \tau
\right) ,
\end{equation*}%
and (non-positive) lower bound given by%
\begin{equation*}
I\left( \left\Vert \mathbf{E}_{k,t}\right\Vert _{F}/\sqrt{p}+\left\Vert \mathbf{A}%
_{0k}^{M}\mathbf{F}_{0k,t}^{M}\left( \mathbf{B}_{0k}^{M}\right) ^{\top }-%
\widehat{\mathbf{A}}_{k}^{M,H}\widehat{\mathbf{F}}_{k,t}^{M,H}\left(
\widehat{\mathbf{B}}_{k}^{M,H}\right) ^{\top }\right\Vert _{F}/\sqrt{p}\leq \tau
\right) -I\left( \left\Vert \mathbf{E}_{k,t}\right\Vert _{F}/\sqrt{p}\leq \tau
\right) .
\end{equation*}%
Hence we can study%
\begin{equation*}
\frac{1}{T}\sum_{t=1}^{T}\mathbf{F}_{0k,t}\mathbf{F}_{0k,t}^{^{\top }}\left[
I\left( \left\Vert \mathbf{E}_{k,t}-\mathbf{A}_{0k}^{M}\left( \widehat{%
\mathbf{F}}_{k,t}^{M,H}-\widehat{\mathbf{S}}_{k}^{M,H}\mathbf{F}_{0k,t}^{M}%
\widehat{\mathbf{S}}_{-k}^{M,H}\right) \left( \mathbf{B}_{0k}^{M}\right)
^{\top }\right\Vert _{F}/\sqrt{p}\leq \tau \right) -I\left( \left\Vert \mathbf{E}%
_{k,t}\right\Vert _{F}/\sqrt{p}\leq \tau \right) \right] ,
\end{equation*}%
as before, by studying%
\begin{eqnarray*}
&&\frac{1}{T}\sum_{t=1}^{T}\mathbf{F}_{0k,t}\mathbf{F}_{0k,t}^{^{\top }}E%
\left[ I\left( \left\Vert \mathbf{E}_{k,t}-\mathbf{A}_{0k}^{M}\left(
\widehat{\mathbf{F}}_{k,t}^{M,H}-\widehat{\mathbf{S}}_{k}^{M,H}\mathbf{F}%
_{0k,t}^{M}\widehat{\mathbf{S}}_{-k}^{M,H}\right) \left( \mathbf{B}%
_{0k}^{M}\right) ^{\top }\right\Vert _{F}/\sqrt{p}\leq \tau \right) -I\left(
\left\Vert \mathbf{E}_{k,t}\right\Vert _{F}/\sqrt{p}\leq \tau \right) \right]  \\
&\leq &\frac{1}{T}\sum_{t=1}^{T}\mathbf{F}_{0k,t}\mathbf{F}_{0k,t}^{^{\top }}%
\left[ P\left( \left\Vert \mathbf{E}_{k,t}\right\Vert
_{F}/\sqrt{p}-\left\Vert \mathbf{A}_{0k}^{M}\mathbf{F}_{0k,t}^{M}\left(
\mathbf{B}_{0k}^{M}\right) ^{\top }-\widehat{\mathbf{A}}_{k}^{M,H}\widehat{%
\mathbf{F}}_{k,t}^{M,H}\left( \widehat{\mathbf{B}}_{k}^{M,H}\right) ^{\top
}\right\Vert _{F}/\sqrt{p}\leq \tau \right) \right.  \\
&&\left. -P\left( \left\Vert \mathbf{E}_{k,t}\right\Vert _{F}/\sqrt{p}\leq
\tau \right) \right] ,
\end{eqnarray*}%
and%
\begin{eqnarray*}
&&\frac{1}{T}\sum_{t=1}^{T}\mathbf{F}_{0k,t}\mathbf{F}_{0k,t}^{^{\top }}E%
\left[ I\left( \left\Vert \mathbf{E}_{k,t}-\mathbf{A}_{0k}^{M}\left(
\widehat{\mathbf{F}}_{k,t}^{M,H}-\widehat{\mathbf{S}}_{k}^{M,H}\mathbf{F}%
_{0k,t}^{M}\widehat{\mathbf{S}}_{-k}^{M,H}\right) \left( \mathbf{B}%
_{0k}^{M}\right) ^{\top }\right\Vert _{F}/\sqrt{p}\leq \tau \right) -I\left(
\left\Vert \mathbf{E}_{k,t}\right\Vert _{F}/\sqrt{p}\leq \tau \right) \right]  \\
&\geq &\frac{1}{T}\sum_{t=1}^{T}\mathbf{F}_{0k,t}\mathbf{F}_{0k,t}^{^{\top
}}P\left[ \left( \left\Vert \mathbf{E}_{k,t}\right\Vert
_{F}/\sqrt{p}+\left\Vert \mathbf{A}_{0k}^{M}\mathbf{F}_{0k,t}^{M}\left(
\mathbf{B}_{0k}^{M}\right) ^{\top }-\widehat{\mathbf{A}}_{k}^{M,H}\widehat{%
\mathbf{F}}_{k,t}^{M,H}\left( \widehat{\mathbf{B}}_{k}^{M,H}\right) ^{\top
}\right\Vert _{F}/\sqrt{p}\leq \tau \right) \right.  \\
&&\left. -P\left( \left\Vert \mathbf{E}_{k,t}\right\Vert _{F}/\sqrt{p}\leq
\tau \right) \right] .
\end{eqnarray*}%
In both cases we can use the fact that the density of $\left\Vert \mathbf{E}%
_{k,t}\right\Vert _{F}/\sqrt{p}$ is bounded, whence by the Mean Value Theorem and the
Law of Total Probability%
\begin{eqnarray*}
&&P\left( p^{-1/2}\left\Vert \mathbf{E}_{k,t}\right\Vert
_{F}-p^{-1/2}\left\Vert \mathbf{A}_{0k}^{M}\mathbf{F}_{0k,t}^{M}\left(
\mathbf{B}_{0k}^{M}\right) ^{\top }-\widehat{\mathbf{A}}_{k}^{M,H}\widehat{%
\mathbf{F}}_{k,t}^{M,H}\left( \widehat{\mathbf{B}}_{k}^{M,H}\right) ^{\top
}\right\Vert _{F}\leq \tau \right)  \\
&&-P\left( p^{-1/2}\left\Vert \mathbf{E}_{k,t}\right\Vert _{F}\leq
\tau \right)  \\
&=&c_{0}p^{-1/2}E\left\Vert \mathbf{A}_{0k}^{M}\mathbf{F}_{0k,t}^{M}\left(
\mathbf{B}_{0k}^{M}\right) ^{\top }-\widehat{\mathbf{A}}_{k}^{M,H}\widehat{%
\mathbf{F}}_{k,t}^{M,H}\left( \widehat{\mathbf{B}}_{k}^{M,H}\right) ^{\top
}\right\Vert _{F},
\end{eqnarray*}%
and%
\begin{eqnarray*}
&&P\left( p^{-1/2}\left\Vert \mathbf{E}_{k,t}\right\Vert
_{F}+p^{-1/2}\left\Vert \mathbf{A}_{0k}^{M}\mathbf{F}_{0k,t}^{M}\left(
\mathbf{B}_{0k}^{M}\right) ^{\top }-\widehat{\mathbf{A}}_{k}^{M,H}\widehat{%
\mathbf{F}}_{k,t}^{M,H}\left( \widehat{\mathbf{B}}_{k}^{M,H}\right) ^{\top
}\right\Vert _{F}\leq \tau \right)  \\
&&-P\left( p^{-1/2}\left\Vert \mathbf{E}_{k,t}\right\Vert _{F}\leq
\tau \right)  \\
&=&-c_{1}p^{-1/2}E\left\Vert \mathbf{A}_{0k}^{M}\mathbf{F}_{0k,t}^{M}\left(
\mathbf{B}_{0k}^{M}\right) ^{\top }-\widehat{\mathbf{A}}_{k}^{M,H}\widehat{%
\mathbf{F}}_{k,t}^{M,H}\left( \widehat{\mathbf{B}}_{k}^{M,H}\right) ^{\top
}\right\Vert _{F}.
\end{eqnarray*}%
Hence%
\begin{eqnarray*}
&&\left\Vert \frac{1}{T}\sum_{t=1}^{T}\mathbf{F}_{0k,t}\mathbf{F}%
_{0k,t}^{^{\top }}E\left[ I\left( \left\Vert \mathbf{E}_{k,t}-\mathbf{A}%
_{0k}^{M}\left( \widehat{\mathbf{F}}_{k,t}^{M,H}-\widehat{\mathbf{S}}%
_{k}^{M,H}\mathbf{F}_{0k,t}^{M}\widehat{\mathbf{S}}_{-k}^{M,H}\right) \left(
\mathbf{B}_{0k}^{M}\right) ^{\top }\right\Vert _{F}/\sqrt{p}\leq \tau \right)
-I\left( \left\Vert \mathbf{E}_{k,t}\right\Vert _{F}/\sqrt{p}\leq \tau \right) \right]
\right\Vert _{F} \\
&\leq &c_{0}\frac{1}{T}\sum_{t=1}^{T}\left\Vert \mathbf{F}_{0k,t}\right\Vert
_{F}^{2}p^{-1/2}E\left\Vert \mathbf{A}_{0k}^{M}\left( \widehat{\mathbf{F}}%
_{k,t}^{M,H}-\widehat{\mathbf{S}}_{k}^{M,H}\mathbf{F}_{0k,t}^{M}\widehat{%
\mathbf{S}}_{-k}^{M,H}\right) \left( \mathbf{B}_{0k}^{M}\right) ^{\top
}\right\Vert _{F} \\
&\leq &c_{0}\frac{1}{T}\sum_{t=1}^{T}\left\Vert \mathbf{F}_{0k,t}\right\Vert
_{F}^{2}E\left\Vert \widehat{\mathbf{F}}_{k,t}^{M,H}-\widehat{\mathbf{S}}%
_{k}^{M,H}\mathbf{F}_{0k,t}^{M}\widehat{\mathbf{S}}_{-k}^{M,H}\right\Vert
_{F} \\
&\leq &c_{0}\left( \frac{1}{T}\sum_{t=1}^{T}\left\Vert \mathbf{F}%
_{0k,t}\right\Vert _{F}^{4}\right) ^{1/2}\left( \frac{1}{T}%
\sum_{t=1}^{T}E\left\Vert \widehat{\mathbf{F}}_{k,t}^{M,H}-\widehat{\mathbf{S%
}}_{k}^{M,H}\mathbf{F}_{0k,t}^{M}\widehat{\mathbf{S}}_{-k}^{M,H}\right\Vert
_{F}^{2}\right) ^{1/2}=o_{P}\left( 1\right) ,
\end{eqnarray*}%
on account of the fact that Assumption \ref{as-1} implies $\left\Vert
\mathbf{F}_{0k,t}\right\Vert _{F}^{4}<\infty $, and that, by adapting the
proof of (\ref{nf2_app}), it follows that%
\begin{equation*}
\frac{1}{T}\sum_{t=1}^{T}\left\Vert \widehat{\mathbf{F}}_{k,t}^{M,H}-%
\widehat{\mathbf{S}}_{k}^{M,H}\mathbf{F}_{0k,t}^{M}\widehat{\mathbf{S}}%
_{-k}^{M,H}\right\Vert _{F}^{2}=O_{P}\left( \widetilde{L}^{-1}\right) .
\end{equation*}%
Finally, by the Law of Large Numbers, and by the fact that the $E_{k,t}$ are
\textit{i.i.d.}, it follows that%
\begin{equation*}
\frac{1}{T}\sum_{t=1}^{T}\mathbf{F}_{0k,t}\mathbf{F}_{0k,t}^{^{\top
}}I\left( \left\Vert \mathbf{E}_{k,t}\right\Vert _{F}/\sqrt{p}\leq \tau \right)
=P\left( \left\Vert \mathbf{E}_{k,0}\right\Vert _{F}/\sqrt{p}\leq \tau \right) \left(
\frac{1}{T}\sum_{t=1}^{T}\mathbf{F}_{0k,t}\mathbf{F}_{0k,t}^{^{\top
}}\right) +o_{P}\left( 1\right) ,
\end{equation*}%
and by Assumption \ref{as-1}%
\begin{equation*}
P\left( \left\Vert \mathbf{E}_{k,0}\right\Vert _{F}/\sqrt{p}\leq \tau \right) \left(
\frac{1}{T}\sum_{t=1}^{T}\mathbf{F}_{0k,t}\mathbf{F}_{0k,t}^{^{\top
}}\right) \rightarrow P\left( \left\Vert \mathbf{E}_{k,0}\right\Vert
_{F}/\sqrt{p}\leq \tau \right) \mathbf{\Sigma }_{0,k}.
\end{equation*}%
Now (\ref{merik1}) finally follows, since $\mathbf{\Sigma }_{0,k}$ has full
rank and $P\left( \left\Vert \mathbf{E}_{k,0}\right\Vert _{F}/\sqrt{p}\leq \tau
\right) >0$ by assumption. Note that, by construction, $\lambda _{j}\left(
I_{a}\right) =0$ for all $j>r_{k}$. Also note that (recall that we set $%
r_{k}=1$ for simplicity)%
\begin{eqnarray*}
\left\Vert I_{b}\right\Vert _{F} &=&\frac{1}{p}\left\Vert \frac{1}{T}%
\sum_{t=1}^{T}\widehat{w}_{k,t}^{H}\mathbf{F}_{0k,t}^{2}\mathbf{A}%
_{0k}\left( \widehat{\mathbf{B}}_{k}^{M,H}-\mathbf{B}_{0k}^{M}\widehat{%
\mathbf{S}}_{-k}^{M,H}\right) ^{^{\top }}\mathbf{B}_{0k}\mathbf{A}%
_{0k}^{\top }\right\Vert _{F} \\
&=&\left( \frac{1}{T}\sum_{t=1}^{T}\widehat{w}_{k,t}^{H}\mathbf{F}%
_{0k,t}^{2}\right) \frac{1}{p}\left\Vert \mathbf{A}_{0k}\left( \widehat{%
\mathbf{B}}_{k}^{M,H}-\mathbf{B}_{0k}^{M}\widehat{\mathbf{S}}%
_{-k}^{M,H}\right) ^{^{\top }}\mathbf{B}_{0k}\mathbf{A}_{0k}^{\top
}\right\Vert _{F} \\
&\leq &c_{0}\frac{1}{p}\left\Vert \mathbf{A}_{0k}\right\Vert
_{F}^{2}\left\Vert \mathbf{B}_{0k}\right\Vert _{F}\left\Vert \widehat{%
\mathbf{B}}_{k}^{M,H}-\mathbf{B}_{0k}^{M}\widehat{\mathbf{S}}%
_{-k}^{M,H}\right\Vert _{F}=O_{P}\left( \widetilde{L}^{-1/2}\right) ,
\end{eqnarray*}%
and the same can be shown for $\left\Vert I_{c}\right\Vert _{F}$ by
symmetry. Similarly%
\begin{eqnarray*}
\left\Vert I_{d}\right\Vert _{F} &=&\frac{1}{p}\left\Vert \frac{1}{Tp_{-k}}%
\sum_{t=1}^{T}\widehat{w}_{k,t}^{H}\mathbf{F}_{0k,t}^{2}\mathbf{A}_{0k}%
\mathbf{B}_{0k}^{\top }\left( \widehat{\mathbf{B}}_{k}^{M,H}-\mathbf{B}%
_{0k}^{M}\widehat{\mathbf{S}}_{-k}^{M,H}\right) \left( \widehat{\mathbf{B}}%
_{k}^{M,H}-\mathbf{B}_{0k}^{M}\widehat{\mathbf{S}}_{-k}^{M,H}\right)
^{^{\top }}\mathbf{B}_{0k}\mathbf{A}_{0k}^{\top }\right\Vert _{F} \\
&=&\left( \frac{1}{T}\sum_{t=1}^{T}\widehat{w}_{k,t}^{H}\mathbf{F}%
_{0k,t}^{2}\right) \frac{1}{p_{-k}p}\left\Vert \mathbf{A}_{0k}\mathbf{B}%
_{0k}^{\top }\left( \widehat{\mathbf{B}}_{k}^{M,H}-\mathbf{B}_{0k}^{M}%
\widehat{\mathbf{S}}_{-k}^{M,H}\right) \left( \widehat{\mathbf{B}}_{k}^{M,H}-%
\mathbf{B}_{0k}^{M}\widehat{\mathbf{S}}_{-k}^{M,H}\right) ^{^{\top }}\mathbf{%
B}_{0k}\mathbf{A}_{0k}^{\top }\right\Vert _{F} \\
&\leq &\left( \frac{1}{T}\sum_{t=1}^{T}\widehat{w}_{k,t}^{H}\mathbf{F}%
_{0k,t}^{2}\right) \frac{1}{p_{-k}p}\left\Vert \mathbf{A}_{0k}\right\Vert
_{F}^{2}\left\Vert \mathbf{B}_{0k}\right\Vert _{F}^{2}\left\Vert \widehat{%
\mathbf{B}}_{k}^{M,H}-\mathbf{B}_{0k}^{M}\widehat{\mathbf{S}}%
_{-k}^{M,H}\right\Vert _{F}^{2}=O_{P}\left( \widetilde{L}^{-1}\right) .
\end{eqnarray*}%
Using Weyls' inequality, this entails that%
\begin{equation}
\lambda _{j}\left( I\right) =c_{0}+O_{P}\left( \widetilde{L}^{-1/2}\right) ,
\label{i-1}
\end{equation}%
for all $j\leq r_{k}$, and recall that
\begin{equation}
\lambda _{j}\left( I\right) =0,  \label{i-2}
\end{equation}%
for $j>r_{k}$. We now turn to bounding the largest eigenvalues of $II$, $III$
and $IV$, setting $r_{k}=1$ for simplicity. Recalling that, by construction,
$\left\Vert \widehat{\mathbf{A}}_{k}^{M}\right\Vert _{F}=c_{0}p_{k}^{1/2}$,
it follows that
\begin{eqnarray*}
\left\Vert II\right\Vert _{F} &\leq &\frac{1}{p}\left\Vert \mathbf{A}%
_{0k}^{M}\left( \mathbf{B}_{0k}^{M}\right) ^{\top }\widehat{\mathbf{B}}%
_{k}^{M,H}\frac{1}{Tp_{-k}}\sum_{t=1}^{T}\widehat{w}_{k,t}^{H}\left(
\widehat{\mathbf{B}}_{k}^{M,H}\right) ^{^{\top }}\mathbf{E}_{k,t}^{^{\top }}%
\mathbf{F}_{0k,t}^{M}\right\Vert _{F} \\
&\leq &\frac{1}{p}\frac{\left\Vert \mathbf{B}_{0k}^{M}\right\Vert
_{F}\left\Vert \widehat{\mathbf{B}}_{k}^{M,H}\right\Vert _{F}}{p_{-k}}%
\left\Vert \frac{1}{T}\sum_{t=1}^{T}\widehat{w}_{k,t}^{H}\mathbf{A}%
_{0k}^{M}\left( \widehat{\mathbf{B}}_{k}^{M,H}\right) ^{^{\top }}\mathbf{E}%
_{k,t}^{^{\top }}\mathbf{F}_{0k,t}^{M}\right\Vert _{F} \\
&\leq &c_{0}\frac{1}{p}\left\Vert \frac{1}{T}\sum_{t=1}^{T}\widehat{w}%
_{k,t}^{H}\mathbf{A}_{0k}^{M}\left( \widehat{\mathbf{B}}_{k}^{M,H}\right)
^{^{\top }}\mathbf{E}_{k,t}^{^{\top }}\mathbf{F}_{0k,t}^{M}\right\Vert _{F}.
\end{eqnarray*}%
Using again (\ref{bkMH}), we may write%
\begin{eqnarray*}
&&\left\Vert \frac{1}{T}\sum_{t=1}^{T}\widehat{w}_{k,t}^{H}\mathbf{A}%
_{0k}^{M}\left( \widehat{\mathbf{B}}_{k}^{M,H}\right) ^{^{\top }}\mathbf{E}%
_{k,t}^{^{\top }}\mathbf{F}_{0k,t}^{M}\right\Vert _{F} \\
&\leq &\left\Vert \frac{1}{T}\sum_{t=1}^{T}\widehat{w}_{k,t}^{H}\mathbf{A}%
_{0k}^{M}\left( \mathbf{B}_{0k}^{M}\widehat{\mathbf{S}}_{-k}^{M,H}\right)
^{^{\top }}\mathbf{E}_{k,t}^{^{\top }}\mathbf{F}_{0k,t}^{M}\right\Vert
_{F}+\left\Vert \frac{1}{T}\sum_{t=1}^{T}\widehat{w}_{k,t}^{H}\mathbf{A}%
_{0k}^{M}\left( \widehat{\mathbf{B}}_{k}^{M,H}-\mathbf{B}_{0k}^{M}\widehat{%
\mathbf{S}}_{-k}^{M,H}\right) ^{^{\top }}\mathbf{E}_{k,t}^{^{\top }}\mathbf{F%
}_{0k,t}^{M}\right\Vert _{F}.
\end{eqnarray*}%
It holds that%
\begin{eqnarray*}
&&\left\Vert \frac{1}{T}\sum_{t=1}^{T}\widehat{w}_{k,t}^{H}\mathbf{A}%
_{0k}^{M}\left( \mathbf{B}_{0k}^{M}\widehat{\mathbf{S}}_{-k}^{M,H}\right)
^{^{\top }}\mathbf{E}_{k,t}^{^{\top }}\mathbf{F}_{0k,t}^{M}\right\Vert
_{F}^{2} \\
&=&\left\Vert \frac{1}{T}\sum_{t=1}^{T}\widehat{w}_{k,t}^{H}\mathbf{A}%
_{0k}^{M}\left( \mathbf{B}_{0k}^{M}\right) ^{^{\top }}\mathbf{E}%
_{k,t}^{^{\top }}\mathbf{F}_{0k,t}^{M}\right\Vert _{F}^{2}\leq c_{0}\frac{1}{%
T}\sum_{t=1}^{T}\left\Vert \mathbf{A}_{0k}\mathbf{B}_{0k}^{^{\top }}\mathbf{E%
}_{k,t}^{^{\top }}\mathbf{F}_{0k,t}\right\Vert _{F}^{2}.
\end{eqnarray*}%
Now it is easy to see that, by simmilar passages as in the above, $%
E\left\Vert \mathbf{A}_{0k}\mathbf{B}_{0k}^{^{\top }}\mathbf{E}%
_{k,t}^{^{\top }}\mathbf{F}_{0k,t}\right\Vert _{F}^{2}=O\left( p\right) $,
and therefore%
\begin{equation}
\frac{1}{p}\left\Vert \frac{1}{T}\sum_{t=1}^{T}\widehat{w}_{k,t}^{H}\mathbf{A%
}_{0k}^{M}\left( \mathbf{B}_{0k}^{M}\widehat{\mathbf{S}}_{-k}^{M,H}\right)
^{^{\top }}\mathbf{E}_{k,t}^{^{\top }}\mathbf{F}_{0k,t}^{M}\right\Vert
_{F}=O_{P}\left( p^{-1/2}\right) .  \label{iia}
\end{equation}%
Also, by the same token%
\begin{eqnarray*}
&&\left\Vert \frac{1}{T}\sum_{t=1}^{T}\widehat{w}_{k,t}^{H}\mathbf{A}%
_{0k}^{M}\left( \widehat{\mathbf{B}}_{k}^{M,H}-\mathbf{B}_{0k}^{M}\widehat{%
\mathbf{S}}_{-k}^{M,H}\right) ^{^{\top }}\mathbf{E}_{k,t}^{^{\top }}\mathbf{F%
}_{0k,t}^{M}\right\Vert _{F}^{2} \\
&\leq &c_{0}\frac{1}{T}\sum_{t=1}^{T}\left\Vert \mathbf{A}_{0k}^{M}\left(
\widehat{\mathbf{B}}_{k}^{M,H}-\mathbf{B}_{0k}^{M}\widehat{\mathbf{S}}%
_{-k}^{M,H}\right) ^{^{\top }}\mathbf{E}_{k,t}^{^{\top }}\mathbf{F}%
_{0k,t}^{M}\right\Vert _{F}^{2} \\
&=&c_{0}\frac{1}{T}\sum_{t=1}^{T}\mathbf{F}_{0k,t}^{2}\text{Tr}\left(
\mathbf{A}_{0k}\left( \widehat{\mathbf{B}}_{k}^{M,H}-\mathbf{B}_{0k}^{M}%
\widehat{\mathbf{S}}_{-k}^{M,H}\right) ^{^{\top }}\mathbf{E}_{k,t}^{^{\top }}%
\mathbf{E}_{k,t}\left( \widehat{\mathbf{B}}_{k}^{M,H}-\mathbf{B}_{0k}^{M}%
\widehat{\mathbf{S}}_{-k}^{M,H}\right) \mathbf{A}_{0k}^{^{\top }}\right)  \\
&\leq &c_{0}\sigma _{\max }^{2}\left( \mathbf{A}_{0k}\right) \left\Vert
\widehat{\mathbf{B}}_{k}^{M,H}-\mathbf{B}_{0k}^{M}\widehat{\mathbf{S}}%
_{-k}^{M,H}\right\Vert _{F}^{2}\frac{1}{T}\sum_{t=1}^{T}\mathbf{F}%
_{0k,t}^{2}\left\Vert \mathbf{E}_{k,t}\right\Vert _{F}^{2} \\
&=&O_{P}\left( 1\right) p^{2}\widetilde{L}^{-1};
\end{eqnarray*}%
hence it holds that%
\begin{equation}
\frac{1}{p}\left\Vert \frac{1}{T}\sum_{t=1}^{T}\widehat{w}_{k,t}^{H}\mathbf{A%
}_{0k}^{M}\left( \widehat{\mathbf{B}}_{k}^{M,H}-\mathbf{B}_{0k}^{M}\widehat{%
\mathbf{S}}_{-k}^{M,H}\right) ^{^{\top }}\mathbf{E}_{k,t}^{^{\top }}\mathbf{F%
}_{0k,t}^{M}\right\Vert _{F}=O_{P}\left( \widetilde{L}^{-1/2}\right) .
\label{iib}
\end{equation}%
Combining (\ref{iia}) and (\ref{iib}), it follows that $\left\Vert
II\right\Vert _{F}=O_{P}\left( \widetilde{L}^{-1/2}\right) $; the same can
be shown, by symmetry, for $\left\Vert III\right\Vert _{F}$. Hence it
follows that%
\begin{equation}
\left\vert \lambda _{\max }\left( II+III\right) \right\vert \leq \left\Vert
II+III\right\Vert _{F}\leq \left\Vert II\right\Vert _{F}+\left\Vert
III\right\Vert _{F}=O_{P}\left( \widetilde{L}^{-1/2}\right) .
\label{l-max-23}
\end{equation}%
Consider now%
\begin{eqnarray*}
&&\frac{1}{Tp_{-k}}\sum_{t=1}^{T}\widehat{w}_{k,t}^{H}\mathbf{E}_{k,t}%
\widehat{\mathbf{B}}_{k}^{M,H}\left( \widehat{\mathbf{B}}_{k}^{M,H}\right)
^{^{\top }}\mathbf{E}_{k,t}^{^{\top }} \\
&=&\frac{1}{Tp_{-k}}\sum_{t=1}^{T}\widehat{w}_{k,t}^{H}\mathbf{E}_{k,t}%
\mathbf{B}_{0k}^{M}\widehat{\mathbf{S}}_{-k}^{M,H}\left( \mathbf{B}_{0k}^{M}%
\widehat{\mathbf{S}}_{-k}^{M,H}\right) ^{^{\top }}\mathbf{E}_{k,t}^{^{\top }}
\\
&&+\frac{1}{Tp_{-k}}\sum_{t=1}^{T}\widehat{w}_{k,t}^{H}\mathbf{E}%
_{k,t}\left( \widehat{\mathbf{B}}_{k}^{M,H}-\mathbf{B}_{0k}^{M}\widehat{%
\mathbf{S}}_{-k}^{M,H}\right) \left( \mathbf{B}_{0k}^{M}\widehat{\mathbf{S}}%
_{-k}^{M,H}\right) ^{^{\top }}\mathbf{E}_{k,t}^{^{\top }} \\
&&+\frac{1}{Tp_{-k}}\sum_{t=1}^{T}\widehat{w}_{k,t}^{H}\mathbf{E}_{k,t}%
\mathbf{B}_{0k}^{M}\widehat{\mathbf{S}}_{-k}^{M,H}\left( \widehat{\mathbf{B}}%
_{k}^{M,H}-\mathbf{B}_{0k}^{M}\widehat{\mathbf{S}}_{-k}^{M,H}\right)
^{^{\top }}\mathbf{E}_{k,t}^{^{\top }} \\
&&+\frac{1}{Tp_{-k}}\sum_{t=1}^{T}\widehat{w}_{k,t}^{H}\mathbf{E}%
_{k,t}\left( \widehat{\mathbf{B}}_{k}^{M,H}-\mathbf{B}_{0k}^{M}\widehat{%
\mathbf{S}}_{-k}^{M,H}\right) \left( \widehat{\mathbf{B}}_{k}^{M,H}-\mathbf{B%
}_{0k}^{M}\widehat{\mathbf{S}}_{-k}^{M,H}\right) ^{^{\top }}\mathbf{E}%
_{k,t}^{^{\top }} \\
&=&IV_{a}+IV_{b}+IV_{c}+IV_{d}.
\end{eqnarray*}%
It holds that%
\begin{eqnarray*}
\lambda _{\max }\left( IV_{a}\right)  &=&\lambda _{\max }\left( \frac{1}{%
pTp_{-k}}\sum_{t=1}^{T}\widehat{w}_{k,t}^{H}\mathbf{E}_{k,t}\mathbf{B}%
_{0k}^{M}\left( \mathbf{B}_{0k}^{M}\right) ^{^{\top }}\mathbf{E}%
_{k,t}^{^{\top }}\right)  \\
&\leq &c_{0}\frac{1}{pTp_{-k}}\sum_{t=1}^{T}\lambda _{\max }\left( \mathbf{E}%
_{k,t}\mathbf{B}_{0k}\mathbf{B}_{0k}^{^{\top }}\mathbf{E}_{k,t}^{^{\top
}}\right) =c_{0}\frac{1}{pTp_{-k}}\sum_{t=1}^{T}\left( \mathbf{B}%
_{0k}^{^{\top }}\mathbf{E}_{k,t}^{^{\top }}\mathbf{E}_{k,t}\mathbf{B}%
_{0k}\right) ,
\end{eqnarray*}%
and since, by standard algebra%
\begin{equation*}
E\left( \mathbf{B}_{k}^{^{\top }}\mathbf{E}_{k,t}^{^{\top }}\mathbf{E}_{k,t}%
\mathbf{B}_{k}\right) =O\left( p\right) ,
\end{equation*}%
it follows that $\lambda _{\max }\left( IV_{a}\right) =O_{P}\left(
p_{-k}^{-1}\right) $. Similarly%
\begin{equation*}
\lambda _{\max }\left( IV_{b}\right) \leq c_{0}\frac{1}{pTp_{-k}}%
\sum_{t=1}^{T}\left( \left( \mathbf{B}_{0k}^{M}\right) ^{^{\top }}\mathbf{E}%
_{k,t}^{^{\top }}\mathbf{E}_{k,t}\left( \widehat{\mathbf{B}}_{k}^{M,H}-%
\mathbf{B}_{0k}^{M}\widehat{\mathbf{S}}_{-k}^{M,H}\right) \right) ,
\end{equation*}%
and since%
\begin{eqnarray*}
\left( \mathbf{B}_{0k}^{M}\right) ^{^{\top }}\mathbf{E}_{k,t}^{^{\top }}%
\mathbf{E}_{k,t}\left( \widehat{\mathbf{B}}_{k}^{M,H}-\mathbf{B}_{0k}^{M}%
\widehat{\mathbf{S}}_{-k}^{M,H}\right)  &=&\text{Tr}\left( \mathbf{E}%
_{k,t}^{^{\top }}\mathbf{E}_{k,t}\left( \widehat{\mathbf{B}}_{k}^{M,H}-%
\mathbf{B}_{0k}^{M}\widehat{\mathbf{S}}_{-k}^{M,H}\right) \left( \mathbf{B}%
_{0k}^{M}\right) ^{^{\top }}\right)  \\
&\leq &\left\Vert \mathbf{E}_{k,t}\right\Vert _{F}\left\Vert \widehat{%
\mathbf{B}}_{k}^{M,H}-\mathbf{B}_{0k}^{M}\widehat{\mathbf{S}}%
_{-k}^{M,H}\right\Vert _{F}\left\Vert \mathbf{B}_{0k}^{M}\right\Vert _{F},
\end{eqnarray*}%
it follows that $\lambda _{\max }\left( IV_{b}\right) =O_{P}\left(
\widetilde{L}^{-1/2}\right) $; the same applies for $\lambda _{\max }\left(
IV_{c}\right) $, and similarly it is easy to see that $\lambda _{\max
}\left( IV_{d}\right) =O_{P}\left( \widetilde{L}^{-1}\right) $. Hence%
\begin{equation}
\lambda _{\max }\left( IV\right) =O_{P}\left( p_{-k}^{-1}\right)
+O_{P}\left( \widetilde{L}^{-1/2}\right) =O_{P}\left( \widetilde{L}%
^{-1/2}\right) \text{.}  \label{iv}
\end{equation}%
Recallig (\ref{l-max-23}), it holds that%
\begin{equation}
\left\vert \lambda _{\max }\left( II+III+IV\right) \right\vert \leq
\left\vert \lambda _{\max }\left( II+III\right) \right\vert +\left\vert
\lambda _{\max }\left( IV\right) \right\vert =O_{P}\left( \widetilde{L}%
^{-1/2}\right) .  \label{remainder-h}
\end{equation}%
The desired result now follows from a routine application of Weyl's
inequality, using (\ref{i-1}), (\ref{i-2}) and (\ref{remainder-h}).
\end{proof}
\end{lemma}

\newpage

\section{Proofs\label{proofs}}

\begin{proof}[Proof of Theorem \protect\ref{ls-theorem}]
First, we divide the parameter space $\Theta $ into $S_{j}=\{\theta \in
\Theta $ : $\left. 2^{j-1}<\sqrt{L}\cdot d\left( \theta ,\theta _{0}\right)
\leq 2^{j}\right\} $, with $j\geq 1$. If $\sqrt{L}d\left( \widehat{\theta }%
,\theta _{0}\right) >2^{V}$ for a certain $V$, then $\widehat{\theta }$ is
in one of the shells $S_{j},j>V$, where the infimum of $\mathbb{M}^{\ast
}(\theta )=\mathbb{M}(\theta )-\mathbb{M}\left( \theta _{0}\right) $ is
nonpositive over this shell since $\mathbb{M}(\widehat{\theta })\leq \mathbb{%
M}\left( \theta _{0}\right) $. Therefore, for every $\eta >0$, we have%
\begin{eqnarray*}
&&\mathbb{P}\left( \sqrt{L}d\left( \widehat{\theta },\theta _{0}\right)
>2^{V}\right) \\
&=&\mathbb{P}\left( 2^{V}\leq \sqrt{L}d\left( \widehat{\theta },\theta
_{0}\right) \leq \sqrt{L}\eta \right) +\mathbb{P}\left( d\left( \widehat{%
\theta },\theta _{0}\right) >\eta \right) \\
&\leq &\sum_{j>V,2^{j-1}\leq \sqrt{L}\eta }\mathbb{P}\left( \inf_{\theta \in
S_{j}}\mathbb{M}^{\ast }\left( \theta \right) \leq 0\right) +\mathbb{P}%
\left( d\left( \widehat{\theta },\theta _{0}\right) >\eta \right) .
\end{eqnarray*}%
According to Lemma \ref{semi-1}, for all $\eta >0$, $\mathbb{P}\left[
d\left( \widehat{\theta },\theta _{0}\right) >\eta \right] $ converges to 0
as $\min \left\{ T,p_{1},\cdots ,p_{K}\right\} \rightarrow \infty $.
Further, for each $\theta $ in $S_{j}$ it holds that
\begin{equation*}
-\overline{\mathbb{M}}^{\ast }(\theta )=-d^{2}\left( \theta ,\theta
_{0}\right) \leq -\frac{2^{2j-2}}{L}.
\end{equation*}%
Then from $\inf_{\theta \in S_{j}}\mathbb{M}^{\ast }(\theta )\leq 0$, we
have
\begin{equation*}
\inf_{\theta \in S_{j}}\mathbb{W}(\theta )=\inf_{\theta \in S_{j}}\left(
\mathbb{M}^{\ast }\left( \theta \right) -\overline{\mathbb{M}}^{\ast }\left(
\theta \right) \right) \leq -\frac{2^{2j-2}}{L}
\end{equation*}%
So
\begin{eqnarray*}
&&\sum_{j>V,2^{j-1}\leq \eta \sqrt{L}}\mathbb{P}\left( \inf_{\theta \in
S_{j}}\mathbb{M}^{\ast }(\theta )\leq 0\right) \\
&\leq &\sum_{j>V,2^{j-1}\leq \eta \sqrt{L}}\mathbb{P}\left( \inf_{\theta \in
S_{j}}\mathbb{W}(\theta )\leq \frac{2^{2j-2}}{L}\right) \\
&\leq &\sum_{j>V,2^{j-1}\leq \eta \sqrt{L}}\mathbb{P}\left( \sup_{\theta \in
S_{j}}|\mathbb{W}(\theta )|\geq \frac{2^{2j-2}}{L}\right) .
\end{eqnarray*}%
Using Lemma \ref{exp-sup} and Markov's inequality, it follows that
\begin{equation*}
\mathbb{P}\left( \sup_{\theta \in S_{j}}|\mathbb{W}(\theta )|\geq \frac{%
2^{2j-2}}{L}\right) \leq c_{0}\frac{L}{2^{2j}}\cdot \mathbb{E}\left(
\sup_{\theta \in S_{j}}|\mathbb{W}(\theta )|\right) \leq c_{1}\frac{L}{2^{2j}%
}\cdot \frac{2^{j}}{L}=2^{-j}.
\end{equation*}%
Hence
\begin{equation*}
\sum_{j>V,2^{j-1}\leq \eta \sqrt{L}}\mathbb{P}\left( \inf_{\theta \in S_{j}}%
\mathbb{M}^{\ast }(\theta )\leq 0\right) \leq c_{0}\sum_{j>V}2^{-j}\leq
c_{1}2^{-V}.
\end{equation*}%
This entails that there exists a $c_{2}<\infty $ such that%
\begin{equation*}
\mathbb{P}\left( \sqrt{L}d\left( \widehat{\theta },\theta _{0}\right)
>2^{V}\right) \leq c_{2}2^{-V},
\end{equation*}%
which eadily entails that $\sqrt{L}\cdot d\left( \widehat{\theta },\theta
_{0}\right) =O_{p}(1)$. The desired result now follows from Lemma \ref%
{lemma2}.
\end{proof}

\begin{proof}[Proof of Theorem \protect\ref{ls-cc}]
Note that $$
\begin{aligned}
\widehat\cS_t-\cS_{0t}=&\widehat\cF_t\times_{k=1}^K\widehat\Ab_k-\cF_{0t}\times_{k=1}^K\Ab_{0k}\\
=&\widehat\cF_t\times_{k=1}^K\widehat\Ab_k-\cF_{0t}\times_{k=1}^K\widehat\Sbb_k\times_{k=1}^K\Ab_{0k}\widehat\Sbb_k\\
=&\sum_{k=1}^K\widehat\cF_t\times_{j=1}^{k-1}\Ab_{0j}\widehat\Sbb_j\times_k\l(\widehat\Ab_k-\Ab_{0k}\widehat\Sbb_k\r)\times_{j=k+1}^K\widehat\Ab_j+\l(\widehat\cF_t-\cF_{0t}\times_{k=1}^K\widehat\Sbb_k\r)\times_{k=1}^K\Ab_{0k}\widehat\Sbb_k.
\end{aligned}$$
Then we have that
$$
\begin{aligned}
	\frac{1}{Tp}\sum_{t=1}^T\lVert\widehat\cS_t-\cS_{0t}\rVert_F^2=&\frac{1}{T}\sum_{t=1}^T\lVert\widehat\cF_t\times_{k=1}^K\widehat\Ab_k-\cF_{0t}\times_{k=1}^K\Ab_{0k}\rVert_F^2\\
	\leq&\frac{1}{Tp}\sum_{t=1}^T\l(\sum_{k=1}^K\lVert\widehat\cF_t\rVert_F^2\prod_{j=1}^{k-1}\lVert\Ab_{0j}\rVert_F^2\lVert\widehat\Ab_k-\Ab_{0k}\widehat\Sbb_k\rVert_F^2+\lVert\widehat\cF_t-\cF_{0t}\times_{k=1}^K\Sbb_k\rVert_F^2\prod_{k=1}^K\lVert\Ab_{0k}\widehat\Sbb_k\rVert_F^2 \r)\\
	\lesssim&\frac{1}{Tp}\sum_{t=1}^T\l(\sum_{k=1}^Kp_{-k}\lVert\widehat\Ab_k-\Ab_{0k}\widehat\Sbb_k\rVert_F^2+p\lVert\widehat\cF_t-\cF_{0t}\times_{k=1}^K\widehat\Sbb_k\rVert_F^2\r)\\
	=&O_p(1/L),
\end{aligned}$$
by Theorem \protect\ref{ls-theorem}, with $L=\min\{p,Tp_{-1},\cdots,Tp_{-K}\}$.

\end{proof}

\begin{proof}[Proof of Theorem \protect\ref{huber-theorem-1}]
Replacing $\mathbb{W}(\theta )$ by $I(\theta )$, $L$ by $L^{\ast }$ and $%
S_{j}$ by $S_{j}^{\ast }=\left\{ \theta \in \Theta :2^{j-1}<\sqrt{L^{\ast }}%
\cdot d\left( \theta ,\theta _{0}\right) \leq 2^{j}\right\} $ in proof of
Theorem \ref{ls-theorem}, we now have%
\begin{eqnarray*}
&&\mathbb{P}\left( \sqrt{L^{\ast }}d\left( \widehat{\theta }^{H},\theta
_{0}\right) >2^{V}\right) \\
&=&\mathbb{P}\left( 2^{V}<\sqrt{L^{\ast }}d\left( \widehat{\theta }%
^{H},\theta _{0}\right) \leq \sqrt{L^{\ast }}\eta \right) +\mathbb{P}\left(
d\left( \widehat{\theta }^{H},\theta _{0}\right) >\eta \right) \\
&\leq &\sum_{j>V,2^{j-1}\leq \sqrt{L^{\ast }}\eta }\mathbb{P}\left(
\inf_{\theta \in S_{j}}L_{3}\left( \theta \right) \leq 0\right) +\mathbb{P}%
\left( d\left( \widehat{\theta }^{H},\theta _{0}\right) >\eta \right) \\
&\leq &\sum_{j>V,2^{j-1}\leq \sqrt{L^{\ast }}\eta }\mathbb{P}\left(
\sup_{\theta \in S_{j}}\left\vert I\left( \theta \right) \right\vert \geq
\frac{2^{2j-2}}{L^{\ast }}\right) +\mathbb{P}\left( d\left( \widehat{\theta }%
^{H},\theta _{0}\right) >\eta \right) \\
&\leq &\sum_{j>V,2^{j-1}\leq \sqrt{L^{\ast }}\eta }\left( \mathbb{P}\left(
\sup_{\theta \in S_{j}}\left\vert I_{1}\left( \theta \right) \right\vert
\geq \frac{2^{2j-2}}{3L^{\ast }}\right) +\mathbb{P}\left( \sup_{\theta \in
S_{j}}\sqrt{\left\vert I_{2,1}\left( \theta \right) \right\vert }\geq \sqrt{%
\frac{2^{2j-2}}{3L^{\ast }}}\right) +\mathbb{P}\left( \sup_{\theta \in
S_{j}}\left\vert I_{2,2}\left( \theta \right) \right\vert \geq \frac{2^{2j-2}%
}{3L^{\ast }}\right) \right) \\
&&+\mathbb{P}\left( d\left( \widehat{\theta }^{H},\theta _{0}\right) >\eta
\right) \\
&\leq &\sum_{j>V,2^{j-1}\leq \sqrt{L^{\ast }}\eta }\frac{3L^{\ast }}{2^{2j-2}%
}\mathbb{E}\sup_{\theta \in S_{j}}\left\vert I_{1}\left( \theta \right)
\right\vert +\sum_{j>V,2^{j-1}\leq \sqrt{L^{\ast }}\eta }\sqrt{\frac{%
3L^{\ast }}{2^{2j-2}}}\mathbb{E}\sup_{\theta \in S_{j}}\sqrt{\left\vert
I_{2,1}\left( \theta \right) \right\vert } \\
&&+\sum_{j>V,2^{j-1}\leq \sqrt{L^{\ast }}\eta }\frac{3L^{\ast }}{2^{2j-2}}%
\mathbb{E}\sup_{\theta \in S_{j}}\left\vert I_{2,2}\left( \theta \right)
\right\vert +\mathbb{P}\left( d\left( \widehat{\theta }^{H},\theta
_{0}\right) >\eta \right) \\
&\leq &c_{0}\left( \sum_{j>V,2^{j-1}\leq \sqrt{L^{\ast }}\eta
}2^{-j/2}\right) +\mathbb{P}\left( d\left( \widehat{\theta }^{H},\theta
_{0}\right) >\eta \right) .
\end{eqnarray*}%
We know by Lemma \ref{distance-huber} that, for arbitrarily small $\eta >0$,
$\mathbb{P}[d\left( \widehat{\theta }^{H},\theta _{0}\right) >\eta ]$
converges to $0$ as $\min \left\{ T,p_{1},...,p_{K}\right\} \rightarrow
\infty $. Also, as $V\rightarrow \infty $, it is easy to show that $%
\sum_{j>V,2^{j-1}\leq \eta \sqrt{L^{\ast }}}2^{-j/2}$ converges to $0$,
which implies $d\left( \widehat{\theta }^{H},\theta _{0}\right) =O_{p}(1/%
\sqrt{L^{\ast }})$. Hence, by making appeal to Lemma \ref{exp-sup}, Theorem %
\ref{huber-theorem-1} follows.
\end{proof}

\begin{proof}[Proof of Theorem \protect\ref{huber-cc-1}]
 The proof is similar to the proof of Theorem \ref{ls-cc} and it is therefore omitted.
\end{proof}

\begin{proof}[Proof of Theorem \protect\ref{huber-theorem-2}]
Consider (\ref{i-dec}); we can write%
\begin{eqnarray*}
I_{1}\left( \theta \right) &=&\frac{1}{Tp}\sum_{t=1}^{T}%
\sum_{i_{1}=1}^{p_{1}}\cdots \sum_{i_{K}=1}^{p_{K}}\zeta _{t}^{0}\left(
\mathcal{F}_{0t}\times _{k=1}^{K}\mathbf{a}_{0k,i_{k}}^{\top }-\mathcal{F}%
_{t}\times _{k=1}^{K}\mathbf{a}_{k,i_{k}}^{\top }\right) e_{t,i_{1},\cdots
,i_{K}} \\
&=&I^{\ast }\left( \theta \right) \sqrt{\frac{1}{Tp}\sum_{t=1}^{T}%
\sum_{i_{1}=1}^{p_{1}}\cdots \sum_{i_{K}=1}^{p_{K}}\left( \zeta
_{t}^{0}\right) ^{2}\left( \mathcal{F}_{0t}\times _{k=1}^{K}\mathbf{a}%
_{0k,i_{k}}^{\top }-\mathcal{F}_{t}\times _{k=1}^{K}\mathbf{a}%
_{k,i_{k}}^{\top }\right) ^{2}e_{t,i_{1},\cdots ,i_{K}}^{2}} \\
&\leq &c_{0}I_{1}^{\ast }\left( \theta \right) d\left( \theta ,\theta
_{0}\right) ,
\end{eqnarray*}%
where we have exploited the fact that $\left\vert \zeta _{t}^{0}\right\vert
\leq c_{0}$\ a.s. and define%
\begin{equation}
I_{1}^{\ast }\left( \theta \right) =\frac{\frac{1}{Tp}\sum_{t=1}^{T}%
\sum_{i_{1}=1}^{p_{1}}\cdots \sum_{i_{K}=1}^{p_{K}}\zeta _{t}^{0}\left(
\mathcal{F}_{0t}\times _{k=1}^{K}\mathbf{a}_{0k,i_{k}}^{\top }-\mathcal{F}%
_{t}\times _{k=1}^{K}\mathbf{a}_{k,i_{k}}^{\top }\right) e_{t,i_{1},\cdots
,i_{K}}}{\sqrt{\frac{1}{Tp}\sum_{t=1}^{T}\sum_{i_{1}=1}^{p_{1}}\cdots
\sum_{i_{K}=1}^{p_{K}}\left( \zeta _{t}^{0}\right) ^{2}\left( \mathcal{F}%
_{0t}\times _{k=1}^{K}\mathbf{a}_{0k,i_{k}}^{\top }-\mathcal{F}_{t}\times
_{k=1}^{K}\mathbf{a}_{k,i_{k}}^{\top }\right) ^{2}e_{t,i_{1},\cdots
,i_{K}}^{2}}}.  \label{i-star}
\end{equation}%
We note that $I_{1}^{\ast }\left( \theta \right) $ is a self-normalised sum.
In order to study it, similarly to \citet{shao1997self}, we will repeatedly
use the following elementary fact%
\begin{equation}
zy=\inf_{b>0}\left( \frac{z^{2}}{b}+y^{2}b\right) ,  \label{elementary}
\end{equation}%
and define the short-hand notation%
\begin{equation*}
\xi _{t,i_{1},\cdots ,i_{K}}=\left( \mathcal{F}_{0t}\times _{k=1}^{K}\mathbf{%
b}_{0k,i_{k}}^{\top }-\mathcal{F}_{t}\times _{k=1}^{K}\mathbf{b}%
_{k,i_{k}}^{\top }\right) e_{t,i_{1},\cdots ,i_{K}}.
\end{equation*}%
Define now an $A>2$; similarly to equation (2.6) in \citet{shao1997self}, it
holds that%
\begin{eqnarray}
\mathbb{P}\left( I\left( \theta ^{\ast }\right) >x\right)
&=&\mathbb{P}%
\left( \sum_{t=1}^{T}\sum_{i_{1}=1}^{p_{1}}\cdots \sum_{i_{K}=1}^{p_{K}}\xi
_{t,i_{1},\cdots ,i_{K}}>x\sqrt{Tp}\sqrt{\sum_{t=1}^{T}%
\sum_{i_{1}=1}^{p_{1}}\cdots \sum_{i_{K}=1}^{p_{K}}\xi _{t,i_{1},\cdots
,i_{K}}^{2}}\right)  \label{shao} \\
&=&\mathbb{P}\left( \sum_{t=1}^{T}\sum_{i_{1}=1}^{p_{1}}\cdots
\sum_{i_{K}=1}^{p_{K}}\xi _{t,i_{1},\cdots ,i_{K}}>x\inf_{b>0}\frac{1}{2b}%
\left( \sum_{t=1}^{T}\sum_{i_{1}=1}^{p_{1}}\cdots \sum_{i_{K}=1}^{p_{K}}\xi
_{t,i_{1},\cdots ,i_{K}}^{2}+b^{2}Tp\right) \right)  \notag \\
&=&\mathbb{P}\left( \sup_{b\geq 0}\sum_{t=1}^{T}\sum_{i_{1}=1}^{p_{1}}\cdots
\sum_{i_{K}=1}^{p_{K}}\left( b\xi _{t,i_{1},\cdots ,i_{K}}-x\left( \xi
_{t,i_{1},\cdots ,i_{K}}^{2}+b^{2}\right) /2\right) >0\right)  \notag \\
&=&\mathbb{P}\left( \sup_{b>4A}\sum_{t=1}^{T}\sum_{i_{1}=1}^{p_{1}}\cdots
\sum_{i_{K}=1}^{p_{K}}\left( b\xi _{t,i_{1},\cdots ,i_{K}}-x\left( \xi
_{t,i_{1},\cdots ,i_{K}}^{2}+b^{2}\right) /2\right) >0\right)  \notag \\
&&+\mathbb{P}\left( \sup_{0\leq b\leq
4A}\sum_{t=1}^{T}\sum_{i_{1}=1}^{p_{1}}\cdots \sum_{i_{K}=1}^{p_{K}}\left(
b\xi _{t,i_{1},\cdots ,i_{K}}-x\left( \xi _{t,i_{1},\cdots
,i_{K}}^{2}+b^{2}\right) /2\right) >0\right)  \notag \\
&=&P_{1}+P_{2}.  \notag
\end{eqnarray}%
We begin by studying $P_{1}$, using a similar approach as in equation (2.7)
in \citet{shao1997self}. It holds that%
\begin{eqnarray*}
P_{1} &=&\mathbb{P}\left(
\sup_{b>4A}\sum_{t=1}^{T}\sum_{i_{1}=1}^{p_{1}}\cdots
\sum_{i_{K}=1}^{p_{K}}\left( b\xi _{t,i_{1},\cdots ,i_{K}}-x\left( \xi
_{t,i_{1},\cdots ,i_{K}}^{2}+b^{2}\right) /2\right) >0\right) \\
&=&\mathbb{P}\left( \sup_{b>4A}\sum_{t=1}^{T}\sum_{i_{1}=1}^{p_{1}}\cdots
\sum_{i_{K}=1}^{p_{K}}\left( b\xi _{t,i_{1},\cdots ,i_{K}}\left( I\left(
\left\vert \xi _{t,i_{1},\cdots ,i_{K}}\right\vert \leq xA\right) +I\left(
\left\vert \xi _{t,i_{1},\cdots ,i_{K}}\right\vert >xA\right) \right) -x%
\frac{\xi _{t,i_{1},\cdots ,i_{K}}^{2}+b^{2}}{2}\right) >0\right) \\
&\leq &\mathbb{P}\left(
\sup_{b>4A}\sum_{t=1}^{T}\sum_{i_{1}=1}^{p_{1}}\cdots
\sum_{i_{K}=1}^{p_{K}}\left( bxA+b\xi _{t,i_{1},\cdots ,i_{K}}I\left(
\left\vert \xi _{t,i_{1},\cdots ,i_{K}}\right\vert >xA\right) -x\frac{\xi
_{t,i_{1},\cdots ,i_{K}}^{2}+b^{2}}{2}\right) >0\right) \\
&\leq &\mathbb{P}\left(
\sup_{b>4A}\sum_{t=1}^{T}\sum_{i_{1}=1}^{p_{1}}\cdots
\sum_{i_{K}=1}^{p_{K}}\left( b\xi _{t,i_{1},\cdots ,i_{K}}I\left( \left\vert
\xi _{t,i_{1},\cdots ,i_{K}}\right\vert >xA\right) -x\frac{\xi
_{t,i_{1},\cdots ,i_{K}}^{2}+b^{2}/2}{2}\right) >0\right) \\
&=&\mathbb{P}\left( \sum_{t=1}^{T}\sum_{i_{1}=1}^{p_{1}}\cdots
\sum_{i_{K}=1}^{p_{K}}\xi _{t,i_{1},\cdots ,i_{K}}I\left( \left\vert \xi
_{t,i_{1},\cdots ,i_{K}}\right\vert >xA\right) \geq \frac{x}{2}%
\inf_{b>4A}\left( \frac{\sum_{t=1}^{T}\sum_{i_{1}=1}^{p_{1}}\cdots
\sum_{i_{K}=1}^{p_{K}}\xi _{t,i_{1},\cdots ,i_{K}}^{2}}{b}+\frac{bTp}{2}%
\right) \right) \\
&\leq &\mathbb{P}\left( \sum_{t=1}^{T}\sum_{i_{1}=1}^{p_{1}}\cdots
\sum_{i_{K}=1}^{p_{K}}\xi _{t,i_{1},\cdots ,i_{K}}I\left( \left\vert \xi
_{t,i_{1},\cdots ,i_{K}}\right\vert >xA\right) \geq \frac{x}{\sqrt{2}}\left(
Tp\sum_{t=1}^{T}\sum_{i_{1}=1}^{p_{1}}\cdots \sum_{i_{K}=1}^{p_{K}}\xi
_{t,i_{1},\cdots ,i_{K}}^{2}\right) ^{1/2}\right) \\
&\leq &\mathbb{P}\left[ \left\vert
\sum_{t=1}^{T}\sum_{i_{1}=1}^{p_{1}}\cdots \sum_{i_{K}=1}^{p_{K}}\xi
_{t,i_{1},\cdots ,i_{K}}^{2}\right\vert ^{1/2}\left\vert
\sum_{t=1}^{T}\sum_{i_{1}=1}^{p_{1}}\cdots \sum_{i_{K}=1}^{p_{K}}I\left(
\left\vert \xi _{t,i_{1},\cdots ,i_{K}}\right\vert >xA\right) \right\vert
^{1/2}\right. \\
&&\left. \geq \frac{x}{\sqrt{2}}\left(
Tp\sum_{t=1}^{T}\sum_{i_{1}=1}^{p_{1}}\cdots \sum_{i_{K}=1}^{p_{K}}\xi
_{t,i_{1},\cdots ,i_{K}}^{2}\right) ^{1/2}\right] \\
&=&\mathbb{P}\left( \sum_{t=1}^{T}\sum_{i_{1}=1}^{p_{1}}\cdots
\sum_{i_{K}=1}^{p_{K}}I\left( \left\vert \xi _{t,i_{1},\cdots
,i_{K}}\right\vert >xA\right) \geq \frac{x^{2}Tp}{2}\right) ,
\end{eqnarray*}%
where we have used the Cauchy-Schwartz inequality in the seventh line. Then,
using the Markov's inequality and recalling that $\exp \left( x\right)
=1+\sum_{j=1}^{\infty }\frac{x^{j}}{j!}$, it follows that%
\begin{eqnarray*}
P_{1} &\leq &\exp \left( -\frac{x^{2}Tp}{2}\right) \mathbb{E}\exp \left(
\sum_{t=1}^{T}\sum_{i_{1}=1}^{p_{1}}\cdots \sum_{i_{K}=1}^{p_{K}}I\left(
\left\vert \xi _{t,i_{1},\cdots ,i_{K}}\right\vert >xA\right) \right) \\
&\leq &\exp \left( -\frac{x^{2}Tp}{2}\right)
\prod\limits_{t=1}^{T}\prod\limits_{i_{1}=1}^{p_{1}}...\prod%
\limits_{i_{K}=1}^{p_{K}}\mathbb{E}\exp \left( I\left( \left\vert \xi
_{t,i_{1},\cdots ,i_{K}}\right\vert >xA\right) \right) \\
&=&\exp \left( -\frac{x^{2}Tp}{2}\right)
\prod\limits_{t=1}^{T}\prod\limits_{i_{1}=1}^{p_{1}}...\prod%
\limits_{i_{K}=1}^{p_{K}}\left( 1+\left( e-1\right) \delta _{t,i_{1},\cdots
,i_{K}}\right) ,
\end{eqnarray*}%
where
\begin{equation*}
\delta _{t,i_{1},\cdots ,i_{K}}=\mathbb{P}\left\{ \left\vert \xi
_{t,i_{1},\cdots ,i_{K}}\right\vert >xA\right\} \leq \delta ,
\end{equation*}%
for sufficiently large $A$ due to the compactness of the parameter space.
Therefore
\begin{equation*}
P_{1}\leq \exp \left( -\frac{x^{2}Tp}{2}\right) \exp \left( Tp\log
(1+(e-1)\delta )\right) \leq \exp \left( -\frac{\left( x^{2}-4\delta \right)
Tp}{2}\right) ,
\end{equation*}%
for small enough $0<\delta <x^{2}/4$. Considering now $P_{2}$ in (\ref{shao}%
) similar to the proof in \cite{shao1997self}, let $0<\alpha <1$, and $%
\Delta =\alpha ^{2}/\left[ (10+60x)A^{4}\right] $, and let $Y$ be a standard
normal random variable independent of $\left( \mathcal{E}_{t},\mathcal{F}%
_{t}\right) $. Then we have%
\begin{eqnarray*}
P_{2} &\leq &\mathbb{P}\left( \max_{1\leq j\leq 1+4A/\Delta }\sup_{\left(
j-1\right) \Delta \leq b\leq j\Delta
}\sum_{t=1}^{T}\sum_{i_{1}=1}^{p_{1}}\cdots \sum_{i_{K}=1}^{p_{K}}\left(
b\xi _{t,i_{1},\cdots ,i_{K}}-x\left( \xi _{t,i_{1},\cdots
,i_{K}}^{2}+b^{2}\right) /2\right) >0\right) \\
&\leq &\mathbb{P}\left( \max_{1\leq j\leq 1+4A/\Delta
}\sum_{t=1}^{T}\sum_{i_{1}=1}^{p_{1}}\cdots \sum_{i_{K}=1}^{p_{K}}\left(
j\Delta \xi _{t,i_{1},\cdots ,i_{K}}-x\left( \xi _{t,i_{1},\cdots
,i_{K}}^{2}+\left( \left( j-1\right) \Delta \right) ^{2}\right) /2\right)
>0\right) \\
&\leq &\sum_{j=1}^{1+4A/\Delta }\mathbb{P}\left(
\sum_{t=1}^{T}\sum_{i_{1}=1}^{p_{1}}\cdots \sum_{i_{K}=1}^{p_{K}}\left(
j\Delta \xi _{t,i_{1},\cdots ,i_{K}}-x\left( \xi _{t,i_{1},\cdots
,i_{K}}^{2}+\left( \left( j-1\right) \Delta \right) ^{2}\right) /2\right)
>0\right) .
\end{eqnarray*}%
By Markov's inequality and (\ref{elementary}), it holds that%
\begin{eqnarray*}
P_{2} &\leq &\sum_{j=1}^{1+4A/\Delta
}\prod\limits_{t=1}^{T}\prod\limits_{i_{1}=1}^{p_{1}}...\prod%
\limits_{i_{K}=1}^{p_{K}}\inf_{\beta \geq 0}\mathbb{E}\exp \left( \beta
\left( j\Delta \xi _{t,i_{1},\cdots ,i_{K}}-x\left( \xi _{t,i_{1},\cdots
,i_{K}}^{2}+\left( \left( j-1\right) \Delta \right) ^{2}\right) /2\right)
\right) \\
&\leq &\sum_{j=1}^{1+4A/\Delta
}\prod\limits_{t=1}^{T}\prod\limits_{i_{1}=1}^{p_{1}}...\prod%
\limits_{i_{K}=1}^{p_{K}}\inf_{\beta \geq 0}\exp \left( \frac{\alpha
^{2}\beta ^{2}}{2}\right) \mathbb{E}\exp \left( \beta \left( j\Delta \xi
_{t,i_{1},\cdots ,i_{K}}-x\left( \xi _{t,i_{1},\cdots ,i_{K}}^{2}+\left(
\left( j-1\right) \Delta \right) ^{2}\right) /2\right) \right) \\
&\leq &\sum_{j=1}^{1+4A/\Delta
}\prod\limits_{t=1}^{T}\prod\limits_{i_{1}=1}^{p_{1}}...\prod%
\limits_{i_{K}=1}^{p_{K}}\inf_{\beta \geq 0}\mathbb{E}\exp \left( \beta
\left( j\Delta \xi _{t,i_{1},\cdots ,i_{K}}+\alpha Y-x\left( \xi
_{t,i_{1},\cdots ,i_{K}}^{2}+\left( \left( j-1\right) \Delta \right)
^{2}\right) /2\right) \right) .
\end{eqnarray*}%
Let $\zeta _{t,i_{1},\cdots ,i_{K},j}=j\Delta \xi _{t,i_{1},\cdots
,i_{K}}+\alpha Y-x\left( \xi _{t,i_{1},\cdots ,i_{K}}^{2}+(j\Delta
)^{2}\right) /2$. By (2.12) in \cite{shao1997self},
\begin{equation*}
\mathbb{E}\exp \left( \beta _{t,i_{1},\cdots ,i_{K},j}\zeta _{t,i_{1},\cdots
,i_{K},j}\right) =\inf_{\beta \geq 0}\mathbb{E}\exp \left( \beta \zeta
_{t,i_{1},\cdots ,i_{K},j}\right) ,
\end{equation*}%
for some bounded $\beta _{t,i_{1},\cdots ,i_{K},j}$ satisfying $0<\beta
_{t,i_{1},\cdots ,i_{K},j}\leq (10+60x)A^{2}/\alpha ^{2}$. Therefore%
\begin{eqnarray*}
P_{2} &\leq &\sum_{j=1}^{1+4A/\Delta
}\prod\limits_{t=1}^{T}\prod\limits_{i_{1}=1}^{p_{1}}...\prod%
\limits_{i_{K}=1}^{p_{K}}\mathbb{E}\exp \left( \beta _{t,i_{1},\cdots
,i_{K},j}\left( j\Delta \xi _{t,i_{1},\cdots ,i_{K}}+\alpha Y-x\left( \xi
_{t,i_{1},\cdots ,i_{K}}^{2}+\left( \left( j-1\right) \Delta \right)
^{2}\right) /2\right) \right) \\
&=&\sum_{j=1}^{1+4A/\Delta
}\prod\limits_{t=1}^{T}\prod\limits_{i_{1}=1}^{p_{1}}...\prod%
\limits_{i_{K}=1}^{p_{K}}\exp \left( \beta _{t,i_{1},\cdots ,i_{K},j}x\left(
j^{2}-\left( j-1\right) ^{2}\right) \Delta ^{2}/2\right) \mathbb{E}\exp
\left( \beta _{t,i_{1},\cdots ,i_{K},j}\zeta _{t,i_{1},\cdots
,i_{K},j}\right) \\
&\leq &\sum_{j=1}^{1+4A/\Delta
}\prod\limits_{t=1}^{T}\prod\limits_{i_{1}=1}^{p_{1}}...\prod%
\limits_{i_{K}=1}^{p_{K}}\exp \left( \beta _{t,i_{1},\cdots ,i_{K},j}j\Delta
^{2}\right) \mathbb{E}\exp \left( \beta _{t,i_{1},\cdots ,i_{K},j}\zeta
_{t,i_{1},\cdots ,i_{K},j}\right) \\
&\leq &\left( 1+4A/\Delta \right) \exp \left( Tp\frac{\Delta \left(
1+4A\right) (10+60x)A^{2}}{\alpha ^{2}}\right)
\prod\limits_{t=1}^{T}\prod\limits_{i_{1}=1}^{p_{1}}...\prod%
\limits_{i_{K}=1}^{p_{K}}\inf_{\beta \geq 0}\mathbb{E}\exp \left( \beta
\zeta _{t,i_{1},\cdots ,i_{K},j}\right) \\
&\leq &\left( 1+4A/\Delta \right) \exp \left( \frac{5}{A}Tp\right)
\prod\limits_{t=1}^{T}\prod\limits_{i_{1}=1}^{p_{1}}...\prod%
\limits_{i_{K}=1}^{p_{K}}\inf_{\beta \geq 0}\mathbb{E}\exp \left( \beta
\zeta _{t,i_{1},\cdots ,i_{K},j}\right) \\
&\leq &\left( 1+4A/\Delta \right) \exp \left( \frac{5}{A}Tp\right)
\prod\limits_{t=1}^{T}\prod\limits_{i_{1}=1}^{p_{1}}...\prod%
\limits_{i_{K}=1}^{p_{K}}\sup_{b\geq 0}\inf_{\beta \geq 0}\mathbb{E}\exp
\left( \beta \left( b\xi _{t,i_{1},\cdots ,i_{K}}+\alpha Y-x\left( \xi
_{t,i_{1},\cdots ,i_{K}}^{2}+b^{2}\right) /2\right) \right) \\
&\leq &\left( 1+4A/\Delta \right) \exp \left( \frac{5}{A}Tp\right)
\prod\limits_{t=1}^{T}\prod\limits_{i_{1}=1}^{p_{1}}...\prod%
\limits_{i_{K}=1}^{p_{K}}\sup_{b\geq 0}\inf_{\beta \geq 0}\exp \left( \frac{%
\alpha ^{2}\beta ^{2}}{2}\right) \mathbb{E}\exp \left( \beta \left( b\xi
_{t,i_{1},\cdots ,i_{K}}-x\left( \xi _{t,i_{1},\cdots
,i_{K}}^{2}+b^{2}\right) /2\right) \right)
\end{eqnarray*}%
As $\alpha \rightarrow 0$, it follows from Lemma $2.1$ in \cite{shao1997self}
that
\begin{equation*}
P_{2t}\leq (1+4A/\Delta )e^{5p/A}\prod_{i_{1}=1}^{p_{1}}\cdots
\prod_{i_{K}=1}^{p_{K}}\sup_{b\geq 0}\inf_{\beta \geq 0}\mathbb{E}\exp
\left\{ \beta \left( b\xi _{t,i_{1},\cdots ,i_{K}}-x\left( \xi
_{t,i_{1},\cdots ,i_{K}}^{2}+b^{2}\right) /2\right) \right\} .
\end{equation*}%
Assumption \ref{as-4} and Proposition $2.1$ in \cite{jing2008towards} show
that, there exist $\beta _{t,i_{1},\cdots ,i_{K}}\left( x^{2}\right) >$ 0
and $b_{t,i_{1},\cdots ,i_{K}}\left( x^{2}\right) >0$ satisfying $\beta
_{t,i_{1},\cdots ,i_{K}}\left( x^{2}\right) \rightarrow \beta _{0}>0$ and $%
b_{t,i_{1},\cdots ,i_{K}}\left( x^{2}\right) \rightarrow 0$ for $%
x\rightarrow 0$,
\begin{equation*}
\sup_{b}\inf_{\beta }\mathbb{E}\exp \left\{ \beta \left( b\xi
_{t,i_{1},\cdots ,i_{K}}-x\left( \xi _{t,i_{1},\cdots
,i_{K}}^{2}+b^{2}\right) /2\right) \right\} =\exp \left\{ -g\left( \beta
_{t,i_{1},\cdots ,i_{K}}\left( x^{2}\right) ,b_{t,i_{1},\cdots ,i_{K}}\left(
x^{2}\right) ;x^{2}\right) \right\} ,
\end{equation*}%
where%
\begin{eqnarray*}
&&g\left( \beta _{t,i_{1},\cdots ,i_{K}}\left( x^{2}\right)
,b_{t,i_{1},\cdots ,i_{K}}\left( x^{2}\right) ;x^{2}\right) \\
&=&\beta _{t,i_{1},\cdots ,i_{K}}\left( x^{2}\right) x^{2}-\log \left(
\mathbb{E}\exp \left[ \beta _{t,i_{1},\cdots ,i_{K}}(x^{2})\left(
2b_{t,i_{1},\cdots ,i_{K}}\left( x^{2}\right) \xi _{t,i_{1},\cdots
,i_{K}}-b_{t,i_{1},\cdots ,i_{K}}^{2}\left( x^{2}\right) \xi
_{t,i_{1},\cdots ,i_{K}}^{2}\right) \right] \right) .
\end{eqnarray*}%
The variable $\xi _{t,i_{1},\cdots ,i_{K}}$ has zero mean, it is
(conditionally) independent across $t$ and $i_{1},\cdots ,i_{K}$,\ and
finally it admits moments up to order $2+\epsilon $; hence, it is easy to
verify that it belongs in the centered Feller class with limit $0$ - see
(1.5) in \cite{jing2008towards}. Hence, by Assumption \ref{as-4}\textit{(i)}%
, we can use Lemma 3.1 in \cite{jing2008towards}, whence
\begin{equation*}
\mathbb{E}\exp \left\{ \beta _{t,i_{1},\cdots ,i_{K}}(x^{2})\left(
2b_{t,i_{1},\cdots ,i_{K}}(x^{2})\xi _{t,i_{1},\cdots
,i_{K}}-b_{t,i_{1},\cdots ,i_{K}}^{2}(x^{2})\xi _{t,i_{1},\cdots
,i_{K}}^{2}\right) \right\} =1+o(x)
\end{equation*}%
as $x\rightarrow 0^{+}$, whence
\begin{eqnarray}
&&\exp \left\{ -g\left( \beta _{t,i_{1},\cdots ,i_{K}}\left( x^{2}\right)
,b_{t,i_{1},\cdots ,i_{K}}\left( x^{2}\right) ;x^{2}\right) \right\}
\label{jing} \\
&=&\exp \left\{ -\beta _{t,i_{1},\cdots ,i_{K}}\left( x^{2}\right)
x^{2}+o\left( x^{2}\right) \right\} =\exp \left\{ -\beta _{0}x^{2}+o\left(
x^{2}\right) \right\} .  \notag
\end{eqnarray}%
Therefore,
\begin{equation*}
P_{2}\leq (1+4A/\Delta )\exp \left\{ -\left( \beta _{0}x^{2}-5/A\right)
Tp+o\left( x^{2}\right) Tp\right\} .
\end{equation*}%
Let $x$ be small enough, and $A$ large enough. Then, as $Tp\rightarrow $ $%
\infty $, using the same arguments as in the proof of Lemma \ref{semi-1} we
have
\begin{equation*}
I_{1}^{\ast }(\theta )\leq \exp \{-C_{13}Tp\}.
\end{equation*}%
Hence, similarly arguments as in the proof of Lemma \ref{semi-1} yield
\begin{equation*}
\Vert \sqrt{Tp}|I_{1}(\theta ^{\ast })|\Vert _{\psi _{2}}\lesssim d(\theta
^{\ast },\theta _{0}),
\end{equation*}%
and then%
\begin{eqnarray}
\mathbb{E}\sup_{\theta \in \Theta }\left\vert I_{1}(\theta ^{\ast
})\right\vert &=&O\left( L^{-1/2}\right) ,  \label{equ:2I1} \\
\mathbb{E}\sup_{\theta \in \Theta \left( \delta \right) }\left\vert
I_{1}(\theta ^{\ast })\right\vert &=&O\left( \delta L^{-1/2}\right) .
\label{equ:2I1b}
\end{eqnarray}%
Considering $I_{2,1}(\theta )$ in (\ref{i-dec}), it is easy to see that this
can be studied in the same way as under Assumption \ref{as-3}. Similarly, as
far as $I_{2,2}(\theta )$ is concerned, assuming $\mathbb{E}%
e_{t,i_{1},...,i_{K}}^{4}<\infty $, we still have%
\begin{eqnarray*}
\mathbb{E}\sup_{\theta \in \Theta }\left\vert I_{2,2}(\theta ^{\ast
})\right\vert &=&O\left( p^{-1/2}\right) , \\
\mathbb{E}\sup_{\theta \in \Theta \left( \delta \right) }\left\vert
I_{2,2}(\theta )\right\vert &\leq &\mathbb{E}\sup_{\theta \in \Theta \left(
\delta \right) }\left\vert D_{t,2}\right\vert =O\left( \delta
p^{-1/2}\right) ;
\end{eqnarray*}%
hence, under $\mathbb{E}e_{t,i_{1},...,i_{K}}^{4}<\infty $, similar
arguments as in the proof under Assumptions \ref{as-1}, \ref{as-2}, and \ref%
{as-3} entail
\begin{equation*}
d(\widehat{\theta }^{H},\theta _{0})=O_{p}(1/\sqrt{L^{\ast }}).
\end{equation*}%
Finally, consider the case $\mathbb{E}\left\vert
e_{t,i_{1},...,i_{K}}\right\vert ^{2+\epsilon }<\infty $. In this case, we
note that $|D_{t,2}|$ is bounded and it can be shown, similarly to the
above, that \newline
\begin{equation*}
\mathbb{E}\sup_{\theta \in \Theta (\delta )}\left\vert \frac{1}{p}%
\sum_{i_{1}=1}^{p_{1}}\cdots \sum_{i_{K}=1}^{p_{K}}\left( \mathcal{F}%
_{0t}\times _{k=1}^{K}\mathbf{a}_{0k,i_{k}}^{\top }-\mathcal{F}_{t}\times
_{k=1}^{K}\mathbf{a}_{k,i_{k}}^{\top }\right) e_{t,i_{1},\cdots
,i_{K}}\right\vert =O\left( \delta /\sqrt{\min \{p_{-1},p_{-2},\cdots
,p_{-K}\}}\right) .
\end{equation*}%
Hence we have%
\begin{eqnarray*}
&&\mathbb{E}\sup_{\theta \in \Theta }\left\vert I_{2,2}(\theta )\right\vert
\\
&\leq &\mathbb{E}\frac{1}{T}\sum_{t=1}^{T}\sup_{\theta \in \Theta
}\left\vert \frac{1}{p}\sum_{i_{1}=1}^{p_{1}}\cdots
\sum_{i_{K}=1}^{p_{K}}\left( \mathcal{F}_{0t}\times _{k=1}^{K}\mathbf{a}%
_{0k,i_{k}}^{\top }-\mathcal{F}_{t}\times _{k=1}^{K}\mathbf{a}%
_{k,i_{k}}^{\top }\right) e_{t,i_{1},\cdots ,i_{K}}\right\vert \\
&=&O\left( 1/\sqrt{\min \{p_{-1},p_{-2},\cdots ,p_{-K}\}}\right) ,
\end{eqnarray*}%
and
\begin{equation*}
\mathbb{E}\sup_{\theta \in \Theta (\delta )}\left\vert I_{2,2}(\theta
)\right\vert =O\left( \delta /\sqrt{\min \{p_{-1},p_{-2},\cdots ,p_{-K}\}}%
\right) .
\end{equation*}%
Following the same logic as in the above, and replacing $L^{\ast }$ as $%
L^{\ast \ast }$, the desired result finally obtains.
\end{proof}

\begin{proof}[Proof of Theorem \protect\ref{huber-cc-2}]
The proof is similar to the proof of Theorem \ref{ls-cc} and it is therefore
omitted.
\end{proof}

\begin{proof}[Proof of Theorem \protect\ref{ls-fn}]
Consider first the case $j<r_{k}$; in this case%
\begin{equation*}
\frac{\lambda _{j}\left( \widehat{\mathbf{M}}_{k}\right) }{\lambda
_{j+1}\left( \widehat{\mathbf{M}}_{k}\right) +c\omega _{k}^{-1/2}}\leq \frac{%
\lambda _{j}\left( \widehat{\mathbf{M}}_{k}\right) }{\lambda _{j+1}\left(
\widehat{\mathbf{M}}_{k}\right) }=O_{P}\left( 1\right) ,
\end{equation*}%
as a consequence of equation (\ref{er-ls-2}) in Lemma \ref{er-ls}.
Similarly, whenever $j>r_{k}$%
\begin{equation*}
\frac{\lambda _{j}\left( \widehat{\mathbf{M}}_{k}\right) }{\lambda
_{j+1}\left( \widehat{\mathbf{M}}_{k}\right) +c\omega _{k}^{-1/2}}\leq
\left( c\omega _{k}^{-1/2}\right) ^{-1}\lambda _{j}\left( \widehat{\mathbf{M}%
}_{k}\right) =O_{P}\left( 1\right) ,
\end{equation*}%
again by equation (\ref{er-ls-2}) in Lemma \ref{er-ls}. Finally, when $%
j=r_{k}$, it holds that there exists a $0<c_{0}<\infty $ such that%
\begin{equation*}
\frac{\lambda _{j}\left( \widehat{\mathbf{M}}_{k}\right) }{\lambda
_{j+1}\left( \widehat{\mathbf{M}}_{k}\right) +c\omega _{k}^{-1/2}}%
=c_{0}\omega _{k}^{1/2}\lambda _{j}\left( \widehat{\mathbf{M}}_{k}\right)
+o_{P}\left( 1\right) \overset{P}{\rightarrow }\infty ,
\end{equation*}%
on account of equation (\ref{er-ls-1}). This completes the proof.
\end{proof}

\begin{proof}[Proof of Theorem \protect\ref{huber-fn}]
The proof is similar to the proof of Theorem \ref{ls-fn}, using Lemma \ref%
{er-h} and it is therefore omitted.
\end{proof}

\end{document}